%% file: vcohdiff2.tex
\documentclass{lmcs}
\pdfoutput=1
\usepackage[utf8]{inputenc}

\usepackage{lastpage}
\lmcsdoi{19}{4}{7}
\lmcsheading{}{\pageref{LastPage}}{}{}%
{Aug.~25,~2022}{Oct.~26,~2023}{}


\usepackage{hyperref}
\usepackage{amsmath}
\usepackage{amsthm}
\usepackage{amsfonts}
\usepackage{amssymb}
\usepackage{cmll}

\usepackage{tikz}
\usetikzlibrary{matrix}
\usetikzlibrary{cd}
\usepackage{pgfplots}
\usepackage{stmaryrd}
\usepackage{ebproof}
\usepackage{url}

\input{mainnot}

\title{A coherent differential PCF}
\thanks{This work was partly supported by the ANR project %
\emph{Probabilistic Programming Semantics (PPS)} ANR-19-CE48-0014.}
\author[T.~Ehrhard]{Thomas Ehrhard\lmcsorcid{0000-0001-5231-5504}}
\address{Université Paris Cité, CNRS, Inria, IRIF, F-75013, Paris, France}
\email{ehrhard@irif.fr}

\begin{document}

\begin{abstract}
  \noindent The categorical models of the differential lambda-calculus
  are additive categories because of the Leibniz rule which requires
  the summation of two expressions and therefore the differential
  lambda-calculus features unrestricted sums of terms. A natural and
  simple operational interpretation of such sums of terms is by finite
  non-determinism although other interpretations, based for instance
  on quantum superposition, might also be possible. In a previous
  paper we introduced a categorical framework for differentiation
  which features a partial form of additivity and is compatible with
  deterministic denotational models such as coherence spaces and
  probabilistic models such as probabilistic coherence spaces. Based
  on this categorical theory we develop a syntax of a deterministic
  version of the differential lambda-calculus. One nice feature of
  this new approach to differentiation is that it is compatible with
  general fixpoints of terms, so our language is actually a
  differential extension of PCF for which we provide a fully
  deterministic operational semantics by means of an adapted version
  of the Krivine Machine.
\end{abstract}

\maketitle

\section{Introduction}
The differential lambda-calculus~\cite{EhrhardRegnier02} extends the
(typed) lambda-calculus with a syntactic notion of differentiation
which is compatible with the basic intuition of differentiation in
Calculus: given \(f:E\to F\) a sufficiently regular function between
two \(\Real\)-vector spaces (finite dimensional, or Banach\dots),
\(f':E\to(\Limpl EF)\) where \(\Limpl EF\) is the space of linear (and
continuous if we are in infinite dimension) maps \(E\to F\) such
that %
\(f(x+u)=f(x)+f'(x)\cdot u+o(\Norm u{})\).
More generally \(f'(x)\) is such that
\(y\mapsto f(x)+f'(x)\cdot(y-x)\) is the best affine approximation of
\(f\) which coincides with \(f\) at \(x\).

Syntactically this means that, given a term \(M\) such that %
\(\Tseq\Gamma M{\Timpl AB}\) and a term \(N\) such that
\(\Tseq\Gamma NA\) we introduce a term \(\Diffsymb M\cdot N\) such
that %
\(\Tseq\Gamma{\Diffsymb M\cdot N}{\Timpl AB}\). Intuitively \(M\)
represents a function \(f:A\to B\), \(N\) an element \(u\) of \(A\)
and \(\Diffsymb M\cdot N\) represents the function \(g:A\to B\) such
that \(g(x)=f'(x)\cdot u\).
This syntactic presentation is a convenient way to express the fact
that the derivative has the same regularity as the differentiated
function (intuitively terms represent ``smooth'' maps, so their
derivatives are themselves smooth) and allows easy iteration of
differentiation (\(n\)-th derivatives).

Differentiation is inherently related to the algebraic operation of
\emph{addition} and the associated operation of subtraction, this is
obvious in the definition of derivative we have all been taught at
school:
\(f'(x)=\lim_{\epsilon\to 0}\frac{f(x+\epsilon)-f(x)}{\epsilon}\). %
More algebraically this connection manifests itself for instance by
the \emph{Leibniz Rule} \((fg)'=f'g+fg'\).
In the differential lambda-calculus, beyond the ordinary
\(\beta\)-reduction, there is a \emph{differential
  \(\beta\)-reduction}: %
\(\Diffsymb(\Abst xAM)\cdot N\Rel\Red\frac{\partial M}{\partial
  x}\cdot N\) %
which uses a \emph{linear substitution} %
\(\frac{\partial M}{\partial x}\cdot N\) defined by induction on
\(M\). %
The definition of this operation involves the Leibniz Rule%
\footnote{Actually the Leibniz Rule is not really related to
  multiplication but more fundamentally to the fact that the parameter
  of a function can be used more than once that is, to the logical
  rule of contraction.} %
in the syntactic constructs where the variable \(x\) can be used at
various places, the most typical example being:
\begin{align}
  \label{eq:app-Leib-diff-lambda}
  \frac{\partial\App PQ}{\partial x}\cdot N
  =\App{\frac{\partial P}{\partial x}\cdot N}{Q}
  +\App{\Diffsymb P\cdot\left(\frac{\partial Q}{\partial x}\cdot N\right)}Q
\end{align}
which is a combination of the Leibniz Rule and of the Chain
Rule.
Because of this feature of the definition of
\(\frac{\partial M}{\partial x}\cdot N\) we had to extend the syntax
of the lambda-calculus with an addition operation typed as follows
\begin{equation}
  \label{eq:term-superposition}
  \begin{prooftree}
    \hypo{\Tseq\Gamma{M_0}A}
    \hypo{\Tseq\Gamma{M_1}A}
    \infer2{\Tseq\Gamma{M_0+M_1}A}
  \end{prooftree}
\end{equation}
This rule is most natural if we have in mind the usual mathematical
intuitions about differentiation.
However the lambda-calculus is not only a nice and convenient syntax
for denoting mathematical functions, it is also an expressive
programming language featuring crucial properties of determinism.
But the most natural operational interpretation of this \(+\)
operation is a kind of non-deterministic superposition%
\footnote{Other interpretations might be possible as well, such as,
  perhaps, quantum states superposition, but we did not explore this
  option yet. It seems clear that ``probabilistic superposition'' is
  not a possible option since it requires non-negative coefficients
  whose sum is \(\leq 1\): indeed \(\PCOH\) is not a model of
  differential \(\LL\).}.
For instance if our calculus has a type \(\Tnat\) of integers with
\(\Tseq\Gamma{\Num n}\Tnat\) for all \(n\in\Nat\) (numerals) then we
can write terms like \(\Num{42}+\Num{57}\) which is a superposition of
two values (nothing to do with \(\Num{99}\) of course!).

In the setting of \emph{differentiable programming} where programs can
be formally differentiated (see for instance~\cite{MazzaPagani21} for
a functional presentation), an addition on terms also plays a crucial
role for the same reason (Leibniz Rule).
This is made possible by the presence of a ground type \(\rho\) of
real numbers (which can be equipped with the usual addition of
numbers) and by the fact that any type is of shape
\(\Impl{A_1}{(\cdots(\Impl{A_n}{\rho})\cdots)}\) and therefore
inherits from \(\rho\) a canonical operation of addition defined
pointwise.
For instance, if the language has also a type \(\Tnat\) of integers,
it is not clear at all how the differential of a term of type
\(\Timpl\rho\Tnat\) or of type \(\Timpl\Tnat\Tnat\) could be defined.
In differentiable programming, differentiation is fundamentally
defined for terms of type
\(\Timpl\rho{(\cdots(\Timpl{\rho}\rho)\cdots)}\), and extended to
higher types by a method close to logical relations (by induction
on types), in sharp contrast with coherent differentiation where
derivatives can be taken with respect to parameters \emph{of any
  type}, and does not rely on any specific numerical datatype.

A natural question is then whether differentiation requires such a
general addition operation on terms and is therefore incompatible with
determinism.
In~\cite{Ehrhard23a} we have provided a semantic negative answer to
this question, based on a new categorical setting that we call
\emph{coherent differentiation}.
The basic idea is to replace \emph{additive categories}%
\footnote{Categories enriched over commutative monoids.}%
\footnote{Or \emph{left-additive categories} which are used for
  axiomatizing differentiation in a general cartesian category \(\cC\)
  which is not necessarily the Kleisli category of a model of \(\LL\),
  see~\cite{BluteCockettSeely09,CockettCruttwell14}.
  Observe that, even in this apparently weaker setting,
  left-additivity also entails that any two morphisms
  \(f_1,f_2\in\cC(X,Y)\) can be added, whatever be the objects \(X,Y\)
  of \(\cC\) so that these categories also feature the kind of
  ``non-determinism'' we are trying to eliminate.
  Axiomatizing coherent differentiation in general cartesian
  categories requires a more general kind of summability structure
  which is presented in another article~\cite{EhrhardWalch23}.} %
with categories equipped with a weaker structure that we call a
\emph{summability structure}: such a category \(\cL\) is endowed with
an endofunctor \(\Sfun\) which intuitively maps any object \(X\) to
the object \(\Sfun X\) of pairs \((x_0,x_1)\) such that \(x_0+x_1\) is
well defined. This functor comes with natural transformations
\(\Sproj 0,\Sproj 1,\Ssum\in\cL(\Sfun X,X)\) satisfying suitable
axioms (intuitively they map \((x_0,x_1)\) to \(x_0\), \(x_1\) and
\(x_0+x_1\) respectively).
Thanks to these axioms it is possible to equip \(\Sfun\) with a monad
structure.

Then assuming that \(\cL\) is a resource category (that is, a
cartesian symmetric monoidal category equipped with a resource
modality comonad \(\Excl\_\)), the differential structure is
axiomatized as a natural morphism
\(\Sdiff_X\in\cL(\Excl{\Sfun X},\Sfun{\Excl X})\) which is a
distributive law between the functor \(\Sfun\) and the comonad
\(\Excl\_\), and also between the monad \(\Sfun\) and the functor
\(\Excl\_\).
This allows one to extend the monad \(\Sfun\) to the Kleisli category
\(\Kl\cL\) into a monad that we denote as
\((\Sdfun,\Sdfunit,\Sdfmult)\).
The category \(\Kl\cL\) has the same objects as \(\cL\) but an element
of \(\Kl\cL(X,Y)\) should not be considered as a linear morphism
\(X\to Y\) as in \(\cL(X,Y)\), but as a morphism which is only
``smooth'' (and actually analytic in several concrete models).
The functor \(\Sdfun\) acts exactly as \(\Sfun\) on objects but its
action on morphisms implements differentiation: given
\(f\in\Kl\cL(X,Y)\), considered as a smooth map \(X\to Y\), the
smooth map
\(\Sdfun f=(\Sfun f)\Compl\Sdiff_X\in\Kl\cL(\Sdfun X\,\Sdfun Y)\) %
intuitively maps \((x_0,x_1)\) to \((f(x_0),f'(x_0)\cdot x_1)\).
The basic observation at the core of this work is indeed that, if
\(x_0+x_1\) is defined, then so must be \(f(x_0)+f'(x_0)\cdot x_1\),
as the beginning of the Taylor expansion of \(f(x_0+x_1)\) at \(x_0\).
The monad structure of \(\Sdfun\) accounts for addition: %
intuitively the natural transformation %
\( \Sdfmult_X\in\Kl\cL(\Sdfun^2X,\Sdfun X) \), which is linear%
\footnote{In the sense that it is obtained from a morphism of
  \(\cL(\Sfun^2X,\Sfun X)\) by composition with the counit
  \(\Der{\Sfun^2X}\in\cL(\Excl{\Sfun^2X},\Sfun^2X)\) of the comonad
  \(\Excl\_\).}, maps %
\(((x_{00},x_{01}),(x_{10},x_{11}))\) to %
\((x_{00},x_{01}+x_{10})\) and \(\Sdfunit_X\in\Kl\cL(X,\Sdfun X)\)
maps %
\(x\) to \((x,0)\).
The naturality of these two transformations in \(\Kl\cL\) expresses
exactly that the differential is linear with respect to this \(0\) and
to this addition.

The major benefit of using the operation \(\Sdfun\) for presenting the
derivative of morphisms of \(\Kl\cL\) is that doing so we preserve the
information of summability: we know that the two components of
\(\Sdfun f\) can be added without requiring that \emph{all} pairs of
morphisms can be added.
The categorical axiomatization described in~\cite{Ehrhard23a} involves
other natural linear transformations:
\(\Sproj i\in\Kl\cL(\Sdfun X,X)\) for \(i=0,1\) (the two obvious
projections), %
\(\Sin i\in\Kl\cL(X,\Sdfun X)\) for \(i=0,1\) (the two injections,
\(\Sin0=\Sdfunit\) and \(\Sin1\) maps \(x\) to \((0,x)\)) and the %
``canonical flip'' \(\Sflip\in\Kl\cL(\Sdfun^2X,\Sdfun^2X)\) which
maps %
\(((x_{00},x_{01}),(x_{10},x_{11}))\) to %
\(((x_{00},x_{10}),(x_{01},x_{11}))\).

As explained in the introduction of~\cite{Ehrhard23a} this categorical
axiomatization has strong formal similarities with the \emph{tangent
  categories} axiomatization of the differential calculus on smooth
manifolds~\cite{Rosicky84} ---~of which differential categories are a
special case as shown in~\cite{CockettCruttwell14}~---.
One crucial difference is that, though tangent categories are not
additive in general, they feature an unrestricted addition operation
available in ``tangent spaces'' (axiomatized by means of a pullback)
which is not available in the coherent differential setting.

\subsection{Contents}
In the present paper we propose a differential lambda-calculus which
uses the idea of coherent differentiation, but now at the syntactical
level.

\subsubsection{Semantic motivations}
Since this work arose from semantic considerations, we first provide
in Section~\ref{sec:sem-motivations} a summary of the categorical
framework for coherent differentiation introduced
in~\cite{Ehrhard23a}.
Our goal is also to make the paper as self-contained as possible.
In this section, we present moreover two concrete instances of this
general notion: the relational model and the probabilistic coherence
space model.
We describe in these well-known models of \LL{} various categorical
constructions of coherent differentiation.
These concrete models are also used technically in the paper, the
first for proving that our reduction rules are ``complete'' for
reducing the terms in Section~\ref{sec:completeness} (see
Section~\ref{sec:intro-intersection} the explanation of this word in
this context)
and the second ---~probabilistic coherence spaces
(PCS~\cite{DanosEhrhard08})~--- for proving denotationally that this
machine is deterministic in Section~\ref{sec:determinism}.

This latter model has been a major incentive for introducing coherent
differentiation.
It is a model of \LL{}, and the morphism of the Kleisli category of
its resource comonad \(\Excl\_\) can be seen as analytic functions.
For instance, in PCS, the type of booleans is interpreted as
the set \(B\) of all sub-probability distributions
\(x=(x_\True,x_\False)\) on the booleans (meaning that
\(x_\True,x_\False\in\Realp\) with \(x_\True+x_\False\leq 1\)) and a
morphism from \(B\) to \(B\) is a function \(f:B\to B\) with
\(f(x)=(f_\True(x),f_\False(x))\) such that
\begin{align*}
  f_\True(x)=\sum_{n,p\in\Nat}a_{n,p}x_\True^nx_\False^p
  \text{\quad and\quad}
  f_\False(x)=\sum_{n,p\in\Nat}b_{n,p}x_\True^nx_\False^p
\end{align*}
for uniquely determined%
\footnote{The coefficients can be retrieved from \(f\) by taking
  partial derivatives.} %
\emph{non-negative} coefficients
\((a_{n,p},b_{n,p}\in\Realp)_{n,p\in\Nat}\).
Of course \(B\) has a natural addition operation, but this addition is
only partially defined because the sum of two sub-probability
distributions is not always a sub-probability distribution.
Nevertheless the morphisms of the Kleisli category being analytic in a
rather standard mathematical sense, it is possible to compute their
derivative, and the purpose of coherent differentiation is precisely
to give a meaning to such derivatives in a partially additive setting.

\subsubsection{The choice of \PCF{}}
The categorical framework introduced in~\cite{Ehrhard23a} and
summarized in Section~\ref{sec:sem-motivations} relies on the
semantics of \LL{}, and this is deeply motivated by the fact
that, so far, all concrete known models of coherent differentiation
are models of \LL{}.
As explained above, in such a model \(\cL\) of \LL{},
differentiation is a distributive law which allows to extend \(\Sfun\)
to a differentiation functor on the cartesian closed Kleisli category
\(\Kl\cL\).
Moreover it is shown in~\cite{Ehrhard23a} that, when \(\cL\) is a
\emph{Lafont category} (that is, has a cofree modality comonad), a
very simple condition suffices to guarantee that it is a model of
coherent differentiation.
In particular the two models which will play a technical role in
the present paper ---~the relational model \(\REL\) and the probabilistic
coherence space model \(\PCOH\)~--- are both Lafont categories.

This fact has guided us in the choice of the \(\lambda\)-calculus to
be extended with coherent differential constructs: it seemed natural to
take a language whose translation in such models is as simple as
possible.
From this point of view \PCF{}~\cite{Plotkin77} is a reasonable option.
It is a simply typed call-by-name calculus and thus can be interpreted
by means of the standard and simple Girard translation of the
\(\lambda\)-calculus in \LL{} whose target is the Kleisli
category \(\Kl\cL\).
This choice allows us also to stress a major feature of this new
presentation of differentiation, which is the fact that it is fully
compatible with \emph{general recursive definitions} whose
formalization in \PCF{} is particularly simple and natural.
This is deeply related to the limitation we put on the \(+\)
operation: contrarily to the ordinary differential lambda-calculus,
our typing rules do not allow one to write a term such that
\(\Abst x\Tnat{(x+\Num{42})}\) which obviously has no fixpoint
(\(\Tnat\) is the type of integers).

Other choices of syntax such as the pure \(\lambda\)-calculus or
call-by-push-value~\cite{LevyP06} would have been possible but would
have required additional categorical ingredients, such as the use of a
reflexive object or of the Eilenberg-Moore
category as in~\cite{EhrhardTasson19}.
Of course we are very interested in studying coherent differentiation
in such languages; this will be the object of further work.

\subsubsection{Syntax}
So our calculus is a differential
extension of \PCF{}%
\footnote{More precisely, of a version of \PCF{} extended with a
  \(\mathsf{let}\) operation restricted to the unique ground type
  \(\Tnat\) of integers; this is particularly relevant for
  probabilistic extensions of this calculus in the spirit
  of~\cite{EhrhardPaganiTasson18b}, which is perfectly possible since
  it admits \(\PCOH\) as a natural denotational model.}  where the
main novelty is a differentiation operation on terms: if %
\(\Tseq\Gamma M{\Timpl AB}\) then %
\(\Tseq\Gamma{\Ldiff M}{\Timpl{\Tdiff A}{\Tdiff B}}\) and this
requires a new type construct \(\Tdiff A\).
The semantics tells us that we should have %
\(\Tdiff{\Timplp AB}=\Timplp{A}{\Tdiff B}\) and we can consider this
equation as a \emph{definition} of %
\(\Tdiff{\Timplp AB}\).
So the only basic differential construct on types that we need is
\(\Tdiffm d\Tnat\) for all \(d\in\Nat\).
This differential term construction induces a new redex, namely the
term %
\(\Ldiff{(\Abst xAM)}\) (for a term \(M\) such that
\(\Tseq{\Gamma,x:A}MB\)) and we stipulate that it reduces to %
\(\Abst x{\Tdiff A}{\Ldletv xM}\) where the term \(\Ldletv xM\) is
defined by induction on \(M\) and satisfies
\(\Tseq{\Gamma,x:\Tdiff A}{\Ldletv xM}{\Tdiff B}\).

This inductive definition of \(\Ldletv xM\) involves the use of
constructs \(\Lsumd dN\), \(\Linjd 0dN\) \Etc~which syntactically
account for the natural transformations alluded to above.
The additional superscript \(d\) is required to express the fact that
the corresponding operations are applied under \(d\) applications of
the \(\Sdfun{}\) functor.
For instance in the rule
\(\Ldletv x{\App NP}=\App{\Lsumd 0{\Ldiff{\Ldletv xN}}}{\Ldletv xP}\),
the \(\Lsumd0{\_}\) construction implements the addition involved in
the Leibniz Rule required by the fact that the variable \(x\) may
occur in both \(N\) and \(P\), as in the
Equation~\Eqref{eq:app-Leib-diff-lambda} of the differential
lambda-calculus.
We also stipulate that
\(\Ldletv x{\Lsumd dN}=\Lsumd{d+1}{\Ldletv xN}\) which shows why the
\(d\) superscripts are required.
The same superscripts appear in the basic ``arithmetic'' constructs
\(\Lsucc dN\), \(\Lpred dN\), \(\Lif dN{P_0}{P_1}\) and also in the
aforementioned \(\mathsf{let}\) construct \(\Llet dyNM\).
When applying the \(\Ldletv x\_\) to these constructs, the \(d\)
superscript is similarly incremented, for instance
\(\Ldletv x{\Lpred dN}=\Lpred{d+1}{\Ldletv xN}\).
The \(\Lflipd d\_\) construct is required because, guided by the
semantics, we set
\(\Ldletv x{\Ldiff N}=\Lflipd 0{\Ldiff{\Ldletv xN}}\).
Notice that the simple typing rules are not sufficient to ``guess''
this rule since in this situation we have %
\(\Tseq{\Gamma,x:A}N{\Timpl BC}\) and hence %
\( \Tseq{\Gamma,x:\Tdiff A}{\Ldiff{\Ldletv xN}} {\Timpl{\Tdiff
    A}{\Tdiffm 2B}} \) %
and also
\( \Tseq{\Gamma,x:\Tdiff A}{\Lflipd 0{\Ldiff{\Ldletv xN}}}
{\Timpl{\Tdiff A}{\Tdiffm 2B}} \), the semantic analysis that we
develop in Section~\ref{sec:semantics} is really mandatory.
The syntax also contains a \emph{projection} construct %
\(\Lprojd rdM\) where \(r\in\Eset{0,1}\) and \(d\in\Nat\) is the depth
we are now acquainted with. This construct allows one to ``extract''
the first (when \(r=0\)) or second (when \(r=1\)) component of a term
of type \(\Tdiff A\) which represents a kind of summable pair%
\footnote{From our \(\LL\) point of view it is important to keep in mind
  that it is not a \emph{multiplicative pair} in which both components
  are actually available, but an \emph{additive pair} which offers two
  possible options among which one must be chosen: this is precisely
  the purpose of our projection construct.}.
The language also contains a \(+\) operation%
\footnote{Of course there is also its associated ``neutral'' constant
  \(0\).}  %
on terms which is used in a single reduction rule, namely %
\( \Lprojd 1d{\Lsumd dM} \Rel{\Rsred\Lang}
\Lprojd0d{\Lprojd1dM}+\Lprojd1d{\Lprojd0dM} \).
So we can understand the construction \(\Lsumd d\_\) as a tag
identifying a place where a sum will have to be introduced when we
will need to extract a component from a ``summable \(4\)-tuple'' of
type \(\Tdiffm2 A\) for some type \(A\) and the reduction rules show
how these tags are produced and modified during computations.
The \(+\) term construct is still necessary, just as in the original
differential lambda-calculus of~\cite{EhrhardRegnier02}, but thanks to
the \(\Lsumd d\_\) construct we can prove subject reduction
\emph{without} the very strong typing~\Eqref{eq:term-superposition}
which was the source of the non-determinism of the differential
lambda-calculus%
\footnote{To be more precise, as it becomes apparent in
  Section~\ref{sec:diff-machine} where a specific reduction strategy
  is presented by means of a Krivine machine, we manage to reduce the
  use of the \(+\) to the ground type \(\Tnat\) of integers.
  This is similar to what happens in differentiable
  programming~\cite{MazzaPagani21}, the main differences being that
  our addition is not at all the arithmetic addition of the ground
  type and that we allow differentiation wrt.~parameters of all types,
  not only ground types.}.

\subsubsection{Categorical semantics and soundness}
\label{sec:intro-soundness}
After having presented the syntax of our calculus \(\Lang\) in
Section~\ref{sec:syntax}, we provide and study in
Section~\ref{sec:semantics} additional categorical constructs which
are definable in the general framework of
Section~\ref{sec:sem-motivations}.
Based on these constructs, we present a general categorical semantics
of our calculus \(\Lang\) in the cartesian closed category associated
with such a model by the familiar Kleisli construction associated with
the \(\Excl\_\) functor and to prove that this interpretation is
invariant under the \(\Rsred\Lang\) reduction relation (soundness).
This requires to prove a number of categorical equations involving in
particular \emph{additive strengths}%
\footnote{The adjective ``additive'' refers to the fact that this
  strength deals with the additive operation of cartesian product and
  not to the multiplicative operation of tensor product.
  It is shown in~\cite{Ehrhard23a} that the monad \(\Sfun\) has
  multiplicative strengths as well which are deeply related to these
  additive strengths.
  Notice that both strong monads are commutative.} %
$\Sdfstr_{X_0,X_1}^0\in\Kl\cL(\Sdfun X_0\IWith
X_1,\Sdfun\Withp{X_0}{X_1})$ and
$\Sdfstr_{X_0,X_1}^1\in\Kl\cL(X_0\IWith\Sdfun
X_1,\Sdfun\Withp{X_0}{X_1})$ %
of the monad $\Sdfun$ on \(\Kl\cL\).
These strengths are linear and extremely easy to define because
\(\Sdfun\) commutes with \(\IWith\), but their properties are not so
straightforward due to the definition of the functor \(\Sdfun\) which
involves the distributive law \(\Sdiff\).
Their main purpose is to define \emph{partial derivatives}: given
\(f\in\Kl\cL(\With{X_0}{X_1},Y)\) the first partial derivative of
\(f\) is %
\((\Sdfun f)\Comp\Sdfstr_{X_0,X_1}^0\in\Kl\cL(\Sdfun X_0\IWith
X_1,Y)\) and the second partial derivative is %
\((\Sdfun f)\Comp\Sdfstr_{X_0,X_1}^1\in\Kl\cL(X_0\IWith \Sdfun
X_1,Y)\).
These constructions are of course crucial in the interpretation of
\(\Lang\) since, for instance, \(\Ldletv xM\) must be interpreted as a
derivative with respect to the variable \(x\) but not to the other
variables occurring in \(M\).
Notice that we use \(\Comp\) for the composition operation in
\(\Kl\cL\) whereas composition in \(\cL\) is denoted as simple
juxtaposition.

Thanks to this series of categorical lemmas we can prove soundness for
the \(\Rsred\Lang\) rewriting, meaning that if
\(M\Rel{\Rsred\Lang}M'\) then \(M\) and \(M'\) have the same
interpretation.
The proof is lengthy due mainly to the number of rewriting rules.

\subsubsection{Intersection types, Krivine machine and completeness of
  the reduction}
\label{sec:intro-intersection}
We address in Section~\ref{sec:completeness} the major issue of
\emph{completeness}:
is our rewriting system \(\Rsred\Lang\) sufficiently rich for
performing all ``required'' reductions?
To give a precise meaning to this question, we need an external
reference, independent from our rewriting system.
Denotational models provide us precisely with such a reference, so we
choose the most basic denotational model of \(\Lang\) which is
\(\REL\) (see Section~\ref{sec:sem-motivations} complemented
by~\ref{sec:complements-REL}) and we present the associated
interpretation of terms by means of a \emph{non-idempotent
  intersection typing system} for \(\Lang\).
Our completeness expresses that if the interpretation in \(\REL\) of a
closed term \(M\) of type \(\Tnat\) contains an integer \(\nu\) then
\(M\) reduces to \(\Num\nu\) using the \(\Rsred\Lang\) reduction
relation.
As usual in this kind of situation we prove this property for a
\emph{particular reduction strategy} that we prefer to present as an
abstract machine in Krivine style.
A state, or \emph{command}, of this machine is a triple %
\(c=\State\delta Ms\) where
\begin{itemize}
\item \(M\) is a closed term of type \(\Tdiffm dF\) where \(F\) is a
  type which is \emph{sharp} in the sense that it is not of shape
  \(\Tdiff A\); in other words %
  \(F=(\Timpl{A_1}{\cdots\Timpl{\Timpl{}{A_k}}\Tnat})\),
\item \(\delta=\Tuple{r_1,\dots,r_d}\) is a sequence of length \(d\)
  of elements of \(\Eset{0,1}\) called \emph{access word}
\item and \(s\) is a stack of type \(F\vdash\Tnat\), meaning that it
  represents an evaluation context \(\Stctx s\) of type \(\Tnat\)
  whose ``hole'' has type \(F\).
\end{itemize}
Then the command \(c\) represents the term %
\(\Tofst c=\Stctx s[\Lprojd{r_1}0{\cdots\Lprojd{r_d}0{M}\cdots}]\) %
which is closed of type \(\Tnat\). %

We introduce a set of reduction rules \(\Rsred\States\) for this
machine and prove that these reduction rules are simulated in the
\(\Rsred\Lang\) rewriting system through the \(c\mapsto\Tofst c\)
translation: this shows that we can really consider this machine just
as a convenient way to express a specific reduction strategy for
applying the reduction \(\Rsred\Lang\).

More precisely this rewriting system acts on \emph{finite multisets of
  commands} which are summable in the sense that the interpretations
of their elements are summable in any model.
We extend the semantics (and, accordingly, the intersection typing
system) to stacks and commands and prove that, if a command is
typeable in the intersection typing system (and then its type is a
natural number \(\nu\)) then its \(\Rsred\States\) reduction leads to
a summable multiset \(C\) of commands which contains the constant
\(\Num\nu\) (or more precisely the command
\(\State{\Tuple{}}{\Num\nu}{\Stempty}\)), thus proving our
completeness result.
This proof follows essentially the standard pattern of a reducibility
argument, complicated by the fact that we have to take into account
arbitrary iterations of the \(\Ldletv xM\) construct.
This method is developed in Section~\ref{sec:int-seq-normalization}.

\subsubsection{Probabilistic semantics and determinism}
In Section~\ref{sec:determinism} we prove that the integer
\(\nu\) above is unique, thus showing that the reduction on commands
is essentially deterministic.
To this end we use the fact that the \(\LL\) model of probabilistic
coherence spaces introduced in~\cite{DanosEhrhard08} is a model of
coherent differentiation.
In that model the type \(\Tnat\) is interpreted as the set of all
sub-probability distributions on \(\Nat\) and we observe%
\footnote{This is due to the fact that the formalism \(\Lang\)
  considered here has no construct for generating random integers as
  in~\cite{EhrhardPaganiTasson18b}.} %
that a summable multiset of commands must be interpreted as such a
sub-probability distribution where \emph{all probabilities belong to
  \(\Nat\)} and hence is either equal to \(0\) or concentrated on a
single element \(\nu\) of \(\Nat\).
So all the elements \(c'\) of \(C\) distinct from the command
\(\State{\Tuple{}}{\Num\nu}{\Stempty}\) must have an \(\emptyset\)
interpretation in \(\REL\) and hence cannot reduce to a value. In
spite of this strong result, the fact that the rewriting system for
this machine has still to deal with multisets of states means that it
still contains a little bit of non-determinism.

Finally using an idea suggested by Guillaume~Geoffroy, we get rid of
this non-determinism by slightly modifying our Krivine machine.
The main change consists in making the access word writable. We can
prove simulation results relating this new machine with the original
one which allow one to prove that the new machine, whilst being fully
deterministic, computes the same thing as the original one, in the
same number of steps.

{
\tableofcontents}

\section{Preliminaries}

\subsection{Notations}
\label{sec:not-multisets}
Given a set \(I\) and \(i,j\in I\), remember that
\(\Kronecker ij\in\Eset{0,1}\) (the Kronecker delta) takes value \(1\)
if \(i=j\) and \(0\) otherwise.

We use \(\Natnz\) for \(\Nat\setminus\Eset 0\).

Given a sequence, or word, \(\alpha=\Tuple{\List i1k}\), we set %
\(\Len\alpha=k\).
We use simple juxtaposition to denote the concatenation of sequences
and also \(i\Tuple{\List i1k}=\Tuple{i,\List i1k}\).
We use \(i\) for the one-element sequence \(\Tuple i\) when there are
no ambiguities.
We use the following notation for circular permutations:
\(\Rcycle\alpha=\Tuple{i_k,i_1,\dots,i_{k-1}}\) and %
\(\Lcycle\alpha=\Tuple{i_2,\dots,i_k,i_1}\).
Of course if \(\Len\alpha=2\) we have \(\Rcycle\alpha=\Lcycle\alpha\).
We also use \(\Rev\alpha\) for the reversed word (that is, if
\(\alpha=\Tuple{i_1,\dots,i_{k}}\) then
\(\Rev\alpha=\Tuple{i_k,\dots,i_1}\)).

We use $\Mfin I$ for the set of finite multisets of elements of a set
$I$.
A multiset is a function $m:I\to\Nat$ such that
$\Supp m=\{i\in I\St m(i)\not=0\}$ is finite. We use additive
notations for operations on multisets ($0$ for the empty multiset,
$m+p$ for their pointwise sum).
We use $\Mset{\List i1k}$ for the multiset $m$ such that
$m(i)=\Card{\{j\in\Nat\St i_j=i\}}$.
If \(m=\Mset{\List j1n}\in\Mfin J\) and \(i\in I\) we set
\(\Mspmap im=\Mset{(i,j_1),\dots,(i,j_n)}\in\Mfin{I\times J}\).

We use \(I\uplus J\) to denote \(I\cup J\) when \(I\cap J=\emptyset\).

Our default notation for composition of morphisms in categories is by
simple juxtaposition, and we use \(X\) to denote the identity morphism
\(X\to X\).
In a category-theoretic context, the notation \(t:F\Tonat G\) means
that \(t\) is a natural transformation from the functor \(F\) to the
functor \(G\).

\subsection{Rewriting}\label{sec:ms-rewriting}
Let \(\cT=(\Rsca\cT,\Rsred\cT)\) be a rewriting system, that is
\(\Rsca\cT\) is a set and \(\mathord{\Rsred\cT}\subseteq\Rsca\cT^2\).
We assume that \(\Rsca\cT\) contains a distinguished element \(0\) and
that there is a binary operation \(+\) on \(\Rsca\cT\):
given \(t_1,t_2\in\Rsca\cT\) there is an element
\(t_1+t_2\in\Rsca\cT\).
We make no further assumptions, in particular we do not assume that,
equipped with \(0\) and \(+\), the set \(\Rsca\cT\) is a monoid.
We define a rewriting system \(\Msrs\cT\) by
\(\Rsca{\Msrs\cT}=\Mfin{\Rsca\cT}\) and rewriting relation defined by
the following rules
\begin{center}
  \begin{prooftree}
    \infer0{\Mset 0\Rel{\Rsred{\Msrs\cT}}\Msetempty}
  \end{prooftree}
  \Treesep
  \begin{prooftree}
    \infer0{\Mset{t_1+t_2}\Rel{\Rsred{\Msrs\cT}}\Mset{t_1,t_2}}
  \end{prooftree}
  \Treesep
  \begin{prooftree}
    \hypo{t\Rel{\Rsred\cT} t'}
    \infer1{\Mset t\Rel{\Rsred{\Msrs\cT}}\Mset{t'}}
  \end{prooftree}
  \Treesep
  \begin{prooftree}
    \hypo{S\Rel{\Rsred{\Msrs\cT}} S'}
    \infer1{S+T\Rel{\Rsred{\Msrs\cT}} S'+T}
  \end{prooftree}
\end{center}
In other words, we have \(S\Rel{\Rsred{\Msrs\cT}}S'\) iff one of the
following conditions hold.
\begin{itemize}
\item \(S=S_0+\Mset 0\) and %
  \(S'=S_0\).
\item \(S=S_0+\Mset t\), \(t\Rel{\Rsred\cT}t'\) and %
  \(S'=S_0+\Mset{t'}\).
\item \(S=S_0+\Mset{t_0+t_1}\) and %
  \(S'=S_0+\Mset{t_0,t_1}\).
\end{itemize}

\section{Semantic motivations}
\label{sec:sem-motivations}

Since the main goal of this paper is to shed some light on the
operational content of coherent differentiation, which was introduced
as a categorical refinement of the semantics of \LL{}
in~\cite{Ehrhard23a} (to which we refer for more details) and
motivated by concrete denotational models, we provide first a bird's
eye view on this categorical setting.
We give then two concrete instances of this notion: the relational
semantics and the probabilistic coherence space semantics of \(\LL\).
The more specific semantic operations and properties which are used
for interpreting our extension $\Lang$ of \PCF{} will be described in
Section~\ref{sec:semantics}.

\subsection{A summary of summable differential categories} %
\label{sec:cohdiff-summary}
A \emph{summable category} is a tuple %
\((\cL,\Sfun,\Sproj0,\Sproj1,\Ssum)\) where %
\(\cL\) is a category with zero morphisms%
\footnote{This means that if \(X,Y\) are objects then there is a
  morphism \(0\in\cL(X,Y)\) such that \(f\Compl 0=0\) and
  \(0\Compl g=0\).}, %
\(\Sfun:\cL\to\cL\) is a functor and %
\(\Sproj0,\Sproj1,\Ssum:\Sfun\Tonat\Id\) are natural transformations
such that %
\(\Sproj0,\Sproj1\in\cL(\Sfun X,X)\) are jointly monic: this means
that a morphism $f\in\cL(X,\Sfun Y)$ is fully characterized by %
\(\Sproj0\Compl f\) and \(\Sproj1\Compl f\). %
Then we say that $f_0,f_1\in\cL(X,Y)$ are \emph{summable} if there
is %
\(h\in\cL(X,\Sfun Y)\) such that \(\Sproj i\Compl h=f_i\) for
$i=0,1$. %
This $h$ is unique and is denoted $\Stuple{f_0,f_1}$. In that
situation the sum $f_0+f_1$ of $f_0,f_1$ is defined as %
\(f_0+f_1=\Ssum\Compl\Stuple{f_0,f_1}\).
There are further axioms \Saxcom, \Saxzero{} and \Saxwit{}
in~\cite{Ehrhard23a} which imply in particular that, equipped with
this partially defined addition, any homset \(\cL(X,Y)\) is a partial
commutative monoid with $0$ as neutral element, and the naturality of
$\Sproj0$, $\Sproj1$ and $\Ssum$ implies that composition commutes
with this partially defined addition, that is, $\cL$ is a ``partially
additive'' category (\Ie~a category enriched in \emph{partial
  commutative monoids} in the sense of~\cite{DuchampPoinsot10}).
These axioms also imply that there is a natural transformation %
\(\Sflip:\Sfun^2\Tonat\Sfun^2\) uniquely characterized by %
\(\Sproj i\Compl\Sproj j\Compl\Sflip=\Sproj j\Compl\Sproj i\) %
for all $i,j\in\Eset{0,1}$, it is called the (canonical)
\emph{flip}.
One also defines uniquely two natural injections %
\(\Sin0=\Stuple{X,0},\Sin1=\Stuple{0,X}:X\to\Sfun X\).

\subsubsection{The associated monad, its monoidal strength and its
  commutativity} %
\label{sec:monad-mult-strength}
Under these assumptions $\Sfun$ has a monad structure with unit %
$\Sin0:\Id\Tonat\Sfun$ and multiplication
\(\Sfunadd:\Sfun^2\Tonat\Sfun\) characterized by %
\(\Sproj0\Compl\Sfunadd=\Sproj0\Compl\Sproj0\) and %
\(\Sproj1\Compl\Sfunadd=\Sproj1\Compl\Sproj0+\Sproj0\Compl\Sproj1\) %
from which it follows easily that %
\(\Sfunadd\Compl\Sflip=\Sfunadd\).
When $\cL$ is symmetric monoidal (with monoidal unit $\One$ and
monoidal product $\ITens$) a further axiom \Saxdist{} expresses
in~\cite{Ehrhard23a} that $\ITens$ distributes over the sum of
morphisms, when defined. It is then possible to define a tensorial
strength
$\Sstr^0_{X_0,X_1}\in\cL(\Tens{(\Sfun
  X_0)}{X_1},\Sfun(\Tens{X_0}{X_1}))$ which is a natural
transformation satisfying further commutations expressing its
compatibility with the $\ITens$ monoidal structure of $\cL$.  Using
the symmetry iso of the monoidal structure of $\cL$ one can define
then
$\Sstr^1_{X_0,X_1}\in\cL(\Tens{X_0}{(\Sfun
  X_1)},\Sfun(\Tens{X_0}{X_1}))$ from $\Sstr^0$.
This strength is fully characterized by %
\(\Sproj i\Compl\Sstr_{X_0,X_1}^0=\Tens{\Sproj i}{X_1}\) for
$i=0,1$. Equipped with this strength the monad
\((\Sfun,\Sin0,\Sfunadd)\) is a commutative monad. More precisely the
following diagram commutes
\begin{equation*}
  \begin{tikzcd}
    &\Tens{\Sfun X_0}{\Sfun X_1}
    \ar[ld,swap,"\Sstr^0_{X_0,\Sfun X_1}"]
    \ar[rd,"\Sstr^1_{\Sfun X_0,X_1}"]
    &
    \\
    \Sfun\Tensp{X_0}{\Sfun X_1}
    \ar[d,swap,"\Sfun\Sstr^1_{X_0,X_1}"]
    &&\Sfun\Tensp{\Sfun X_0}{X_1}
    \ar[d,"\Sfun\Sstr^0_{X_0,X_1}"]
    \\
    \Sfun^2\Tensp{X_0}{X_1}
    \ar[rr,"\Sflip_{\Tens{X_0}{X_1}}"]
    &&
    \Sfun^2\Tensp{X_0}{X_1}
  \end{tikzcd}
\end{equation*}
as indeed %
\(
\Sproj i\Compl\Sproj j\Compl(\Sfun\Sstr^1_{X_0,X_1})
\Compl\Sstr^0_{X_0,\Sfun X_1}=\Tens{\Sproj j}{\Sproj i}
\) and %
\(
\Sproj i\Compl\Sproj j\Compl(\Sfun\Sstr^0_{X_0,X_1})
\Compl\Sstr^1_{\Sfun X_0,X_1}=\Tens{\Sproj i}{\Sproj j}
\). %
The induced natural symmetric lax monoidality morphism %
\(
\Smont_{X_0,X_1}
=\Sfunadd\Compl(\Sfun\Sstr^0_{X_0,X_1})\Compl\Sstr^1_{\Sfun X_0,X_1}
=\Sfunadd(\Sfun\Sstr^1_{X_0,X_1})\Compl\Sstr^0_{X_0,\Sfun X_1}
\in\cL(\Tens{\Sfun X_0}{\Sfun X_1},\Sfun\Tensp{X_0}{X_1})
\) is fully characterized by %
\begin{align*}
  \Sproj 0\Compl\Smont_{X_0,X_1}=\Tens{\Sproj0}{\Sproj0}
  \text{\quad and\quad}
  \Sproj 1\Compl\Smont_{X_0,X_1}
  =\Tens{\Sproj1}{\Sproj0}+\Tens{\Sproj0}{\Sproj1}
\end{align*}
and its $0$-ary version is simply %
\(\Smontz=\Sin0\in\cL(\One,\Sfun\One)\).

If the SMC $\cL$ is also closed then we require in~\cite{Ehrhard23a}to
the summability structure to satisfy a further condition~\Saxfun.  We
use %
\((\Limpl XY,\Evlin\in\cL(\Tens{\Limplp XY}{X},Y))\) %
for the internal hom object of $X$ and $Y$ in $\cL$ and %
\(\Curlin f\in\cL(Z,\Limpl XY)\) for the curried version of %
\(f\in\cL(\Tens ZX,Y)\). With these notations, \Saxfun{} says that %
\(\Sfun\Limplp XY\) and \(\Limpl X{\Sfun Y}\) are isomorphic %
(more precisely the morphism %
\(\Sfun\Limplp XY\to\Limplp X{\Sfun Y}\) that one can define using %
\(\Sstr^0_{\Limpl XY,X}\) is an iso).
Intuitively this means that, on morphisms, summability is defined
``pointwise''.

\subsection{Differentiation as a double distributive law on a
  resource category} %
\label{sec:diff-distr-law}
We assume now that $\cL$ is a resource category which means that $\cL$
is a symmetric monoidal closed category which is also cartesian (with
cartesian product %
\((\Bwith_{i\in I}X_i,(\Proj i)_{i\in I})\) for any finite family of
objects \((X_i)_{i\in I}\) of $\cL$, the case $I=\emptyset$ yielding
the terminal object $\Top$ of $\cL$.
Being a resource category means also that $\cL$ is equipped with a
\emph{resource comonad}, that is a tuple
$(\Excl\_,\Der{},\Digg{},\Seelyz,\Seelyt)$ where $\Excl\_$ is a
functor $\cL\to\cL$ which is a comonad with counit $\Der{}$ (for
\emph{dereliction}) and comultiplication $\Digg{}$ (for
\emph{digging}), and $\Seelyz\in\cL(\Sone,\Excl\Top)$ and
$\Seelyt\in\cL(\Tens{\Excl X}{\Excl Y},\Excl{(\With XY)})$ are the
Seely isomorphisms subject to conditions that we do not recall here,
see for instance~\cite{Mellies09}, apart for the following which
explains how $\Digg{}$ interacts with $\Seelyt$.
\begin{equation}\label{eq:seely-digg-comm}
  \begin{tikzcd}
    \Tens{\Excl{X_0}}{\Excl{X_0}}
    \ar[rr,"\Tens{\Digg{X_0}}{\Digg{X_1}}"]
    \ar[d,swap,"\Seelyt_{X_0,X_1}"] &[1.8em] &[2em]
    \Tens{\Excll{X_0}}{\Excll{X_1}}
    \ar[d,"\Seelyt_{\Excl{X_0},\Excl{X_1}}"]
    \\
    \Excl{\Withp{X_0}{X_1}} \ar[r,"\Digg{\With{X_0}{X_1}}"] &
    \Excll{\Withp{X_0}{X_1}}
    \ar[r,"\Excl{\Tuple{\Excl{\Proj0},\Excl{\Proj1}}}"] &
    \Excl{\Withp{\Excl{X_0}}{\Excl{X_1}}}
  \end{tikzcd}
\end{equation}
Then $\Excl\_$ inherits a \emph{lax} symmetric monoidality on the SMC
$(\cL,\mathord\ITens)$.
This means that one can define %
$\Monz\in\cL(1,\Excl\Sone)$ and %
$\Mont_{X_0,X_1}\in\cL(\Tens{\Excl{X_0}}{\Excl{X_1}},\Excl{\Tensp{X_0}{X_1}})$
satisfying suitable coherence commutations.
Explicitly these morphisms are given by
\[
  \begin{tikzcd}
    \Sone\ar[r,"\Seelyz"]
    &[-1em]
    \Excl\Top\ar[r,"\Digg\Top"]
    &
    \Excll\Top\ar[r,"\Excl{\Invp{\Seelyz}}"]
    &
    \Excl\Sone
  \end{tikzcd}
\]
\[
  \begin{tikzcd}
    \Tens{\Excl{X_0}}{\Excl{X_1}}\ar[d,"\Seelyt_{X_0,X_1}"]
    \\
    \Excl{\Withp{X_0}{X_1}}\ar[d,"\Digg{\With{X_0}{X_1}}"]
    \\
    \Excll{\Withp{X_0}{X_1}}\ar[d,"\Excl{\Invp{\Seelyt_{X_0,X_1}}}"]
    \\
    \Excl{\Tensp{\Excl{X_0}}{\Excl{X_1}}}
    \ar[d,"\Excl{\Tensp{\Der{X_0}}{\Der{X_1}}}"]
    \\
    \Excl{\Tensp{{X_0}}{{X_1}}}
  \end{tikzcd}
\]
And by combining these morphisms in an arbitrary way one can define
uniquely %
\(
\Mong n_{\List X0{n-1}}
\in\cL(\Excl{X_0}\ITens\cdots\ITens\Excl{X_{n-1}},
\Exclp{X_0\ITens\cdots\ITens X_{n-1}})
\)
thanks to the coherence diagrams that we left implicit.

Remember that, using $\Kl\cL$ for the Kleisli category of this comonad
where we use the notation $g\Comp f$ to denote composition of
morphisms, we have a functor %
$\Kllin:\cL\to\Kl\cL$ defined by %
$\Kllin(X)=X$ for objects and, given $f\in\cL(X,Y)$, we set %
$\Kllin f=f\Compl\Der X\in\Kl\cL(X,Y)$.
Functoriality results from the fact that $(\Excl\_,\Der{},\Digg{})$ is
a comonad.
The intuition is that this functor allows one to see morphisms of
$\cL$ (considered as linear) as morphisms in $\Kl\cL$
where morphisms are not linear in general.
This is why this functor is often faithful but of course not full in
general.

The Kleisli category $\Kl\cL$ is cartesian with cartesian product of a
family $(X_i)_{i\in I}$ of objects of $\cL$ given by %
$(\Bwith_{i\in I}X_i,(\Proje i=\Kllin(\Proj i))_{i\in I})$. Given a
family of morphisms %
$f_i\in\Kl\cL(Y,X_i)$ for $i\in I$, the morphism %
$\Tuple{f_i}_{i\in I}\in\cL(\Excl Y,\Bwith_{i\in I}X_i)
=\Kl\cL(Y,\Bwith_{i\in I}X_i)$ is uniquely characterized by the fact
that %
$\Proje j\Comp\Tuple{f_i}_{i\in I}=f_j$ for each $j\in I$. 

Notice that if now %
$f_i\in\Kl\cL(X_i,Y_i)$ for each $i\in I$ we can define functorially %
$\Bwithe_{i\in I}f_i\in\Kl\cL(\Bwith_{i\in I}X_i,\Bwith_{i\in I}Y_i)$
by %
\begin{align*}
  \Bwithe_{i\in I}f_i
  =\Tuple{f_i\Comp\Proje i}_{i\in I}
  =\Tuple{f_i\Compl\Excl{\Proj i}}_{i\in I}
  =(\Bwith_{i\in I}f_i)\Compl\Tuple{\Excl{\Proj i}}_{i\in I}\,.
\end{align*}
\begin{lemma}
  Let $f\in\cL(X,Y)$ and $g\in\Kl\cL(Y,Z)$, we have %
  $g\Comp\Kllin(f)=g\Compl\Excl f$.
\end{lemma}
\begin{proof}
  We have $\Excl{\Der X}\Compl\Digg X=\Id_{\Excl X}$.
\end{proof}

Given \(f\in\cL(\Excl X,Y)\) we set %
$\Prom f=\Excl f\Compl\Digg X\in\cL(\Excl X,\Excl Y)$ which is
sometimes called the $\oc$-lifting or the \emph{promotion} of
$f$.
Given %
\( f\in\cL(\Excl{X_0}\ITens\cdots\ITens\Excl{X_{n-1}},Y) \) %
one can define more generally %
\( \Prom f\in\cL(\Excl{X_0}\ITens\cdots\ITens\Excl{X_{n-1}},\Excl Y)
\) %
using \(\Mong n_{\List X0{n-1}}\).
Notice also that if %
\(f\in\cL(X,Y)\) one has \(\Prom{\Kllin(f)}=\Excl f\).

We assume furthermore that \(\cL\) is equipped with a summability
structure satisfying the axioms summarized in
Section~\ref{sec:cohdiff-summary} (we use the same notations) and we
also assume that the summability functor $\Sfun$ preserves cartesian
products and to simplify notations we assume that it preserves them
strictly, that is %
\begin{align} %
  \label{eq:Sfun-preserves-With}
  \Sfun(\Bwith_{i\in I}X_i)=\Bwith_{i\in I}\Sfun X_i
  \text{\quad and \quad}\forall i\in I\ \Sfun\Proj i=\Proj i\,.
\end{align}
A \emph{coherent differential structure} on such a summable resource
category consists of a natural transformation %
\(\Sdiff_X\in\cL(\Excl{\Sfun X},\Sfun\Excl X)\) satisfying a few
axioms that we recall here.

\begin{Axicond}{\Daxlocal}{ax:daxlocal}
  \(
    \begin{tikzcd}
      \Excl{\Sfun X}\ar[r,"\Sdiff_X"]\ar[dr,swap,"\Excl{\Sproj0}"]
      &\Sfun\Excl X\ar[d,"\Sproj0"]\\
      &\Excl X
    \end{tikzcd}
  \)
\end{Axicond}

\begin{Axicond}{\Daxlin}{ax:daxlin}
  \(
    \begin{tikzcd}
      \Excl X\ar[d,swap,"\Excl{\Sin 0}"]\ar[rd,"\Sin0"]
      &\\
      \Excl{\Sfun X}\ar[r,"\Sdiff_X"]
      &\Sfun\Excl X
    \end{tikzcd}
    \quad\quad
    \begin{tikzcd}
      \Excl{\Sfun^2X}\ar[r,"\Sdiff_{\Sfun X}"]\ar[d,swap,"\Excl{\Sfunadd}"]
      &\Sfun{\Excl{\Sfun X}}\ar[r,"\Sfun\Sdiff_X"]
      &\Sfun^2\Excl X\ar[d,"\Sfunadd"]\\
      \Excl{\Sfun X}\ar[rr,"\Sdiff_X"]
      &&
      \Sfun\Excl X
    \end{tikzcd}
  \)
\end{Axicond}
That is, $\Sdiff$ is a distributive law between the monad $\Sfun$ and
the functor $\Excl\_$.
This means essentially that derivatives commute with sums and with
$0$, that is, are linear.
This allows one to extend the comonad \(\Excl\_\) to the Kleisli
category of the monad \(\Sfun\) which is again a resource category.
It can be understood as an \emph{infinitesimal} extension of \(\cL\);
this construction, as well as its syntactical outcomes, will be
studied in further work.

\begin{Axicond}{\Daxchain}{ax:daxchain}
  \(
    \begin{tikzcd}
      \Excl{\Sfun X}\ar[r,"\Sdiff_X"]\ar[rd,swap,"\Der{\Sfun X}"]
      &\Sfun\Excl X\ar[d,"\Sfun\Der X"]\\
      &\Sfun X
    \end{tikzcd}
    \quad
    \begin{tikzcd}
      \Excl{\Sfun X}\ar[rr,"\Sdiff_X"]\ar[d,swap,"\Digg{\Sfun X}"] &
      &\Sfun\Excl X\ar[d,"\Sfun\Digg X"]\\
      \Excll{\Sfun X}\ar[r,"\Excl{\Sdiff_X}"] &\Excl{\Sfun\Excl
        X}\ar[r,"\Sdiff_{\Excl X}"] &\Sfun\Excll X
    \end{tikzcd}
  \)
\end{Axicond}
That is, $\Sdiff$ is a distributive law between the comonad $\Excl\_$
and the functor $\Sfun$.
This allows one to extend $\Sfun$ to an endofunctor on $\Kl\cL$.

\begin{Axicond}{\Daxwith}{ax:daxwith}
\(
\begin{tikzcd}
  \Excl{\Sfun\Top}\ar[rr,"\Sdiff_\Top"]\ar[d,swap,"\Excl 0"]
  &&\Sfun{\Excl\Top}\ar[d,"\Sfun\Inv{(\Seelyz)}"]\\
  \Excl\Top\ar[r,"\Inv{(\Seelyz)}"]&\Sone\ar[r,"\Sin0"]&\Sfun\Sone
\end{tikzcd}
\)

\(
\begin{tikzcd}
  \Excl{\Sfun{\Withp{X_0}{X_1}}}\ar[r,"\Sdiff_{\With{X_0}{X_1}}"]
  \ar[d,swap,"\Excl{\Tuple{\Sfun{\Proj0},\Sfun{\Proj1}}}"]
  &\Sfun\Excl{\Withp{X_0}{X_1}}\ar[r,"\Sfun\Inv{(\Seelyt)}"]
  &[1.2em]\Sfun\Tensp{\Excl{X_0}}{\Excl{X_1}}\\
  \Excl{\Withp{\Sfun X_0}{\Sfun X_1}}\ar[r,"\Inv{(\Seelyt)}"]
  &\Tens{\Excl{\Sfun X_0}}{\Excl{\Sfun X_1}}
  \ar[r,"\Tens{\Sdiff_{X_0}}{\Sdiff_{X_1}}"]
  &\Tens{\Sfun{\Excl{X_0}}}{\Sfun{\Excl{X_1}}}\ar[u,swap,"\Smont_{\Excl{X_0},\Excl{X_1}}"]
\end{tikzcd}
\)

Notice that our assumption that $\Sfun$ preserves $\IWith$ strictly (in
the sense of~\Eqref{eq:Sfun-preserves-With}) means that the
morphism %
\(\Tuple{\Sfun{\Proj0},\Sfun{\Proj1}}\) is the identity.
  
\end{Axicond}
Intuitively, this diagram expresses that a derivative wrt.~a pair of
variables ($\Sdiff_{\With{X_0}{X_1}}$) can be expressed as a sum of
partial derivatives ($\Sdiff_{X_0}$ and $\Sdiff_{X_1}$):
this is deeply related to the Leibniz Rule.

\begin{Axicond}{\Daxschwarz}{ax:daxschwarz}
  \(
    \begin{tikzcd}
      \Excl{\Sfun^2X}\ar[r,"\Sdiff_{\Sfun X}"]\ar[d,swap,"\Excl\Sflip"]
      &\Sfun\Excl{\Sfun X}\ar[r,"\Sfun\Sdiff_X"]
      &\Sfun^2\Excl X\ar[d,"\Sflip"]\\
      \Excl{\Sfun^2X}\ar[r,"\Sdiff_{\Sfun X}"]
      &\Sfun\Excl{\Sfun X}\ar[r,"\Sfun\Sdiff_X"]
      &\Sfun^2\Excl X
    \end{tikzcd}
  \)
\end{Axicond}
This expresses that the second derivative is a symmetric bilinear function.

When a resource category \(\cL\) is equipped with a summability
structure \((\Sfun,\Sproj0,\Sproj1,\Ssum)\) and a natural
transformation \(\Sdiff\) satisfying these axioms, it is called a
\emph{differential summable resource category}.

\subsubsection{The induced coherent differentiation monad}
Thanks to~\ref{ax:daxchain} 
we extend the functor $\Sfun$ to a functor %
\(\Sdfun:\Kl\cL\to\Kl\cL\) on the Kleisli category of the comonad
$\Excl\_$ as follows:
\(\Sdfun X=\Sfun X\) and, if \(f\in\Kl\cL(X,Y)\) then %
\(\Sdfun f=(\Sfun f)\Compl\Sdiff_X\in\Kl\cL(\Sdfun X,\Sdfun Y)\).
Then we can define $\Sdfunit_X=\Kllin{\Sin0}\in\Kl\cL(X,\Sdfun X)$ and
$\Sdfmult_X=\Kllin{\Sfunadd}\in\Kl\cL(\Sdfun^2 X,\Sdfun X)$ and the
condition~\ref{ax:daxlin} 
entails that these morphisms are natural;
the intuitive meaning of that condition is that the differential of a
map of the Kleisli category is linear in the sense that it commutes
with the algebraic operation represented by \(\Sdfunit\) and
\(\Sdfmult\).
These natural transformations are easily seen to equip $\Sdfun$ with a
monad structure.

\subsubsection{The elementary situation} %
\label{sec:elementary-summability}
In the models of coherent differentiation we know for the time being,
the summability and differential structures arise from more basic
properties%
\footnote{And not from an additional \emph{structure}:
  it is interesting to observe that the summability structure we are
  describing in this section, when it exists, is completely
  characterized by pre-existing structures of the category \(\cL\), so
  having such a summability structure is indeed a \emph{property} of
  \(\cL\).} %
of one specific object of a resource category with zero-morphisms
\(\cL\), namely
\[
  \Into=\With\Sone\Sone\,.
\]
Notice first that it is always possible to define
\(\Win 0,\Win 1,\Wdiag\in\cL(\Sone,\Into)\) by
\(\Win0=\Tuple{\Sone,0}\), \(\Win1=\Tuple{0,\Sone}\) and
\(\Wdiag=\Tuple{\Sone,\Sone}\).
We say that \(\cL\) is \emph{elementarily summable} if the two
following properties hold.
Our first assumption is that \(\Win0,\Win1\) are jointly epic,
meaning that morphisms \(f\in\cL(\Into,X)\) are fully characterized by
\(f\Compl\Win0,f\Compl\Win1\in\cL(\Sone,X)\).
Our second assumption, corresponding to~\Saxwit{} of~\cite{Ehrhard23a},
is that if \(f_0,f_1,g\in\cL(\Tens X\Into,Y)\) are such that the
following diagrams commute
\begin{equation*}
  \begin{tikzcd}
    \Tens X\Sone\ar[r,"\Tens X\Wdiag"]\ar[d,swap,"\Tens X{\Win r}"]
    &[1.2em]
    \Tens X\Into\ar[d,"f_r"]
    \\
    \Tens X\Into\ar[r,"g"]
    &
    Y
  \end{tikzcd}
\end{equation*}
for \(r=0,1\) then there is a (necessarily unique by our first
assumption) morphism \(h\in\cL(\Tens{\Tens X\Into}\Into,Y)\) such that
the following diagram commutes
\begin{equation*}
  \begin{tikzcd}
    \Tens{\Tens X\Into}\Sone
    \ar[r,"\Tens{\Tens X\Into}{\Win r}"]
    \ar[d,swap,"\Rightu"]
    &[2em]
    \Tens{\Tens X\Into}\Into\ar[d,"h"]
    \\
    \Tens X\Into\ar[r,"f_r"]
    &
    Y
  \end{tikzcd}
\end{equation*}
for \(r=0,1\).
Under these assumptions, the functor \(\Sfun=(\Limpl\Into\_)\),
equipped with natural transformations
\(\Sproj0,\Sproj1,\Ssum:\Sfun\Tonat\Id\) easily defined by
precomposition with \(\Win0,\Win1,\Wdiag\), is a summability structure
on \(\cL\) which is called \emph{the elementary summability structure}
of \(\cL\).

It is also possible to equip \(\Into\) with a structure of
cocommutative comonoid, with counit \(\Proj0\in\cL(\Into,\Sone)\) and
comultiplication \(\Scmont\in\cL(\Into,\Tens\Into\Into)\) fully
characterized by
\begin{equation*}
  \begin{tikzcd}
    \Sone
    \ar[r,"\Win0"]
    \ar[d,swap,"\Inv\Leftu=\Inv\Rightu"]
    &[1.4em]
    \Into
    \ar[d,"\Scmont"]
    \\
    \Tens\Sone\Sone
    \ar[r,"\Tens{\Win0}{\Win0}"]
    &
    \Tens\Into\Into
  \end{tikzcd}
  \Treesep
  \begin{tikzcd}
    \Sone
    \ar[r,"\Win1"]
    \ar[d,swap,"\Inv\Leftu=\Inv\Rightu"]
    &[4.4em]
    \Into
    \ar[d,"\Scmont"]
    \\
    \Tens\Sone\Sone
    \ar[r,"\Tens{\Win0}{\Win1}+\Tens{\Win1}{\Win0}"]
    &
    \Tens\Into\Into
  \end{tikzcd}
\end{equation*}
where the sum is defined in terms of the elementary summability
structure, and exists by our assumptions.

It is proven in~\cite{Ehrhard23a} that endowing this elementary
summability structure \((\Sfun,\Sproj0,\Sproj1,\Ssum)\) with a
coherent differentiation structure amounts to equipping \(\Into\) with
a \(\oc\)-coalgebra structure \(\Sdiffca\in\cL(\Into,\Excl\Into)\)
which satisfies also some compatibility with the comonoid structure
\((\Proj0,\Scmont)\) of \(\Into\).
We don't need here to go into these technicalities, we only need to
know that, when \(\cL\) is a Lafont
category~\cite{Mellies09,Ehrhard23a}, such a \(\Sdiffca\) always
exists and is induced by the comonoid structure of \(\Into\):
remember indeed that a resource category is Lafont if, roughly
speaking, any cocommutative comonoid is a \(\oc\)-coalgebra, in a
unique way. See~\cite{Ehrhard23a},
Theorem~21 %
for more details.

So, for a Lafont resource category with zero-morphisms, being an
elementary coherent differential category is a property, not an
additional structure: it simply means that it satisfies the two
conditions of elementary summability.

\subsubsection{\(\Cpolike\) summable categories} %
\label{sec:Scott-summable}
One distinctive feature of the present approach to differentiation is
its built-in compatibility with general recursion at all types.
In the models which motivated this work, general recursion is
implemented by means of fixpoint operators which arise from a cpo
structure of homsets as usual in domain-theoretic models of
\PCF{}~\cite{Plotkin77}.
We explain how the associated order relation is induced by the
summability structure.

Let $(\cL,\Sfun)$ be a summable category.
Let $f_0,f_1\in\cL(X,Y)$, we write $f_0\leq f_1$ if there exists
$h\in\cL(X,\Sfun Y)$ such that $\Sproj0\Compl=f_0$ and
$\Ssum\Compl h=f_1$.
In other words:
there is $g\in\cL(X,Y)$ such that $f_0,g$ are summable and $f_1=f_0+g$.
\begin{lemma}
  The relation $\leq$ on $\cL(X,Y)$ is a preorder relation for which
  \(0\) is the least element.
\end{lemma}

\begin{definition}
  The summable category $(\cL,\Sfun)$ is \emph{\Cpolike} if, equipped
  with $\leq$, any homset $\cL(X,Y)$ is a poset (with $0$ as least
  element as we have seen), which is \(\omega\)-complete is the sense
  that any monotone $\omega$-sequence of elements of $\cL(X,Y)$ has a
  least upper bound%
  \footnote{On purpose we do not ask all directed sets to have a lub
    as it is usual in domain theory because we have in mind
    probability based models where such a requirement would prevent us
    from applying the monotone convergence theorem.}.
  Morphism composition is required to be continuous.
  This means in particular that \(\cL\) is enriched in \(\omega\)-cpos.
  Moreover the $\ITens$ tensor is assumed to be locally continuous (in
  both arguments) when $\cL$ is an SMC.
  Next, the functor $\Sfun$ must be locally continuous, and, in the
  case where $\cL$ is a resource category, the functor $\Excl\_$ is
  also assumed to be locally continuous.
\end{definition}

Assume that \(\cL\) is a \Cpolike{} summable resource category.
In the CCC $\Kl\cL$, for any object $X$, we can define a sequence of
morphisms %
$\Sfix_n^X\in\Kl\cL(\Simpl XX,X)$ by induction as follows
\begin{align*}
  \Sfix^X_0 &=0\\
  \Sfix^X_{n+1} &= \Ev\Comp\Tuple{\Simpl XX,\Sfix^X_n}
\end{align*}
and an easy induction, using the minimality of $0$ and the fact that
all categorical operations are monotone wrt.~$\leq$, shows that the
sequence $(\Sfix^X_n)_{n\in\Nat}$ is monotone.
We set
\begin{align*}
  \Sfix^X=\sup_{n\in\Nat}\Sfix^X_n\in\Kl\cL(\Simpl XX,X)
\end{align*}
and by continuity of all categorical operations we have
\begin{align}
  \label{eq:Sfix-charact}
  \Sfix^X=\Ev\Comp\Tuple{\Simpl XX,\Sfix^X}
\end{align}
which means that \(\Sfix^X\) is a fixpoint operator.

We will illustrate now these general definitions in two concrete
models of \(\LL\) which will also be essential in the technical
developments of this paper.

\subsection{The relational model} %
\label{sec:rel-model}
The simplest example of a resource category is the category $\REL$ of
sets and relations which is a well known model of classical \(\LL\).
In~\cite{Ehrhard23a} it is proven that \(\REL\) is also an
elementary differential summable resource category;
actually we proved this result for the category $\NUCS$ of non-uniform
coherence spaces, but in that category the operations on morphisms are
exactly the same as in $\REL$ and the operations on objects are the
same as in $\REL$ as far as the webs are concerned%
\footnote{A non-uniform coherence space $X$ is a triple consisting of
  a set $\Web X$, the web of $X$, and two disjoint binary symmetric
  relations on that web, the strict coherence and the strict
  incoherence relations.
  Remember that \(\REL\) is a model of \emph{differential \LL{}}
  ---~which is not the case of \(\NUCS\)~---, so it is not surprising
  at all that it is also a coherent differential category.}.

Let us briefly recall the definition of this model of $\LL$. An object
of $\REL$ is a set and $\REL(X,Y)=\Part{X\times Y}$, composition being
the usual composition of relations denoted $v\Compl u$ when
$u\in\REL(X,Y)$ and $v\in\REL(Y,Z)$.
The identity morphism is the diagonal relation.
The isos in this category are the bijections.

This category $\REL$ is symmetric monoidal with
$\Sone=\Eset{\Sonelem}$ as tensorial unit and
$\Tens{X_0}{X_1}=X_0\times X_1$, and given $u_i\in\REL(X_i,Y_i)$ for
$i=0,1$, one sets %
$\Tens{u_0}{u_1}=\Eset{((a_0,a_1),(b_0,b_1))\St (a_i,b_i)\in u_i\text{
    for }i=0,1)}$ defining a functor $\REL^2\to\REL$ which has obvious
natural isos %
$\Leftu_X\in\REL(\Tens\Sone X,X)$, $\Rightu_X\in\REL(\Tens
X\Sone,X)$, %
$\Assoc_{X_0,X_1,X_2}
\in\REL(\Tens{\Tensp{X_0}{X_1}}{X_2},\Tens{X_0}{\Tensp{X_1}{X_2}})$
and %
$\Sym_{X_0,X_1}\in\REL(\Tens{X_0}{X_1},\Tens{X_1}{X_0})$.
This SMC is closed with internal hom from $X$ to $Y$ the pair %
$(\Limpl XY,\Evlin)$ where %
$\Limpl XY=X\times Y$ and
$\Evlin=\Eset{(((a,b),a),b)\St a\in X\text{ and }b\in
  Y}\in\REL(\Tens{\Limplp XY}{X},Y)$. Given any morphism %
$u\in\REL(\Tens ZX,Y)$, the associated morphism (Curry transpose of
$u$) %
$\Curlin u\in\REL(Z,\Limpl XY)$ is simply %
$\Curlin u=\Eset{(a,(b,c))\St ((a,b),c)\in u}$.
This SMCC is \Stauto{} with dualizing object $\Bott=\Sone$, so that
the ``linear negation'' of an object $X$ is simply $X$.

The category $\REL$ is cartesian: the cartesian product of a family
$(X_i)_{i\in I}$ of objects of $\REL$ is %
$(\Bwith_{i\in I}X_i,(\Proj i)_{i\in I})$ where %
$\Bwith_{i\in I}X_i=\Union_{i\in I}\Eset i\times X_i$ and %
$\Proj i=\Eset{((i,a),a)\St i\in I\text{ and }a\in X_i}
\in\REL(\Bwith_{j\in I}X_j,X_i)$ is the $i$th projection.
The tupling of a family of morphisms %
$(u_i\in\REL(Y,X_i))_{i\in I}$ is the morphism %
$\Tuple{u_i}_{i\in I}\in\REL(Y,\Bwith_{i\in I}X_i)$ given by %
$\Tuple{u_i}_{i\in I} =\Eset{(b,(i,a))\St i\in I\text{ and }(b,a)\in
  u_i}$.
The coproduct %
$(\Bplus_{i\in I}X_i,(\Inj i)_{i\in I})$ also exists and is given by %
$\Bplus_{i\in I}X_i=\Bwith_{i\in I}X_i$ and %
$\Inj i\in\REL(X_i,\Bplus_{j\in I}X_j)$ is given by %
$\Inj i=\Eset{(a,(i,a))\St i\in I\text{ and }a\in X_i}$;
it is the $i$th injection.
The cotupling of morphisms $(u_i\in\REL(X_i,Y))_{i\in I}$ is %
$\Cotuple{u_i}_{i\in I}\in\REL(\Bplus_{i\in I}X_i,Y)$ given by %
$\Cotuple{u_i}_{i\in I}=\Eset{((i,a),b)\St (a,b)\in u_i}$.
Notice that the terminal (and initial) object of $\REL$ is
$\Top=\Zero=\emptyset$.

The SMC $\REL$ is a Lafont category~\cite{Mellies09,Ehrhard23a}. The
exponential functor is given by $\Excl X=\Mfin X$ and, if
$s\in\REL(X,Y)$ then %
\(
\Excl s =\{(\Mset{\List a1k},\Mset{\List b1k})\St k\in\Nat \text{
  and }\forall i\ (a_i,b_i)\in s)\}
\), %
defining a functor $\REL\to\REL$.
The comonad structure of that functor is given by the natural
transformations %
$\Der X=\Eset{(\Mset a,a)\St a\in X}\in\REL(\Excl X,X)$ and %
$\Digg X=\Eset{(\Mset{\List m1n},m_1+\cdots+m_n) \St n\in\Nat\text{
    and }\List m1n\in\Excl X}\in\REL(\Excl X,\Excll X)$.
The Seely isomorphisms are %
$\Seelyz=\Eset{(\Sonelem,\Emptymset)}\in\REL(\Sone,\Excl\Top)$ and %
\begin{multline*}
  \Seelyt_{X,Y}
  =\{((\Mset{\List a1n},\Mset{\List b1k})
  ,\Mset{(1,a_1),\dots,(1,a_n)},(2,b_1),\dots,(2,b_k)) \St\\
  \List
  a1n\in X\text{ and }\List b1k\in Y\}\in\REL(\Tens{\Excl X}{\Excl
  Y},\Exclp{\With XY})\,.
\end{multline*}
The Kleisli category \(\Kl\REL\) can be
directly described as %
\(\Kl\REL(X,Y)=\Mfin X\times Y\), with identity %
\(\Id_X=\Eset{(\Mset a,a)\St a\in X}\in\Kl\REL(X,X)\) %
and, given \(u\in\Kl\REL(X,Y)\) and \(v\in\Kl\REL(Y,Z)\), %
composition given by
\(
v\Comp u
=\Eset{(m_1+\cdots+m_k,c)\St k\in\Nat\text{ and }
  \exists\List b1k\in Y\
  (\Mset{\List b1k},c)\in v
  \text{ and }
(m_i,b_i)\in u\text{ for }i=1,\dots,k}
\).

It is easy to see that $\REL$ is elementarily summable (see
Section~\ref{sec:elementary-summability}) and has therefore exactly
one coherent differential structure.
The cocommutative comonoid structure of
\(\Into=\With\Sone\Sone=\Eset{0,1}\) has counit %
\(\Proj0=\Eset{(0,\Sonelem)}\in\REL(\Into,\Sone)\) %
and comultiplication %
\begin{multline*}
\Scmont
=\Eset{(0,(0,0)),(1,(1,0)),(1,(0,1))}\\
=\{(r,(r_0,r_1))\St r,r_0,r_1\in\Into\text{ and }r=r_0+r_1)\}
\in\REL(\Into,\Tens\Into\Into)
\end{multline*} %
as explained in %
Section~5.1 
of~\cite{Ehrhard23a} %
which, by the Lafont property, induces on \(\Into\) the $\oc$-coalgebra
structure \(\Sdiffca\in\REL(\Into,\Excl\Into)\) given by %
\begin{multline*}
  \Sdiffca =\Eset{(0,k\Mset 0)\St k\in\Nat} \cup\Eset{(1,\Mset
    1+k\Mset 0)\St k\in\Nat} \\
  =\Eset{(r,\Mset{\List
    r1k})\in\Into\times\Mfin\Into \St r=\sum_{i=1}^kr_i}\,.
\end{multline*}
As explained in Section~\ref{sec:elementary-summability}, in this
elementary setting, the associated summability functor is %
\[
  \Scfun=(\Limpl\Into\_)\,,
\]
and is therefore given by %
\(\Scfun X=\Eset{0,1}\times X\) and, for \(u\in\REL(X,Y)\), %
\(\Scfun u=\Eset{((i,a),(i,b))\St i\in\Into\text{ and }(a,b)\in u}\).
The associated natural transformations are %
\(
\Sproj i=\Eset{((i,a),a)\St a\in X},
\Ssum
=\Sproj0\cup\Sproj1
=\Eset{((i,a),a)\St i\in\Into\text{ and }a\in X}
\in\REL(\Scfun X,X)
\) %
and the two injections are %
\(\Sin i=\Eset{(a,(i,a))\St a\in X}\in\REL(X,\Scfun X)\) %
for $i=0,1$. %
So any two morphisms \(f_0,f_1\in\REL(X,Y)\) are summable, and their
sum is \(f_0\cup f_1\).

The monad structure of \(\Scfun\) has
\(
\Sin0
=\Eset{(a,(0,a))\St a\in X}\in\REL(X,\Scfun X)
\) as unit and %
\begin{align*}
\Sfunadd
&=\Eset{((0,0,a),(0,a))\St a\in X}
\cup\Eset{((1,0,a),(1,a))\St a\in X}
\cup\Eset{((0,1,a),(1,a))\St a\in X}\\
&=\Eset{((r_0,r_1,a),(r,a))
  \St a\in X,\ r,r_0,r_1\in\Into \text{ and }r=r_0+r_1}
\in\REL(\Scfun^2X,\Scfun X)
\end{align*} %
as multiplication. %
The flip is %
\(
\Sflip=
\{((r_0,r_1,a),(r_1,r_0,a))\St r_0,r_1\in\Into\text{ and }a\in X\}
\in\REL(\Scfun^2X,\Scfun^2X)
\).

Notice that there is a natural bijection between %
\(\Excl{\Scfun X}\) and \(\Mfin X\times\Mfin X\), mapping the multiset %
\(\Mset{(0,a_1),\dots,(0,a_l),(1,b_1),\dots,(1,b_r)}\) to %
\((\Mset{\List a1l},\Mset{\List b1r})\).
The natural transformation %
\(\Sdiff_X\in\REL(\Excl{\Scfun X},\Scfun{\Excl X})\) is defined as the
Curry transpose of the following composition of morphisms
\[
  \begin{tikzcd}
    \Tens{\Exclp{\Limpl\Into X}}{\Into}
    \ar[r,"\Tens{\Exclp{\Limpl\Into X}}{\Sdiffca}"]
    &[2em]
    \Tens{\Exclp{\Limpl\Into X}}{\Excl\Into}
    \ar[r,"\Mont_{\Limpl\Into X,\Into}"]
    &[0.6em]ms
    \Exclp{\Tens{\Limplp\Into X}{\Into}}
    \ar[r,"\Excl\Evlin"]
    &[-1em]
    \Excl X
  \end{tikzcd}\,.
\]
(see~\cite{Ehrhard23a},
Theorem~20 %
and hence %
\begin{multline*}
  \Sdiff_X =\Eset{((m,\Emptymset),(0,m))\St m\in\Mfin X}\\
  \cup\Eset{((m,\Mset a),(1,m+\Mset a))\St m\in\Mfin X\text{ and }a\in
    X}\,.
\end{multline*}
It will also be convenient to write equivalently
\begin{multline*}
\Sdiff_X=
  \{(m',(r,m))\in\Mfin{\Into\times X}\times(\Into\times\Mfin X)
  \St\\
  m=\Mset{\List a1k},\ m'=\Mset{(r_1,a_1),\dots,(r_k,a_k)}
  \text{ and }r=r_1+\cdots+r_k\}\,.
\end{multline*}

Therefore the functor \(\Sdfun:\Kl\REL\to\Kl\REL\), which is defined
on objects by %
\(\Sdfun X=\Scfun X\) and on morphisms by %
\(\Sdfun u=(\Scfun u)\Compl\Sdiff_X\in\REL(\Excl{\Scfun X},\Scfun Y)\)
for %
\(u\in\Kl\REL(X,Y)\) satisfies
\begin{align*}
  \Sdfun u&=
  \Eset{((m,\Emptymset),(0,b))\St (m,b)\in u}
            \cup\Eset{((m,\Mset a),(1,b))\St (m+\Mset a,b)\in u}\\
          &=\{(m',(r,b))\St m'
            =\Mset{(r_1,a_1),\dots,(r_k,a_k)}\in\Mfin{\Into\times X},\\
          &\Textsep\Textsep  (\Mset{\List a1k},b)\in u\text{ and }
            r=r_1+\cdots+r_k\in\Into\}\,.
\end{align*} %
Then the monad structure of \(\Scfun\) is extended to \(\Sdfun\)
by %
\(
\Sdfunit_X
=\{(\Mset a,(0,a))\St a\in X\}\in\Kl\REL(X,\Sdfun X)
\) %
and %
\begin{align*}
  \Sdfmult_X
  &=\Eset{(\Mset{(0,0,a)},(0,a))\St a\in X}\\
  &\Textsep\cup\Eset{(\Mset{(1,0,a)},(1,a))\St a\in X}
    \cup\Eset{(\Mset{(0,1,a)},(1,a))\St a\in X}\\
  &=\Eset{(\Mset{(r_0,r_1,a)},(r,a))\St
    a\in X,\ r,r_0,r_1\in\Into\text{ and }r=r_0+r_1}
    \in\Kl\REL(\Sdfun^2X,\Sdfun X)\,.
\end{align*}

Concerning fixpoints, observe that \((\REL,\Scfun)\) is \Cpolike{} (in
the sense of Section~\ref{sec:Scott-summable}) with $\subseteq$ as
associated order relation on morphisms. One checks easily that the %
\((\Sfix_n^X\in\Kl\REL(\Simpl XX,X))_{n\in\Nat}\) %
are given by %
\(\Sfix_0^X=\emptyset\) and %
\( \Sfix_{n+1}^X =\Eset{(m_1+\cdots+m_k+\Mset{(\Mset{\List
      a1k},a)},a)\St k\in\Nat\text{ and }
  (m_1,a_1),\dots,(m_k,a_k)\in\Sfix_n^X} \) so that %
\(\Sfix^X\in\Kl\REL(\Simpl XX,X)\) is the least set such that %
for all \(k\in\Nat\), \((m_1,a_1),\dots,(m_k,a_k)\in\Sfix^X\) %
and \(a\in X\), one has %
\((m_1+\cdots+m_k+\Mset{(\Mset{\List a1k},a)},a)\in\Sfix^X\).

\subsection{Probabilistic coherence spaces as an elementary coherent
  differential category} %
\label{sec:PCS-definition}
We present now the model which motivated this whole investigation, the
model of probabilistic coherence spaces (PCS~\cite{DanosEhrhard08}),
and explain why it is an elementary model of coherent differentiation.
One can see \(\PCOH\) as obtained by adding to the purely combinatorial
skeleton \(\REL\) some ``numerical flesh'' whose purpose is to
describe probabilities of computational events.

Given an at most countable set $A$ and $u,u'\in\Realpcto A$, we set
$\Eval u{u'}=\sum_{a\in A}u_au'_a\in\Realpc$ where \(\Realpc\) is the
completed half real line. Given $P\subseteq\Realpcto A$, we define
$\Orth P\subseteq\Realpcto A$ as
\begin{align*}
  \Orth P=\{u'\in\Realpcto A\St\forall u\in P\ \Eval u{u'}\leq 1\}\,.
\end{align*}
Observe that if $P$ satisfies
\( \forall a\in A\,\exists x\in P\ x_a>0 \) and
\( \forall a\in A\,\exists m\in\Realp \forall x\in P\ x_a\leq m \)
then $\Orth P\in\Realpto I$ and $\Orth P$ satisfies the same two
properties that we call \emph{local boundedness} which can also be
rephrased as
\begin{align*}
  \forall a\in A\quad 0<\sup_{x\in P}x_a<\infty\,.
\end{align*}

A probabilistic pre-coherence space (pre-PCS) is a pair
$X=(\Web X,\Pcoh X)$ where $\Web X$ is an at most countable
set\footnote{This restriction is not technically necessary, but very
  meaningful from a philosophic point of view; the non countable case
  should be handled via measurable spaces and then one has to consider
  more general objects as in~\cite{EhrhardPaganiTasson18} for
  instance.} and $\Pcoh X\subseteq\Realpcto{\Web X}$ satisfies
$\Biorth{\Pcoh X}=\Pcoh X$. A probabilistic coherence space (PCS) is a
pre-PCS $X$ such that \(\Pcoh X\) is locally bounded.

Given a PCS $X$ and \(x\in\Pcoh X\) we set
$\Norm x_X=\sup_{x'\in\Pcoh{\Orth X}}\Eval x{x'}\in\Intcc01$. This
operation obeys the usual properties of a norm: %
\(\Norm x=0\Implies x=0\), %
\(\Norm{x_0+x_1}\leq\Norm{x_0}+\Norm{x_1}\) and %
\(\Norm{\lambda x}=\lambda\Norm x\) for all \(\lambda\in\Intcc 01\).

\begin{remark}%
  \label{rem:pcoh-comp-notation}
  Given \(x\in\Pcoh X\) and \(a\in\Web X\) we use the notations
  \(x_a\) or \(x(a)\) for the corresponding element of \(\Realp\),
  depending on the context. In some situations \(x_i\) can denote an
  element of \(\Pcoh X\) and in such a situation we will prefer the
  notation \(x_i(a)\) to denote the \(a\)-component of \(x_i\) to
  avoid the ugly \({x_i}_a\).
\end{remark}

Given $t\in\Realpcto{A\times B}$ considered as a matrix (where $A$ and
$B$ are at most countable sets) and $u\in\Realpcto A$, we define
$\Matappa tu\in\Realpcto B$ by
$(\Matappa tu)_b=\sum_{a\in A}t_{a,b}u_a$ (usual formula for applying
a matrix to a vector), and if $s\in\Realpcto{B\times C}$ we define the
product%
\footnote{We write this product in the reverse order wrt.~the usual
  algebraic conventions on matrices, because it is the notion of
  composition in our category and we respect the standard order of
  factors when writing a composition in a category.
  This is a well known and unfortunate mismatch of conventions.}
$\Matapp st\in\Realpcto{A\times C}$ of the matrix $s$ and $t$ as usual
by $(\Matapp st)_{a,c}=\sum_{b\in B}t_{a,b}s_{b,c}$. This is an
associative operation.

Let $X$ and $Y$ be PCSs, a morphism from $X$ to $Y$ is a matrix
$t\in\Realpto{\Web X\times\Web Y}$ such that
$\forall x\in\Pcoh X\ \Matappa tx\in\Pcoh Y$. It is clear that the
identity (diagonal) matrix is a morphism from $X$ to $X$ and that the matrix
product of two morphisms is a morphism and therefore, PCSs equipped
with this notion of morphism form a category $\PCOH$.

The condition $t\in\PCOH(X,Y)$ is equivalent to
\(
\forall x\in\Pcoh X\,\forall y'\in\Pcoh{\Orth Y}\ \Eval{\Matappa tx}{y'}\leq 1
\)
and observe that $\Eval{\Matappa tx}{y'}=\Eval t{\Tens x{y'}}$ where
$(\Tens x{y'})_{(a,b)}=x_ay'_b$. We define
$\Limpl XY=(\Web X\times\Web Y,\{t\in\Realpto{\Web{\Limpl
    XY}}\St\forall x\in\Pcoh X\ \Matappa tx\in\Pcoh Y\})$: this is a
pre-PCS by this observation, and checking that it is indeed a PCS is
easy.

We define then $\Tens XY=\Orth{\Limplp X{\Orth Y}}$; %
this is a PCS which satisfies
\(
\Pcohp{\Tens XZ}=\Biorth{\{\Tens xz\St x\in\Pcoh X\text{ and }z\in\Pcoh Z\}}
\)
where $(\Tens xz)_{(a,c)}=x_az_c$.
Then it is easy to see that we have equipped in that way the category
$\PCOH$ with a symmetric monoidal structure for which it is
$\ast$-autonomous with the dualizing object
$\Bott=\One=(\{*\},[0,1])$, which coincides with the unit of $\ITens$.
The $\ast$-autonomy follows easily from the observation that
$(\Limpl X\Bott)\Isom\Orth X$.

\begin{lemma}
  Given \(s,t\in\PCOH(X_1\ITens\cdots\ITens X_k,Y)\), if %
  for all \((x_i\in\Pcoh{X_i})_{i=1}^k\) one has %
  \(\Matappa s{(x_1\ITens\cdots\ITens x_k)}
  =\Matappa t{(x_1\ITens\cdots\ITens x_k)}\)
  then \(s=t\).
\end{lemma}

The category $\PCOH$ is cartesian: if $(X_j)_{j\in J}$ is an at most
countable family of PCSs, then
$(\Bwith_{j\in J}X_j,(\Proj j)_{j\in J})$ is the cartesian product of
the $X_j$'s, with
$\Web{\Bwith_{j\in J}X_j}=\Bunion_{j\in J}\{j\}\times\Web{X_j}$,
$(\Proj j)_{(k,a),a'}=1$ if $j=k$ and $a=a'$ and
$(\Proj j)_{(k,a),a'}=0$ otherwise, and
$x\in\Pcohp{\Bwith_{j\in J}X_j}$ if $\Matappa{\Proj j}x\in\Pcoh{X_j}$
for each $j\in J$ (for $x\in\Realpto{\Web{\Bwith_{j\in
      J}X_j}}$). Given $(t_j\in\PCOH(Y,X_j))_{j\in J}$, the unique
morphism $t=\Tuple{t_j}_{j\in J}\in\PCOH(Y,\Bwith_{j\in J}X_j)$ such
that $\Proj j\Compl t=t_j$ is simply defined by
$t_{b,(j,a)}=(t_j)_{a,b}$. The dual operation $\Bplus_{j\in J}X_j$,
which is a coproduct, is characterized by
$\Web{\Bplus_{j\in J}X_j}=\Bunion_{j\in J}\{j\}\times\Web{X_j}$ and
$x\in\Pcohp{\Bplus_{j\in J}X_j}$ if $x\in\Pcohp{\Bwith_{j\in J}X_j}$
and $\sum_{j\in J}\Norm{\Matappa{\Proj j}x}_{X_j}\leq 1$.

A particular case is $\Snat=\Bplus_{\nu\in\Nat}X_\nu$ where
$X_\nu=\One$ for each $\nu\in\Nat$. So that $\Web\Snat=\Nat$ and
$x\in\Realpto\Nat$ belongs to $\Pcoh\Snat$ if
$\sum_{\nu\in\Nat}x_\nu\leq 1$ (that is, $x$ is a sub-probability
distribution on $\Nat$). For each $\nu\in\Nat$ we have
$\Base \nu\in\Pcoh\Snat$ which is the distribution concentrated on the
integer $\nu$.  There are successor and predecessor morphisms
$\Ssuc,\Spred\in\PCOH(\Snat,\Snat)$ given by
$\Ssuc_{\nu,\nu'}=\Kronecker{\nu+1}{\nu'}$ %
(see~Section~\ref{sec:not-multisets}) %
and $\Spred_{\nu,\nu'}=1$ if $\nu=\nu'=0$ or $\nu=\nu'+1$ (and
$\Spred_{\nu,\nu'}=0$ in all other cases). An element of
$\PCOH(\Snat,\Snat)$ is a (sub)stochastic matrix and the very idea of
this model is to represent programs as transformations of this kind,
and their generalizations.

As to the exponentials, one sets $\Web{\Excl X}=\Mfin{\Web X}$ and
$\Pcohp{\Excl X}=\Biorth{\{\Prom x\St x\in\Pcoh X\}}$ where, given
$m\in\Mfin{\Web X}$, $\Prom x_m=x^m=\prod_{a\in\Web X}x_a^{m(a)}$. A
morphism
$t\in\PCOH(\Excl X,Y)=\Pcohp{\Limpl{\Excl X}{Y}}$ is
completely characterized by the associated analytic function
\begin{align*}
  \Fun t:\Pcoh X&\to\Pcoh Y\\
  x&\mapsto\Matappa t{\Prom x}=\sum_{m\in\Web{\Excl
      X},b\in\Web Y}t_{m,b}x^m\,\Base b\,.
\end{align*}

\begin{lemma}%
  \label{lemma:kl-tensor-maps-charact}
  Let %
  \(t\in\Realpto{\Web{\Limpl{\Excl{X_1}\ITens\cdots\ITens\Excl{X_k}}{Y}}}\).
  One has \(t\in\PCOH(\Excl{X_1}\ITens\cdots\ITens\Excl{X_k},Y)\) %
  iff for all \((x_i\in\Pcoh{X_i})_{i=1}^k\) one has %
  \(\Matappa t{(\Prom{x_1}\ITens\cdots\ITens\Prom{x_k})}\in\Pcoh Y\).
\end{lemma}

\begin{lemma}%
  \label{lemma:pcoh-kl-morph-charact}
  If \(s,t\in\PCOH(\Excl{X_1}\ITens\cdots\ITens\Excl{X_k},Y)\)
  satisfy %
  \(
  \Matappa s{(\Prom{x_1}\ITens\cdots\ITens\Prom{x_k})}
  =
  \Matappa t{(\Prom{x_1}\ITens\cdots\ITens\Prom{x_k})}
  \) for all %
  \((x_i\in\Pcoh{X_i})_{i=1}^k\) then \(s=t\).
\end{lemma}
This very useful property uses crucially the local boundedness
property of PCSs.

Then given $t\in\PCOH(X,Y)$, we explain now how to define
$\Excl t\in\PCOH(\Excl X,\Excl Y)$.
Let \(m\in\Mfin{\Web X}\) and \(p\in\Mfin{\Web Y}\).
We use \(\Mstrans mp\) for the set of all
\(r\in\Mfin{\Web X\times\Web Y}\) such that
\begin{align*}
  \forall a\in\Web X\ m(a)=\sum_{b\in\Web Y}r(a,b)
  \text{\quad and\quad}
  \forall b\in\Web Y\ p(b)=\sum_{a\in\Web X}r(a,b)\,.
\end{align*}
Notice that if \(r\in\Mfin{\Web X\times\Web Y}\) then %
\(\Mscard r=\Mscard m=\Mscard p\) so that \(\Mstrans mp\) is non-empty
iff \(\Mscard m=\Mscard p\).
When \(r\in\Mstrans np\) we set
\begin{align*}
  \Multinomb pr
  =\prod_{b\in\Web Y}\frac{\Factor{p(b)}}{\prod_{a\in\Web X}\Factor{r(a,b)}}
\end{align*}
which belongs to \(\Nat\setminus\Eset 0\); this is a generalized
multinomial coefficient.
Then we have
\begin{align*}
  (\Excl t)_{m,p}=\sum_{r\in\Mstrans mp}\Multinomb pr t^r
\end{align*}
where we recall that %
\(t^r=\prod_{(a,b)\in\Web X\times\Web Y}t_{a,b}^{r(a,b)}\).
The main feature of this definition is that for all \(x\in\Pcoh X\) one
has %
\(
\Fun{\Excl t}(x)
=\Matappa{\Excl t}{\Prom x}=\Prom{(\Matappa tx)}
\). %
This property fully characterizes \(\Excl t\). %
The comonad structure is given by %
\(\Der X\in\Realpto{\Web{\Limpl{\Excl X}{X}}}\) given by %
\((\Der X)_{m,a}=\Kronecker{m}{\Mset a}\) so that %
\(\forall x\in\Pcoh X\ \Matappa{\Der X}{\Prom x}=x\in\Pcoh X\) and
therefore \(\Der X\in\Pcoh(\Excl X,X)\).
Similarly one defines %
\(\Digg X\in\Realpto{\Web{\Limpl{\Excl X}{\Excll X}}}\) by %
\((\Digg X)_{(m,\Mset{\List m1n})}=\Kronecker m{m_1+\cdots+m_n}\) so
that %
\(\forall x\in\Pcoh X\ \Matappa{\Digg X}{\Prom x}=\Promm x\) and
hence, again, \(\Digg X\in\Pcoh(\Excl X,\Excll X)\).
The equations required to prove that \((\Excl\_,\Der{},\Digg{})\) is
indeed a comonad result from Lemma~\ref{lemma:pcoh-kl-morph-charact}.
For instance, let \(t\in\Pcoh(X,Y)\), we have %
\( \Matappa{(\Digg Y\Compl\Excl t)}{\Prom x} =\Matappa{\Digg
  Y}{(\Prom{\Matappa tx})} =\Matappa{\Digg X}{\Prom{(\Matappa tx)}}
=\Promm{(\Matappa tx)} \) and %
\( \Matappa{(\Excll t\Compl\Digg X)}{\Prom x} =\Matappa{\Excll
  t}{(\Matappa{\Digg X}{\Prom x})} =\Matappa{\Excll t}{\Promm x}
=\Prom{(\Matappa{\Excl t}{\Prom x})} =\Promm{(\Matappa tx)} \) which
shows that \(\Digg{}\) is a natural transformation.
As another example, we have %
\( \Matappa{(\Digg{\Excl X}\Compl\Digg X)}{\Prom x}
=\Matappa{\Digg{\Excl X}}{\Promm x} =\Prommm x \) and %
\( \Matappa{(\Excl{\Digg X}\Compl\Digg X)}{\Prom x}
=\Matappa{\Excl{\Digg X}}{\Promm x} =\Prom{(\Matappa{\Digg X}{\Prom
    x})} =\Prom{(\Promm x)} =\Prommm x \) %
and hence %
\(\Digg{\Excl X}\Compl\Digg X=\Excl{\Digg X}\Compl\Digg X\) %
which is one of the required comonad commutations.
The others are proven similarly.

The monoidality Seely isomorphisms %
\(\Seelyz\in\PCOH(\Sone,\Excl\Top)\) and %
\( \Seelyt_{X_1,X_2}
\in\PCOH(\Tens{\Excl{X_1}}{\Excl{X_2}},\Exclp{\With{X_1}{X_2}}) \) %
are given by %
\(\Seelyz_{\Sonelem,\Emptymset}=1\) and %
\(\Seelyt_{((m_1,m_2),m)}=\Kronecker{\Mspmap 1{m_1}+\Mspmap
  2{m_2}}{m}\) %
where, for a multiset \(m=\Mset{\List a1k}\) we set %
\(\Mspmap im=\Mset{(i,a_1),\dots,(i,a_k)}\), see
Section~\ref{sec:not-multisets}.
It is obvious that \(\Seelyz\) is an iso. To check that
\(\Seelyt_{X_1,x_2}\) is a morphism we use
Lemma~\ref{lemma:kl-tensor-maps-charact}: let \(x_i\in\Pcoh{X_i}\) for
\(i=1,2\), one has %
\( \Matappa{\Seelyt_{X_1,X_2}}{(\Tens{\Prom{x_1}}{\Prom{x_2}})}
=\Prom{\Tuple{x_1,x_2}} \in\Pcoh{\Exclp{\With{X_1}{X_2}}} \). %
Conversely, defining
\(
s\in\Realpto{\Limpl{\Exclp{\With{X_1}{X_2}}}{\Tensp{\Excl{X_1}}{\Excl{X_2}}}}
\) by %
\(s_{m,(m_1,m_2)}=\Kronecker{\Mspmap1{m_1}+\Mspmap2{m_2}}{m}\) we
have %
\( \Matappa s{\Prom{\Tuple{x_1,x_2}}} =\Tens{\Prom{x_1}}{\Prom{x_2}}
\in\Pcohp{\Tens{\Excl{X_1}}{\Excl{X_2}}} \) for all
\(x_i\in\Pcoh{X_i}\) (\(i=1,2\)), and hence %
\(s\in\PCOH(\Exclp{\With{X_1}{X_2}},\Tensp{\Excl{X_1}}{\Excl{X_2}})\). %
It is obvious that \(s\) is the inverse of \(\Seelyt_{X_1,X_2}\) which
is therefore an iso in \(\PCOH\).
Proving that it is natural and that it satisfies all the required
commutations for turning \(\PCOH\) into a model of \LL{} is routine
(using crucially Lemma~\ref{lemma:pcoh-kl-morph-charact}).

The induced lax monoidality %
\(
\Mong k\in\PCOH(
\Excl{X_1}\ITens\cdots\ITens\Excl{X_k},\Excl{(X_1\ITens\cdots\ITens X_k)})
\) is such that %
\((\Mong k)_{(\List m1k),m}=1\)
if %
\(m=\Mset{(a_1^1,\dots,a_k^1),\dots,(a_1^n,\dots,a_k^n))}\) and %
\((m_i=\Mset{a_i^1,\dots,a_i^n})_{i=1}^k\), and %
\((\Mong k)_{(\List m1k),m}=0\) otherwise.

\begin{theorem}%
  \label{th:pcoh-Lafont}
  The SMC \(\PCOH\) is a Lafont category.
\end{theorem}
\begin{proof}
  This is the object of~\cite{CrubilleEhrhardPaganiTasson16}. 
\end{proof}

\subsubsection{Elementary differential structure of \(\PCOH\)}%
\label{sec:pcoh-can-diff-struct}
The category \(\PCOH\) has zero-morphisms (we have the \(0\) matrix
in \(\PCOH(X,Y)\) for any two objects \(X\) and \(Y\)).

The object \(\Into=\With\Sone\Sone\) can be described as %
\(\Web\Into=\Eset{0,1}\) and %
\(\Pcoh\Into=\Intcc01^2\).
Then the morphisms \((\Win r\in\PCOH(\Sone,\Into))_{r=0,1}\) are
characterized by %
\(\Matappa{\Win 0}u=(u,0)\) and \(\Matappa{\Win 1}u=(0,u)\) for %
\(u\in\Pcoh\Sone=\Intcc01\). These two morphisms are jointly epic
because, for any \(t\in\PCOH(\Into,X)\) and \((u_0,u_1)\in\Pcoh\Into\)
one has %
\(\Matappa t{(u_0,u_1)}=\Matappa t{(u_0,0)}+\Matappa t{(0,u_1)}\) by
linearity of \(t\).

Given a PCS \(X\), the PCS \(\Scfun X=\Limplp\Into X\) is
characterized by \(\Web{\Scfun X}=\Eset{0,1}\times\Web X\) and %
\(\Pcoh{\Scfun X}=\{(x_0,x_1)\in\Pcoh X^2\St x_0+x_1\in\Pcoh X\}\) (to
be more precise, an element \(x\in\Realpto{\Web{\Scfun X}}\) belongs
to \(\Pcoh{\Scfun X}\) if \(x_0+x_1\in\Pcoh X\), where
\(x_r\in\Pcoh X\) is given by \(x_r(a)=x(r,a)\), we refer to
Remark~\ref{rem:pcoh-comp-notation} for the notation).
Given \(x\in\Pcohp{\Scfun X}\), the morphism
\(x\Compl\Wdiag\in\Compl\PCOH(\Sone,X)\), considered as an element of
\(\Pcoh X\), is simply \(x_0+x_1\).
So the natural transformations %
\(\Sproj0,\Sproj1,\Ssum\in\PCOH(\Scfun X,X)\) are characterized by %
\(\Matappa{\Sproj r}x=x_r\) and \(\Matappa\Ssum x=x_0+x_1\).

Therefore two morphisms \(s_0,s_1\in\PCOH(X,Y)\) are summable iff %
\(\forall x\in\Pcoh X\ \Matappa{s_0}x+\Matappa{s_1}x\in\Pcoh Y\) which
is equivalent to \(s_0+s_1\in\PCOH(X,Y)\) since %
\(\Matappa{s_0}x+\Matappa{s_1}x=\Matappa{(s_0+s_1)}x\).
Then the witness of summability is %
\(\Stuple{s_0,s_1}\in\PCOH(X,\Scfun Y)\) characterized by %
\(\Matappa{\Stuple{s_0,s_1}}x=(\Matappa{s_0}x,\Matappa{s_1}x)\).
Let $s_{00},s_{01},s_{10},s_{11}\in\PCOH(X,Y)$ be morphisms such that
$(s_{00},s_{01})$ and $(s_{10},s_{11})$ are summable, and moreover
$(s_{00}+s_{01},s_{10}+s_{11})$ is summable.
Then the witnesses
$\Stuple{s_{00},s_{01}},\Stuple{s_{10},s_{11}}\in\PCOH(X,\Sfun X)$ are
summable because %
\(\Stuple{s_{00},s_{01}}+\Stuple{s_{10},s_{11}}=
\Stuple{s_{00}+s_{10},s_{01}+s_{11}}\) as easily checked.
So \Saxwit{} holds (see Section~\ref{sec:elementary-summability}
and~\cite{Ehrhard23a}) which shows that \(\PCOH\) is an elementary
summable category.

As explained in Section~\ref{sec:elementary-summability}
and~\cite{Ehrhard23a} %
Section~5.1 
\(\Into\) is equipped with a commutative comonoid structure given by
the two \(\PCOH\) morphisms %
\(\Proj0\in\PCOH(\Into,\Sone)\) and %
\(\Scmont\in\PCOH(\Into,\Tens\Into\Into)\), which are given by %
\((\Proj0)_{r,\Sonelem}=\Kronecker r0\) and %
\(\Scmont_{r,(r_0,r_1)}=\Kronecker{r}{r_0+r_1}\) %
for \(r,r_0,r_1\in\Into\).
Therefore, by Theorem~\ref{th:pcoh-Lafont}, \(\Into\) has an induced
\(\oc\)-coalgebra structure \(\Sdiffca\in\PCOH(\Into,\Excl\Into)\),
which is given by
\begin{align*}
  \Sdiffca_{r,\Mset{\List r1k}}=\Kronecker{r}{\sum_{i=1}^kr_k}\,,
\end{align*}
in other words \(\Sdiffca_{0,\Mset{\List r1k}}\) is equal to \(1\) if
all the \(r_i\)'s are \(=0\) and to \(0\) otherwise.
And %
\(\Sdiffca_{1,\Mset{\List r1k}}\) is equal to \(1\) if exactly one
among the \(r_i\)'s is equal to \(1\) an all the others are equal to
\(0\), and to \(0\) otherwise.

By %
Theorem~21 %
of~\cite{Ehrhard23a}
we know that \(\Sdiffca\) (denoted~\(\Sdiffcaold\) in that paper)
defines a coherent differential structure on \(\PCOH\).
Let us describe explicitly the associated natural %
\(\Sdiff_X\in\PCOH(\Excl{\Sdfun X},\Sdfun{\Excl X})\).
We know that %
\(\Sdiff_X=\Curlin d\) where %
\(d\in\PCOH(\Excl{\Limplp\Into X}\ITens\Into,\Excl X)\) is defined as
the following composition of morphisms in \(\PCOH\):
\begin{center}
  \begin{tikzcd}
    \Tens{\Exclp{\Limpl\Into X}}{\Into}
    \ar[r,"\Tens{\Exclp{\Limpl\Into X}}{\Sdiffca}"]
    &[2em]
    \Tens{\Exclp{\Limpl\Into X}}{\Excl\Into}
    \ar[r,"\Mont_{\Limpl\Into X,\Into}"]
    &[0.6em] \Exclp{\Tens{\Limplp\Into X}{\Into}}
    \ar[r,"\Excl\Evlin"]
    &[-1em] \Excl X
  \end{tikzcd}\,.
\end{center}
Let
\(
d'=
\Mont_{\Limpl\Into X,\Into}\Compl(\Tens{\Exclp{\Limpl\Into X}}{\Sdiffca})
\in\PCOH(\Tens{\Exclp{\Limpl\Into X}}{\Into},
\Exclp{\Tens{\Limplp\Into X}{\Into}})
\), we have %
\begin{align*}
  d'_{(\Mset{(r_1,a_1),\dots,(r_n,a_n))},r),m}
  =
  \begin{cases}
    1 &\text{if }m=\Mset{((r_1,a_1),r'_1),\dots,((r_n,a_n),r'_n)}
    \text{ and }r=r'_1+\cdots+r'_n\\
    0 &\text{otherwise}
  \end{cases}
\end{align*}
On the other hand, to have
\((\Excl{\Evlin})_{\Mset{((r_1,a_1),r'_1),\dots,((r_n,a_n),r'_n)},p}\not=0\) %
we need %
\((r_i=r'_i)_{i=1}^n\) and \(p=\Mset{\List a1n}\).
So to have %
\(d_{(\Mset{(r_1,a_1),\dots,(r_n,a_n))},r),p}\not=0\) we need %
\(r=r_1+\cdots+r_n\) and \(p=\Mset{\List a1n}\), and then
\begin{align*}
  d_{(\Mset{(r_1,a_1),\dots,(r_n,a_n))},r),p}
  =(\Excl\Evlin)_{\Mset{((r_1,a_1),r_1),
  \dots,((r_n,a_n),r_n)},\Mset{\List a1n}}
\end{align*}
Notice that %
\(\Mstrans{\Mset{((r_1,a_1),r_1), \dots,((r_n,a_n),r_n)}}{\Mset{\List
    a1n}}\) contains exactly one element \(h\) such that
\(\Evlin^h\not=0\), namely the multiset %
\(h=\Mset{(((r_1,a_1),r_1),a_1),\dots,(((r_n,a_n),r_n),a_n)}\), and of
course \(\Evlin^h=1\).
If \(r=0\) we have \(r_1=\cdots=r_n=0\) and hence
\begin{align*}
  d_{(\Mset{(r_1,a_1),\dots,(r_n,a_n)},r),p}
  =\Multinomb{p}{h}=1\,.
\end{align*}
If \(r=1\), there is exactly one \(i\in\Eset{1,\dots,n}\) such that
\(r_i=1\), and we have \(r_j=0\) for \(j\not=i\). Then we have
\begin{align*}
  d_{(\Mset{(r_1,a_1),\dots,(r_n,a_n)},r),p}
  =\Multinomb{p}{h}=p(a_i)\,.
\end{align*}
To summarize
\begin{align*}
  (\Sdiff_X)_{m',(r,m)}
  =
  \begin{cases}
    1 &\text{if }r=0\text{ and }m'=\Mspmap0m\\
    m(a) &\text{if }r=1,\ a\in\Supp m
    \text{ and }m'=(\Mspmap 0{(m-\Mset a)})+\Mset{(1,a)}\\
    0&\text{otherwise.}
  \end{cases}
\end{align*}
Let \(t\in\Kl\PCOH(X,Y)=\Pcohp{\Limpl{\Excl X}Y}\). %
Then \(\Sdfun t\in\Kl\PCOH(\Sdfun X,\Sdfun Y)\) is defined as %
\((\Sfun t)\Compl\Sdiff_X\), so we have %
\begin{align*}
  (\Sdfun t)_{m',(r,b)}
  =
  \begin{cases}
    t_{m,b}&\text{if }r=0\text{ and }m'=\Mspmap 0m\\
    (m(a)+1)t_{m+\Mset a,b}&\text{if }r=1
    \text{ and }m'=(\Mspmap0m)+\Mset{(1,a)}\\
    0&\text{otherwise.}
  \end{cases}
\end{align*}
Notice that in the above trichotomy the multiset \(m\) is completely
determined by the condition on \(m'\): in the first case
\(m'=\Mset{(0,a_1),\dots,(0,a_n)}\) and then \(m=\Mset{\List a1n}\). In
the second case \(m'=\Mset{(r_1,a_1),\dots,(r_n,a_n)}\) and there is
exactly one index \(i\) such that \(r_i=1\), and we have \(r_j=0\) for
\(j\not=i\). Then we have
\(m=\Mset{a_1,\dots,a_{i-1},a_{i+1},\dots,a_n}\) and \(a=a_i\).

Let \(t\in\Kl\PCOH(X,Y)\), to describe the function %
\(\Fun{\Sdfun t}:\Pcohp{\Sdfun X}\to\Pcohp{\Sdfun Y}\), remember first
that \(\Pcohp{\Sdfun X}\) can be identified with the set of all pairs
\((x,u)\in\Pcoh X^2\) such that \(x+u\in\Pcoh X\). With this
identification, the element \(\Prom{(x,u)}\) of
\(\Pcohp{\Excl{\Scfun X}}\) is given by
\begin{align*}
  \Prom{(x,u)}_{m'}=\prod_{a\in\Web X}x_a^{m'(0,a)}
  \prod_{a\in\Web X}u_a^{m'(1,a)}\,.
\end{align*}
If \(m'=\Mspmap 0m\) then \(\Prom{(x,u)}_{m'}=x^m\)
and if \(m'=\Mspmap0m+\Mset{(1,a)}\) then \(\Prom{(x,u)}_{m'}=x^mu_a\).
It follows that
\begin{align*}
  \Fun{\Sdfun t}(x,u)
  =\left(\Fun t(x),\sum_{m\in\Mfin X,a\in\Web X,b\in\Web Y}
  (m(a)+1)t_{m+\Mset a,b}x^mu_a\,\Base b\right)
\end{align*}
and notice that the second component of this tuple is nothing but the
\(u\)-linear component of the powerseries \(\Fun t(x+u)\)
(see~\cite{Ehrhard19}). So, as expected, if we set \(f=\Fun t\) then
\[
  \Fun{\Sdfun t}(x,u)=(f(x),f'(x)\cdot u)
\]
in the ordinary sense of mathematical differentiation.

\begin{Example}
  There is a morphism \(t\in\Kl\PCOH(\Sone,\Sone)\) such that,
  identifying \(\Pcoh\Sone\) with \(\Intcc01\), one has
  \(\Fun t(x)=1-\sqrt{1-x}=\sum_{n\in\Nat}t_nx^n=f(x)\) for a sequence
  \((t_n)_{n\in\Nat}\) of non-negative real numbers that we could
  write explicitly. Then %
  \( \Fun{\Sdfun t}(x,u) =(\Fun t(x),\sum_{n\in\Nat}(n+1)t_{n+1}x^nu)
  =(f(x),f'(x)u).  \) %
  In this case it is interesting to notice that
  \(f'(x)=\frac1{2\sqrt{1-x}}\) is not defined for \(x=1\) but that
  \(f'(x)u\) is defined even for \(x=1\) (and takes value \(0\))
  because of the constraint that \(x+u=1\). And indeed we know that %
  \(\Sdfun t\in\Kl\PCOH(\Sdfun\Sone,\Sdfun\Sone)\). The function \(f\)
  is entirely defined by the equation \(f(x)=\frac 12x+\frac12f(x)^2\)
  and by the fact that the corresponding series must have only
  non-negative coefficients. It is easy to write in a probabilistic
  version of \PCF{} with a unit type a recursive program which is
  interpreted as \(t\).
\end{Example}

\section{Syntax of \texorpdfstring{\(\Lang\)}{Λcd}}\label{sec:syntax}
Our choice of notations for \(\Lang\) is fully coherent with the
notations chosen to describe the model, suggesting a straightforward
denotational interpretation.
The types are \( A,B,\dots \Bnfeq \Tdnat d \Bnfor \Timpl AB \) (with
$d\in\Nat$) and then for any type $A$ we define $\Tdiff A$ as follows:
$\Tdiff{(\Tdnat d)}=\Tdnat{d+1}$ and
$\Tdiff{(\Timpl AB)}=(\Timpl A{\Tdiff B})$.
Terms are given by
\begin{align*}
  M,N,\dots
  \Bnfeq x
  &\Bnfor \Abst xAM
    \Bnfor \App MN
    \Bnfor \Lfix M
    \Bnfor \Num n
    \Bnfor \Lsucc dM
    \Bnfor \Lpred dM
    \Bnfor \Lift AdMPQ\\
    &\Bnfor \Llett AdxMP
    \Bnfor \Ldiff M
    \Bnfor \Lprojd idM
    \Bnfor \Linjd idM
     \Bnfor \Lsumd dM
     \Bnfor \Lflipdl dlM
    \Bnfor \Lzerot A
    \Bnfor\Lplus MN
\end{align*}
where $n,d,l\in\Nat$ and $i\in\Eset{0,1}$, so that our syntax has
countably many constructs.

\subsection{The typing system}

The typing system uses a reduction relation $\Linred$ expressing that
most constructs are linear wrt.~$0$ and %
addition of terms; it is specified in Figure~\ref{fig:lin-eq} and is
based on the following notion of \emph{linear context}
\begin{equation} %
  \label{eq:linear-context}
  \begin{split}
  L \Bnfeq\Echole
  & \Bnfor\Abst xAL
    \Bnfor\App LN
    \Bnfor\Lsucc dL
    \Bnfor\Lpred dL
    \Bnfor\Lif dLPQ
    \Bnfor\Llet dxLP
  \\
  & \Bnfor\Ldiff L
    \Bnfor \Lprojd idL
    \Bnfor \Linjd idL
    \Bnfor \Lsumd dL
    \Bnfor \Lflipdl dlL\,.
  \end{split}
\end{equation}
The \emph{height} $\Lcht L$ of a linear context $L$ is the distance
between its hole and its root, in other words $\Lcht\Echole=0$,
$\Lcht{\Abst xAL}=1+\Lcht L$, $\Lcht{\Lif dLPQ}=1+\Lcht L$ \emph{etc}.

\begin{remark}
  We have decorated the conditional and the let constructs with a type, which is
  intended to be the type of its last parameter(s).
  The only purpose of this decoration is to provide a type for the
  resulting $\Lzero$ in the linear reduction of $\Lif d\Lzero PQ$ and
  $\Llet dx\Lzero P$ in Figure~\ref{fig:lin-eq}.
  Most often, we will drop this type decoration which can easily be
  retrieved from the context.
\end{remark}

\begin{lemma}
  For any linear context \(L\) we have %
  \(L[0]\Trcl\Linred0\) and %
  \(L[M_0+M_1]\Trcl\Linred L[M_0]+L[M_1]\).
\end{lemma}
\begin{proof}
  By induction on \(\Lcht L\). If \(\Lcht L=0\) we use the fact that
  \(R\Trcl\Linred R\).
  Otherwise we have \(L=K[L']\) where \(\Lcht K=1\) and
  \(\Lcht{L'}=\Lcht L-1\).
  By inductive hypothesis %
  \(L'[M_0+M_1]\Trcl\Linred L'[M_0]+L'[M_1]\) and hence by definition
  of \(\Linred\) we have %
  \(L[M_0+M_1]\Trcl\Linred K[L'[M_0]+L'[M_1]]\) and by definition of
  \(\Linred\) again we have %
  \(K[L'[M_0]+L'[M_1]]\Linred L[M_0]+L[M_1]\). %
  The case of \(L[0]\) is similar.
\end{proof}

\begin{figure}
  \centering
  \begin{prooftree}
    \infer0{L[0]\Linred 0}
  \end{prooftree}
  \Treesep
  \begin{prooftree}
    \infer0{L[M_0+M_1]\Linred L[M_0]+L[M_1]}
  \end{prooftree}
  \Treesep
  \begin{prooftree}
    \hypo{M\Linred M'}
    \infer1{L[M]\Linred L[M']}
  \end{prooftree}
  \caption{Linear reduction, $L$ must be a linear context of height $1$.}
  \label{fig:lin-eq}  
\end{figure}

\begin{figure}
  \begin{center}
    \begin{prooftree}
      \hypo{i\in\Eset{1,\dots,k}}
      \infer1[\Trvar]{\Tseq{(x_1:A_1,\dots,x_k:A_k)}{x_i}{A_i}}
    \end{prooftree}
    \labeltext{$\Trvar$}{rl:trvar}
    \Treesep
    \begin{prooftree}
      \hypo{\Tseq{\Gamma,x:A}{M}{B}}
      \infer1[\Trabs]{\Tseq{\Gamma}{\Abst xAM}{\Timpl AB}}
    \end{prooftree}
    \labeltext{$\Trabs$}{rl:trabs}
  \end{center}
  \begin{center}
    \begin{prooftree}
      \hypo{\Tseq\Gamma M{\Timpl AB}}
      \hypo{\Tseq\Gamma NA}
      \infer2[\Trapp]{\Tseq\Gamma{\App MN}B}
    \end{prooftree}
    \labeltext{$\Trapp$}{rl:trapp}
    \Treesep
    \begin{prooftree}
      \hypo{\Tseq\Gamma M{\Timpl AA}}
      \infer1[\Trfix]{\Tseq\Gamma{\Lfix M}A}
    \end{prooftree}
    \labeltext{$\Trfix$}{rl:trfix}
    \Treesep
    \begin{prooftree}
      \hypo{n\in\Nat}
      \infer1[\Trnum]{\Tseq\Gamma{\Num n}\Tnat}
    \end{prooftree}
    \labeltext{$\Trnum$}{rl:trnum}
  \end{center}
  \begin{center}
    \begin{prooftree}
      \hypo{\Tseq\Gamma M{\Tdiffm d\Tnat}}
      \infer1[\Trsuc]{\Tseq\Gamma{\Lsucc dM}{\Tdiffm d\Tnat}}
    \end{prooftree}
    \labeltext{$\Trsuc$}{rl:trsuc}
    \Treesep
    \begin{prooftree}
      \hypo{\Tseq\Gamma M{\Tdiffm d\Tnat}}
      \infer1[\Trpred]{\Tseq\Gamma{\Lpred dM}{\Tdiffm d\Tnat}}
    \end{prooftree}
    \labeltext{$\Trpred$}{rl:trpred}
  \end{center}
  \begin{center}
    \begin{prooftree}
      \hypo{\Tseq\Gamma M{\Tdiffm d\Tnat}}
      \hypo{\Tseq\Gamma PA}
      \hypo{\Tseq\Gamma QA}
      \infer3[\Trif]{\Tseq\Gamma{\Lift AdMPQ}{\Tdiffm d A}}
    \end{prooftree}
    \labeltext{$\Trif$}{rl:trif}
    \Treesep
    \begin{prooftree}
      \infer0[\Trzero]{\Tseq\Gamma{\Lzerot A}A}
    \end{prooftree}
    \labeltext{$\Trzero$}{rl:trzero}
  \end{center}
  \begin{center}
    \begin{prooftree}
      \hypo{\Tseq\Gamma{M}{\Tdiffm{d+1}A}}
      \hypo{i\in\Eset{0,1}}
      \infer2[\Trprojo]{\Tseq\Gamma{\Lprojd idM}{\Tdiffm dA}}
    \end{prooftree}
    \labeltext{$\Trprojo$}{rl:trprojo}
    \Treesep
    \begin{prooftree}
      \hypo{\Tseq\Gamma M{\Tdiffm{d+1}A}}
      \infer1[\Trprojt]{\Tseq\Gamma{\Lplus{\Lprojd0dM}{\Lprojd1dM}}{\Tdiffm dA}}
    \end{prooftree}
    \labeltext{$\Trprojt$}{rl:trprojt}
  \end{center}
  \begin{center}
    \begin{prooftree}
      \hypo{\Tseq\Gamma{M_0+M_1}{\Tdiffm{d+1}A}}
      \infer1[\Trprojd]{\Tseq\Gamma{\Lprojd1d{M_0}+\Lprojd0d{M_1}}{\Tdiffm dA}}
    \end{prooftree}
    \labeltext{$\Trprojd$}{rl:trprojd}
    \Treesep
    \begin{prooftree}
      \hypo{\Tseq\Gamma MA}
      \hypo{M\Linred M'}
      \infer2[\Trlin]{\Tseq\Gamma{M'}A}
    \end{prooftree}
    \labeltext{$\Trlin$}{rl:trlin}
  \end{center}
  \begin{center}
    \begin{prooftree}
      \hypo{\Tseq\Gamma M{\Tdiffm dA}}
      \hypo{i\in\Eset{0,1}}
      \infer2[\Trinj]{\Tseq\Gamma{\Linjd idM}{\Tdiffm{d+1}A}}
    \end{prooftree}
    \labeltext{$\Trinj$}{rl:trinj}
    \Treesep
    \begin{prooftree}
      \hypo{\Tseq\Gamma M{\Tdiffm{d+2}A}}
      \infer1[\Trsum]{\Tseq\Gamma{\Lsumd dM}{\Tdiffm{d+1}A}}
    \end{prooftree}
    \labeltext{$\Trsum$}{rl:trsum}
    \Treesep
    \begin{prooftree}
      \hypo{\Tseq\Gamma M{\Tdiffm{d+l+2}A}}
      \infer1[\Trcirc]{\Tseq\Gamma{\Lflipdl dlM}{\Tdiffm{d+l+2}A}}
    \end{prooftree}
    \labeltext{$\Trcirc$}{rl:trcirc}
  \end{center}
  \begin{center}
    \begin{prooftree}
      \hypo{\Tseq\Gamma M{\Timpl AB}}
      \infer1[\Trdiff]{\Tseq\Gamma{\Ldiff M}{\Timpl{\Tdiff A}{\Tdiff B}}}
    \end{prooftree}
    \labeltext{$\Trdiff$}{rl:trdiff}
    \Treesep
    \begin{prooftree}
      \hypo{\Tseq\Gamma M{\Tdiffm d\Tnat}}
      \hypo{\Tseq{\Gamma,x:\Tnat}NB}
      \infer2[\Trlet]{\Tseq\Gamma{\Llett AdxMN}{\Tdiffm dB}}
    \end{prooftree}
    \labeltext{$\Trlet$}{rl:trlet}
  \end{center}
  \caption{Typing rules}
  \label{fig:typing-rules}
\end{figure}
One should think of a term of type $\Tdiffm kA$ as a complete binary
tree of height $k$ whose leaves have type $A$.
In constructs such as $\Lsucc dM$, the integer $d$ represents the
``depth'' at which the corresponding operation is performed in a tree
of type $\Tdiffm kA$ with $k\geq d$.
The main intuitive feature of such a tree is that its leaves are
summable. %
When $d=0$ we often drop the superscript.

We provide a typing system in Figure~\ref{fig:typing-rules} allowing
one to prove typing judgments $\Tseq\Gamma MA$.
Notice that in general, when %
$\Tseq\Gamma{N_0}A$ and $\Tseq\Gamma{N_1}A$, %
it is not necessarily true that %
$\Tseq\Gamma{N_0+N_1}A$.

Some examples of terms are provided in Section~\ref{sec:term-examples}
together with their relational semantics.
More examples will be provided in forthcoming articles.

\begin{remark}
  We use the notation \(\Ldiff M\) for the syntactic differentiation
  of a term \(M\) instead of \(\mathsf D M\) as in an earlier version
  of this paper, and similarly for types.
  In that way we avoid a clash with the notations used in the
  (categorical and additive) setting of cartesian differential
  categories for denoting an operation which, in our syntax, would
  correspond to \(\Lproj 1{\Ldiff M}\).
  We use fully compatible notations in the final version
  of~\cite{Ehrhard23a}.
\end{remark}

\subsection{Differential}
Given a variable $x$ and a term $N$, we define a term %
$\Ldletv xM$ in Figure~\ref{fig:diff-subst} which is called the
\emph{differential} of \(M\) with respect to \(x\).

As it is usual in the \(\lambda\)-calculus, this definition requires
some \(\alpha\)-conversions to be performed on the fly.
\begin{figure}
  \begin{align*}
    \Ldletv xy
    &=
      \begin{cases}
        x & \text{if }y=x\\
        \Linj0 y & \text{otherwise}
      \end{cases}
          &%
            \Ldletv x{\Abst yBP}&=\Abst yB{\Ldletv xP}
    \text{ if }x\not=y\\
    \Ldletv x{\Ldiff M}
    &=\Lflip{\Ldiff{\Ldletv xM}}
          &%
            \Ldletv x{\App PQ}
                 &={\App{\Lsum{\Ldiff{\Ldletv xP}}}{\Ldletv xQ}}\\
    \Ldletv x{\Lfix M}
    &=\Lfix{(\Lsum{\Ldiff{\Ldletv xM}})}
          &%
            \Ldletv{x}{\Num n}&=\Linj0{\Num n}\\
    \Ldletv x{\Lsucc dM}&=\Lsucc{d+1}{\Ldletv xM}
          &%
            \Ldletv x{\Lpred dM}&=\Lpred{d+1}{\Ldletv xM}\\
    \Ldletv x{\Lif dMPQ}
    &=\Lsum{\Lflipl d{\Lif{d+1}{\Ldletv xM}{\Ldletv xP}{\Ldletv xQ}}}
      \hspace{-16em}
          &\\
    \Ldletv x{\Llet dyPQ}
    &=\Lsum{\Lflipl d{\Llet{d+1} y{\Ldletv xP}{\Ldletv xQ}}}
      \hspace{-16em}
          &\\
    \Ldletv x\Lzero
    &=\Lzero
          &\Ldletv x{\Lplus{M_0}{M_1}}
                 &=\Lplus{\Ldletv x{M_0}}{\Ldletv x{M_1}}\\
    \Ldletv x{\Lprojd idM}
    &=\Lprojd i{d+1}{\Ldletv xM}
          &%
            \Ldletv x{\Lsumd dM}
                 &=\Lsumd{d+1}{\Ldletv xM}\\
    \Ldletv x{\Linjd idM}&=\Linjd i{d+1}{\Ldletv xM}
          &%
            \Ldletv x{\Lflipdl dlM}
                 &=\Lflipdl{d+1}l{\Ldletv xM}
  \end{align*}
  \caption{Inductive definition of the differential of a term}
  \label{fig:diff-subst}  
\end{figure}

\begin{lemma}\label{lemma:ldlet-lin-context}
  Let \(L\) be a linear context. There is a linear context %
  \(\Ldletv xL\) %
  such that, for any term \(M\), we have %
  \(\Ldletv x{L[M]}=\Ldletv xL[\Ldletv xM]\).
\end{lemma}
\begin{proof}
  Simple analysis of the definition of \(\Ldletv xM\) in
  Figure~\ref{fig:diff-subst}.
  The definition of \(\Ldletv xL\) is given in
  Figure~\ref{fig:context-diff}.
\end{proof}


\begin{figure}
  \centering
  \begin{align*}
    \Ldletv x\Echole
    &=\Echole
    &
      \Ldletv x{\Abst yAL}
    &=\Abst yA{\Ldletv xL}
      \text{ if }x\not=y\\
    \Ldletv x{\Ldiff L}
    &=\Lflip{\Ldiff{\Ldletv xL}}
    &
      \Ldletv x{\App LN}
    &={\App{\Lsum{\Ldiff{\Ldletv xL}}}{\Ldletv xN}}\\
    \Ldletv x{\Lsucc dL}&=\Lsucc{d+1}{\Ldletv xL}
          &%
            \Ldletv x{\Lpred dL}&=\Lpred{d+1}{\Ldletv xL}\\
    \Ldletv x{\Lif dLPQ}
    &=\Lsum{\Lflipl d{\Lif{d+1}{\Ldletv xL}{\Ldletv xP}{\Ldletv xQ}}}
      \hspace{-16em}
          &\\
    \Ldletv x{\Llet dyLQ}
    &=\Lsum{\Lflipl d{\Llet{d+1} y{\Ldletv xL}{\Ldletv xQ}}}
      \hspace{-16em}
          &\\
    \Ldletv x{\Lprojd idL}
    &=\Lprojd i{d+1}{\Ldletv xL}
          &%
            \Ldletv x{\Lsumd dL}
                 &=\Lsumd{d+1}{\Ldletv xL}\\
    \Ldletv x{\Linjd idL}&=\Linjd i{d+1}{\Ldletv xL}
          &%
            \Ldletv x{\Lflipdl dlL}
                 &=\Lflipdl{d+1}l{\Ldletv xL}
  \end{align*}
  \caption{Inductive definition of the differential of a context}
  \label{fig:context-diff}
\end{figure}

\begin{lemma}\label{lemma:linred-in-context}
  If \(R\Linred R'\) and \(L\) is a linear context then %
  \(L[R]\Linred L[R']\). %
  We also have \(L[0]\Trcl\Linred 0\) and
  \(L[R_0+R_1]\Trcl\Linred L[R_0]+L[R_1]\).
\end{lemma}
\begin{proof}
  Straightforward inductions on \(\Lcht L\).
\end{proof}

\begin{lemma}\label{lemma:ldlet-linred}
  If $R\Linred R'$ then %
  $\Ldletv xR\Trcl\Linred\Ldletv x{R'}$.
\end{lemma}
\begin{proof}
  By induction on the derivation of \(R\Linred R'\).
  Assume that \(R=L[R_0+R_1]\) and \(R'=L[R_0]+L[R_1]\) with
  \(\Lcht L=1\). Using Lemma~\ref{lemma:ldlet-lin-context}, we have %
  \begin{align*}
    \Ldletv xR
    &=\Ldletv xL[\Ldletv x{R_0+R_1}]\\
    &=\Ldletv xL[\Ldletv x{R_0}+\Ldletv x{R_1}]\\
    &\Trcl\Linred\Ldletv xL[\Ldletv x{R_0}]+\Ldletv xL[\Ldletv x{R_1}]\\
    &=\Ldletv x{L[R_0]}+\Ldletv x{L[R_1]}\\
    &=\Ldletv x{L[R_0]+L[R_1]}
  \end{align*} %
  by Lemma~\ref{lemma:linred-in-context}. %

  Assume now that \(R=L[M]\), \(R'=L[M']\) and \(M\Linred M'\). By
  inductive hypothesis we know that %
  \(\Ldletv xM\Trcl\Linred\Ldletv x{M'}\).
  We have %
  \(\Ldletv xR=\Ldletv xL[\Ldletv xM]\) and %
  \(\Ldletv x{R'}=\Ldletv xL[\Ldletv x{M'}]\) %
  by Lemma~\ref{lemma:ldlet-lin-context} and hence %
  \(\Ldletv xR\Trcl\Linred\Ldletv x{R'}\) by
  Lemma~\ref{lemma:linred-in-context}.
\end{proof}

\begin{lemma}\label{lemma:ldlet-typing}
  If $\Tseq{\Gamma,x:A}{M}{B}$ %
  then $\Tseq{\Gamma,x:\Tdiff A}{\Ldletv xM}{\Tdiff B}$.
\end{lemma}
\begin{proof}
  We consider the following cases, the others are left to the reader.
  
  \Proofcase %
  Assume first that $M=\Lif dP{Q_0}{Q_1}$ and that the last typing
  rule is~\ref{rl:trif} 
  so that $\Tseq{\Gamma,x:A}{P}{\Tdiffm d\Tnat}$ and
  $\Tseq{\Gamma,x:A}{Q_i}{B}$ for $i=0,1$. By inductive hypothesis we
  have %
  $\Tseq{\Gamma,x:\Tdiff A}{\Ldletv xP}{\Tdiffm{d+1}\Tnat}$ and %
  $\Tseq{\Gamma,x:\Tdiff A}{\Ldletv x{Q_i}}{\Tdiff{B}}$ %
  for $i=0,1$. Applying the rule~\ref{rl:trif} 
  we get
  $\Tseq{\Gamma,x:\Tdiff A} {\Lif{d+1}{\Ldletv xP}{\Ldletv
      x{Q_0}}{\Ldletv x{Q_1}}} {\Tdiffm{d+2}B}$ and hence we have
  $\Tseq{\Gamma,x:\Tdiff A}{\Lflipl d{\Lif{d+1}{\Ldletv xP}{\Ldletv
        x{Q_0}}{\Ldletv x{Q_1}}}}{\Tdiffm{d+2}B}$. %
  Therefore
  \[
    \Tseq{\Gamma,x:\Tdiff A}{\Lsum{\Lflipl d{\Lif{d+1}{\Ldletv xP}{\Ldletv
            x{Q_0}}{\Ldletv x{Q_1}}}}}{\Tdiffm{d+1}B}\,.
  \]
  Notice finally that $\Tdiffm{d+1}B=\Tdiff{\Tdiffm dB}$ is exactly
  the type expected for $\Ldletv xM$ in that case.

  \Proofcase %
  Assume that $M=\Ldiff P$ and that the last typing rule %
  is~\ref{rl:trdiff} 
  so that $\Tseq{\Gamma,x:A}P{\Simpl CD}$,
  $\Tseq{\Gamma,x:A}{M}{\Simpl{\Tdiff C}{\Tdiff D}}$ and
  $B=\Simplp{\Tdiff C}{\Tdiff D}$. By inductive hypothesis we have
  $\Tseq{\Gamma,x:\Tdiff A}{\Ldletv xP}{\Simpl C{\Tdiff D}}$ and hence
  $\Tseq{\Gamma,x:\Tdiff A}{\Ldiff{\Ldletv xP}}{\Simpl{\Tdiff C}{\Tdiffm
      2D}}=\Tdiffm2{\Simplp{\Tdiff C}{D}}=\Tdiff B$. It follows that
  $\Tseq{\Gamma,x:\Tdiff A}
  {\Lflip{\Ldiff{\Ldletv xP}}}{\Tdiff B}$ as required.
  
  \Proofcase %
  Assume next that $M=\Abst yCP$ and that the last typing rule %
  is~\ref{rl:trabs} 
  so that $\Tseq{\Gamma,x:A,y:C}{P}{D}$ (and hence
  $B=\Timplp CD$). By inductive hypothesis
  $\Tseq{\Gamma,x:\Tdiff A,y:C}{\Ldletv xP}{\Tdiff D}$ and hence
  $\Tseq{\Gamma,x:\Tdiff A}
  {\Abst yB{\Ldletv xP}}{\Timplp C{\Tdiff D}=\Tdiff B}$
  as required.

  \Proofcase %
  Assume now that $M=\App PQ$ and that the last typing rule %
  is~\ref{rl:trapp} 
  with $\Tseq{\Gamma,x:A}{P}{\Timpl CB}$ and
  $\Tseq{\Gamma,x:A}{Q}C$. Then by inductive hypothesis we have
  $\Tseq{\Gamma,x:\Tdiff A}
  {\Ldletv xP}{\Tdiff{(\Timpl CB)}}=\Timplp{C}{\Tdiff B}$
  and $\Tseq{\Gamma,x:\Tdiff A}{\Ldletv xQ}{\Tdiff C}$. Therefore
  \[
  \Tseq{\Gamma,x:\Tdiff A}{\Ldiff{\Ldletv xP}}
  {\Timpl{\Tdiff C}{\Tdiffm2B}=\Tdiffm 2{\Timplp{\Tdiff C}B}}
  \]  and hence %
  \(
  \Tseq{\Gamma,x:\Tdiff A}{\Lsum{\Ldiff{\Ldletv xP}}}
  {\Timpl{\Tdiff C}{\Tdiff B}}
  \) so that %
  \(
  \Tseq{\Gamma,x:\Tdiff A}{\App{\Lsum{\Ldiff{\Ldletv xP}}}{\Ldletv xQ}}
  {\Tdiff B}
  \)
  by the rules~\ref{rl:trsum} 
  and~\ref{rl:trapp}. 

  \Proofcase %
  Assume that $M=\Lflipdl dlP$ and that the last typing rule %
  is~\ref{rl:trcirc} 
  with $\Tseq{\Gamma,x:A}P{\Tdiffm{l+d+2}C=B}$. By inductive
  hypothesis we have %
  $\Tseq{\Gamma,x:\Tdiff A}{\Ldletv xP}{\Tdiffm{l+d+3}C}$ and hence %
  $\Tseq{\Gamma,x:\Tdiff A}
  {\Lflipdl{d+1}l{\Ldletv xP}}{\Tdiffm{l+d+3}C=\Tdiff B}$ by
  applying the rule~\ref{rl:trcirc}. 
  
  \Proofcase %
  Assume that $M=\Lfix P$ and that the last typing rule %
  is~\ref{rl:trfix} 
  with $\Tseq{\Gamma,x:A}P{\Timpl BB}$ so that
  $\Tseq{\Gamma,x:\Tdiff A}{\Ldletv xP}{\Timpl B{\Tdiff B}}$ and hence
  $\Tseq{\Gamma,x:\Tdiff A} {\Ldiff{\Ldletv xP}}{\Timpl{\Tdiff
      B}{\Tdiffm2B}}$ by~\ref{rl:trdiff} 
  and therefore
  $\Tseq{\Gamma,x:\Tdiff A}{\Lsum{\Ldiff{\Ldletv xP}}}{\Timpl{\Tdiff
      B}{\Tdiff B}}$ by~\ref{rl:trsum} 
  and finally
  $\Tseq{\Gamma,x:\Tdiff A}{\Lfix{(\Lsum{\Ldiff{\Ldletv xP}})}}{\Tdiff
    B}$ by~\ref{rl:trfix}. 

  \Proofcase %
  Assume that the last typing rule is~\ref{rl:trprojo} 
  meaning that we have
  $M=\Lprojd idP$ and $B=\Tdiffm dC$ with %
  $\Tseq{\Gamma,x:A}{P}{\Tdiffm{d+1}C}$. %
  Then by inductive hypothesis we have %
  $\Tseq{\Gamma,x:\Tdiff A}{\Ldletv x{P}}{\Tdiffm{d+2}C}$ and hence %
  $\Tseq{\Gamma,x:\Tdiff A}
  {\Lprojd i{d+1}{\Ldletv x{P}}}{\Tdiffm{d+1}C}$ by the
  rule~\ref{rl:trprojo}. 

  \Proofcase %
  Assume that the last typing rule is~\ref{rl:trprojt} 
  so that %
  $M=\Lprojd 0dP+\Lprojd 1dP$ and $B=\Tdiffm dC$ with %
  $\Tseq{\Gamma,x:A}{P}{\Tdiffm{d+1} C}$. %
  By inductive hypothesis we have %
  $\Tseq{\Gamma,x:\Tdiff A}{\Ldletv xP}{\Tdiffm{d+2}B}$ and hence %
  $\Tseq{\Gamma,x:\Tdiff A} {\Lprojd 0{d+1}{\Ldletv xP}+\Lprojd
    1{d+1}{\Ldletv xP}} {\Tdiffm{d+1}C}$ by the %
  rule~\ref{rl:trprojt}. 
  That is %
  $\Tseq{\Gamma,x:\Tdiff A}{\Ldletv xM}{\Tdiff B}$ as expected.

  \Proofcase %
  Assume that the last typing rule is~\ref{rl:trprojd} 
  so that %
  $M=\Lprojd0d{P_0}+\Lprojd1d{P_1}$ with %
  $\Tseq{\Gamma,x:A}{P_0+P_1}{\Tdiffm{d+1}C}$ and $B=\Tdiffm dC$. %
  By inductive hypothesis we have %
  \[
    \Tseq{\Gamma,x:\Tdiff A}
    {\Ldletv x{P_0}+\Ldletv x{P_1}}{\Tdiffm{d+2}C}
  \] and hence %
  $\Tseq{\Gamma,x:\Tdiff A}{\Lprojd 0{d+1}{\Ldletv x{P_0}}
    +\Lprojd 1{d+1}{\Ldletv x{P_1}}}{\Tdiffm{d+1}C}$ that is %
  $\Tseq{\Gamma,x:\Tdiff A}{\Ldletv xM}{\Tdiff B}$ as expected.
  
  \Proofcase %
  Assume that the last typing rule is~\ref{rl:trlin} 
  so that %
  $\Tseq{\Gamma,x:A}PB$ and $P\Linred M$. %
  By inductive hypothesis %
  $\Tseq{\Gamma,x:\Tdiff A}{\Ldletv xP}{\Tdiff B}$ and we have %
  $\Ldletv xP\Trcl\Linred\Ldletv xM$ by Lemma~\ref{lemma:ldlet-linred}
  and hence %
  $\Tseq{\Gamma,x:\Tdiff A}{\Ldletv xM}{\Tdiff B}$ by the %
  rule~\ref{rl:trlin}. 
\end{proof}

\subsection{Reduction rules}
We define a rewriting system \(\Lang\). The elements of \(\Rsca\Lang\)
are the terms of the syntax introduced above. 
The main reduction rules are given in Figure~\ref{fig:main-red-rules}.
\renewcommand\Red{\Rel{\Rsred\Lang}}
\begin{figure}
  \begin{align*}
    \App{\Abst xAM}{N}
    &\Red\Subst MNx
    &\Ldiffp{\Abst xAM}
    &\Red\Abst x{\Tdiff A}{\Ldletv xM}\\
    \Lsucc0{\Num n}&\Red\Num{n+1}
    &\Lpred0{\Num 0}&\Red\Num 0\\
    \Lpred0{\Num{n+1}}&\Red\Num n
    &\Lif0{\Num 0}{P}{Q}&\Red P\\
    \Lif0{\Num{n+1}}PQ&\Red Q
    &\Llet 0x{\Num n}P&\Red\Subst P{\Num n}x\\
    \Lfix P&\Red\App P{\Lfix P}
  \end{align*}
  \caption{Main reduction rules}
  \label{fig:main-red-rules}  
\end{figure}
A second series of reduction rules given in
Figure~\ref{fig:proj-red-rules} specifies how the projections
$\Lprojd jdM$ interact with the other constructs. They are crucially
used for ``reading'' the result of a computation by accessing leaves
of a ``tree'' of type $\Tdiffm dA$ (complete binary tree of height
$d$; the leaves can themselves be trees if $A=\Tdiffm eB$ with
$e>0$).

\begin{remark}
  The two rules
  $\Lprojd id{\Lprojd jeM}\Red\Lprojd j{e-1}{\Lprojd idM}$ %
  if $d<e$ and %
  $\Lprojd id{\Lprojd jeM}\Red\Lprojd j{e}{\Lprojd i{d+1}M}$ %
  if $e\leq d$ lead clearly to infinite sequences of computations so
  it would be tempting to remove one of them from the rewriting
  system. However both seem necessary in order to prove the soundness
  of the stack machine that we introduce in Section~\ref{sec:stacks}.
\end{remark}

We also need the reduction rule
\begin{center}
  \begin{prooftree}
    \hypo{M\Linred M'}
    \infer1{M\Red M'}
  \end{prooftree}
\end{center}

\begin{figure}
  {\footnotesize
    \begin{align*}
    \Lprojd id{\Abst xAM}
    &\Red\Abst xA{\Lprojd idM}
    &%
      \Lprojd id{\App MN}
    &\Red \App{\Lprojd idM}N\\
    \Lprojd id{\Lsucc{e} M}&\Red\Lsucc{e-1}{\Lprojd idM}
                             \text{\quad if }d<e
    &%
      \Lprojd id{\Lpred{e}M}&\Red\Lpred{e-1}{\Lprojd idM}
                              \text{\quad if }d<e\\
    \Lprojd id{\Lif{e} MPQ}
    &\Red\Lif{e-1}{\Lprojd idM}PQ
      \text{\quad if }d<e
    &%
      \Lprojd id{\Lif eMPQ}
    &\Red\Lif{e}M{\Lprojd i{d-e}P}{\Lprojd i{d-e}Q}
    \\
    &&&\text{\quad\quad\quad if }e\leq d\\
    \Lprojd id{\Llet{e}xMP}
    &\Red\Llet{e-1}x{\Lprojd idM}P
      \text{\quad if }d<e
    &%
      \Lprojd id{\Llet{e}xMP}&\Red \Llet exM{\Lprojd i{d-e}P}
                               \text{\quad if }e\leq d\\
    \Lprojd 0d{\Lsumd dM}
    &\Red \Lprojd0d{\Lprojd0d M}
    &%
      \Lprojd 1d{\Lsumd dM}
    &\Red \Lprojd 1d{\Lprojd 0dM}+\Lprojd 0d{\Lprojd 1dM}\\
    \Lprojd id{\Lsumd eM}
    &\Red \Lsumd{e-1}{\Lprojd idM}
      \text{\quad if }d<e
    &%
      \Lprojd id{\Lsumd eM}
    &\Red \Lsumd{e}{\Lprojd i{d+1}M}
      \text{\quad if }e<d\\
    &&\hspace{-18em}\Lprojd{i_{l+1}}d{\cdots\Lprojd{i_0}d{\Lflipdl dlM}}
       \Red\Lprojd{i_0}d{\Lprojd{i_{l+1}}d{\cdots\Lprojd{i_1}d{M}}}
    \\
    \Lprojd id{\Lflipdl elM}&\Red\Lflipdl{e-1}l{\Lprojd idM}
                              \text{\quad if }d<e
    &%
      \Lprojd id{\Lflipdl elM}&\Red\Lflipdl{e}l{\Lprojd idM}
                                \text{\quad if }e+l+2\leq d
    \\
    \Lprojd id{\Linjd jdM}
    &\Red M\text{\quad if }i=j
    &%
      \Lprojd id{\Linjd jdM}
    &\Red 0\text{\quad if }i\not=j\\
    \Lprojd id{\Linjd jeM}
    &\Red \Linjd j{e-1}{\Lprojd idM}\text{\quad if }d<e
    &%
      \Lprojd id{\Linjd jeM}
    &\Red \Linjd j{e}{\Lprojd i{d-1}M}\text{\quad if }e<d\\
    \Lprojd id{\Lprojd jeM}&\Red\Lprojd j{e-1}{\Lprojd idM}
                             \text{\quad if }d<e
    &\Lprojd id{\Lprojd jeM}&\Red\Lprojd j{e}{\Lprojd i{d+1}M}
                              \text{\quad if }e\leq d\\
    \Lprojd i{d+1}{\Ldiff M}&\Red\Ldiff{\Lprojd idM}
    &\Ldiff{\Lprojd idM}&\Red\Lprojd i{d+1}{\Ldiff M}
  \end{align*}}
  \caption{Projection reduction rules}
  \label{fig:proj-red-rules}    
\end{figure}

\begin{example}
  Let \(T=\Abst{f}{\Timpl AA}{\Abst xA{\App f{\App fx}}}\) so that
  \begin{align*}
    \Tseq{}{\Ldiff T}{\Timpl{(\Timpl A{\Tdiff A})}{\Timpl A{\Tdiff A}}}\,,
  \end{align*}
  we have
  \begin{align*}
    \Ldiff T
    &\Rel\Red
      \Abst{f}{\Timpl{A}{\Tdiff A}}{\Ldletv f{\Abst xA{\App f{\App fx}}}}\\
    &=\Abst{f}{\Timpl{A}{\Tdiff A}}{\Abst xA{\Ldletv f{\App f{\App fx}}}}\\
    &=\Abst{f}{\Timpl{A}{\Tdiff A}}{\Abst xA{
      \App{\Lsum{\Ldiff{\Ldletv ff}}}{\Ldletv{f}{\App fx}}}}\\
    &=\Abst{f}{\Timpl{A}{\Tdiff A}}{\Abst xA{
      \App{\Lsum{\Ldiff f}}{\App{\Lsum{\Ldiff f}}{\Ldletv fx}}}}\\
    &=\Abst{f}{\Timpl{A}{\Tdiff A}}{\Abst xA{
      \App{\Lsum{\Ldiff f}}{\App{\Lsum{\Ldiff f}}{\Linj0 x}}}}
  \end{align*}
  which shows that, contrarily to what happens in the differential
  \(\lambda\)-calculus, even if the variable \(f\) occurs twice, no
  actual sum is created during this computation of the differential of
  \(T\).
\end{example}

\subsubsection{Reducing sums, and the evaluation contexts} %
\label{sec:reducing-sums} %
These reduction rules can be applied almost anywhere in a term (taking
care as usual of not binding free variables of $N$ in the ordinary
substitution $\Subst MNx$ and in the differential substitution
$\Ldletv xM$).

However, in order to make the proof of subject reduction possible, we
forbid reductions within subterms of the shape $M_0+M_1$.
Indeed by the very nature of the coherence we want to implement in
this programming language, we have provided very restricted ways to
type sums.
For that reason allowing one for instance to reduce $M_0$ to some
$M'_0$ by performing, say, a $\beta$-reduction would lead to a term
$M'_0+M_1$ whose typeability is not at all obvious (imagine for
instance that $M_i=\Lproj iM$ for some $M$ such that
$\Tseq\Gamma M{\Tdiff A}$).
One option would be to develop a theory of ``parallel'' reductions
generalizing the observation that in the example at hand the
$\beta$-reduction performed in $M_0$ is also available in $M_1$
because both come from the same term $M$.
This kind of approach
will certainly be developed in further work.
For the time being we adopt a much simpler and conservative approach.
So here is the syntax of our evaluation contexts:
\begin{align*}
  E
  \Bnfeq \Echole
  & \Bnfor \Abst xAE
    \Bnfor \App EN
    \Bnfor \App ME
    \Bnfor \Lfix E
    \Bnfor \Lsucc dE
    \Bnfor \Lpred dE
  \\
  &  \Bnfor \Lif dEPQ
    \Bnfor \Lif dMEQ
    \Bnfor \Lif dMPE
    \Bnfor \Llet dxEP
    \Bnfor \Llet dxME
  \\
  & \Bnfor \Ldiff E
    \Bnfor \Lprojd idE
    \Bnfor \Linjd idE
    \Bnfor \Lsumd dE
    \Bnfor \Lflipdl dlE
\end{align*}
and the associated inference rule is as usual
\begin{center}
  \begin{prooftree}
    \hypo{M\Red M'}
    \infer1{\Ecfilled EM\Red\Ecfilled E{M'}}
  \end{prooftree}
\end{center}

We will need however to perform reduction within sums at some point
otherwise our computations will remain stuck for artificial
reasons. So we do allow such reductions but only at ``toplevel'': this
is precisely the purpose of the associated rewriting system %
\(\Msrs\Lang\) defined in Section~\ref{sec:ms-rewriting}.

\subsubsection{Term multiset typing}
Let \(S=\Mset{\List M1k}\in\Rsca{\Msrs\Lang}\), \(\Gamma\) be a
context and \(A\) be a type. We write \(\Tseq\Gamma SA\) if
\(\Tseq\Gamma{M_i}A\) for \(i=1,\dots,k\). This notion of typing for
multisets (which represent sums of terms) is quite weak: %
\(\Tseq\Gamma{\Mset{M_0,M_1}}A\) does not imply
\(\Tseq\Gamma{M_0+M_1}A\).
\renewcommand\Topred{\Rel{\Rsred{\Msrs\Lang}}} %
It is only for that reason that we will be able to prove subject
reduction for \(\Topred\). %
This is not really an issue because we will prove that the semantics
is invariant by reduction (including the $\Topred$ reduction) ---~a
property which is called as usual \emph{soundness}--- so we know that
actually the terms that we obtain by performing the $\Topred$
reduction belong to the expected type even if we are not necessarily
able to prove it syntactically.

\subsubsection{Subject reduction}

\begin{lemma}\label{lemma:typing-subst-lemma}
  If $\Tseq{\Gamma,x:A}MB$ and $\Tseq\Gamma NA$ then
  $\Tseq\Gamma{\Subst MNx}B$.
\end{lemma}
\begin{proof}
  Straightforward induction on the typing derivation of $M$.
\end{proof}

\begin{theorem}[Subject reduction]\label{th:subj-reduction}
  If $\Tseq\Gamma MA$ and $M\Red M'$ then $\Tseq\Gamma{M'}A$.
\end{theorem}
\begin{proof}
  The last possible typing rules for the derivation of %
  $\Tseq\Gamma MA$ cannot be any of the %
  rules~\ref{rl:trprojt}, 
  \ref{rl:trprojd} 
  or~\ref{rl:trlin} 
  since the reduction $\Red$ does not apply to sums and to
  $\Lzero$. For the remaining typing rules, observe that the typing
  system is syntax directed, we consider a few reductions. The proof
  is by induction on the derivation of $M\Red M'$.

  \Proofcase %
  Assume that $M=\Ldiff{(\Abst xBN)}$ and
  $M'=\Abst x{\Tdiff B}{\Ldletv xN}$ so that the last typing rule %
  is~\ref{rl:trdiff}, 
  with $\Tseq{\Gamma,x:B}{N}{C}$ and hence
  $\Tseq\Gamma{\Abst xBN}{\Timpl BC}$ and
  $A=(\Timpl{\Tdiff B}{\Tdiff C})$. Then we have
  $\Tseq{\Gamma,x:\Tdiff B}{\Ldletv xN}{\Tdiff C}$ by
  Lemma~\ref{lemma:ldlet-typing}, and therefore
  $\Tseq\Gamma{\Abst x{\Tdiff B}{\Ldletv xN}}{\Timpl{\Tdiff B}{\Tdiff
      C}}$ by~\ref{rl:trabs}. 
  All the other reduction rules of Figure~\ref{fig:main-red-rules} are
  dealt with as usual in the typed $\lambda$-calculus, using
  Lemma~\ref{lemma:typing-subst-lemma}.

  The fact that if $\Tseq\Gamma MA$ and $M\Linred M'$ then
  $\Tseq\Gamma{M'}A$ is by a straightforward application of %
  rule~\ref{rl:trlin}. 

  So we consider now some of the rules of Figure~\ref{fig:proj-red-rules}.

  \Proofcase %
  Assume that $M=\Lprojd id{\Abst xBN}$ with
  $\Tseq{\Gamma,x:B}N{\Tdiffm{d+1}C}$ so that
  $A=(\Timpl B{\Tdiffm{d}C})=\Tdiffm{d}{(\Timpl BC)}$. Then we
  have $\Tseq{\Gamma,x:B}{\Lprojd idN}{\Tdiffm dC}$ and hence
  $\Tseq\Gamma{\Abst xB{\Lprojd idN}}{A}$. And on the other
  hand $\Tseq\Gamma{\Lprojd id{\Abst xBN}}{A}$.

  \Proofcase %
  Assume that $M=\Lprojd id{\App NP}$ with %
  $\Tseq\Gamma N{\Timpl B{\Tdiffm{d+1}C}=\Tdiff{(\Timpl B{\Tdiffm
        dC})}}$ %
  and $\Tseq\Gamma PB$ %
  so that $\Tseq\Gamma{\App NP}{\Tdiffm{d+1}C}$ and hence
  $\Tseq\Gamma{M}{A}$ where $A=\Tdiffm dC$, and on the other hand
  $\Tseq\Gamma{\Lprojd idN}{\Timpl B{\Tdiffm dC}}$ and hence
  $\Tseq\Gamma{\App{\Lprojd idN}{P}}{A}$.

  \Proofcase %
  Assume that $M=\Lprojd id{\Lsucc eN}$ with %
  $\Tseq\Gamma N{\Tdiffm e\Tnat}$. For $M$ to be typeable we need to have %
  $d<e$ and then $\Tseq\Gamma M{\Tdiffm{e-1}\Tnat=A}$. Then we have %
  $\Tseq\Gamma{\Lprojd idN}{\Tdiffm{e-1}\Tnat}$ and hence %
  $\Tseq\Gamma{\Lsucc{e-1}{\Lprojd idN}}A$.

  \Proofcase %
  Assume that $M=\Lprojd id{\Lif eNPQ}$ and $d<e$, with
  $\Tseq\Gamma N{\Tdiffm e\Tnat}$ and $\Tseq\Gamma PB$ and
  $\Tseq\Gamma QB$ so that %
  $\Tseq\Gamma{\Lif eNPQ}{\Tdiffm eB}$ and hence %
  $\Tseq{\Gamma}{M}{\Tdiffm{e-1}B}$, which means that %
  $A=\Tdiffm{e-1}B$. On the other hand we have %
  $\Tseq{\Gamma}{\Lprojd idN}{\Tdiffm{e-1}\Tnat}$ and hence %
  $\Tseq{\Gamma}{\Lif{e-1}{\Lprojd idN}{P}{Q}}{\Tdiffm{e-1}B}$ as
  expected.

  \Proofcase %
  Assume that $M=\Lprojd id{\Lif eNPQ}$ and $e\leq d$, with %
  $\Tseq\Gamma N{\Tdiffm e\Tnat}$ and assume that $\Tseq\Gamma PB$ and
  $\Tseq\Gamma QB$ so that %
  $\Tseq{\Gamma}{\Lif eNPQ}{\Tdiffm eB=A}$. %
  For the term $M$ to be typeable, we need $B$ to be of
  shape $\Tdiffm{d-e+1}C$ for some (uniquely defined) type $C$ so that
  $A=\Tdiffm{d+1}C$ and hence $\Tseq\Gamma{\Lprojd idM}{\Tdiffm
    dC}$. On the other hand we have
  $\Tseq{\Gamma}{\Lprojd i{d-e}P}{\Tdiffm{d-e}C}$ and hence
  $\Tseq{\Gamma}{\Lif e{N}{\Lprojd i{d-e}P}{\Lprojd i{d-e}Q}}{\Tdiffm
    dC}$ as expected.

  \Proofcase %
  Assume that $M=\Lprojd 0d{\Lsumd dN}$ with %
  $\Tseq\Gamma{N}{\Tdiffm{d+2}C}$, and hence %
  $\Tseq\Gamma{\Lsumd dN}{\Tdiffm{d+1}C}$ and therefore %
  $\Tseq\Gamma{\Lprojd 0d{\Lsumd dN}}{C}$, so that $A=\Tdiffm dC$. %
  Then we have %
  $\Tseq\Gamma{\Lprojd 0dN}{\Tdiffm{d+1}C}$ by~\ref{rl:trprojo} 
  and hence
  $\Tseq\Gamma{\Lprojd0d{\Lprojd0d N}}A$ by~\ref{rl:trprojo} 
  again.
  
  \Proofcase %
  Assume that $M=\Lprojd 1d{\Lsumd dN}$ with
  $\Tseq\Gamma{N}{\Tdiffm{d+2}C}$, and hence %
  $\Tseq\Gamma{\Lsumd dN}{\Tdiffm{d+1}C}$ and therefore %
  $\Tseq\Gamma{\Lprojd 1d{\Lsumd dN}}{C}$, so that $A=\Tdiffm dC$. %
  Then we have %
  $\Tseq\Gamma{\Lprojd 0dN+\Lprojd 1dN}{\Tdiffm{d+1}C}$ %
  by~\ref{rl:trprojt} 
  and hence
  $\Tseq\Gamma{\Lprojd1d{\Lprojd0d N}+\Lprojd0d{\Lprojd1d N}}A$ %
  by~\ref{rl:trprojd}. 

  \Proofcase %
  Assume that $M=\Lprojd id{\Lsumd eN}$ with $d<e$. So we must have
  $\Tseq\Gamma N{\Tdiffm{e+2}A}$ so that
  $\Tseq\Gamma{\Lsumd
    eN}{\Tdiffm{e+1}A=\Tdiffm{d+1}{\Tdiffm{e-d}A}}$ and hence
  $\Tseq\Gamma{\Lprojd id{\Lsumd eN}}{\Tdiffm eA}$. We have
  $\Tseq\Gamma{\Lprojd idN}{\Tdiffm{e+1}A}$ and hence
  $\Tseq\Gamma{\Lsumd{e-1}{\Lprojd idN}}{\Tdiffm eA}$ as required.

  \Proofcase %
  Assume that $M=\Lprojd id{\Lsumd eN}$ with $e<d$. We must have
  $\Tseq\Gamma N{\Tdiffm{e+2}A}$ so that
  $\Tseq\Gamma{\Lsumd eN}{\Tdiffm{e+1}A}$ and for
  $\Lprojd id{\Lsumd eN}$ to be typeable we need $A$ to be of shape
  $\Tdiffm{d-e}B$ (for a uniquely defined type $B$) so that
  $\Tseq\Gamma{\Lsumd eN}{\Tdiffm{d+1}B}$ and hence
  $\Tseq\Gamma{M=\Lprojd id{\Lsumd eN}}{\Tdiffm dB=\Tdiffm eA}$. We
  have $\Tseq\Gamma N{\Tdiffm {d+2}B}$ and hence
  $\Tseq\Gamma{\Lprojd i{d+1}N}{\Tdiffm{d+1}B}$ and therefore (using
  the fact that $d>0$)
  $\Tseq\Gamma{\Lsumd e{\Lprojd i{d+1}N}}{\Tdiffm dB=\Tdiffm eA}$.

  \Proofcase %
  Assume that %
  $M=\Lprojd{i_{l+1}}d{\cdots\Lprojd{i_0}d{\Lflipdl dlN}}$ with %
  $\Tseq\Gamma N{\Tdiffm{l+d+2}C}$ and hence %
  $\Tseq\Gamma M{\Tdiffm{d}C=A}$ by~\ref{rl:trcirc} 
  and $l+2$ applications of~\ref{rl:trprojo}. 
  Then we have %
  $\Tseq\Gamma{\Lprojd{i_0}d{\Lprojd{i_{l+1}}d{\cdots\Lprojd{i_1}d{N}}}}A$
  by $l+2$ applications of~\ref{rl:trprojo}. 

  \Proofcase %
  Assume that %
  $M=\Lprojd id{\Lflipdl elN}$ with $d<e$ so that we have %
  $\Tseq\Gamma N{\Tdiffm{e+l+2}C}$ for a type $C$ such that %
  $A=\Tdiffm{e+l+1}C$. Then we have %
  $\Tseq\Gamma{\Lprojd idN}{\Tdiffm{e+l+1}C}$ (since $d<e+l+1$) and hence %
  $\Tseq\Gamma{\Lflipdl {e-1}l{\Lprojd idN}}A$ since %
  $e>0$ and hence $e+l+1=(e-1)+l+2$.

  \Proofcase %
  Assume that %
  $M=\Lprojd id{\Lflipdl elN}$ with $e+l+2\leq d$ so that we have %
  $\Tseq\Gamma N{\Tdiffm{e+l+2}C}$ for a type $C$ such that %
  $A=\Tdiffm{e+l+1}C$, and moreover %
  $A=\Tdiffm dD$ for some type $D$, meaning that %
  $C=\Tdiffm{d-e-l-1}D$ (remember that $d-e-l-1>0$). Then we have %
  $\Tseq\Gamma{\Lprojd idN}{\Tdiffm{e+l+1}C=A}$
  by~\ref{rl:trprojo} 
  and hence %
  $\Tseq\Gamma{\Lflipdl el{\Lprojd idN}}A$ %
  by~\ref{rl:trcirc} 
  which can be applied since %
  $A=\Tdiffm{e+l+2}{\Tdiffm{d-e-l-2}{D}}$.

  The remaining cases are similar.
\end{proof}

Given a derivation $\delta$ in the typing system we use
$\Tdersize\delta$ for the number of inference rule occurrences $\delta$
contains.
\begin{lemma}\label{lemma:ty-der-sum}
  Let $\delta$ be a typing derivation of %
  $\Tseq\Gamma{\Lplus{M_0}{M_1}}A$. %
  For $j=0,1$, there is a derivation $\delta_j$ of %
  $\Tseq\Gamma{M_j}A$ such that %
  $\Tdersize{\delta_j}\leq\Tdersize{\delta}$.
\end{lemma}
\begin{proof}
  By induction on $\delta$. The following cases can arise.

  \Proofcase %
  The last rule of $\delta$ is~\ref{rl:trprojt} 
  so that %
  $M_j=\Lprojd jdM$, $A=\Tdiffm dB$ and %
  $\Tseq\Gamma M{\Tdiffm{d+1}B}$ by a derivation $\delta'$ such that %
  $\Tdersize{\delta'}=\Tdersize\delta-1$ that we can extend with a
  rule~\ref{rl:trprojo} 
  to get the required derivation $\delta_j$ of %
  $\Tseq\Gamma{M_j}A$ which satisfies %
  $\Tdersize{\delta_j}=\Tdersize{\delta'}+1=\Tdersize\delta$.

  \Proofcase %
  The last rule of $\delta$ is~\ref{rl:trprojd} 
  so that %
  $M_0=\Lprojd 1d{N_0}$, $M_1=\Lprojd 0d{N_1}$, %
  $A=\Tdiffm dB$ and %
  $\Tseq\Gamma{N_0+N_1}{\Tdiffm{d+1}{B}}$ by a derivation %
  $\delta'$ such that %
  $\Tdersize{\delta'}=\Tdersize{\delta}-1$ %
  and hence by inductive hypothesis, for $j=0,1$, we have a
  derivation %
  $\delta'_j$ of $\Tseq\Gamma{N_j}{\Tdiffm{d+1}B}$ such that %
  $\Tdersize{\delta'_j}\leq\Tdersize{\delta'}$ that we can extend with
  a~\ref{rl:trprojo} 
  rule to get a derivation %
  $\delta_j$ of %
  $\Tseq\Gamma{M_j}A$. We have
  $\Tdersize{\delta_j}=\Tdersize{\delta'_j}+1\leq\Tdersize\delta$.

  \Proofcase %
  The last rule of $\delta$ is~\ref{rl:trlin} 
  so that there is a linear
  context $L$ of height $1$ and terms $N_0$, $N_1$ such that %
  $M_j=L[N_j]$ and %
  $\Tseq\Gamma{L[N_0+N_1]}A$ by a derivation $\delta'$ such that %
  $\Tdersize{\delta'}=\Tdersize\delta-1$. %
  This implies (by a simple inspection of the various possibilities for
  $L$ which has height $1$) that for some context $\Delta$ and some
  type $B$ one has %
  $\Tseq\Delta{N_0+N_1}B$ by a derivation $\delta''$ such that %
  $\Tdersize{\delta''}=\Tdersize{\delta'}-k_L$ where %
  $k_L\in\Natnz$ \emph{depends only on $L$} (if for instance
  $L=\Lif d{\Echole}{P_0}{P_1}$ then $k_L=1+k_0+k_1$ where $k_i$ is
  the size of the typing derivation of $P_i$). %
  So that by inductive hypothesis we have derivations $\delta''_j$
  of %
  $\Tseq\Delta{N_j}B$ for $j=0,1$ such that %
  $\Tdersize{\delta''_j}\leq\Tdersize{\delta''}$. %
  We can extend $\delta''_j$ with exactly the typing rule associated
  with $L$ to get a derivation %
  $\delta_j$ of $\Tseq\Gamma{L[N_j]}A$ such that %
  $\Tdersize{\delta_j}=\Tdersize{\delta''_j}+k_L
  \leq\Tdersize{\delta''}+k_L
  =\Tdersize{\delta'}=\Tdersize{\delta}-1$ and hence %
  $\Tdersize{\delta_j}<\Tdersize\delta$.
\end{proof}

\begin{theorem}[Subject reduction for multisets] %
  \label{th:subj-reduction-ms}
  Assume that \(\Tseq\Gamma SA\) where %
  \(S\in\Rsca{\Msrs\Lang}\) and that %
  \(S\Topred S'\). Then %
  \(\Tseq\Gamma{S'}A\).
\end{theorem}
\begin{proof}
  The following cases are possible.

  \Proofcase %
  \(S=S_0+\Mset M\), \(M\Red 0\) and %
  \(S'=S_0\). We have \(\Tseq\Gamma{S'}{A}\) since all elements of
  \(S'\) belong to \(S\).

  \Proofcase %
  \(S=S_0+\Mset M\), \(M\Red M'\) and %
  \(S'=S_0+\Mset{M'}\). Then we have \(\Tseq\Gamma{M'}A\) by
  Theorem~\ref{th:subj-reduction} and hence \(\Tseq\Gamma{S'}A\).

  \Proofcase %
  \(S=S_0+\Mset M\), 
  \(M\Red M_0+M_1\) and %
  \(S'=S_0+\Mset{M_0,M_1}\). Since \(\Tseq\Gamma MA\), we have %
  \(\Tseq\Gamma{M_i}A\) for \(i=0,1\) by Lemma~\ref{lemma:ty-der-sum}
  and hence \(\Tseq\Gamma{S'}A\).
\end{proof}



\subsection{Examples of terms}
\label{sec:term-examples}
Given \(M\) such that %
\(\Tseq\Gamma M{\Timpl AB}\) we set %
\begin{align*}
  \Ldiffd M=\Lproj1{\Ldiff M}
\end{align*}
so that \(\Tseq\Gamma{\Ldiffd M}{\Timpl{\Tdiff A}{B}}\) and
\(\Ldiffd M\) can be understood as the derivative of \(M\) taking in
\(\Tdiff A\) an argument which combines the point at which it is
computed and the ``direction'' to which it is (linearly) applied.
Given moreover \(N\) such that \(\Tseq\Gamma NA\) we set
\begin{align*}
  \Linapp MN=\App{\Ldiffd M}{\Linj 1N}
\end{align*}
so that \(\Tseq\Gamma{\Linapp MN}B\).
This application of \(M\) to \(N\) will be ``successful'' only if
\(M\) uses its argument linearly. Otherwise, it will produce the
result \(\Lzero\).

\begin{remark}
  We can understand the term \(\Lzero\) (of any type) as an
  uncatchable exception which is raised when some constraint of
  linearity expressed by the use of the \(\Ldiff{}\) operator is not
  respected.
  In the semantics, it does not seem possible to distinguish
  \(\Lzerot A\) from \(\Lfix{\Abst xAx}\) (an ever-looping term) so
  that this exception cannot be caught \emph{within} the language,
  but, in the implementation, nothing forces us to enter into an
  infinite loop when a term reduces to \(\Lzero\): we can simply abort
  the execution of the program and return the information that it
  failed (and possibly more information about this failure: this a
  mere matter of implementation).

  Syntactically, whether these two kinds of zero have necessarily to
  be identified remains an open question, but since our approach is
  strongly based on denotational semantics and since none of our
  models allows to distinguish them, we will stick to this
  identification for the time being.
  Nevertheless we could even imagine that the execution of our
  ``differential programs'' is controlled by an external
  ``meta-program'' which is allowed to handle such exceptions.
\end{remark}

To try to provide some intuitions about the operational behavior of
our example terms, we use the intersection typing system of
Section~\ref{sec:inter-types-terms} even if we are aware that this
system cannot be fully understood by the reader at this stage. We
believe nevertheless that most of its rules convey an intuitive
meaning which can be useful to understand these terms.

The only typing derivations for \(\Ldiffd f\) in the intersection
typing system of Section~\ref{sec:inter-types-terms} (see in
particular Figure~\ref{fig:int-typing-rules}) lead to judgments of
shape
\begin{align*}
  \Tseqi{\Contz\Gamma,f:\Mset{(\Mset{\List a1k,a},b)}:\Timpl AB}
  {\Ldiffd f}{(\Mset{\Tseqact 0{a_1},\dots,\Tseqact 0{a_1},\Tseqact 1a},b)}
  {\Timpl{\Tdiff A}B}
\end{align*}
where the notation \(\Tseqact ia\) is defined in
Section~\ref{sec:REL-deduction}.

The only typing derivations for \(\Linapp fx\) in the intersection typing
system of Figure~\ref{fig:int-typing-rules} lead to judgments of shape
\begin{align*}
  \Tseqi{\Contz\Gamma,f:\Mset{(\Mset a,b)}:\Timpl AB,x:\Mset a:A}
  {\Linapp fx}{b}{B}
\end{align*}
where \(a\in\Tsemrel A\) and \(b\in\Tsemrel B\) whereas the typings of
\(\App fx\) are
\begin{align*}
  \Tseqi{\Contz\Gamma,f:\Mset{(m,b)}:\Timpl AB,x:m:A}
  {\App fx}{b}{B}
\end{align*}
for all \(m\in\Mfin{\Tsemrel A}\) and \(b\in\Tsemrel B\).

To illustrate this construction we define a term \(M\) such that %
\(\Tseq{}{M}{\Timpl{(\Timpl\Tnat\Tnat)}{\Timpl\Tnat\Tnat}}\) by
\begin{align*}
  M=\Lfix{\Abst F{\Timpl{(\Timpl\Tnat\Tnat)}{\Timpl\Tnat\Tnat}}
  {\Abst f{\Timpl\Tnat\Tnat}{\Abst x\Tnat
  {\Llet0yx{\Lif0y{\Num 0}{\Succ{\App F{f\Argsep\Linapp fy}}}}}}
  }}\,.
\end{align*}
The intuition is that \(\App M{f\Argsep{\Num\nu}}\) iterates
\(f:\Timpl\Tnat\Tnat\) starting from \(\nu\) until the value \(0\)
is reached, and returns the number of steps. Moreover, the program is
written in such a way that the argument \(f\) is used as a linear
morphism.
In the PCS model of Section~\ref{sec:PCS-definition} we can see \(M\)
as implementing a random walk in a Markov chain (on the integers)
represented by its first argument (which is used as a
\(\Nat\times\Nat\) sub-stochastic matrix thanks to the linear
application \(\Linapp fy\)).

In the intersection typing system of
Figure~\ref{fig:int-typing-rules}, we can derive
\begin{align*}
  \Tseqi{}
  {M}
  {(\Mset{(\Mset{\nu_1},\nu_2),(\Mset{\nu_2},\nu_3),\dots,
  (\Mset{\nu_{k}},0)},(\Mset{\nu_1},k))}
  {\Timpl{(\Timpl\Tnat\Tnat)}{\Timpl\Tnat\Tnat}}
\end{align*}
if \(k\in\Nat\) and \(\nu_1,\dots,\nu_k\in\Natnz\).
For \(k=0\) this means in particular that
\( \Tseqi{}{M}{(\Msetempty,(\Mset 0,0))}
{\Timpl{(\Timpl\Tnat\Tnat)}{\Timpl\Tnat\Tnat}}
\).
It is important to notice that the linearity of \(M\) in its second
parameter (called \(x\) in the definition of \(M\)) comes from our use
of the \(\Llet 0yx\_\) construct.
Then we can take the derivative of \(M\), typed as follows:
\begin{align*}
  \Tseq{}{\Ldiffd M}{\Timpl{(\Timpl\Tnat{\Tdiff\Tnat})}{\Timpl\Tnat\Tnat}}\,.
\end{align*}

We describe now a way of building an argument for such a derivative.
Let \(\cT\subseteq\Nat\)
be recursive and implemented by some term of type
\(\Tseq{}T{\Tnat\to\Tnat}\), in the sense that
\begin{align*}
  \Psem T{}=\{(\Mset\nu,0)\St\nu\in\cT\}\cup\{(\Mset\nu,1)\St\nu\notin\cT\}\,.
\end{align*}
Given terms \(M_0,M_1\) such that
\((\Tseq\Gamma{M_i}{\Timpl\Tnat A})_{i=0,1}\) (for some type \(A\) and
typing context \(\Gamma\)) we define a term \(\Tuple{M_0,M_1}^T\) such
that %
\(\Tseq\Gamma{\Tuple{M_0,M_1}^T}{\Timpl\Tnat{\Tdiff A}}\) given by %
\begin{align*}
  \Tuple{M_0,M_1}^T
  =\Abst x{\Tnat}
  {\Llet0yx{\Lif0{\App Ty}{\Lin0{\App{M_0}y}}{\Lin1{\App{M_1}y}}}}
\end{align*}
so that
\begin{equation*}
  \begin{prooftree}
    \hypo{\nu\in \cT_i}
    \infer1{\Tseqi{f_i:\Mset{(k\Mset\nu,a)}:\Timpl\Tnat A,
  f_{1-i}:\Msetempty:\Timpl\Tnat A}
  {\Tuple{f_0,f_1}^T}{(\Mset\nu,\Tseqact ia)}{\Timpl\Tnat{\Tdiff A}}}
  \end{prooftree}
\end{equation*}
where \(\cT_0=\cT\) and \(\cT_1=\Nat\setminus \cT\).
Let us take \(A=\Tnat\) so that
\(
\Tseq{}{\App{\Ldiffd M}{\Tuple{M_0,M_1}^T}}{\Timpl\Tnat\Tnat}
\).
%
Then one has
\begin{multline*}
  f_0:\Mset{(\Mset{\nu_1},\nu_2),(\Mset{\nu_2},\nu_3),
    \dots,(\Mset{\nu_{i-1}},\nu_i),(\Mset{\nu_{i+1}},\nu_{i+2}),
  \dots,(\Mset{\nu_k},0)}:\Timpl\Tnat\Tnat,
  \\
  f_1:\Mset{(\Mset{\nu_{i}},\nu_{i+1})}:\Timpl\Tnat\Tnat
  \vdash
  {\App{\Ldiffd M}{\Tuple{f_0,f_1}^T}}:(\Mset{\nu_1},k):\Timpl\Tnat\Tnat
\end{multline*}
as soon as \(\nu_j\in\cT\setminus\Eset 0\) for
\(j\in\Eset{1,\dots,i_{i-1},i_{i+1},\dots,k}\), \(\nu_i\notin \cT\)
and \(\nu_i\not=0\).
In other words, in order to induce a successful computation, the
stochastic matrices \(M_0\) and \(M_1\) in %
\(\App{\Ldiffd M}{\Tuple{M_0,M_1}^T}\Appsep\Num\nu\) %
need to satisfy the property that the iteration of \(M_0\) from \(\nu\)
will leave \(\cT\) exactly once, and \(M_1\) will return in \(\cT\)
from that external value.
During the computation, \(M_1\) is used exactly once whereas there are
no restrictions on the number of uses of \(M_0\).

More interesting examples should arise by considering typically a term
\(M\) such that %
\[\Tseq{\Gamma,g:\Timpl{\Tdiff A}{\Tdiff B},x:A}{M}{B}\,.\]
Then we have
\begin{align*}
  \Tseq{\Gamma,f:\Timpl AB,x:A}{\Subst{M}{\Ldiff f}{g}}{B}
\end{align*}
and hence
\begin{align*}
  \Tseq{\Gamma}{P=\Lfix{(\Abst{f}{\Timpl AB}
  {\Abst{x}{A}{\Subst{M}{\Ldiff f}{g}}})}}{\Timpl AB}
\end{align*}
meaning that we are defining a term \(P\) such that
\(\Tseq\Gamma P{\Timpl AB}\) in terms of its own derivative, something
which looks a bit like a ``differential equation''.
The study of such definitions, typical of our new coherent
differential setting, seems complex and is postponed to further work.

%


\section{Semantics} %
\label{sec:semantics}

Before explaining in Section~\ref{sec:types-terms-interp} how types
and terms are interpreted in the Kleisli category of a coherent
differential summable resource category, we provide more information
about the categorical constructions that we will need. These
Sections~\ref{sec:partial-derivatives}~--~\ref{sec:syn-basic-constr}
can be
seen as a complement to Section~\ref{sec:cohdiff-summary}
and~\ref{sec:diff-distr-law}.

\subsection{Partial derivatives}
\label{sec:partial-derivatives}
We assume to be given a summable resource category $\cL$ which is
closed (wrt.~its symmetric monoidal structure) and is equipped with a
differentiation $\Sdiff$, see Section~\ref{sec:cohdiff-summary}.
We generalize the lax monoidality of the $\Sdfun$ functor to a natural
transformation
$\Smont_n\in\cL(\Sdfun{X_0}\ITens\cdots\ITens\Sdfun{X_n},
\Sdfun(X_0\ITens\cdots\ITens X_n))$ by induction on $n$ (there are
various possible definitions, all leading to the same morphisms), for
instance $\Smont_0=\Sin0$ and
$\Smont_{n+1} =\Smont_{X_0\ITens\cdots\ITens
  X_n,X_{n+1}}\Compl\Tensp{\Smont_n}{\Sdfun{X_{n+1}}}$. The resulting
morphism is fully characterized by the following property.
\begin{lemma}
  $\Sproj 0\Compl\Smont_n=\Sproj0\ITens\cdots\ITens\Sproj0$ and
  $\Sproj 1\Compl\Smont_n=\Sproj1\ITens\Sproj
  0\ITens\cdots\ITens\Sproj0+\cdots+\Sproj
  0\ITens\cdots\ITens\Sproj0\ITens\Sproj1$.
\end{lemma}

\subsubsection{Additive strength} %
\label{sec:additive-strength}
We define morphisms
$\Sdfstr_{X_0,X_1}^0\in\Kl\cL(\Sdfun X_0\IWith
X_1,\Sdfun\Withp{X_0}{X_1})$ and
$\Sdfstr_{X_0,X_1}^1\in\Kl\cL(X_0\IWith\Sdfun
X_1,\Sdfun\Withp{X_0}{X_1})$ %
of $\Sdfun$ as
\begin{align*}
  \Sdfstr_{X_0,X_1}^0=\Kllin\Withp{\Sfun X_0}{\Sin 0}
  \text{\quad and\quad}
  \Sdfstr_{X_0,X_1}^1=\Kllin\Withp{\Sin 0}{\Sfun X_1}\,.
\end{align*}

\begin{lemma}
  The morphism
  $\Sdfstr^0_{X_0,X_1}\in\Kl\cL(\With{\Sdfun
    X_0}{X_1},\Sdfun\Withp{X_0}{X_1})$ is natural in $X_0,X_1$ and
  similarly for $\Sdfstr^1_{X_0,X_1}$.
\end{lemma}
\begin{proof}
  Let $f_i\in\Kl\cL(X_i,Y_i)$ for $i=0,1$, %
  we must show that the two following morphisms are equal:
  \begin{align*}
    g&=\Sdfun\Withep{f_0}{f_1}\Comp\Sdfstr^0_{X_0,X_1}\\
       &=\Withp{\Sfun f_0}{\Sfun f_1}
       \Compl\Tuple{\Sfun\Excl{\Proj 0},\Sfun\Excl{\Proj 1}}
       \Compl\Sdiff_{\With{X_0}{X_1}}
       \Compl\Excl{\Withp{\Sfun{X_0}}{\Sin0}}\\
    h&=\Sdfstr_{\With{Y_0}{Y_1}}^0
       \Comp\Withep{\Sdfun f_0}{f_1}\\
       &=\Withp{\Sfun Y_0}{\Sin0}
       \Compl\Withp{(\Sfun{f_0})\Compl\Sdiff_{X_0}}{f_1}
       \Compl\Tuple{\Excl{\Proj0},\Excl{\Proj1}}\\
     &=\Withp{\Sfun Y_0}{\Sin0}
       \Compl\Withp{\Sfun{f_0}}{f_1}
       \Compl\Withp{\Sdiff_{X_0}}{\Excl{X_1}}
       \Compl\Tuple{\Excl{\Proj0},\Excl{\Proj1}}
  \end{align*}
  and for this it suffices to prove that %
  $\Proj i\Compl g=\Proj i\Compl h$ for $i=0,1$. We have
  \begin{align*}
    \Proj 0\Compl g
    &=(\Sfun f_0)
      \Compl(\Sfun\Excl{\Proj0})
      \Compl\Sdiff_{\With{X_0}{X_1}}
      \Compl\Excl{\Withp{\Sfun{X_0}}{\Sin0}}\\
    &=(\Sfun f_0)
      \Compl\Sdiff_{X_0}
      \Compl\Excl{\Proj 0}
      \Compl\Excl{\Withp{\Sfun{X_0}}{\Sin0}}
    \text{\quad by naturality of }\Sdiff\\
    &=(\Sfun f_0)
      \Compl\Sdiff_{X_0}
      \Compl\Excl{\Proj 0}=\Proj0\Compl h
  \end{align*}
  and, by a similar computation
  \begin{align*}
    \Proj 1\Compl g
    =(\Sfun f_1)
    \Compl\Sdiff_{X_1}
    \Compl\Excl{\Sin0}
    \Compl\Excl{\Proj 1}
    =(\Sfun f_1)
    \Compl\Sin0
    \Compl\Excl{\Proj1}
    =\Sin0
    \Compl f_1
    \Compl\Excl{\Proj1}
  \end{align*}
  by~\ref{ax:daxlin} 
  and by naturality of $\Sin0$ and hence %
  $\Proj 1\Compl g=\Proj 1\Compl h$.
\end{proof}

There is a simple connection between this additive strength and the
tensorial strength of the monad $\Sfun$ introduced %
in~\cite{Ehrhard23a},
Section~4, 
see also Section~\ref{sec:monad-mult-strength} in the present paper,
through the strong symmetric monoidal structure of $\Excl\_$ (Seely
isomorphisms).
\begin{theorem}\label{th:sdfstr-sstr-comm}
  The following diagram commutes.
  \begin{equation*}
    \begin{tikzcd}
      \Tens{\Excl{\Sfun X_0}}{\Excl{X_1}}
      \ar[r,"\Seelyt_{\Sfun X_0,X_1}"]
      \ar[d,swap,"\Tens{\Sdiff_{X_0}}{\Excl{X_1}}"]
      &[1em]
      \Exclp{\With{\Sfun X_0}{X_1}}
      \ar[r,"\Promp{\Sdfstr^0_{X_0,X_1}}"]
      &[1em]
      \Excl{\Sfun\Withp{X_0}{X_1}}
      \ar[d,"\Sdiff_{\With{X_0}{X_1}}"]
      \\
      \Tens{\Sfun{\Excl{X_0}}}{\Excl{X_1}}
      \ar[r,"\Sstr_{\Excl{X_0},\Excl{X_1}}^0"]
      &
      \Sfun\Tensp{\Excl{X_0}}{\Excl{X_1}}
      \ar[r,"\Sfun\Seelyt_{X_0,X_1}"]
      &
      \Sfun\Exclp{\With{X_0}{X_1}}
    \end{tikzcd}
  \end{equation*}
  and similarly for $\Sdfstr_{X_0,X_1}^1$ and $\Sstr_{X_0,X_1}^1$.
\end{theorem}
Remember that %
$\Promp{\Sdfstr^0_{X_0,X_1}}$ is the promotion of %
$\Sdfstr^0_{X_0,X_1}\in\cL(\Exclp{\With{\Sfun
    X_0}{X_1}},\Sfun\Withp{X_0}{X_1})$, so that actually %
$\Promp{\Sdfstr^0_{X_0,X_1}}=\Exclp{\With{\Sfun X_0}{\Sin0}}$.
\begin{proof}
  By~\ref{ax:daxwith} 
  we have
  \begin{align*}
    \Sdiff_{\With{X_0}{X_1}}
    =(\Sfun{\Seelyt_{X_0,X_1}})
    \Compl\Smont_{\Excl{X_0},\Excl{X_1}}
    \Compl\Tensp{\Sdiff_{X_0}}{\Sdiff_{X_1}}
    \Compl\Invp{\Seelyt_{\Sfun X_0,\Sfun X_1}}
  \end{align*}
  so that
  \begin{align*}
    \Sdiff_{\With{X_0}{X_1}}
    &\Compl\Promp{\Sdfstr^0_{X_0,X_1}}
    \Compl\Seelyt_{\Sfun X_0,X_1}\\
    &=(\Sfun{\Seelyt_{X_0,X_1}})
      \Compl\Smont_{\Excl{X_0},\Excl{X_1}}
      \Compl\Tensp{\Sdiff_{X_0}}{\Sdiff_{X_1}}
      \Compl\Invp{\Seelyt_{\Sfun X_0,\Sfun X_1}}
      \Compl\Exclp{\With{\Sfun X_0}{\Sin0}}
      \Compl\Seelyt_{\Sfun X_0,X_1}\\
    &=(\Sfun{\Seelyt_{X_0,X_1}})
      \Compl\Smont_{\Excl{X_0},\Excl{X_1}}
      \Compl\Tensp{\Sdiff_{X_0}}{\Sdiff_{X_1}}
      \Tensp{\Exclp{\Sfun X_0}}{\Excl{\Sin0}}\\
    &=(\Sfun{\Seelyt_{X_0,X_1}})
      \Compl\Smont_{\Excl{X_0},\Excl{X_1}}
      \Compl\Tensp{\Sdiff_{X_0}}{\Sin0}
      \text{\quad by~\ref{ax:daxlin} 
      }\\
    &=(\Sfun{\Seelyt_{X_0,X_1}})
      \Compl\Sfunadd
      \Compl(\Sfun\Sstr_{\Excl{X_0},\Excl{X_1}}^1)
      \Compl\Sstr_{\Excl{X_0},\Sfun\Excl{X_1}}^0
      \Compl\Tensp{\Sdiff_{X_0}}{\Sin0}\\
    &\text{\quad\quad by definition of }\Smont
      \text{ in~\cite{Ehrhard23a},
      Section~4}\\ 
    &=(\Sfun{\Seelyt_{X_0,X_1}})
      \Compl\Sfunadd
      \Compl(\Sfun\Sstr_{\Excl{X_0},\Excl{X_1}}^1)
      \Compl\Sfun\Tensp{\Excl{X_0}}{\Sin0}
      \Compl\Sstr_{\Excl{X_0},\Excl{X_1}}^0
      \Compl\Tensp{\Sdiff_{X_0}}{\Excl{X_1}}\\
    &=(\Sfun{\Seelyt_{X_0,X_1}})
      \Compl\Sfunadd
      \Compl(\Sfun\Sin0)
      \Compl\Sstr_{\Excl{X_0},\Excl{X_1}}^0
      \Compl\Tensp{\Sdiff_{X_0}}{\Excl{X_1}}\\
    &\text{\quad\quad since }\Sstr^0
      \text{ is a strength of the monad }\Sfun
      \text{ whose unit is }\Sin0\\
    &=(\Sfun{\Seelyt_{X_0,X_1}})
      \Compl\Sstr_{\Excl{X_0},\Excl{X_1}}^0
      \Compl\Tensp{\Sdiff_{X_0}}{\Excl{X_1}}
  \end{align*}
  by one of the monad commutations.
\end{proof}

\begin{lemma}\label{lemma:sdfstr-description}
  The morphism $\Sdfstr^0_{X_0,X_1}$ is equal to the following
  composition of morphisms in $\cL$:
\begin{equation}
  \label{eq:Sdfstr-def}
  \begin{tikzcd}
    \Excl{\Withp{\Sfun X_0}{X_1}}\ar[r,"\Invp{\Seelyt}"]
    &[2em]
    \Tens{\Excl{\Sfun X_0}}{\Excl{X_1}}
    \ar[r,"\Tens{\Sdiff_{X_0}}{\Excl{X_1}}"]
    &
    \Tens{\Sfun\Excl{X_0}}{\Excl{X_1}}
    \ar[d,"\Sstr^0_{\Excl{X_0},\Excl{X_1}}"]
    \\
    \Sfun{\Withp{X_0}{X_1}}
    &
    \Sfun\Excl{\Withp{X_0}{X_1}}
    \ar[l,"\Sfun\Der{\With{X_0}{X_1}}"]
    &
    \Sfun\Tensp{\Excl{X_0}}{\Excl{X_1}}
    \ar[l,"\Sfun\Seelyt"]
  \end{tikzcd}
\end{equation}
\end{lemma}
\begin{proof}
  By Theorem~\ref{th:sdfstr-sstr-comm}, the
  morphism~\Eqref{eq:Sdfstr-def} is equal to
  \begin{align*}
    (\Sfun\Der{\With{X_0}{X_1}})
    \Compl\Sdiff_{\With{X_0}{X_1}}
    \Compl\Promp{\Sdfstr^0_{X_0,X_1}}
    &=\Der{\Sfun\Withp{X_0}{X_1}}
      \Compl\Promp{\Sdfstr^0_{X_0,X_1}}
      \text{\quad by~\ref{ax:daxchain} 
      }\\
    &=\Sdfstr^0_{X_0,X_1}\,.\qedhere
  \end{align*}
\end{proof}

\begin{lemma}
\begin{align}
  \Sproj0\Compl\Sdfstr^0_{X_0,X_1}=\With{\Sproj0}{X_1}
  &\text{\quad and\quad}
    \Sproj1\Compl\Sdfstr^0_{X_0,X_1}=\With{\Sproj1}0\label{eq:sproj-sdfstr0}\\
  \Sproj0\Compl\Sdfstr^1_{X_0,X_1}=\With{X_0}{\Sproj0}
  &\text{\quad and\quad}
    \Sproj1\Compl\Sdfstr^1_{X_0,X_1}=\With0{\Sproj1}\label{eq:sproj-sdfstr1}\,.
\end{align}    
\end{lemma}
\begin{proof}
  Immediate consequence of the definitions.
\end{proof}

\begin{theorem}
  The natural morphisms $\Sdfstr^0,\Sdfstr^1$ are strengths for the
  monad $(\Sdfun,\Sdfunit,\Sdfmult)$ on the category $\Kl\cL$.
\end{theorem}
This means that the following diagrams commute in $\Kl\cL$.
\begin{equation*}
  \begin{tikzcd}
    \With{X_0}{X_1}
    \ar[d,swap,"\With{\Sdfunit_{X_0}}{X_1}"]
    \ar[dr,"\Sdfunit_{\With{X_0}{X_1}}"]
    &[1em]
    \\
    \With{\Sdfun X_0}{X_1}
    \ar[r,"\Sdfstr^0_{X_0,X_1}"]
    &
    \Sdfun\Withp{X_0}{X_1}
  \end{tikzcd}
\end{equation*}
\begin{equation*}
  \begin{tikzcd}
    \With{\Sdfun^2X_0}{X_1}
    \ar[d,swap,"\With{\Sdfmult_{X_0}}{X_1}"]
    \ar[r,"\Sdfstr^0_{\Sdfun X_0,X_1}"]
    &[1em]
    \Sdfun\Withp{\Sdfun X_0}{X_1}
    \ar[r,"\Sdfun\Sdfstr_{X_0,X_1}"]
    &[1em]
    \Sdfun^2\Withp{X_0}{X_1}
    \ar[d,"\Sdfmult_{\With{X_0}{X_1}}"]
    \\
    \With{\Sdfun X_0}{X_1}
    \ar[rr,"\Sdfstr^0_{X_0,X_1}"]
    &&
    \Sdfun\Withp{X_0}{X_1}    
  \end{tikzcd}
\end{equation*}
\begin{equation*}
  \begin{tikzcd}
    \With{\Sdfun X}{\Top}
    \ar[r,"\Sdfstr_{X,\Top}"]
    \ar[dr,swap,"\Proj0"]
    &
    \Sdfun\Withp{X}{\Top}
    \ar[d,"\Sdfun\Proj0"]
    \\
    &
    \Sdfun X
  \end{tikzcd}
  \Treesep
  \begin{tikzcd}
    \Sdfun X_0\IWith X_1\IWith X_2
    \ar[r,"\With{\Sdfstr^0_{X_0,X_1}}{X_2}"]
    \ar[rd,swap,"\Sdfstr^0_{X_0,\With{X_1}{X_2}}"]
    &[2em]
    \With{\Sdfun\Withp{X_0}{X_1}}{X_2}
    \ar[d,"\Sdfstr^0_{\With{X_0}{X_1},X_2}"]
    \\
    &
    \Sdfun(X_0\IWith X_1\IWith X_2)
  \end{tikzcd}
\end{equation*}
where we keep the associativity isomorphisms of $\IWith$ implicit.

\begin{proof}
  It suffices to prove the corresponding commutations in %
  $\cL$ rather than $\Kl\cL$ since all the involved morphisms are
  images of morphisms in $\cL$ through $\Kllin$, and this is quite
  easy.
\end{proof}

The commutativity of this strength takes a particularly strong form in
this setting.
\begin{lemma}\label{lemma:Sdfstr-flip}
  The following diagram commutes in $\cL$
  \begin{equation*}
    \begin{tikzcd}
      \Sfun{X_0}\IWith\Sfun{X_1}
      \ar[r,"\Sdfstr_{\Sfun X_0,X_1}^1"]
      \ar[d,swap,"\Sdfstr^0_{X_0,\Sfun X_1}"]
      &
      \Sfun\Withp{\Sfun X_0}{X_1}\ar[r,"\Sfun\Sdfstr^0_{X_0,X_1}"]
      &
      \Sfun^2\Withp{X_0}{X_1}\ar[d,"\Sflip"]\\
      \Sfun\Withp{X_0}{\Sfun X_1}
      \ar[rr,"\Sfun\Sdfstr^1_{X_0,X_1}"]
      &&
      \Sfun^2\Withp{X_0}{X_1}
    \end{tikzcd}
  \end{equation*}
\end{lemma}
\begin{proof}
  We prove that for each $j,k\in\Eset{0,1}$ one has
  \[
    \Sproj k \Compl\Sproj j \Compl(\Sfun\Sdfstr^1_{X_0,X_1})
  \Compl\Sdfstr^0_{X_0,\Sfun X_1} = \Sproj j \Compl\Sproj k
  \Compl(\Sfun\Sdfstr^0_{X_0,X_1}) \Compl\Sdfstr^1_{X_0,\Sfun X_1}
  \]
  which will prove our contention since
  $\Sproj k\Compl\Sproj j\Compl\Sflip=\Sproj j\Compl\Sproj k$. This
  amounts to proving that
  \begin{align*}
    \Sproj k
    \Compl\Sdfstr^1_{X_0,X_1}
    \Compl\Sproj j
    \Compl\Sdfstr^0_{X_0,\Sfun X_1}   
    =
    \Sproj j
    \Compl\Sdfstr^0_{X_0,X_1}
    \Compl\Sproj k
    \Compl\Sdfstr^1_{\Sfun X_0,X_1}
  \end{align*}
  for which we apply Equations~\Eqref{eq:sproj-sdfstr0}
  and~\Eqref{eq:sproj-sdfstr1}.  We have
  \begin{align*}
    \Sproj 0
    \Compl\Sdfstr^1_{X_0,X_1}
    \Compl\Sproj 0
    \Compl\Sdfstr^0_{X_0,\Sfun X_1}
    &=\Withp{X_0}{\Sproj0}
      \Compl\Withp{\Sproj0}{\Sfun X_1}
      =\With{\Sproj0}{\Sproj0}
      =\Sproj 0
      \Compl\Sdfstr^0_{X_0,X_1}
      \Compl\Sproj 0
      \Compl\Sdfstr^1_{\Sfun X_0,X_1}
    \\
    \Sproj 1
    \Compl\Sdfstr^1_{X_0,X_1}
    \Compl\Sproj 1
    \Compl\Sdfstr^0_{X_0,\Sfun X_1}
    &=\Withp{0}{\Sproj1}
      \Compl\Withp{\Sproj1}{0}
      =0
      =\Sproj 1
      \Compl\Sdfstr^0_{X_0,X_1}
      \Compl\Sproj 1
      \Compl\Sdfstr^1_{\Sfun X_0,X_1}
    \\
    \Sproj 0
    \Compl\Sdfstr^1_{X_0,X_1}
    \Compl\Sproj 1
    \Compl\Sdfstr^0_{X_0,\Sfun X_1}
    &=\Withp{X_0}{\Sproj0}\Compl\Withp{\Sproj1}{0}
      =\With{\Sproj1}0\\
    &=\Withp{\Sproj1}{0}
      \Compl\Withp{\Sfun X_0}{\Sproj0}
      =\Sproj 1
      \Compl\Sdfstr^0_{X_0,X_1}
      \Compl\Sproj 0
      \Compl\Sdfstr^1_{\Sfun X_0,X_1}    
  \end{align*}
  and the last case is symmetrical.
\end{proof}

As a consequence
\[
  \Sfunadd\Compl(\Sfun\Sstr_1^{X_0,X_1})\Compl\Sstr_0^{X_0,\Sfun X_1}
=\Sfunadd\Compl(\Sfun\Sstr_0^{X_0,X_1})\Compl\Sstr_1^{\Sfun
  X_0,X_1}\in\cL(\Sfun X_0\IWith\Sfun X_1,\Sfun\Withp{X_0}{X_1})\,.
\]

Thanks to our assumption that $\Sfun$ preserves $\IWith$ strictly, in
the sense of~\Eqref{eq:Sfun-preserves-With}, this morphism is actually
the identity.
\begin{theorem}\label{th:sdstr-mont-identity}
  The morphism
  $\Sfunadd\Compl(\Sfun\Sdfstr^1_{X_0,X_1})\Compl\Sdfstr^0_{X_0,\Sfun X_1}
  =\Sfunadd\Compl(\Sfun\Sdfstr^0_{X_0,X_1})\Compl\Sdfstr^1_{\Sfun
    X_0,X_1}$ is the identity morphism.
\end{theorem}
\begin{proof}
  The first equation results from Lemma~\ref{lemma:Sdfstr-flip}. From
  the proof of that lemma we get
  \begin{align*}
    \Sproj0\Compl\Sfunadd
    \Compl(\Sfun\Sdfstr^1_{X_0,X_1})
    \Compl\Sdfstr^0_{X_0,\Sfun X_1}
    &=\Sproj0
      \Compl\Sproj0
      \Compl(\Sfun\Sdfstr^1_{X_0,X_1})
      \Compl\Sdfstr^0_{X_0,\Sfun X_1}\\
    &=\With{\Sproj0}{\Sproj0}
  \end{align*}
  and
  \begin{align*}
    \Sproj1\Compl\Sfunadd
    \Compl(\Sfun\Sdfstr^1_{X_0,X_1})
    \Compl\Sdfstr^0_{X_0,\Sfun X_1}
    &=\Sproj1
      \Compl\Sproj0
      \Compl(\Sfun\Sdfstr^1_{X_0,X_1})
      \Compl\Sdfstr^0_{X_0,\Sfun X_1}+
      \Sproj0
      \Compl\Sproj1
      \Compl(\Sfun\Sdfstr^1_{X_0,X_1})
      \Compl\Sdfstr^0_{X_0,\Sfun X_1}\\
    &=\Withp0{\Sproj1}+\Withp{\Sproj1}0=\With{\Sproj1}{\Sproj1}
  \end{align*}
  by linearity of $\IWith$ on morphisms.
\end{proof}

More generally given objects $\List X0n$ we have an additive strength
morphism %
\begin{align*}
  \Sdfstr^i\in\Kl\cL(X_0\IWith\cdots\IWith\Sdfun{X_i}\IWith\cdots\IWith X_n,
  \Sdfun(X_0\IWith\cdots\IWith X_n))\,.
\end{align*}
which is actually linear and comes from
$\Sdfstr^i\in\cL(X_0\IWith\cdots\IWith\Sfun{X_i}\IWith\cdots\IWith
X_n, \Sfun(X_0\IWith\cdots\IWith X_n))$.  Up to the identification of
$\Sfun(X_0\IWith\cdots\IWith X_n)$ with
$\Sfun{X_0}\IWith\cdots\IWith\Sfun{X_n}$, this morphism of $\cL$ can
simply be written as
\begin{align}
  \label{eq:sdfstrgen-def}
  \Sdfstr^i=\Sin0^{X_0}\IWith\cdots\IWith\Sin0^{X_{i-1}}\IWith X_i
  \IWith\Sin0^{X_{i+1}}\IWith\cdots\IWith\Sin0^{X_n}\,.
\end{align}
When we will need to be explicit as to the list of
objects $\List X0n$, we will write $\Sdfstr^i_{\List X0n}$ instead of
$\Sdfstr^i$.

\begin{lemma}\label{lemma:sdfstr-commute-general}
  Let $i,l\in\Eset{0,\dots,n}$. If $i\not=l$ we have
  \begin{align*}
    (\Sfun\Sdfstr^i_{\List X0n})
    \Compl\Sdfstr^l_{X_0,\dots,\Sfun X_i,\dots,X_n}
    =\Sflip\Compl
    (\Sfun\Sdfstr^l_{\List X0n})
    \Compl\Sdfstr^i_{X_0,\dots,\Sfun X_l,\dots,X_n}\,.
  \end{align*}
  And for any $i,l\in\Eset{0,\dots,n}$, we have
  \begin{align*}
    \Sfunadd\Compl(\Sfun\Sdfstr^i_{\List X0n})
    \Compl\Sdfstr^l_{X_0,\dots,\Sfun X_i,\dots,X_n}
    &=\Sfunadd\Compl
  (\Sfun\Sdfstr^l_{\List X0n}) \Compl\Sdfstr^i_{X_0,\dots,\Sfun
    X_l,\dots,X_n}\\
  \Sfunadd\Compl(\Sfun\Sdfstr^i_{\List X0n})
    \Compl\Sdfstr^i_{X_0,\dots,\Sfun X_i,\dots,X_n}
    &=\Sdfstr^i_{\List
      X0n}\Compl(X_0\IWith\cdots\IWith\Sfunadd\IWith\cdots\IWith X_n)\,.
  \end{align*}
\end{lemma}
\begin{proof}
  The proof of the first equation is exactly as the one of
  Lemma~\ref{lemma:Sdfstr-flip}. In the case $l\not=i$, the second
  equation follows from the first one and from
  $\Sfunadd\Compl\Sflip=\Sfunadd$ and in the case $l=i$, it is
  trivial. We prove the last equation. We have
  \begin{align*}
    \Sproj 0\Compl\Sfunadd\Compl(\Sfun\Sdfstr^i_{\List X0n})
    \Compl\Sdfstr^i_{X_0,\dots,\Sfun X_i,\dots,X_n}
    &=\Sproj0
      \Compl\Sproj0
      \Compl(\Sfun\Sdfstr^i_{\List X0n})
      \Compl\Sdfstr^i_{X_0,\dots,\Sfun X_i,\dots,X_n}\\
    &=\Sproj0
      \Compl\Sdfstr^i_{\List X0n}
      \Compl\Sproj0
      \Compl\Sdfstr^i_{X_0,\dots,\Sfun X_i,\dots,X_n}\\
    &=X_0\IWith\cdots\IWith(\Sproj0\Compl\Sproj 0)\IWith\cdots\IWith X_n
  \end{align*}
  and
  \begin{align*}
    \Sproj0
    \Compl\Sdfstr^i_{\List X0n}
    \Compl(X_0\IWith\cdots\IWith\Sfunadd\IWith\cdots\IWith X_n)
    &=X_0\IWith\cdots\IWith(\Sproj0\Compl\Sfunadd)\IWith\cdots\IWith X_n\\
    &=X_0\IWith\cdots\IWith(\Sproj0\Compl\Sproj0)\IWith\cdots\IWith X_n\,.
  \end{align*}
  Next we have
  \begin{align*}
    &\Sproj 1\Compl\Sfunadd\Compl(\Sfun\Sdfstr^i_{\List X0n})
    \Compl\Sdfstr^i_{X_0,\dots,\Sfun X_i,\dots,X_n}\\
    &\Textsep=(\Sproj0\Compl\Sproj1+\Sproj1\Compl\Sproj0)
      \Compl(\Sfun\Sdfstr^i_{\List X0n})
      \Compl\Sdfstr^i_{X_0,\dots,\Sfun X_i,\dots,X_n}\\
    &\Textsep=\Sproj0
      \Compl\Sdfstr^i_{\List X0n}
      \Compl\Sproj1
      \Compl\Sdfstr^i_{X_0,\dots,\Sfun X_i,\dots,X_n}
      +\Sproj1
      \Compl\Sdfstr^i_{\List X0n}
      \Compl\Sproj0
      \Compl\Sdfstr^i_{X_0,\dots,\Sfun X_i,\dots,X_n}\\
    &\Textsep=(X_0\IWith\cdots\IWith\Sproj 0\IWith\cdots\IWith X_n)
      \Compl(0\IWith\cdots\IWith\Sproj1\IWith\cdots\IWith 0)\\
    &\Textsep\quad+(0\IWith\cdots\IWith\Sproj1\IWith\cdots\IWith 0)
      \Compl(X_0\IWith\cdots\IWith\Sproj 0\IWith\cdots\IWith X_n)\\
    &\Textsep=0\IWith\cdots\IWith(\Sproj0\Compl\Sproj1+\Sproj1\Compl\Sproj0)
      \IWith\cdots\IWith 0
  \end{align*}
  Notice that $\IWith$ is not a multilinear operation on morphisms, so
  in the last equality we are crucially using the fact that all
  factors but the $i$th are equal to $0$ in both summands.  On the
  other hand we have
  \begin{align*}
    \Sproj1
    \Compl\Sdfstr^i_{\List X0n}
    \Compl(X_0\IWith\cdots\IWith\Sfunadd\IWith\cdots\IWith X_n)
    &=0\IWith\cdots\IWith(\Sproj1\Compl\Sfunadd)\IWith\cdots\IWith 0\\
    &=0\IWith\cdots\IWith(\Sproj0\Compl\Sproj1+\Sproj1\Compl\Sproj0)
      \IWith\cdots\IWith 0
  \end{align*}
  proving our contention by the fact that $\Sproj0,\Sproj1$ are jointly monic.
\end{proof}

Given $f\in\Kl\cL(X_0\IWith\cdots\IWith X_n,Y)$, we define the %
\emph{$i$-th partial derivative of $f$} %
as
\(
\Sdfunpart if
=\Sdfun f\Comp\Sdfstr^i
\in\Kl\cL(X_0\IWith\cdots\IWith\Sdfun{X_i}\IWith\cdots\IWith X_n,\Sdfun Y)
\).

\begin{theorem}\label{th:partial-der-commute}
  Let $i,l\in\Eset{0,\dots,n}$. If $i\not=l$ then
  \begin{align*}
    \Sdfunpart i{\Sdfunpart lf}=\Sflip\Comp\Sdfunpart l{\Sdfunpart if}
  \end{align*}
  so that for any $i,l\in\Eset{0,\dots,n}$ we have
  \( \Sdfmult\Comp\Sdfunpart i{\Sdfunpart lf} =\Sdfmult\Comp\Sdfunpart
  l{\Sdfunpart if} \). Moreover
  $\Sdfmult\Comp\Sdfunpart i{\Sdfunpart if}=\Sdfunpart
  if\Comp(X_0\IWith\cdots\IWith\Sdfmult\IWith\cdots\IWith X_n)$.
\end{theorem}
This is an immediate consequence of
Lemma~\ref{lemma:sdfstr-commute-general} and of the naturality of
$\Sflip$ and of $\Sdfmult$ in the category $\Kl\cL$.

Notice that the morphism %
\(\Sdfmult\Comp\Sdfunpart i{\Sdfunpart lf}\) %
in the statement of this result belongs to
$\Kl\cL(X_0\IWith\cdots\IWith\Sdfun X_i\IWith
X_{i+1}\IWith\cdots\IWith\Sdfun X_l\IWith\cdots\IWith X_n,\Sdfun Y)$
if $i<l$ and to
$\Kl\cL(X_0\IWith\cdots\IWith\Sdfun^2 X_i\IWith\cdots\IWith X_n,\Sdfun
Y)$ if $i=l$.

\begin{theorem}\label{th:sdfun-sdfunpart-sfunadd}
  Let $f\in\Kl\cL(\With{X_0}{X_1},Y)$ so that
  $\Sdfun f\in\Kl\cL(\With{\Sdfun X_0}{\Sdfun X_1},\Sdfun Y)$. Then
  \begin{align*}
    \Sdfun f=\Sdfmult\Comp\Sdfunpart1\Sdfunpart0f
    =\Sdfmult\Comp\Sdfunpart0\Sdfunpart1f\,.
  \end{align*}
\end{theorem}
\begin{proof}
  The second equation holds by Theorem~\ref{th:partial-der-commute}. Next
  we have
  \begin{align*}
    \Sdfmult\Comp\Sdfunpart1\Sdfunpart0f
    &=\Sdfmult
      \Comp\Sdfun(\Sdfunpart0f\Comp\Sdfstr^0_{X_0,X_1})\\
    &=\Sdfmult
      \Comp\Sdfun^2f
      \Comp\Sdfun{\Sdfstr^0_{X_0,X_1}}
      \Comp\Sdfstr^1_{\Sdfun X_0,X_1}\\
    &=\Sdfun f
      \Comp\Sdfmult
      \Comp\Sdfun{\Sdfstr^0_{X_0,X_1}}
      \Comp\Sdfstr^1_{\Sdfun X_0,X_1}\\
    &=\Sdfun f
  \end{align*}
  by Theorem~\ref{th:sdstr-mont-identity}.
\end{proof}
\begin{remark}
  The intuitive meaning of this result is that the derivative of a
  function acting on pairs is obtained as the sum of its partial
  derivatives. This sum is computed by the \(\Sdfmult\) natural
  transformation.
\end{remark}

Given $e\in\Nat$ we can more generally define a linear
\begin{align*}
  \Sdfstr^i_{\List X0n}(e)
  \in\Kl\cL(X_0\IWith\cdots\IWith\Sdfun^eX_i\IWith\cdots\IWith X_n,
  \Sdfun^e X_0\IWith\cdots\IWith\Sdfun^eX_n)
\end{align*}
by induction on $e$ (we give only the definition for $n=1$, the
generalization is easy and not really required for our purpose):
we set $\Sdfstr^0_{X_0,X_1}(0)=\Id$ and %
$\Sdfstr^0_{X_0,X_1}(e+1)
=\Sdfun\Sdfstr^0_{X_0,X_1}(e)\Comp\Sdfstr^0_{\Sdfun^e X_0,X_1}$ typed
as follows:
\begin{center}
  \begin{tikzcd}
    \Sdfun^{e+1}X_0\IWith X_1\ar[r,"\Sdfstr^0_{\Sdfun^e X_0,X_1}"]
    &[2em]
    \Sdfun^{e+1}X_0\IWith \Sdfun X_1\ar[r,"\Sdfun\Sdfstr^0_{X_0,X_1}(e)"]
    &[2em]
    \Sdfun^{e+1}X_0\IWith \Sdfun^{e+1} X_1    
  \end{tikzcd}
\end{center}
and similarly for \(\Sdfstr^1(e)\). %
We can easily give a direct description of this morphism.
\begin{lemma}\label{lemma:sdfstr-it-charact}
  $\Sdfstr^i_{\List X0n}(e)
  =\Sin0(e)\IWith\cdots\IWith\Sdfun^eX_i\IWith\cdots\IWith\Sin0(e)$
  where $\Sin0(e)\in\cL(X,\Sfun^{e}X)$ is defined inductively by %
  $\Sin0(e)=\Id$ and $\Sin0(e+1)=(\Sfun{\Sin0(e)})\Compl\Sin0$.
\end{lemma}
That is, in $\Kl\cL$, $\Sin0(e+1)=(\Sdfun\Sin0(e))\Comp\Sin0$. The
proof is a straightforward induction.

\begin{lemma}\label{lemma:sdfunpart-it-expr}
  If $f\in\Kl\cL(\With{X_0}{X_1},Y)$, $d\in\Nat$ and $i\in\Eset{0,1}$,
  we have %
  $\Sdfunpart i^df=\Sdfun^df\Comp\Sdfstr^i_{X_0,X_1}(d)$.
\end{lemma}
\begin{proof}
  Immediate consequence of the functoriality of $\Sdfun$ in $\Kl\cL$
  and of the definition of $\Sdfstr^i_{X_0,X_1}(d)$.
\end{proof}

\begin{lemma}\label{lemma:proj-sin-deep}
  If $d<e$ we have %
  $\Sdfun^d\Sproj 0\Comp\Sin0(e)=\Sin0(e-1)$ and %
  $\Sdfun^d\Sproj 1\Comp\Sin0(e)=0$.
\end{lemma}
\begin{proof}
  By induction on $d$. Notice first that we are actually dealing with
  linear morphisms so that we can do our computations in $\cL$. For
  the base case we have, since $e>0$: %
  $\Sproj i\Compl\Sin0(e) =\Sproj i\Compl\Sfun(\Sin0(e-1))\Compl\Sin0
  =\Sin0(e-1)\Compl\Sproj i\Compl\Sin0$ %
  and we have $\Sproj i\Compl\Sin0=\Kronecker i0\Id$. %

  For the inductive case observe first that since $d+1<e$ we have
  $e\geq 2$. Then %
  \begin{align*}
    (\Sfun^{d+1}\Sproj i)\Compl\Sin0(e)
    &=(\Sfun^{d+1}\Sproj i)\Compl\Sfun(\Sin0(e-1))\Comp\Sin0\\
    &=\Sfun((\Sfun^d\Sproj i)\Compl\Sin0(e-1))\Compl\Sin0\\
    &=\Sfun(\Kronecker i0\Sin0(e-2))\Compl\Sin0
      \text{\quad by ind.~hypothesis}\\
    &=\Kronecker i0\Sin0(e-1)
  \end{align*}
  as contended.
\end{proof}

We generalize the canonical flip
$\Sflip\in\cL(\Sfun^2X,\Sfun^2X)$ to an iso
$\Sflipl l\in\cL(\Sfun^{l+2}X,\Sfun^{l+2}X)$ for each $l\in\Nat$
defined inductively by
\begin{align*}
  \Sflipl0=\Sflip
  \text{\quad and\quad}
  \Sflipl{l+1}=\Sflip\Compl(\Sfun\Sflipl l)\,.
\end{align*}

\begin{lemma}\label{lemma:sflipl-circular-permutation}
  Given $\List i0{l+1}\in\Eset{0,1}$, one has
  \begin{align*}
    \Sproj{i_{l+1}}
    \Compl\cdots\Compl\Sproj{i_0}
    \Compl\Sflipl l
    =\Sproj{i_0}
    \Compl\Sproj{i_{l+1}}
    \Compl\cdots\Compl\Sproj{i_1}\,.
  \end{align*}
\end{lemma}
\begin{proof}
  By induction on $l$. For $l=0$ the property holds by the very
  definition of $\Sflip$. Assume that the property holds for $l$ and
  let us prove it for $l+1$. So let $\List i0{l+2}\in\Eset{0,1}$. We
  have
  \begin{align*}
    \Sproj{i_{l+2}}
    \Compl\cdots\Compl\Sproj{i_0}
    \Compl\Sflipl{l+1}
    &=\Sproj{i_{l+2}}
      \Compl\cdots\Compl\Sproj{i_0}
      \Compl\Sflip\Compl(\Sfun\Sflipl l)\\
    &=\Sproj{i_{l+2}}
      \Compl\cdots\Compl\Sproj{i_2}\Compl\Sproj{i_0}
      \Compl\Sproj{i_1}
      \Compl(\Sfun\Sflipl l\\
    &=\Sproj{i_{l+2}}
      \Compl\cdots\Compl\Sproj{i_2}\Compl\Sproj{i_0}
      \Compl\Sflipl l
      \Compl\Sproj{i_1}
      \text{\quad by nat.~of }\Sproj{i_1}\\
    &=\Sproj{i_0}
      \Compl\Sproj{i_{l+2}}
      \Compl\cdots\Compl\Sproj{i_2}\Compl\Sproj{i_1}
      \text{\quad by ind.~hyp.}\qedhere
  \end{align*}
\end{proof}
This means that $\Sflipl l$ implements a circular permutation of
length $l+2$ on the indices.

\begin{lemma}\label{lemma:sdfun-circular}
  Let $f\in\Kl\cL(\With{X_0}{X_1},Y)$ and let $k\in\Nat$. Then
  $\Sdfunpart 1{\Sdfunpart0^{k+1} f},\Sdfunpart0^{k+1}{\Sdfunpart 1
    f}\in\Kl\cL(\With{\Sdfun^{k+1}X}{\Sdfun
    X},\Sdfun^{k+2}Y)$ satisfy the relation
  $\Sdfunpart 1{\Sdfunpart0^{k+1}{f}}
  =\Sflipl k\Comp\Sdfunpart0^{k+1}{\Sdfunpart 1f}$.
\end{lemma}
\begin{proof}
  For $k=0$, this is just Theorem~\ref{th:partial-der-commute}. For
  the inductive step we have
  \begin{align*}
    \Sdfunpart 1{\Sdfunpart0^{k+2}{f}}
    &=\Sflip\Comp\Sdfunpart0\Sdfunpart1\Sdfunpart0^{k+1}f
    \text{\quad by Theorem~\ref{th:partial-der-commute}}\\
    &=\Sflip\Comp\Sdfunpart0(\Sflipl k\Comp\Sdfunpart0^{k+1}\Sdfunpart1f)
    \text{\quad by ind.~hyp.}\\
    &=\Sflip\Comp\Sdfun\Sflipl k\Comp\Sdfunpart0^{k+2}\Sdfunpart1f
    \text{\quad by def.~of }\Sdfunpart0\\
    &=\Sflipl{k+1}\Comp\Sdfunpart0^{k+2}\Sdfunpart1f
      \text{\quad by def.~of }\Sflipl{k+1}\,.%
      \qedhere
  \end{align*}
\end{proof}

\subsection{Differentiation in the closed case} %
\label{sec:differentiation-closed}
Since our purpose is to provide the categorical foundations of
$\PCFD$, we require the category $\cL$ to be closed wrt.~its SMC
structure.

Remember that we consider the isos
$\Sdfun(\With ZX)\Isom\Sdfun Z\IWith\Sdfun X$ and
$\Sdfun{(\Simpl XY)}\Isom\Simpl X{\Sdfun Y}$ as identities: this is
our \Saxfun{} axiom of~\cite{Ehrhard23a} and we assume that the
corresponding iso, which is
$\Cur(\Sdfunpart0\Ev)=\Cur((\Sdfun\Ev)\Comp\Sdfstr^0_{\Simpl
  XY,X})\in\Kl\cL(\Sdfun(\Simpl XY),\Simpl X{\Sdfun Y})$, is the
identity morphism. With these identifications we have the following
equation.

\begin{lemma}\label{lemma:Sdfun-curry}
  Let $f\in\Kl\cL(\With ZX,Y)$ so that $\Cur f\in\Kl\cL(Z,\Simpl
  XY)$, %
  $\Sdfunpart0 f\in\Kl\cL(\With{\Sdfun Z}{X},\Sdfun Y)$ and %
  $\Sdfun(\Cur f)\in\Kl\cL(\Sdfun Z,\Simpl X{\Sdfun Y})$. Then %
  $\Sdfun(\Cur f)=\Cur(\Sdfunpart0 f)$.
\end{lemma}
\begin{proof}
  More precisely we must prove that %
  $\Cur(\Sdfunpart0\Ev)\Comp\Sdfun(\Cur f)=\Cur(\Sdfunpart0f)$ which
  boils down to the naturality of $\Sdfstr^0$ by simple computations
  in the CCC $\Kl\cL$.
\end{proof}

\begin{lemma}\label{lemma:sstr-sin0}
  The following diagram commutes in $\cL$
  \begin{equation*}
    \begin{tikzcd}
      \Tens{X}{Y}\ar[r,"\Tens{\Sin 0}Y"]\ar[dr,swap,"\Sin0"]
      &\Tens{\Sfun X}Y\ar[d,"\Sstr_0^{X,Y}"]\\
      &\Sfun\Tensp XY
    \end{tikzcd}
  \end{equation*}
\end{lemma}
This is easily proven using as usual the fact that $\Sproj0,\Sproj1$
are jointly monic.

\begin{lemma}\label{lemma:sdfun-Ev-expression}
  The following diagram commutes in $\Kl\cL$
  \begin{equation*}
    \begin{tikzcd}
      \With{\Simplp X{\Sdfun Y}}{\Sdfun X}
      \ar[r,"\Sdfun\Ev"]
      \ar[d,swap,"\Sdfstr^1_{\Simpl X{\Sdfun Y},X}"]
      &
      \Sdfun Y\\
      \With{\Simplp X{\Sdfun^2Y}}{\Sdfun X}
      \ar[r,"\Sdfun\Ev"]
      &
      \Sdfun^2Y\ar[u,swap,"\Sdfmult"]
    \end{tikzcd}
  \end{equation*}
\end{lemma}
\begin{proof}
  We need to come back to the definition of the functor $\Sdfun$. We
  expand the definition of
  $\Sdfun\Ev:\With{\Simplp X{\Sdfun Y}}{\Sdfun X}\to\Sdfun Y$ in $\cL$.
  \begin{align*}
    \Sdfun\Ev
    &=(\Sfun\Ev)\Compl\Sdiff_{\With{\Limplp{\Excl X}{Y}}{X}}\\
    &=(\Sfun\Evlin)
      \Compl\Sfun\Tensp{\Der{\Limpl{\Excl X}{Y}}}{\Excl X}
      \Compl\Sfun\Invp\Seelyt
      \Compl\Sdiff_{\With{\Limplp{\Excl X}{Y}}{X}}
    \text{\quad by def.~of }\Ev\\
    &=(\Sfun\Evlin)
      \Compl\Sfun\Tensp{\Der{\Limpl{\Excl X}{Y}}}{\Excl X}
      \Compl\Smont_{\Excl{\Limplp{\Excl X}{Y},\Excl X}}
      \Compl\Tensp{\Sdiff_{\Limpl{\Excl X}Y}}{\Sdiff_X}
      \Compl\Invp\Seelyt
      \text{\quad by~\ref{ax:daxwith} 
      }\\
    &=(\Sfun\Evlin)
      \Compl\Smont_{{\Limplp{\Excl X}{Y},\Excl X}}
      \Compl\Tensp{\Sfun\Der{\Limpl{\Excl X}{Y}}}{\Sfun\Excl X}
      \Compl\Tensp{\Sdiff_{\Limpl{\Excl X}Y}}{\Sdiff_X}
      \Compl\Invp\Seelyt\\
    &=(\Sfun\Evlin)
      \Compl\Smont_{{\Limplp{\Excl X}{Y},\Excl X}}
      \Compl\Tensp{\Der{\Sfun\Limplp{\Excl X}{Y}}}{\Sdiff_X}
      \Compl\Invp\Seelyt
      \text{\quad by~\ref{ax:daxchain} 
      }\\
    &=(\Sfun\Evlin)
      \Compl\Sfunadd
      \Compl(\Sfun\Sstr^0_{\Comma{\Limplp{\Excl X}{Y}}{\Excl X}})
      \Compl\Sstr^1_{\Limpl{\Excl X}{\Sfun Y},\Excl X}
      \Compl\Tensp{\Der{\Sfun\Limplp{\Excl X}{Y}}}{\Sdiff_X}
      \Compl\Invp\Seelyt
    \text{\quad by def.~of }\Smont\\
    &=\Sfunadd
      \Compl(\Sfun^2\Evlin)
      \Compl(\Sfun\Sstr^0_{\Comma{\Limplp{\Excl X}{Y}}{\Excl X}})
      \Compl\Sstr^1_{\Limpl{\Excl X}{\Sfun Y},\Excl X}
      \Compl\Tensp{\Der{\Sfun\Limplp{\Excl X}{Y}}}{\Sdiff_X}
      \Compl\Invp\Seelyt\\
    &=\Sfunadd
      \Compl(\Sfun\Evlin)
      \Compl\Sstr^1_{\Limpl{\Excl X}{\Sfun Y},\Excl X}
      \Compl\Tensp{\Der{\Sfun\Limplp{\Excl X}{Y}}}{\Sdiff_X}
      \Compl\Invp\Seelyt
  \end{align*}
  by the identification
  $\Sfun\Limplp{\Excl X}{Y}=\Limplp{\Excl X}{\Sfun Y}$. On the other
  hand we have
  \begin{align*}
    (\Sdfun\Ev)\Comp\Sdfstr^1_{\Simpl X{\Sdfun Y},X}
    &=(\Sfun\Evlin)
      \Compl\Smont_{{\Limplp{\Excl X}{\Sfun Y},\Excl X}}
      \Compl\Tensp{\Der{\Sfun\Limplp{\Excl X}{\Sfun Y}}}{\Sdiff_X}
      \Compl\Invp\Seelyt
      \Compl\Excl{\Sdfstr^1_{\Simpl X{\Sdfun Y},X}}
    \text{\quad as above}\\
    &=(\Sfun\Evlin)
      \Compl\Smont_{{\Limplp{\Excl X}{\Sfun Y},\Excl X}}
      \Compl\Tensp{\Der{\Sfun\Limplp{\Excl X}{\Sfun Y}}}{\Sdiff_X}
      \Compl\Invp\Seelyt
      \Compl\Excl{\Withp{\Sin0}{\Sfun X}}\\
      &\text{\Textsep by def.~of }\Sdfstr^1_{\Simpl X{\Sdfun Y},X}\\
    &=(\Sfun\Evlin)
      \Compl\Smont_{{\Limplp{\Excl X}{\Sfun Y},\Excl X}}
      \Compl\Tensp{\Der{\Sfun\Limplp{\Excl X}{\Sfun Y}}}{\Sdiff_X}
      \Compl{\Tensp{\Excl{\Sin0}}{\Excl{\Sfun X}}}
      \Compl\Invp\Seelyt\\
    &=(\Sfun\Evlin)
      \Compl\Smont_{{\Limplp{\Excl X}{\Sfun Y},\Excl X}}
      \Compl{\Tensp{\Sin0}{\Sdiff_X}}
      \Compl\Tensp{\Der{\Limplp{\Excl X}{\Sfun Y}}}{\Excl{\Sfun X}}
      \Compl\Invp\Seelyt\\
    &=(\Sfun\Evlin)
      \Compl\Sfunadd
      \Compl(\Sfun\Sstr^1_{\Limpl{\Excl X}{\Sfun Y},\Excl X})
      \Compl\Sstr^0_{\Limpl{\Excl X}{\Sfun Y},\Sfun\Excl X}
      \Compl{\Tensp{\Sin0}{\Sfun\Excl X}}\\
    &\hspace{4em}
      \Compl{\Tensp{\Limplp{\Excl X}{\Sfun Y}}{\Sdiff_X}}
      \Compl\Tensp{\Der{\Limplp{\Excl X}{\Sfun Y}}}{\Excl{\Sfun X}}
      \Compl\Invp\Seelyt
    \text{\quad by def.~of }\Smont\\
    &=(\Sfun\Evlin)
      \Compl\Sfunadd
      \Compl(\Sfun\Sstr^1_{\Limpl{\Excl X}{\Sfun Y},\Excl X})
      \Compl\Sin0
      \Compl\Tensp{\Der{\Limplp{\Excl X}{\Sfun Y}}}{\Sdiff_X}
      \Compl\Invp\Seelyt\text{\ \ by Lemma~\ref{lemma:sstr-sin0}}\\
    &=(\Sfun\Evlin)
      \Compl\Sfunadd
      \Compl\Sin0
      \Compl\Sstr^1_{\Limpl{\Excl X}{\Sfun Y},\Excl X}
      \Compl\Tensp{\Der{\Limplp{\Excl X}{\Sfun Y}}}{\Sdiff_X}
      \Compl\Invp\Seelyt\\
    &=(\Sfun\Evlin)
      \Compl\Sstr^1_{\Limpl{\Excl X}{\Sfun Y},\Excl X}
      \Compl\Tensp{\Der{\Limplp{\Excl X}{\Sfun Y}}}{\Sdiff_X}
      \Compl\Invp\Seelyt
  \end{align*}
  which proves our contention.
\end{proof}

\subsection{The case of multilinear morphisms} %
\label{sec:diff-multilinear}
For each $i\in\{0,\dots,n\}$ we can define a tensorial generalized
strength
\[
  \Sstr^i_{\List X0n}
  \in\cL(X_0\ITens\cdots\ITens\Sfun{X_i}\ITens
  \cdots\ITens
  X_n,\Sfun{(X_0\ITens\cdots\ITens X_n)})\,.
\]

Let $l\in\cL(X_0\ITens\cdots\ITens X_n,Y)$, then we define
$\Mlin l\in\Kl\cL(X_0\IWith\cdots\IWith X_n,Y)$ as the following
composition of morphisms
\begin{center}
  \begin{tikzcd}
    \Exclp{X_0\IWith\cdots\IWith X_n}
    \ar[r,"\Inv{(\Seely^n)}"]
    &\Excl{X_0}\ITens\cdots\ITens\Excl{X_n}
    \ar[r,"\Der{X_0}\ITens\cdots\ITens\Der{X_n}"]
    &[3.8em]X_0\ITens\cdots\ITens X_n
    \ar[r,"l"]
    &[-1em]Y    
  \end{tikzcd}
\end{center}
A morphism in $\Kl\cL(X_0\IWith\cdots\IWith X_n,Y)$ definable in that
way can be called an \emph{$n+1$-linear} morphism (that is, a
multilinear morphisms with $n+1$ arguments) for the following reason.

\begin{lemma}\label{lemma:Kleisli-multilin}
  With these notations, we have
  \[
    \Mlin l\Comp(X_0\IWith\cdots\IWith X_{i-1}\IWith 0\IWith
    X_{i+1}\cdots\IWith X_n)=0
  \] and if
  $f_0,f_1\in\cL(Z,X_i)$ are summable then so are
  \[
    \Mlin l\Comp(X_0\IWith\cdots\IWith X_{i-1}\IWith f_j\IWith
    X_{i+1}\cdots\IWith X_n)
  \] for $j=0,1$ and
  \begin{multline*}
    \Mlin l\Comp(X_0\IWith\cdots\IWith X_{i-1}\IWith(f_0+f_1)\IWith
    X_{i+1}\cdots\IWith X_n)\\
    =\Mlin l\Comp(X_0\IWith\cdots\IWith
    X_{i-1}\IWith f_0\IWith X_{i+1}\cdots\IWith X_n)\\
    +\Mlin
    l\Comp(X_0\IWith\cdots\IWith X_{i-1}\IWith f_1\IWith
    X_{i+1}\cdots\IWith X_n)\,.
  \end{multline*}
\end{lemma}
\begin{proof}
  For the summability, we have
  $\Mlin l\Comp(X_0\IWith\cdots\IWith X_{i-1}\IWith f_j\IWith
  X_{i+1}\cdots\IWith X_n)=\Mlin{l_j}$ where
  $l_j=l\Compl(X_0\ITens\cdots\ITens f_j\ITens\cdots\ITens X_n)$ and
  since $f_0,f_1$ are summable so are
  $X_0\ITens\cdots\ITens f_j\ITens\cdots\ITens X_n$ for $j=0,1$ %
  by \Saxdist{} of~\cite{Ehrhard23a} with
  $X_0\ITens\cdots\ITens(f_0+f_1)\ITens\cdots\ITens X_n$ as sum. The
  result follows by
  Lemma~12 %
  of~\cite{Ehrhard23a}. For the property relative to $0$ we use similarly
  the fact that $0$ is absorbing for $\ITens$ and for composition in
  $\cL$.
\end{proof}

\begin{theorem} %
  \label{th:sdiff-multilin}
  Let \(f\in\cL(X_0\IWith\cdots\IWith X_n,Y)\) be $n+1$-linear. %
  Then
  \(
  \Sdfun f\in\Kl\cL(\Sdfun X_0\IWith\cdots\IWith\Sdfun X_n,\Sdfun Y)
  \) %
  is also $n+1$-linear.
\end{theorem}
\begin{proof}
  We assume $n=1$ for the sake of readability, the general case is not
  more difficult. Let $l\in\cL(\Tens{X_0}{X_1},Y)$ be such that %
  $f=\Mlin l$. Then we set %
  \( l'=(\Sfun l)\Compl\Smont_{X_0,X_1} \in\cL(\Tens{\Sfun X_0}{\Sfun
    X_1},\Sfun Y) \) %
  and using~\ref{ax:daxwith} 
  and~\ref{ax:daxchain} 
  one shows that %
  \(\Mlin{l'}=\Sdfun f\).
\end{proof}

\begin{lemma}\label{lemma:sproj-sdiff-Seely}
  Given $\List X0n\in\Obj\cL$ and $X=X_0\IWith\cdots\IWith X_n$, we have
  \begin{align*}
    \Sproj 0\Compl\Sdiff_X\Compl\Seely^n
    &=\Seely^n\Compl(\Excl{\Sproj 0}\ITens\cdots\ITens\Excl{\Sproj 0})\\
    \Sproj 1\Compl\Sdiff_X\Compl\Seely^n
    &=\Seely^n\Compl((\Sproj 1\Compl\Sdiff_{X_0})\ITens\Excl{\Sproj 0}\ITens
      \cdots\ITens\Excl{\Sproj 0}
      +\cdots
      +\Excl{\Sproj 0}\ITens\cdots\ITens\Excl{\Sproj 0}
      \ITens(\Sproj 1\Compl\Sdiff_{X_0}))\,.
  \end{align*}
  which both belong to
  $\cL(\Excl{\Sfun{X_0}}\ITens\cdots\ITens\Excl{\Sfun{X_n}},\Excl{X})$.
\end{lemma}
\begin{proof}
  For the sake of readability we assume that $n=1$, the general case
  is not more difficult. By Axiom~\ref{ax:daxwith} 
  we have
  $\Sdiff_X\Compl\Seelyt
  =\Sfun\Seelyt\Compl\Smont_2\Compl\Tensp{\Sdiff_{X_0}}{\Sdiff_{X_1}}$
  hence
  \begin{align*}
    \Sproj 0\Compl\Sdiff_X\Compl\Seelyt
    &=\Sproj 0\Compl\Sfun\Seelyt\Compl\Smont_2
      \Compl\Tensp{\Sdiff_{X_0}}{\Sdiff_{X_1}}\\
    &=\Seelyt\Compl\Sproj0\Compl\Smont_2
      \Compl\Tensp{\Sdiff_{X_0}}{\Sdiff_{X_1}}\\
    &=\Seelyt\Compl\Tensp{\Sproj0}{\Sproj0}
      \Compl\Tensp{\Sdiff_{X_0}}{\Sdiff_{X_1}}\\
    &=\Seelyt\Tensp{\Excl{\Sproj 0}}{\Excl{\Sproj 0}}
  \end{align*}
  and
  \begin{align*}
    \Sproj1\Compl\Sdiff_X\Compl\Seelyt
    &=\Seelyt\Compl\Sproj1\Compl\Smont_2
      \Compl\Tensp{\Sdiff_{X_0}}{\Sdiff_{X_1}}\\
    &=\Seelyt\Compl(\Tens{\Sproj1}{\Sproj0}+\Tens{\Sproj0}{\Sproj1})
      \Compl\Tensp{\Sdiff_{X_0}}{\Sdiff_{X_1}}\\
    &=\Seelyt\Compl(\Tens{(\Sproj1\Compl\Sdiff_{X_0})}{\Excl{\Sproj0}}
      +\Tens{\Excl{\Sproj0}}{(\Sproj1\Compl\Sdiff_{X_1})})\,.
      \qedhere
  \end{align*}
\end{proof}

\begin{theorem}\label{th:sproj-sdfun-mlin}
  With the same notations,
  $\Sdfun\Mlin l\in\Kl\cL(\Sdfun X_0\IWith\cdots\IWith\Sdfun
  X_n,\Sdfun Y)$ satisfies
  \begin{align*}
    \Sproj0\Comp\Sdfun\Mlin l
    &=\Mlin l\Comp(\Sproj0\IWith\cdots\IWith\Sproj 0)\\
    \Sproj1\Comp\Sdfun\Mlin l
    &=\Mlin l\Comp(\Sproj1\IWith\cdots\IWith\Sproj 0)
    +\cdots+\Mlin l\Comp(\Sproj0\IWith\cdots\IWith\Sproj 1)\,.
  \end{align*}
\end{theorem}
\begin{proof}
  Immediate consequence of Lemma~\ref{lemma:sproj-sdiff-Seely} and of
  the naturality of the $\Sproj j$'s wrt.~$\Sfun$ in $\cL$.
\end{proof}

For $i\in\Eset{1,\dots,n}$, we have
\(
(\Sfun l)\Compl\Sstr^i
\in\cL(X_0\ITens\cdots\ITens\Sfun{X_i}\ITens\cdots\ITens
X_n,\Sfun Y)
\). %
Remember that $\Sproj j\in\cL(\Sfun X,X)$. Given a
``linear'' morphism $h\in\cL(X,Y)$, we use the same notation $h$ for
the corresponding morphism $\Kllin h=h\Compl\Der X\in\Kl\cL(X,Y)$.

Let $l\in\cL(X_0\ITens\cdots\ITens X_n,Y)$, since $\Mlin l$ is
multilinear, its partial derivatives should be ``trivial'', the
purpose of the next result is to state precisely this
triviality. Given $i\in\Eset{0,\dots,n}$, we define the $i$-th
``partial application'' of the functor $\Sfun$ to $l$ as
$\Sfunpart il=(\Sfun l)\Compl\Sstr^i\in\cL(X_0\ITens\cdots\ITens\Sfun
X_i\ITens\cdots\ITens X_n,\Sfun Y)$.
\begin{theorem}\label{th:diffpart-multilin}
  For each $i\in\{0,\dots,n\}$ we have
  $\Sdfunpart i{\Mlin l}=\Mlin{\Sfunpart il}$ and for
  $j\in\Eset{0,1}$, we have
  $\Sproj j\Comp\Sdfunpart i{\Mlin l}=\Mlin
  l\Comp(X_0\IWith\cdots\IWith\Sproj j\IWith\cdots\IWith X_n)$.
\end{theorem}
\begin{proof}
  For the sake of readability we take $n=1$. The general case is
  conceptually not more difficult to deal with, just harder to read
  due to cumbersome notations. We first prove the second equation for
  $i=1$ (the case \(i=0\) is similar), namely, in $\cL$:
  \begin{align*}
    \Sproj j\Compl\Sdfunpart 1{\Mlin l}
    =\Mlin l\Compl\Excl{(\With{X_0}{\Sproj j})}
  \end{align*}
  where $l\in\cL(\Tens{X_0}{X_1},Y)$ so that
  $\Mlin l=l\Compl\Tensp{\Der{X_0}}{\Der{X_1}}\Compl\Invp\Seelyt$ is a
  bilinear morphism in $\Kl\cL$. We have
  \begin{align*}
    \Sproj j\Compl\Sdfunpart 1{\Mlin l}
    =\Sproj j
      \Compl\Sfun\Mlin l
      \Compl\Sdiff_{\With{X_0}{X_1}}
      \Compl\Excl{\Sdfstr^1}
    =\Mlin l\Compl\Sproj j
      \Compl \Sdiff_{\With{X_0}{X_1}}
      \Compl\Excl{\Sdfstr^1}\,.
  \end{align*}
  Therefore
  \begin{align*}
    \Sproj 0\Compl\Sdfunpart 1{\Mlin l}
    &=\Mlin l\Compl\Excl{\Withp{\Sproj0}{\Sproj0}}\Compl\Excl{\Sdfstr^1}
      \text{\quad by~\ref{ax:daxlocal} 
      }\\
    &=\Mlin l\Compl\Excl{\Withp{X_0}{\Sproj0}}
  \end{align*}
  since
  $\Withp{\Sproj0}{\Sproj0}\Compl\Sdfstr^1=\Withp{X_0}{\Sproj0}$. Next
  \begin{align*}
    \Sproj 1\Compl\Sdfunpart 1{\Mlin l}
    &=\Mlin l\Compl\Seelyt
      \Compl(\Tensp{(\Sproj1\Compl\Sdiff_{X_0})}{\Excl{\Sproj0}}
      +\Tensp{\Excl{\Sproj0}}{(\Sproj1\Compl\Sdiff_{X_1})})
      \Compl\Invp{\Seelyt}\Compl\Excl{\Sdfstr^1}\\
    &\text{\Textsep by~\ref{ax:daxwith} 
      and def.~of }\Smont\\
    &=l\Compl\Tensp{\Der{X_0}}{\Der{X_1}}
      \Compl\Invp{\Seelyt}
      \Compl\Seelyt
      \Compl(\Tensp{(\Sproj1\Compl\Sdiff_{X_0})}{\Excl{\Sproj0}}
      +\Tensp{\Excl{\Sproj0}}{(\Sproj1\Compl\Sdiff_{X_1})})
      \Compl\Invp{\Seelyt}\Compl\Excl{\Sdfstr^1}\\
    &\text{\Textsep by def.~of }\Mlin l\\
    &=l
      \Compl
      (\Tens{(\Der{X_0}\Compl\Sproj1\Compl\Sdiff_{X_0})}
            {(\Der{X_1}\Compl\Excl{\Sproj0})}
      +\Tens{(\Der{X_0}\Compl\Excl{\Sproj0})}
            {(\Der{X_1}\Compl\Sproj1\Compl\Sdiff_{X_1})})
      \Compl\Compl\Invp{\Seelyt}\Compl\Excl{\Sdfstr^1}\\
    &=l
      \Compl
      (\Tens{(\Sproj1\Compl\Sfun{\Der{X_0}}\Compl\Sdiff_{X_0})}
            {(\Der{X_1}\Compl\Excl{\Sproj0})}\\
    &\Textsep  +\Tens{(\Der{X_0}\Compl\Excl{\Sproj0})}
            {(\Sproj1\Compl\Sfun{\Der{X_1}}\Compl\Sdiff_{X_1})})
      \Compl\Compl\Invp{\Seelyt}\Compl\Excl{\Sdfstr^1}
    \text{\quad by nat.~of }\Sproj1\\
    &=l
      \Compl
      (\Tens{(\Sproj1\Compl\Der{\Sfun X_0})}
            {(\Der{X_1}\Compl\Excl{\Sproj0})}
      +\Tens{(\Der{X_0}\Compl\Excl{\Sproj0})}
            {(\Sproj1\Compl\Der{\Sfun X_1})})
      \Compl\Compl\Invp{\Seelyt}\Compl\Excl{\Sdfstr^1}\\
    &\text{\Textsep by~\ref{ax:daxchain} 
      }\\
    &=l
      \Compl
      (\Tens{(\Der{X_0}\Compl\Excl{\Sproj1})}
            {(\Der{X_1}\Compl\Excl{\Sproj0})}
      +\Tens{(\Der{X_0}\Compl\Excl{\Sproj0})}
            {(\Der{X_1}\Compl\Excl{\Sproj1})})
      \Compl\Compl\Invp{\Seelyt}\Compl\Excl{\Sdfstr^1}\\
    &=\Mlin l
      \Compl\Seelyt
      \Compl(\Tens{\Excl{\Sproj 0}}{\Excl{\Sproj1}}
      +\Tens{\Excl{\Sproj 1}}{\Excl{\Sproj0}})
      \Compl\Compl\Invp{\Seelyt}\Compl\Excl{\Sdfstr^1}\\
    &=\Mlin l
      \Compl(\Excl{\Withp{\Sproj0}{\Sproj1}}+\Excl{\Withp{\Sproj1}{\Sproj0}})
      \Compl\Excl{\Sdfstr^1}\\
    &=\Mlin l\Compl\Excl{\Withp{X_0}{\Sproj 1}}
      +\Mlin l\Compl\Excl{\Withp{0}{\Sproj0}}
  \end{align*}
  since $\Withp{\Sproj0}{\Sproj1}\Sdfstr^1=\Withp{X_0}{\Sproj1}$ and
  $\Withp{\Sproj1}{\Sproj0}\Sdfstr^1=\Withp{0}{\Sproj0}$. Finally we have
  \begin{align*}
    \Mlin l\Compl\Excl{\Withp{0}{\Sproj0}}
    &=l\Compl\Tensp{\Der{X_0}}{\Der{X_1}}
      \Compl\Invp{\Seelyt}
      \Compl\Excl{\Withp{0}{\Sproj0}}\\
    &=l\Compl\Tensp{\Der{X_0}}{\Der{X_1}}
      \Compl\Tensp{\Excl 0}{\Excl{\Sproj0}}
      \Compl\Invp\Seelyt\\
    &=l\Compl\Tensp0{\Sproj0}
      \Compl\Invp\Seelyt=0      
  \end{align*}
  proving our contention. Now we prove the first equation for $i=1$
  (the case \(i=0\) is similar). For $j\in\Eset{0,1}$ we have %
  \begin{align*}
    \Sproj j\Compl\Mlin{(\Sfun l)\Compl\Sstr^1}
    &=\Sproj j
      \Compl(\Sfun l)
      \Compl\Sstr^1
      \Compl\Tensp{\Der{X_0}}{\Der{\Sfun X_1}}
      \Compl\Invp\Seelyt\\
    &=l
      \Compl\Sproj j
      \Compl\Sstr^1
      \Compl\Tensp{\Der{X_0}}{\Der{\Sfun X_1}}
      \Compl\Invp\Seelyt\\
    &=l
      \Compl\Tensp{X_0}{\Sproj j}
      \Compl\Tensp{\Der{X_0}}{\Der{\Sfun X_1}}
      \Compl\Invp\Seelyt
    \text{\quad by def.~of }\Sstr^1\\
    &=l
      \Compl\Tensp{\Der{X_0}}{\Der{X_1}}
      \Compl\Tensp{\Excl{X_0}}{\Excl{\Sproj j}}
      \Compl\Invp\Seelyt\\
    &=l
      \Compl\Tensp{\Der{X_0}}{\Der{X_1}}
      \Compl\Invp\Seelyt
      \Compl\Excl{\Withp{X_0}{\Sproj j}}\\
    &=\Mlin l\Compl\Excl{\Withp{X_0}{\Sproj j}}
    \text{\quad by definition of }\Mlin l\\
    &=\Sproj j\Compl\Sdfunpart 1{\Mlin l}
    \text{\quad as we have just proven,}
  \end{align*}
  which proves our first equation by the fact that $\Sproj0,\Sproj1$
  are jointly epic.
\end{proof}

\subsection{The basic multilinear constructs}
\label{sec:basic-categ-constr}
Now we introduce the multilinear operations which will interpret the
basic constructs of $\PCFD$.
We make the following assumption about $\cL$.
\begin{Axicond}{\Laxint}{ax:laxint}
  The functor $X\mapsto\Plus\Sone X$ from $\cL$ to $\cL$ has an initial
  algebra $\Snat$.
\end{Axicond}
This means that there is a morphism
$\chi\in\cL(\Plus\Sone\Snat,\Snat)$ such that for any
$f\in\cL(\Plus\Sone X,X)$ there is exactly one morphism
$g\in\cL(\Snat,X)$ such that $f\Compl(\Plus\Sone g)=g\Compl\chi$.
We know that there is only one such morphism $\chi$, and that this
morphism is an iso (Lambek's Lemma). We assume that $\chi$ is the
identity to simplify notations, so that $\Snat=\Plus\Sone\Snat$ ``on
the nose''.
Given $f\in\cL(\Plus\Sone X,X)$ we use $\Snatit f$ for the unique
element of $\cL(\Snat,X)$ such that
\( \Snatit f=f\Compl(\Plus\Sone{\Snatit f}) \).

We set $\Ssuc=\Inj 1\in\cL(\Snat,\Plus\Sone\Snat)=\cL(\Snat,\Snat)$
which represents the successor constructor on integers and
$\Szero=\Inj0\in\cL(\Sone,\Snat)$ which represents the zero
constant.
It follows that for each $\nu\in\Nat$ we can define the constants
$\Snum\nu\in\cL(\Sone,\Snat)$ by $\Snum0=\Szero$ and
$\Snum{\nu+1}=\Ssuc{(\Snum\nu)}$.

Next we define the predecessor morphism
$\Spred=\Cotuple{\Inj0,\Snat}\in\cL(\Plus\Sone\Snat,\Snat)$, that is
$\Spred\in\cL(\Snat,\Snat)$.
We have $\Spred\Compl\Snum0=\Snum 0$ and
$\Spred\Compl\Snum{\nu+1}=\Snum\nu$.

Next notice that we have a morphism
$\Coalg\Snat\in\cL(\Snat,\Excl\Snat)$ which turns $\Snat$ into a
$\Excl\_$-coalgebra (that is, an object of $\Em\cL$, the
Eilenberg-Moore category of the comonad $\Excl\_$).
This morphism is $\Coalg\Snat=\Snatit f$ where
$f=\Cotuple{\Excl{\Inj0}\Compl\Coalg\Sone
  ,\Excl\Snat}\in\cL(\Plus\Sone{\Excl\Snat},\Excl\Snat)$ where
$\Coalg\Sone
=\Excl{\Invp\Seelyz}\Compl\Digg\Top\Compl\Seelyz\in\cL(\Sone,\Excl\Sone)$
is the canonical $\Excl\_$-coalgebra structure of $\Sone$.
This allows one in particular to define an erasing morphism
$\Coalgw\Snat=\Weak\Snat\Compl\Coalg\Snat\in\cL(\Snat,\Sone)$ as well
as a duplicating morphism
$\Coalgc\Snat
=\Tensp{\Der\Snat}{\Der\Snat}\Compl\Contr\Snat\Compl\Coalg\Snat
\in\cL(\Snat,\Tens\Snat\Snat)$.
Remember indeed that, for any object \(X\) of \(\cL\), the object
\(\Excl X\) has a canonical structure of cocommutative comonoid with
counit \(\Weak X\in\cL(\Excl X,\Sone)\) and comultiplication
\(\Contr X\in\cL(\Excl X,\Tens{\Excl X}{\Excl X})\).

Given an object $X$ we set
$\Slet=\Evlin\Compl\Sym\Compl\Tensp{\Coalg\Snat}{(\Limpl{\Excl\Snat}{X})}
\in\cL(\Tens\Snat{(\Limpl{\Excl\Snat}{X})},X)$.

Last we define
$\Sif=\Evlin\Tensp g\Snat\in\cL(\Snat\ITens(\With XX),X)$ where
\[
  g=\Cotuple{\Curlinp{\Proj0\Compl\Leftu},
  \Curlinp{\Proj1\Compl\Leftu}\Compl\Coalgw\Snat}
  \in\cL(\Snat,\Limpl{\With XX}{X})
\]
where $\Proj j\Compl\Leftu$ is typed as follows
\begin{center}
  \begin{tikzcd}
    \Tens{\Sone}{\Withp XX}\ar[r,"\Leftu"]
    &\Withp XX\ar[r,"\Proj j"]
    &X
  \end{tikzcd}
\end{center}
so that the two following diagrams commute
\begin{equation}\label{eq:diag-conditional}
  {\footnotesize
    \begin{tikzcd}
    \Tens\Sone{\Withp XX}\ar[r,"\Tens\Szero{\Withp XX}"]
    \ar[d,swap,"\Leftu"]
    &[2em]\Tens\Snat{\Withp XX}\ar[d,"\Sif"]\\
    \With XX\ar[r,"\Proj0"]
    &X
  \end{tikzcd}
  \\
  \begin{tikzcd}
    \Tens\Snat{\Withp XX}\ar[rr,"\Tens\Ssuc{\Withp XX}"]
    \ar[d,swap,"\Tens{\Coalgw\Snat}{\Withp XX}"]
    &&\Tens\Snat{\Withp XX}\ar[d,"\Sif"]\\
    \Tens\Sone{\Withp XX}\ar[r,"\Leftu"]
    &\With XX\ar[r,"\Proj1"]
    &X
  \end{tikzcd}}
\end{equation}

\subsection{Syntactic constructs in the model}
\label{sec:syn-basic-constr}
We introduce now semantic constructs on morphisms which exactly
mimic the syntax so as to make the translation from syntax to
semantics straightforward.

First, given $n\in\Nat$ we also use the notation $\Snum n$ for the
morphism $\Snum n\Compl\Weak{Z}\in\Kl\cL(Z,\Snat)$.

Given $f\in\Kl\cL(Z,\Snat)$ we define
$\Ssuc(f)=\Ssuc\Compl f\in\Kl\cL(Z,\Snat)$ and similarly
$\Spred(f)=\Spred\Compl f$.

More generally given $d\in\Nat$ and $f\in\Kl\cL(Z,\Sdfun^d\Snat)$ we
set %
$\Ssuc^d(f)=(\Sdfun^d\Ssuc)\Comp f\in\Kl\cL(Z,\Sdfun^d\Snat)$. %
We define similarly $\Spred^d(f)\in\Kl\cL(Z,\Sdfun^d\Snat)$.

We have defined %
$\Mlin\Sif\in\Kl\cL(\Snat\IWith\Withp XX,X)$. So we have %
\[
  \Sdfunpart 0^d\Mlin\Sif\in\Kl\cL(\Sdfun^d\Snat\IWith\Withp
  XX,\Sdfun^d X)
\]
(notice that this is not a trilinear morphism, but a bilinear one,
separately linear in $\Sdfun^k\Snat$ and $\With XX$).
Let $g\in\Kl\cL(Z,\Sdfun^d\Snat)$ and $f_j\in\Kl\cL(Z,X)$ for $j=0,1$.
We set %
\[
  \Sif^d(g,f_0,f_1)={\Sdfunpart0^d\Mlin\Sif}\Comp\Tuple{g,f_0,f_1}
  \in\Kl\cL(Z,\Sdfun^k X)\,.
\] %
Notice that
$\Sdfunpart0^d{\Mlin\Sif}=\Mlin{\Sfunpart0^d\Sif}$ %
(this notation is introduced in Section~\ref{sec:diff-multilinear}).

We have defined %
$\Mlin\Slet\in\Kl\cL(\With\Snat{(\Simpl{\Snat}X)},X)$ so that %
\[
  \Sdfunpart0^d\Mlin\Slet
  \in\Kl\cL(\With{\Sdfun^d\Snat}{(\Simpl{\Snat}X)},\Sdfun^dX)\,.
\]
Let %
$g\in\Kl\cL(Z,\Sdfun^d\Snat)$ and %
$f\in\Kl\cL(\With Z\Snat,X)$ so that %
$\Cur f\in\Kl\cL(Z,\Simpl\Snat X)$, we set %
\[
  \Slet^d(g,f)=\Sdfunpart0^d\Mlin\Slet\Comp\Tuple{g,\Cur f}
  \in\Kl\cL(Z,\Sdfun^dX)\,.
\]

If $f\in\Kl\cL(Z,\Simpl XY)$ and $g\in\Kl\cL(Z,X)$ then we define
$\App fg\in\Kl\cL(Z,Y)$ as $\App fg=\Ev\Comp\Tuple{f,g}$.

If $f\in\Kl\cL(Z,\Simpl XY)$ we have
$\Ev\Comp\Withp fX\in\Kl\cL(\With ZX,Y)$ and hence
$\Sdfunpart 1{(\Ev\Comp\Withp fX)}\in\Kl\cL(\With Z{\Sdfun X},\Sdfun
Y)$ so that we set
\[
  \Sdfunc(f)=\Cur{(\Sdfunpart 1{(\Ev\Comp\Withp
      fX)})}\in\Kl\cL(Z,\Simpl{\Sdfun X}{\Sdfun Y})\,.
\]
Notice that
\begin{align*}
  \Sdfunpart 1{(\Ev\Comp\Withp fX)}
  &=\Sdfun\Ev\Comp\Withp{\Sdfun f}{\Sdfun X}\Comp\Sdfstr_{Z,X}^1\\
  &=\Sdfun\Ev\Comp\Sdfstr_{\Simpl XY,X}^1\Comp(\With f{\Sdfun X})
\end{align*}
by naturality of $\Sdfstr^1$ and hence %
\begin{equation}
  \label{eq:dcur-dint}
  \Sdfunc(f)=\Sdfunint^{X,Y}\Comp f
\end{equation}
where %
\[
  \Sdfunint^{X,Y}= \Cur{(\Sdfun\Ev\Comp\Sdfstr_{\Simpl
      XY,X}^1)}\in\Kl\cL(\Simpl XY,\Simpl{\Sdfun X}{\Sdfun Y})
\]
is the ``internalization'' of the functor $\Sdfun$ made possible by
its strength.

Remember that a morphism %
\(f\in\Kl\cL(X,Y)\) is linear if \(f=\Kllin g\) for some %
\(g\in\cL(X,Y)\) and that, when this $g$ exists, it is unique. 
\begin{lemma} %
  \label{lemma:dint-linear}
  The morphism
  $\Sdfunint^{X,Y}=\Cur{(\Sdfun\Ev\Comp\Sdfstr_{\Simpl XY,X}^1)}
  \in\Kl\cL(\Simpl XY,\Simpl{\Sdfun X}{\Sdfun Y})$ is linear.
\end{lemma}
\begin{proof}
  This results from the fact that $\Ev$ is left-linear and %
  $\Sdfstr^1_{\Simpl XY,X}$ is linear on %
  $\With{\Simplp XY}{\Sdfun X}$.
\end{proof}
So we shall also consider tacitly \(\Sdfunint\) as an element of %
\(\cL(\Simpl XY,\Simpl{\Sdfun X}{\Sdfun Y})\).

If $f\in\Kl\cL(Z,\Sdfun X)$ and $j\in\Eset{0,1}$ we set
$\Sproj j(f)=\Sproj j\Compl f\in\Kl\cL(Z,X)$, and if
$f\in\Kl\cL(Z,\Sdfun^2X)$ we set
$\Sfunadd(f)=\Sfunadd\Compl f\in\Kl\cL(Z,\Sdfun X)$ and
$\Sflip(f)=\Sflip\Compl f\in\Kl\cL(Z,\Sdfun^2 X)$. Last if
$f\in\Kl\cL(Z,X)$ we set
$\Sin j(f)=\Sin j\Compl f\in\Kl\cL(Z,\Sdfun Y)$.

\begin{lemma}\label{lemma:sdfun-nat}
  For any object $X$ of $\cL$ we have
  \begin{align*}
    \Sdfunpart 1{\Mlin\Sif_X}
    &=\Mlin\Sif_{\Sdfun X}\in\Kl\cL(\Snat\IWith\Sdfun\Withp XX,\Sdfun X)\\
    \Sdfunpart 1{\Mlin\Slet_X}
    &=\Mlin\Slet_{\Sdfun X}\in\Kl\cL(\Snat\IWith\Sdfun\Simplp\Snat X,\Sdfun X)
  \end{align*}
\end{lemma}
\begin{proof}
  We have
  \begin{align*}
    \Sproj0\Comp\Sdfunpart 1{\Mlin\Sif_X}
    &=\Sproj 0
      \Comp\Sdfun{\Mlin\Sif_X}
      \Comp\Sdfstr^1_{\Snat,\With XX}\\
    &=\Mlin\Sif_{X}
      \Comp{\Withp{\Sproj 0}{\Withp{\Sproj 0}{\Sproj 0}}}
      \Comp\Sdfstr^1_{\Snat,\With XX}
    \text{\quad by Theorem~\ref{th:sproj-sdfun-mlin}}\\
    &=\Mlin\Sif_{X}
      \Comp{\Withp{\Snat}{\Withp{\Sproj 0}{\Sproj 0}}}\\
    &=\Sproj0\Comp\Mlin\Sif_{\Sdfun X}
  \end{align*}
  by naturality of $\Sif_X$ with respect to $X$. Next
  \begin{align*}
    \Sproj1\Comp\Sdfunpart 1{\Mlin\Sif_X}
    &=\Sproj 1
      \Comp\Sdfun{\Mlin\Sif_X}
      \Comp\Sdfstr^1_{\Snat,\With XX}\\
    &=\Mlin\Sif_{X}
      \Comp{\Withp{\Sproj 1}{\Withp{\Sproj 0}{\Sproj 0}}}
      \Comp\Sdfstr^1_{\Snat,\With XX}\\
    &\Textsep +
      \Mlin\Sif_{X}
      \Comp{\Withp{\Sproj 0}{\Withp{\Sproj 1}{\Sproj 1}}}
      \Comp\Sdfstr^1_{\Snat,\With XX}
    \text{\quad by Theorem~\ref{th:sproj-sdfun-mlin}}\\
    &=\Mlin\Sif_{X}
      \Comp{\Withp{0}{\Withp{\Sproj 0}{\Sproj 0}}}+
      \Mlin\Sif_{X}
      \Comp{\Withp{\Snat}{\Withp{\Sproj 1}{\Sproj 1}}}\\
    &=\Mlin\Sif_{X}
      \Comp{\Withp{\Snat}{\Withp{\Sproj 1}{\Sproj 1}}}
    \text{\quad by bilinearity of }\Mlin\Sif\\
    &=\Sproj1\Comp\Mlin\Sif_{\Sdfun X}
      \text{\quad by naturality}
  \end{align*}
  and the contention follows by joint monicity of
  $\Sproj0,\Sproj1$. The case of $\Slet$ is completely similar.
\end{proof}

\begin{lemma}\label{lemma:sif-slet-sdfunpart}
  For $k\in\Nat$ we have
  \begin{align*}
    \Sdfunpart1{\Sdfunpart0^{k+1}\Mlin\Sif_X}
    &=\Sflipl k\Comp\Sdfunpart0^{k+1}\Mlin\Sif_{\Sdfun X}\\
    \Sdfunpart1{\Sdfunpart0^{k+1}\Mlin\Slet_X}
    &=\Sflipl
      k\Comp\Sdfunpart0^{k+1}\Mlin\Slet_{\Sdfun X}\,.
  \end{align*}
\end{lemma}
\begin{proof}
  By Lemma~\ref{lemma:sdfun-circular} and Lemma~\ref{lemma:sdfun-nat}.
\end{proof}

\subsubsection{Recursion}
\label{sec:cat-recursion}

From now on we assume that $\cL$ is a differential summable resource
category which is \Cpolike{}, see Section~\ref{sec:Scott-summable}.

\begin{theorem}\label{th:sdfun-sfix}
  For any object $X$ we have
  \begin{align*}
    \Sdfun\Sfix^X=\Sfix^{\Sdfun X}\Comp\Cur{(\Sdfun\Ev)}
  \end{align*}
\end{theorem}
Observe that this equation is well typed: we have %
$\Ev:\With{\Simplp XX}{X}\to X$ and hence %
$\Sdfun\Ev:\With{\Simplp X{\Sdfun X}}{\Sdfun X}\to \Sdfun X$ so that %
$\Cur(\Sdfun\Ev):\Simplp X{\Sdfun X}\to\Simplp{\Sdfun X}{\Sdfun X}$
and hence both sides of the equation are morphisms %
$\Simplp X{\Sdfun X}\to\Sdfun X$.
\begin{proof}
  By induction on $n\in\Nat$ we prove that %
  $\forall n\in\Nat\ \Sdfun\Sfix^X_n
  =\Sfix^{\Sdfun X}_n\Comp\Cur{(\Sdfun\Ev)}$ and the
  result follows by continuity. For $n=0$ the equation is obvious so
  assume that it holds for some $n\in\Nat$.

  We have
  \begin{align*}
    \Sfix^{\Sdfun X}_{n+1}\Comp\Cur{(\Sdfun\Ev)}
    &= \Ev^{X,\Sdfun X}
      \Comp\Tuple{\Simpl X{\Sdfun X},\Sfix^{\Sdfun X}_n}
      \Comp\Cur{(\Sdfun\Ev^{X,X})}\\
    &= \Ev^{X,\Sdfun X}
      \Comp\Tuple{\Cur(\Sdfun\Ev^{X,X}),
      \Sfix^{\Sdfun X}_n\Comp\Cur(\Sdfun\Ev^{X,X})}\\
    &= \Ev^{X,\Sdfun X}
      \Comp\Tuple{\Cur(\Sdfun\Ev^{X,X}),\Sdfun\Sfix^{X}_n}
      \text{\quad by ind.~hyp.}\\
    &= \Ev^{X,\Sdfun X}
      \Comp\Withp{\Cur(\Sdfun\Ev^{X,X})}{\Sdfun X}
      \Comp\Tuple{\Simpl X{\Sdfun X},\Sdfun\Sfix^{X}_n}\\
    &= \Sdfun\Ev^{X,X}
      \Comp\Tuple{\Simpl X{\Sdfun X},\Sdfun\Sfix^{X}_n}
      \text{\quad by cart.~closedness}\\
    &=\Sdfun(\Ev\Comp\Tuple{\Simpl XX,\Sfix^X_n})\\
    &=\Sdfun\Sfix^X_{n+1}
  \end{align*}
  as contended, using also the fact that $\Sdfun$ is a functor which
  commutes with cartesian products.
\end{proof}

\begin{remark}
  This theorem is a remarkable property of the fixpoint operator at
  any type \(X\): its derivative can be simply expressed by means of
  a fixpoint operator at type \(\Tdiff X\).
\end{remark}

\subsection{Interpreting types and terms} %
\label{sec:types-terms-interp}
The translation of any type $A$ into an object $\Tsem A$ of $\Kl\cL$
(that is, of $\cL)$ is given by %
$\Tsem{\Tdiffm d\Tnat}=\Sdfun^d\Snat$ and %
$\Tsem{\Timpl AB}=\Simplp{\Tsem A}{\Tsem B}$ so that %
$\Tsem{\Tdiffm dA}=\Sdfun^d\Tsem A$ holds for all type $A$ and all
$d\in\Nat$ thanks to our identification of $\Simpl X{\Sdfun Y}$ with
$\Sdfun\Simplp XY$.

A context $\Gamma=(x_1:A_1,\dots,x_k:A_k)$ is interpreted as %
$\Tsem\Gamma=\Tsem{A_1}\IWith\cdots\IWith\Tsem{A_k}$ considered as an
object of $\Kl\cL$.


The next theorem also provides our definition of the interpretation of
terms.
\begin{theorem}\label{th:sem-defined-sum}
  Given a term $M$, a type $A$ and a context $\Gamma$ such that %
  $\Tseq\Gamma MA$ for some typing derivation $\delta$ (so that $A$ is
  actually determined by $M$) %
  one can define $\Psem M\Gamma\in\Kl\cL(\Tsem\Gamma,\Tsem A)$ in
  such a way that
  \begin{itemize}
  \item $\Psem M\Gamma\in\Kl\cL(\Tsem\Gamma,\Tsem A)$ depends only on
    $M$ and not on $\delta$
  \item and if $M=M_0+M_1$ then %
    $\Psem{M_0}\Gamma,\Psem{M_1}\Gamma$ are summable in %
    $\Kl\cL(\Tsem\Gamma,\Tsem A)$ and %
    $\Psem M\Gamma=\Psem{M_0}\Gamma+\Psem{M_1}\Gamma$ (this makes
    sense by Lemma~\ref{lemma:ty-der-sum}).
\end{itemize}
\end{theorem}
\begin{proof}
  By induction on $\Tdersize\delta$ where $\delta$ is a derivation of
  the typing judgment $\Tseq\Gamma MA$.
  We proceed by cases, according to the last rule in $\delta$.

  \Proofcase %
  If $M=x_i$ for some $i\in\Eset{1,\dots,k}$ we set %
  $\Psem M\Gamma=\Proj i$.

  \Proofcase %
  If $M=\Abst xBN$ then we have $A=\Timplp BC$ and
  $\Tseq{\Gamma,x:B}NC$ so that by inductive hypothesis %
  $\Psem N{\Gamma,x:B}\in\Kl\cL(\With{\Tsem\Gamma}{\Tsem B},\Tsem C)$
  and we set %
  $\Psem M\Gamma=\Cur{\Psem N{\Gamma,x:B}}
  \in\Kl\cL(\Tsem\Gamma,\Simpl{\Tsem B}{\Tsem C})$ by inductive
  hypothesis.

  \Proofcase %
  If $M=\App NP$ with %
  $\Tseq\Gamma N{\Timpl BA}$ and %
  $\Tseq\Gamma PB$ then we have by inductive hypothesis %
  $\Psem N\Gamma\in\Kl\cL(\Tsem\Gamma,\Simpl{\Tsem B}{\Tsem A})$ and %
  $\Psem P\Gamma\in\Kl\cL(\Tsem\Gamma,\Tsem B)$ and hence we set %
  $\Psem M\Gamma=\App{\Psem N\Gamma}{\Psem P\Gamma}
  =\Ev\Comp\Tuple{\Psem N\Gamma,\Psem P\Gamma}
  \in\Kl\cL(\Tsem\Gamma,\Tsem A)$.

  \Proofcase %
  If $M=\Lfix N$ with %
  $\Tseq\Gamma N{\Timpl AA}$ so that by inductive hypothesis %
  $\Psem N\Gamma\in\Kl\cL(\Tsem\Gamma,\Simpl{\Tsem A}{\Tsem A})$ and
  so we set %
  \(
  \Psem M\Gamma
  =\Sfix^{\Tsem A}\Comp\Psem N\Gamma\in\Kl\cL(\Tsem\Gamma,\Tsem A)
  \)
  as required.

  \Proofcase %
  If $M=\Num n$ for some $n\in\Nat$, we set %
  $\Psem M\Gamma=\Snum n\in\Kl\cL(\Tsem\Gamma,\Snat)$.

  \Proofcase %
  If $M=\Lsucc dN$ so that $\Tseq\Gamma N{\Tdiffm d\Tnat}$ and hence %
  $\Psem N\Gamma\in\Kl\cL(\Tsem\Gamma,\Sdfun^d\Snat)$ by inductive
  hypothesis, we set %
  $\Psem M\Gamma=\Ssuc^d(\Psem
  N\Gamma)\in\Kl\cL(\Tsem\Gamma,\Sdfun^d\Snat)$.
  Of course we set similarly %
  $\Psem{\Lpred dN}\Gamma=\Spred^d(\Psem
  N\Gamma)\in\Kl\cL(\Tsem\Gamma,\Sdfun^d\Snat)$.

  \Proofcase %
  If $M=\Lif dN{P_0}{P_1}$ with %
  $\Tseq\Gamma N{\Tdiffm d\Tnat}$ and %
  $\Tseq\Gamma {P_j}A$ for $j=0,1$ so that by inductive hypothesis %
  $\Psem N\Gamma\in\Kl\cL(\Tsem\Gamma,\Sdfun^d\Snat)$ and %
  $\Psem{P_j}\Gamma\in\Kl\cL(\Tsem\Gamma,\Tsem A)$ for $j=0,1$.
  So we set %
  \( \Psem M\Gamma =\Sif^d (\Psem
  N\Gamma,\Psem{P_0}\Gamma,\Psem{P_1}\Gamma)
  \in\Kl\cL(\Tsem\Gamma,\Sdfun^d\Tsem A=\Tsem{\Tdiffm dA}) \) %
  where we use the notation \(\Sif^d\) introduced in
  Section~\ref{sec:syn-basic-constr}.

  \Proofcase %
  If $M=\Llet d xNP$ with %
  $\Tseq\Gamma N{\Tdiffm d\Tnat}$ and %
  $\Tseq{\Gamma,x:\Tnat}{P}{A}$ so that by inductive hypothesis %
  $\Psem N\Gamma\in\Kl\cL(\Tsem\Gamma,\Sdfun^d\Snat)$ and %
  $\Psem P{\Gamma,x:\Tnat}\in\Kl\cL(\With{\Tsem\Gamma}{\Snat},\Tsem
  A)$ %
  and we set %
  $\Psem M\Gamma =\Slet^d(\Psem N\Gamma,\Cur{\Psem P{\Gamma,x:\Tnat}})
  \in\Kl\cL(\Tsem\Gamma,\Sdfun^d\Psem A\Gamma)$
  where we use the notation \(\Slet^d\) introduced in
  Section~\ref{sec:syn-basic-constr}.

  \Proofcase %
  We set $\Psem{\Lzerot A}\Gamma=0\in\Kl\cL(\Tsem\Gamma,\Tsem A)$.

  \Proofcase %
  If $M=\Lprojd idN$ then %
  $\Tseq\Gamma N{\Tdiffm{d+1}B}$ with $A=\Tdiffm dB$ so that %
  $\Psem N\Gamma\in\Kl\cL(\Tsem\Gamma,\Sdfun^{d+1}\Tsem B)$ and we
  set %
  $\Psem M\Gamma=\Sdfun^d\Sproj i\Comp\Psem N\Gamma
  \in\Kl\cL(\Tsem\Gamma,\Sdfun^d\Tsem B=\Tsem A)$.

  \Proofcase %
  If $M=\Linjd idN$ then we have %
  $\Tseq\Gamma N{\Tdiffm dB}$ with $A=\Tdiffm{d+1}B$ so that %
  $\Psem N\Gamma\in\Kl\cL(\Tsem\Gamma,\Sdfun^d\Tsem B)$ and we set %
  $\Psem M\Gamma=\Sdfun^d\Sin i\Comp\Psem N\Gamma
  \in\Kl\cL(\Tsem\Gamma,\Sdfun^{d+1}\Tsem B=\Tsem A)$.

  \Proofcase %
  If $M=\Lsumd dN$ then we have %
  $\Tseq\Gamma N{\Tdiffm{d+1}B}$ with $A=\Tdiffm dB$ so that %
  $\Psem N\Gamma\in\Kl\cL(\Tsem\Gamma,\Sdfun^{d+1}\Tsem B)$ and we
  set %
  \(
  \Psem M\Gamma
  =\Sdfun^d\Sdfmult\Comp\Psem N\Gamma
  \in\Kl\cL(\Tsem\Gamma,\Sdfun^d\Tsem B=\Tsem A)
  \).

  \Proofcase %
  If $M=\Lflipdl dlN$ then we have %
  $\Tseq\Gamma N{\Tdiffm{d+l+2}B}$ with $A=\Tdiffm{d+l+2}B$ and %
  $\Psem N\Gamma\in\Kl\cL(\Tsem\Gamma,\Sdfun^{d+l+2}\Tsem B)$ and we
  set %
  $\Psem M\Gamma=\Sdfun^d\Sflipl l\Comp\Psem N\Gamma
  \in\Kl\cL(\Tsem\Gamma,\Sdfun^{d+l+2}\Tsem B=\Tsem A)$.

  \Proofcase %
  If $M=\Ldiff N$ then we have %
  $\Tseq\Gamma N{\Timpl BC}$ and $A=(\Timpl{\Tdiff B}{\Tdiff C})$ %
  and hence %
  $\Psem N\Gamma\in\Kl\cL(\Tsem\Gamma,\Simpl{\Tsem B}{\Tsem C})$ and
  we set %
  $\Psem M\Gamma=\Sdfunint^{\Tsem B,\Tsem C}\Comp\Psem N\Gamma
  \in\Kl\cL(\Tsem\Gamma,\Simplp{\Sdfun\Tsem B}{\Sdfun\Tsem C}=\Tsem
  A)$.

  Assume now that $M=M_0+M_1$.
  We distinguish the same subcases as in the proof of
  Lemma~\ref{lemma:ty-der-sum}.

  \Proofcase %
  The last rule of $\delta$ is~\ref{rl:trprojt} 
  so that %
  $M_j=\Lprojd jdN$, $A=\Tdiffm dB$ and %
  $\Tseq\Gamma N{\Tdiffm{d+1}B}$ by a derivation $\delta'$ such that %
  $\Tdersize{\delta'}=\Tdersize\delta-1$.
  By inductive hypothesis we have %
  $f=\Psem N\Gamma\in\Kl\cL(\Tsem\Gamma,\Sdfun^{d+1}\Tsem B)$ and we
  know that the morphisms %
  $\Sdfun^d\Sproj 0\Comp f=\Psem{M_0}\Gamma$ and %
  $\Sdfun^d\Sproj 1\Comp f=\Psem{M_1}\Gamma$ are summable, with sum %
  $\Sdfun^d\Ssum\Comp f$.
  We set %
  $\Psem M\Gamma=\Sdfun^d\Ssum\Comp\Psem N\Gamma
  \in\Kl\cL(\Tsem\Gamma,\Sdfun^d\Tsem B)$ so that actually %
  $\Psem M\Gamma=\Psem{M_0}\Gamma+\Psem{M_1}\Gamma$.

  \Proofcase %
  The last rule of $\delta$ is~\ref{rl:trprojd} 
  so that %
  $M_0=\Lprojd 1d{N_0}$, $M_1=\Lprojd 0d{N_1}$, %
  $A=\Tdiffm dB$ and %
  $\Tseq\Gamma{N_0+N_1}{\Tdiffm{d+1}{B}}$ by a derivation %
  $\delta'$ such that %
  $\Tdersize{\delta'}=\Tdersize{\delta}-1$.
  By inductive hypothesis we have defined two summable morphisms %
  $\Psem{N_j}\Gamma\in\Kl\cL(\Tsem\Gamma,\Sdfun^{d+1}\Tsem{B})$ for
  $j=0,1$.
  It follows that the 4 morphisms
  $(\Sdfun^d\Sproj i)\Compl\Psem{N_j}\Gamma$ (for $i,j\in\Eset{0,1}$)
  are summable, and hence %
  $\Sdfun^d\Sproj 1\Comp\Psem{N_0}\Gamma =\Psem{M_0}\Gamma$ and %
  $\Sdfun^d\Sproj 0\Comp\Psem{N_1}\Gamma =\Psem{M_1}\Gamma$ are
  summable.
  We set %
  $\Psem M\Gamma=\Psem{M_0}\Gamma+\Psem{M_1}\Gamma$.

  \Proofcase %
  The last rule of $\delta$ is~\ref{rl:trlin} 
  so that there is a linear
  context $L$ of height $1$ and terms $N_0$, $N_1$ such that %
  $M_j=L[N_j]$ and %
  $\Tseq\Gamma{L[N_0+N_1]}A$ by a derivation $\delta'$ such that %
  $\Tdersize{\delta'}=\Tdersize\delta-1$. %
  This implies (by a simple inspection of the various possibilities for
  $L$ which has height $1$) that for some context $\Delta$ and some
  type $B$ one has %
  $\Tseq\Delta{N_0+N_1}B$ by a derivation $\delta''$ such that %
  $\Tdersize{\delta''}=\Tdersize{\delta'}-k_L$ where %
  $k_L\in\Natnz$ \emph{depends only on $L$} (if for instance
  $L=\Lif d{\Echole}{P_0}{P_1}$ then $k_L=1+k_0+k_1$ where $k_i$ is
  the size of the typing derivation of $P_i$). %
  Now we consider the various possibilities for $L$.
  \begin{itemize}
  \item $L=\Abst xC\Echole$ and we have %
    $\Delta=(\Gamma,x:B)$, $A=\Timplp CB$.
    By inductive hypothesis we have %
    $\Psem{N_j}{\Gamma,x:C}\in\Kl\cL(\With{\Tsem\Gamma}{\Tsem C},\Tsem
    B)$ %
    for $j=0,1$, %
    $\Psem{N_0}{\Gamma,x:C}$ and $\Psem{N_1}{\Gamma,x:C}$ %
    are summable and also that %
    $\Psem{N_0+N_1}{\Gamma,x:C}
    =\Psem{N_0}{\Gamma,x:C}+\Psem{N_1}{\Gamma,x:C}$. %
    We have $\Psem{M_j}\Gamma=\Cur{\Psem{N_j}{\Gamma,x:C}}$ %
    because we know that there is a derivation $\delta_j$ of %
    $\Tseq{\Gamma,x:C}{N_j}B$ such that %
    $\Tdersize{\delta_j}\leq\Tdersize{\delta}$ by %
    Lemma~\ref{lemma:ty-der-sum}.
    It follows that %
    $\Psem{M_0}\Gamma$ and $\Psem{M_1}\Gamma$ %
    are summable and we can set %
    $\Psem M\Gamma=\Psem{M_0}\Gamma+\Psem{M_1}\Gamma$
  \item $L=\App{\Echole}P$ and we have %
    $\Delta=\Gamma$, $B=\Timplp CA$, $\Tseq\Gamma PC$ by a derivation
    of size $k_P>0$ and %
    the derivation $\delta'$ of %
    $\Tseq\Gamma{N_0+N_1}{\Timpl CA}$ satisfies %
    $\Tdersize{\delta}=\Tdersize{\delta'}+k_P+1$.
    So by inductive hypothesis %
    $\Psem P\Gamma\in\Kl\cL(\Tsem\Gamma,\Tsem C)$, and %
    $\Psem{N_j}\Gamma\in\Kl\cL(\Tsem\Gamma,\Simpl{\Tsem C}{\Tsem
      A})$ %
    (for $j=0,1$) are summable and we have %
    $\Psem{N_0+N_1}\Gamma=\Psem{N_0}\Gamma+\Psem{N_1}\Gamma$.
    We have %
    $\Psem{M_j}\Gamma=\Ev\Comp\Tuple{\Psem{N_j}\Gamma,\Psem
      P\Gamma}$ %
    because the derivation $\delta_j$ of $\Tseq\Gamma{M_j}A$
    satisfies %
    $\Tdersize{\delta_j}\leq\Tdersize\delta$ and hence %
    $\Psem{M_0}\Gamma$ and $\Psem{M_1}\Gamma$ are summable with %
    $\Psem{M_0}\Gamma+\Psem{M_1}\Gamma
    =\Ev\Comp\Tuple{\Psem{N_0}\Gamma+\Psem{N_1}\Gamma,\Psem
      P\Gamma}$ %
    by left-linearity of $\Ev$.
    We set %
    $\Psem M\Gamma=\Psem{M_0}\Gamma+\Psem{M_1}\Gamma$.
  \item $L=\Lif d\Echole{P_0}{P_1}$ and we have %
    $\Delta=\Gamma$, $B=\Tdiffm d\Tnat$, $A=\Tdiffm dC$ and %
    $\Tseq\Gamma{P_i}C$ for $i=0,1$ by derivations of sizes %
    $k_0$ and $k_1$ respectively so that, denoting by %
    $\delta'$ the derivation of %
    $\Tseq\Gamma{N_0+N_1}{\Tdiffm d\Tnat}$, we have %
    $\Tdersize\delta=\Tdersize{\delta'}+k_0+k_1+1$. %
    It follows by inductive hypothesis that %
    $\Psem{P_i}\Gamma\in\Kl\cL(\Tsem\Gamma,\Tsem C)$ for $i=0,1$, %
    and that %
    $\Psem{N_j}\Gamma\in\Kl\cL(\Tsem\Gamma,\Sdfun^d\Snat)$ for $j=0,1$
    are summable with %
    $\Psem{N_0+N_1}\Gamma=\Psem{N_0}\Gamma+\Psem{N_1}\Gamma$.
    For $j=0,1$ we have %
    $\Psem{M_j}\Gamma=\Sdfunpart 0^d\Mlin\Sif
    \Comp\Tuple{\Psem{N_j}\Gamma,\Tuple{\Psem{M_0}\Gamma,\Psem{M_1}\Gamma}}$ %
    because the derivation $\delta_j$ of $\Tseq\Gamma{M_j}A$
    satisfies %
    $\Tdersize{\delta_j}\leq\Tdersize\delta$ %
    (by Lemma~\ref{lemma:ty-der-sum}) and hence, by left-linearity
    of %
    $\Sdfunpart 0^d\Mlin\Sif$, %
    $\Psem{M_0}\Gamma$ and $\Psem{M_1}\Gamma$ are summable and satisfy
    $\Psem{M_0}\Gamma+\Psem{M_1}\Gamma =\Sdfunpart 0^d\Mlin\Sif
    \Comp\Tuple{\Psem{N_0}\Gamma+\Psem{N_1}\Gamma,
      \Tuple{\Psem{P_0}\Gamma,\Psem{P_1}\Gamma}}$.
    We set %
    $\Psem M\Gamma=\Psem{M_0}\Gamma+\Psem{M_1}\Gamma$.
  \item $L=\Llet dx{\Echole}P$ and we have %
    $\Delta=\Gamma$, $B=\Tdiffm d\Tnat$, %
    $\Tseq{\Gamma,x:\Tnat}{P}{C}$ by a derivation of size $k$ and
    $A=\Tdiffm dC$ so that denoting by $\delta'$ the derivation of %
    $\Tseq\Gamma{N_0+N_1}{\Tdiffm d\Tnat}$ we have %
    $\Tdersize\delta=\Tdersize{\delta'}+k+1$.
    This case is completely similar to the previous one.
    By inductive hypothesis we have %
    $\Psem P{\Gamma,x:\Tdiffm d\Tnat}
    \in\Kl\cL(\Tsem\Gamma\IWith\Sdfun^d\Snat,\Tsem C)$ and, for
    $j=0,1$ we have %
    $\Psem{N_j}\Gamma\in\Kl\cL(\Tsem\Gamma,\Tdiffm\Snat)$ which are
    summable with %
    $\Psem{N_0+N_1}\Gamma=\Psem{N_0}\Gamma+\Psem{N_1}\Gamma$.
    For $j=0,1$ we have %
    $\Psem{M_j}\Gamma=\Sdfunpart 0^d\Mlin\Slet
    \Comp\Tuple{\Psem{N_j}\Gamma,\Cur\Psem P{\Gamma,x:\Tdiffm
        d\Tnat}}$ %
    because the derivation $\delta_j$ of $\Tseq\Gamma{M_j}A$
    satisfies %
    $\Tdersize{\delta_j}\leq\Tdersize\delta$ %
    (by Lemma~\ref{lemma:ty-der-sum}) and hence, by left-linearity
    of %
    $\Sdfunpart 0^d\Mlin\Slet$, %
    $\Psem{M_0}\Gamma$ and $\Psem{M_1}\Gamma$ are summable and satisfy
    $\Psem{M_0}\Gamma+\Psem{M_1}\Gamma =\Sdfunpart 0^d\Mlin\Slet
    \Comp\Tuple{\Psem{N_0}\Gamma+\Psem{N_1}\Gamma, \Cur{\Psem
        P{\Gamma,x:C}}}$.
    We set %
    $\Psem M\Gamma=\Psem{M_0}\Gamma+\Psem{M_1}\Gamma$.
  \item $L=\Ldiff{\Echole}$ and we have %
    $\Delta=\Gamma$, $B=\Timplp CE$ and
    $A=\Timplp{\Tdiff C}{\Tdiff E}$ %
    and we use $\delta'$ for the derivation of %
    $\Tseq\Gamma{N_0+N_1}{\Timpl CE}$ so that %
    $\Tdersize{\delta}=\Tdersize{\delta'}+1$ %
    and hence by inductive hypothesis we have %
    $\Psem{N_j}\Gamma\in\Kl\cL(\Tsem\Gamma,\Simpl{\Tsem C}{\Tsem
      E})$ %
    for $j=0,1$, these two morphisms are summable and we have %
    $\Psem{N_0+N_1}\Gamma=\Psem{N_0}\Gamma+\Psem{N_1}\Gamma$.
    For $j=0,1$ we have %
    $\Psem{M_j}\Gamma=\Sdfunint^{\Tsem C,\Tsem E}
    \Comp\Psem{N_j}\Gamma$ %
    because the derivation $\delta_j$ of $\Tseq\Gamma{M_j}A$
    satisfies %
    $\Tdersize{\delta_j}\leq\Tdersize\delta$ %
    (by Lemma~\ref{lemma:ty-der-sum}) and hence, by linearity of %
    $\Sdfunint^{\Tsem C,\Tsem E}$, %
    $\Psem{M_0}\Gamma$ and $\Psem{M_1}\Gamma$ are summable and satisfy
    $\Psem{M_0}\Gamma+\Psem{M_1}\Gamma =\Sdfunint^{\Tsem C,\Tsem E}
    \Comp(\Psem{N_0}\Gamma+\Psem{N_1}\Gamma)$. 
    We set %
    $\Psem M\Gamma=\Psem{M_0}\Gamma+\Psem{M_1}\Gamma$.
  \item $L=\Lprojd id\Echole$ and we have %
    $\Delta=\Gamma$, $B=\Tdiffm{d+1}C$ and $A=\Tdiffm dC$ %
    and we use $\delta'$ for the derivation of %
    $\Tseq\Gamma{N_0+N_1}{\Tdiffm{d+1}C}$ so that %
    $\Tdersize{\delta}=\Tdersize{\delta'}+1$ %
    and hence by inductive hypothesis we have %
    $\Psem{N_j}\Gamma\in\Kl\cL(\Tsem\Gamma,\Sdfun^{d+1}\Tsem C)$ %
    for $j=0,1$, these two morphisms are summable and we have %
    $\Psem{N_0+N_1}\Gamma=\Psem{N_0}\Gamma+\Psem{N_1}\Gamma$.
    For $j=0,1$ we have %
    $\Psem{M_j}\Gamma=\Sdfun^d{\Sproj i} \Comp\Psem{N_j}\Gamma$ %
    because the derivation $\delta_j$ of $\Tseq\Gamma{M_j}A$
    satisfies %
    $\Tdersize{\delta_j}\leq\Tdersize\delta$ %
    (by Lemma~\ref{lemma:ty-der-sum}) and hence, by linearity of %
    $\Sdfun^d{\Sproj i}$, %
    $\Psem{M_0}\Gamma$ and $\Psem{M_1}\Gamma$ are summable and satisfy
    $\Psem{M_0}\Gamma+\Psem{M_1}\Gamma =\Sdfun^d\Sproj i
    \Comp(\Psem{N_0}\Gamma+\Psem{N_1}\Gamma)$. We set %
    $\Psem M\Gamma=\Psem{M_0}\Gamma+\Psem{M_1}\Gamma$.
  \item The remaining cases are similar: 
    in each of them we see that we can sensibly set
    $\Psem M\Gamma=\Psem{M_0}\Gamma+\Psem{M_1}\Gamma$. %
    \qedhere
  \end{itemize}
\end{proof}

\begin{figure}
  \centering
  \begin{align*}
    \Psem{x_i}\Gamma
    &=\Proj i
    &
      \Psem{\Abst xBN}\Gamma
    &=\Cur{\Psem N{\Gamma,x:B}}\\
    \Psem{\App NP}\Gamma
    &=\Ev\Comp\Tuple{\Psem N\Gamma,\Psem P\Gamma}
    &
      \Psem{\Lfix N}\Gamma
    &=\Sfix^{\Tsem A}\Comp\Psem N\Gamma\\
    \Psem{\Num n}\Gamma
    &=\Snum n
    &
      \Psem{\Lsucc dN}\Gamma
    &=\Ssuc^d(\Psem N\Gamma)\\
    \Psem{\Lpred dN}\Gamma
    &=\Ssuc^d(\Psem N\Gamma)
    &
      \Psem{\Lif dN{P_0}{P_1}}\Gamma
    &=\Sif^d (\Psem N\Gamma,\Psem{P_0}\Gamma,\Psem{P_1}\Gamma)\\
    \Psem{\Llet d xNP}\Gamma
    &=\Slet^d(\Psem N\Gamma,\Cur{\Psem P{\Gamma,x:\Tnat}})
    &
      \Psem{\Lprojd idN}\Gamma
    &=\Sdfun^d\Sproj i\Comp\Psem N\Gamma\\
    \Psem{\Linjd idN}\Gamma
    &=\Sdfun^d\Sin i\Comp\Psem N\Gamma
    &
      \Psem{\Lsumd dN}\Gamma
    &=\Sdfun^d\Sdfmult\Comp\Psem N\Gamma\\
    \Psem{\Lflipdl dlN}\Gamma
    &=\Sdfun^d\Sflipl l\Comp\Psem N\Gamma
    &
      \Psem{\Ldiff N}\Gamma
    &=\Sdfunint^{\Tsem B,\Tsem C}\Comp\Psem N\Gamma\\
    \Psem{M_0+M_1}\Gamma
    &=\Psem{M_0}\Gamma+\Psem{M_1}\Gamma
  \end{align*}
  \caption{Interpretation of terms (see the proof of
    Theorem~\ref{th:sem-defined-sum})}
  \label{fig:term-interp}
\end{figure}

\subsubsection{Substitution lemmas}

The first substitution lemma is completely standard in a
$\lambda$-calculus setting.

\begin{lemma}\label{lemma:sem-context-extension}
  If $\Tseq\Gamma MB$, so that $\Tseq{\Gamma,x:A}MB$, then we have %
  $\Psem M{\Gamma,x:A}=\Psem M\Gamma\Comp\Proj 0$ where %
  $\Proj 0\in\Kl\cL(\With{\Tsem\Gamma}{\Tsem A},\Tsem\Gamma)$ is the
  first projection.
\end{lemma}
\begin{proof}
  By induction on the typing derivation of \(M\).
\end{proof}

\begin{lemma}[Ordinary substitution]\label{lemma:sem-subst}
  If $\Tseq\Gamma N{A}$ and %
  $\Tseq{\Gamma,x:A}MB$ then %
  $\Psem{\Subst MNx}\Gamma\in\Kl\cL(\Tsem\Gamma,\Tsem B)$ satisfies
  \begin{align*}
    \Psem{\Subst MxN}\Gamma=
    \Psem M{\Gamma,x:A}
    \Comp\Tuple{\Tsem\Gamma,\Psem N\Gamma}\,.
  \end{align*}
\end{lemma}
\begin{proof}
  By induction on the typing derivation $M$.
\end{proof}

\begin{lemma}[Semantics of the differential]\label{lemma:sem-diff-subst}
  If %
  $\Tseq{\Gamma,x:A}MB$ then %
  $\Psem{\Ldletv xM}{\Gamma,x:\Tdiff A}
  \in\Kl\cL(\Tsem\Gamma\IWith\Ldiff{\Tsem A},\Sdfun\Tsem B)$ satisfies
  \begin{align*}
    \Psem{\Ldletv xM}{\Gamma,x:\Tdiff A}=
    \Sdfunpart 1{\Psem M{\Gamma,x:A}}\,.
  \end{align*}
\end{lemma}
If
$\Gamma=(\List A1k)$ and %
$\Delta=(\List A1{i-1},\Tdiff{A_i},\List A{i+1}k)$ %
for some
$i\in\Eset{1,\dots,k}$ and %
$\Tseq\Gamma MB$, and then we have %
$\Tseq\Delta{\Ldletv{x_i}M}{\Tdiff B}$ and in this slightly more
general situation the lemma states that %
\(
\Psem{\Ldletv{x_i}M}\Delta
=\Sdfunpart i^{\List Z1k}{\Psem M\Delta}
\) %
where %
$\Sdfunpart i^{\List Z1k}f=\Sdfun f\Comp\Sdfstr^i_{\List Z1k}$ %
is the ``$i$th partial derivative of $f$''.
This slightly more general statement is equivalent to the lemma by the
symmetry of the cartesian product $\IWith$.
\begin{proof}
  By induction on \(M\), %
  \emph{and not on its typing derivation $\delta$}. %
  This is possible thanks to Theorem~\ref{th:sem-defined-sum} which
  states that the interpretation does not depend on the typing
  derivation. This point is crucial when dealing with sums. We use the
  following notations: $f$ for $\Psem M{\Gamma,x:A}$, %
  $Z$ for $\Tsem\Gamma$, $U$ for $\Tsem A$, %
  and $X$ for $\Tsem B$. %
  Sometimes we also write $\Gamma=(x_1:A_1,\dots,x_k:A_k)$ and in that
  case we set %
  $Z_i=\Tsem{A_i}$ so that $Z=Z_1\IWith\cdots\IWith Z_k$.

  \Proofcase %
  Assume that $M=x$ and hence $\delta$ must end with~\ref{rl:trvar} 
  . %
  We have $X=U$ and $f=\Proj 1\in\Kl\cL(\With ZU,U)$, we have
  $\Sdfunpart1 f=(\Sdfun\Proj 1)\Comp\Sdfstr^1=\Proj
  1\Comp\Withp{\Sin0}{\Sdfun U}=\Proj 1\in\Kl\cL(\With Z{\Sdfun
    U},\Sdfun U)$ since $\Sdfun$ commutes with cartesian products in
  $\Kl\cL$.
  It follows that %
  \( \Sdfunpart 1f=\Psem{\Ldletv xM}{\Gamma,x:\Tdiff A} \) %
  since $\Ldletv xM=M$.

  \Proofcase %
  Assume that $M=x_i$ for some $i\in\Eset{1,\dots,k}$ and hence
  $\delta$ must end with~\ref{rl:trvar} 
  So we have %
  $f=\Proj i\in\Kl\cL(Z_1\IWith\cdots\IWith Z_k\IWith U,Z_i)$, %
  $X=Z_i$ %
  and we have %
  $\Sdfunpart 1f =\Sdfun f\Comp\Sdfstr^1_{Z,U} =\Sdfun
  f\Comp\Sdfstr^{k+1}_{\List Z1k,U} =\Sdfun\Proj
  i\Comp(\Sin0\IWith\cdots\IWith\Sin0\IWith U) =\Proj
  i\Comp(\Sin0\IWith\cdots\IWith\Sin0\IWith U) =\Sin0\Comp\Proj i$ %
  using the fact that $\Sfun$ preserves cartesian products and the
  expression~\Eqref{eq:sdfstrgen-def} of the $k+1$-ary $\Sdfstr^i$. %
  Therefore %
  \(
  \Sdfunpart 1f
  =\Psem{\Ldletv xM}{\Gamma,x:\Tdiff A}
  \) %
  since %
  $\Ldletv xM=\Lin0{x_i}$.

  \Proofcase %
  Assume that %
  $M=\Abst yCP$ so that $\delta$ must end with~\ref{rl:trabs} 
  applied to %
  $\Tseq{\Gamma,x:A,y:C}P{E}$ and we have $B=\Timplp CE$ and hence %
  $X=\Simplp VY$ where $\Tsem C=V$ and $\Tsem E=Y$. %
  Let $g=\Psem P{\Gamma,x:A,y:C}\in\Kl\cL(Z\IWith U\IWith V,Y)$. %
  Notice that
  \(
  \Sdfunpart 1g=\Sdfun g\Comp\Sdfstr^1_{Z,U,V}
  \in\Kl\cL(Z\IWith\Sdfun U\IWith V,\Sdfun Y)
  \). %
  We have
  \begin{align*}
    \Sdfunpart1{f}
    &= \Sdfunpart1(\Cur g)\\
    &= \Sdfun(\Cur g)\Comp\Sdfstr^1_{Z,U}\\
    &= \Cur((\Sdfun g)\Comp\Sdfstr^0_{\With ZU,V})\Comp\Sdfstr^1_{Z,U}
      \text{\quad by Lemma~\ref{lemma:Sdfun-curry}}
  \end{align*}
  where
  \(
  \Sdfstr^0_{\With ZU,V}
  \in\Kl\cL(\Sdfun Z\IWith\Sdfun U\IWith V,
  \Sdfun Z\IWith\Sdfun U\IWith\Sdfun V)
  \). %
  So
  \(
  \Sdfunpart1{f}
  =\Cur((\Sdfun g)\Comp\Sdfstr^0_{\With ZU,V}\Comp(\With{\Sdfstr^1_{Z,U}}{V}))
  =\Cur((\Sdfun g)\Comp\Sdfstr_{Z,U,V}^1)=\Cur(\Sdfunpart 1g)
  \). %
  By inductive hypothesis we have %
  $\Sdfunpart 1g
  =\Psem{\Ldletv xP}{\Gamma,x:\Tdiff A,y:C}$ and hence %
  \(
  \Sdfunpart 1f
  =\Psem{\Abst yC{\Ldletv xP}}{\Gamma,x:\Tdiff A}
  \) %
  as required.

  \Proofcase %
  Assume that $M=\Ldiff P$ so that $\delta$ must end %
  with~\ref{rl:trdiff} 
  applied to %
  $\Tseq{\Gamma,x:A}P{\Timpl CE}$ and we have %
  $B=\Timplp CE$ and hence $X=\Simplp VY$ where %
  $\Tsem C=V$ and $\Tsem E=Y$. %
  Let $g=\Psem P{\Gamma,x:A}\in\Kl\cL(\With ZU,\Simpl VY)$ so that %
  \(
  f=\Sdfunc g
  =\Sdfunint^{V,Y}\Comp g
  \in\Kl\cL(\With ZU,\Simpl{\Sdfun V}{\Sdfun Y})
  \) where %
  \(
  \Sdfunint^{V,Y}=\Cur{(\Sdfun\Ev\Comp\Sdfstr^{\Simpl VY,V}_1)}
  \).
Then
  \begin{align*}
    \Sdfunpart1f
    &=\Sdfun{\Sdfunc g}\Comp\Sdfstr^1_{Z,U}\\
    &=\Sdfun{\Cur{(\Sdfun\Ev\Comp\Sdfstr^{\Simpl V{Y},V}_1)}}
      \Comp\Sdfun g\Comp\Sdfstr^1_{Z,U}\\
    &=\Cur(\Sdfun(\Sdfun\Ev\Comp\Sdfstr^{\Simpl V{Y},V}_1)
      \Comp\Sdfstr^0_{\Simpl VY,\Sdfun V})
      \Comp\Sdfunpart 1 g
      \text{\quad by Lemma~\ref{lemma:Sdfun-curry}.}
  \end{align*}
  Next we have
  \begin{align*}
    \Sdfun^2\Ev
    \Comp\Sdfun\Sdfstr^1_{\Simpl VY,V}
    \Comp\Sdfstr^0_{\Simpl VY,\Sdfun V}
    &=\Sdfun^2\Ev
      \Comp\Sflip
      \Comp\Sdfun\Sdfstr^0_{\Simpl VY,V}
      \Comp\Sdfstr^1_{\Simpl V{\Sdfun Y},V}\\
    &=\Sflip
      \Comp\Sdfun^2\Ev
      \Comp\Sdfun\Sdfstr^0_{\Simpl VY,V}
      \Comp\Sdfstr^1_{\Simpl V{\Sdfun Y},V}
  \end{align*}
  by Lemma~\ref{lemma:Sdfstr-flip} and naturality of $\Sflip$. By our
  identification of $\Sdfun(\Simpl VY)$ with $\Simpl V{\Sdfun Y}$
  through the iso
  $\Cur(\Sdfunpart0\Ev)$ we have
  $\Sdfun\Ev\Comp\Sdfstr^0_{\Simpl VY,V}=\Ev\in\Kl\cL((\Simpl V{\Sdfun
    Y})\IWith V,\Sdfun Y)$ and hence
  \begin{align*}
    \Sdfunpart1f
    &=\Cur(\Sflip\Comp\Sdfun\Ev\Comp\Sdfstr^1_{\Simpl V{\Sdfun Y},V})
    \Comp\Sdfunpart 1g\\
    &=\Sflip\Comp\Cur(\Sdfun\Ev\Comp\Sdfstr^1_{\Simpl V{\Sdfun Y},V})
    \Comp\Sdfunpart 1g\\
    &=\Sflip\Comp\Sdfunc\Sdfunpart1 g\\
    &=\Psem{\Lflip{\Ldiff{\Ldletv xP}}}{\Gamma,x:\Tdiff A}\,.
  \end{align*}
  Notice that we have identified $\Sflip_{\Simpl VY}$ with
  $\Simpl V{\Sflip_Y}$ in accordance with our convention of
  identifying $\Sdfun^2(\Simpl VY)$ with $\Simpl V{\Sdfun^2Y}$.

  \Proofcase %
  Assume that $M=\App PQ$ so that $\delta$ must end %
  with~\ref{rl:trapp} 
  applied to %
  $\Tseq{\Gamma,x:A}P{\Timpl CB}$ and %
  $\Tseq{\Gamma,x:A}Q{C}$ and let $Y=\Tsem C$. %
  Let $p=\Psem P\Gamma\in\Kl\cL(\With ZU,\Simpl YX)$ and %
  $q=\Psem Q\Gamma\in\Kl\cL(\With ZU,Y)$ so that %
  \( f =\Psem M\Gamma =\App pq =\Ev\Comp\Tuple{p,q} \in\Kl\cL(\With
  ZU,X) \), %
  $\Sdfunpart1p\in\Kl\cL(\With Z{\Sdfun U},\Simpl Y{\Sdfun X})$ and
  $\Sdfunpart1q\in\Kl\cL(\With Z{\Sdfun U},\Sdfun Y)$. %
  We have
  $\Sdfunc{\Sdfunpart1 p}\in\Kl\cL(\With Z{\Sdfun U},\Simpl{\Sdfun
    Y}{\Sdfun^2 X})$ so that %
  \(\Sdfmult\Comp\Sdfunc{\Sdfunpart1 p}\in\Kl\cL(\With Z{\Sdfun
    U},\Simpl{\Sdfun Y}{\Sdfun X})\) and hence
  $\App{\Sdfmult\Comp\Sdfunc{\Sdfunpart1p}}
  {\Sdfunpart1q}\in\Kl\cL(\With Z{\Sdfun U},\Sdfun X)$, we prove that
  \begin{equation*}
    \Sdfunpart1(\App pq)
    =\App{\Sdfmult\Comp\Sdfunc{\Sdfunpart1p}}{\Sdfunpart1q}\,.
  \end{equation*}

  We have
  \begin{align*}
    \Sdfunpart1(\App pq)
    &=(\Sdfun\Ev)\Comp(\Sdfun\Tuple{p,q})\Comp\Sdfstr^1_{Z,U}\\
    &=(\Sdfun\Ev)\Comp\Tuple{\Sdfun p,\Sdfun q}\Comp\Sdfstr^1_{Z,U}
    \text{\quad since $\Sdfun$ preserves cart.~prod.}\\
    &=(\Sdfun\Ev)\Comp\Tuple{\Sdfunpart 1p,\Sdfunpart 1q}\,.
  \end{align*}
  where
  $\Sdfun\Ev\in\Kl\cL(\With{\Simplp X{\Sdfun Y}}{\Sdfun X},\Sdfun Y)$
  and we know by Lemma~\ref{lemma:sdfun-Ev-expression} that
  \(
  \Sdfun\Ev
  =\Sdfmult
  \Comp(\Sdfun\Ev')
  \Comp\Sdfstr^1_{\Simpl{Y}{\Sdfun X},Y}
  \) %
  where $\Ev'$ is the
  evaluation morphism in
  $\Kl\cL(\With{\Simplp Y{\Sdfun X}}{Y},\Sdfun X)$.

  On the other hand
  \begin{align*}
    \App{\Sdfunc{\Sdfunpart1p}}{\Sdfunpart1q}
    &=\Ev\Comp\Tuple{\Sdfunc{\Sdfunpart1p},\Sdfunpart1q}\\
    &=\Ev\Comp\Withp{\Sdfunint^{Y,\Sdfun X}}{\Sdfun X}
      \Comp\Tuple{\Sdfunpart1p,\Sdfunpart1q}
  \end{align*}
  where we recall that
  $\Sdfunint^{Y,\Sdfun X} =\Cur((\Sdfun\Ev')\Comp\Sdfstr^1_{\Simpl
    Y{\Sdfun X},Y}) \in\Kl\cL(\Simpl Y{\Sdfun X},\Simpl{\Sdfun
    Y}{\Sdfun^2X})$ and we have used Equation~\Eqref{eq:dcur-dint}.
  Therefore
  \begin{align*}
    \App{\Sdfunc{\Sdfunpart1p}}{\Sdfunpart1q}
    =\Sdfun\Ev'
    \Comp\Sdfstr^1_{\Simpl Y{\Sdfun X},Y}
    \Comp\Tuple{\Sdfunpart1p,\Sdfunpart1q}\,.
  \end{align*}
  It follows that
  \begin{align*}
    \App{\Sdfmult\Comp\Sdfunc{\Sdfunpart1p}}{\Sdfunpart1q}
    &=\Sdfmult\Comp\App{\Sdfunc{\Sdfunpart1p}}{\Sdfunpart1q}\\
    &=\Sdfmult\Comp(\Sdfun\Ev')
      \Comp\Sdfstr^1_{\Simpl Y{\Sdfun X},Y}
      \Comp\Tuple{\Sdfunpart1p,\Sdfunpart1q}\\
    &=\Sdfun\Ev\Comp\Tuple{\Sdfunpart1p,\Sdfunpart1q}
      \text{\quad by Lemma~\ref{lemma:sdfun-Ev-expression}}\\
    &=\Sdfunpart 1f
  \end{align*}
  where the first equation results from our identification of %
  \(\Simpl{\Sdfun Y}{\Sdfmult}\) with \(\Sdfmult\) %
  and we have also used the fact that
  \( f=\App pq=\Ev\Comp\Tuple{p,q} \). %
  Finally, using again Equation~\Eqref{eq:dcur-dint}, %
  \begin{align*}
    \Sdfunpart 1f
    &=\App{\Sdfmult
      \Comp\Sdfunint
      \Comp{\Sdfunpart1p}}{\Sdfunpart1q}\\
    &=\Psem{\App{\Lsum{\Ldiff{\Ldletv xP}}}{\Ldletv xQ}}{\Gamma,x:\Tdiff A}
  \end{align*}
  as contended, using of course also the inductive hypothesis.

  \Proofcase %
  Assume that $M=\Lsucc dP$ so that the last rule of $\delta$ %
  is~\ref{rl:trsuc} 
  and that we have %
  $\Tseq{\Gamma,x:A}P{\Tdiffm d\Tnat=B}$ and hence $X=\Sdfun^d\Snat$. %
  Let $g=\Psem P{\Gamma,x:A}$ so that %
  $f=\Sdfun^d\Mlin\Ssuc\Comp g$, we have %
  \begin{align*}
    \Sdfunpart1f
    &=\Sdfunpart 1(\Sdfun^d\Mlin\Ssuc\Comp f)\\
    &=\Sdfun^{d+1}\Mlin\Ssuc
      \Comp\Sdfun g
      \Comp\Sdfstr^1_{Z,U}\\
    &=\Sdfun^{d+1}\Mlin\Ssuc
      \Comp\Sdfunpart1 g\\
    &=\Sdfun^{d+1}\Mlin\Ssuc
      \Comp\Psem{\Ldletv xP}{\Gamma,x:\Tdiff A}
    \text{\quad   by inductive hypothesis}\\
    &=\Psem{\Lsucc{d+1}{\Ldletv xP}}{\Gamma,x:\Tdiff A}
  \end{align*}

  The cases where $M$ is of shape $\Lpred dP$, $\Lprojd idP$,
  $\Linjd idP$, $\Lsumd dP$ and $\Lflipd dP$ are similarly dealt with.

  \Proofcase %
  Assume that $M=\Lif dP{Q_0}{Q_1}$ so that $\delta$ ends %
  with~\ref{rl:trif} 
  and that we have %
  $\Tseq{\Gamma,x:A}P{\Tdiffm d\Tnat}$, %
  $\Tseq{\Gamma,x:A}{Q_i}C$ for $i=0,1$ so that %
  $B=\Tdiffm dC$ and $X=\Sdfun^dY$ where $Y=\Tsem C$. Let %
  $p=\Psem P{\Gamma,x:A}\in\Kl\cL(\With ZU,\Sdfun^d\Tnat)$ and %
  $(q_i=\Psem{Q_i}{\Gamma,x:A}\in\Kl\cL(\With ZU,Y))_{i=0,1}$.
  We have
  $\Sdfunpart0^d{\Mlin\Sif_X}\in\cL(\Sdfun^d\Snat\IWith\Withp
  YY,\Sdfun^d Y=X)$ and
  \(f=\Psem M{\Gamma,x:A}=\Sdfunpart0^d\Mlin\Sif_X
  \Comp\Tuple{p,\Tuple{q_0,q_1}}\in\Kl\cL(\With ZU,X) 
  \).
  We have
  \begin{align*}
    &\Sdfunpart1(\Sdfunpart0^d\Mlin\Sif_X
    \Comp\Tuple{p,\Tuple{q_0,q_1}})\\
    &\Textsep=\Sdfun\Sdfunpart0^d\Mlin\Sif_X
      \Comp\Sdfun\Tuple{p,\Tuple{q_0,q_1}}
      \Comp\Sdfstr^1_{Z,U}\text{\quad by def.~of }\Sdfunpart1
      \text{ and functoriality of }\Sdfun\\
    &\Textsep=\Sdfmult
      \Comp\Sdfunpart1\Sdfunpart0^{d+1}\Mlin\Sif_X
      \Comp\Tuple{\Sdfunpart1p,\Tuple{\Sdfunpart1{q_0},\Sdfunpart1{q_1}}}
      \text{\quad by Theorem~\ref{th:sdfun-sdfunpart-sfunadd}}\\
    &\Textsep=\Sdfmult
      \Comp\Sflipl d
      \Comp\Sdfunpart0^{d+1}\Mlin\Sif_{\Sdfun X}
      \Comp\Tuple{\Sdfunpart1p,\Tuple{\Sdfunpart1{q_0},\Sdfunpart1{q_1}}}
    \text{\quad by Lemma~\ref{lemma:sif-slet-sdfunpart}}
  \end{align*}
  therefore
  \begin{align*}
    \Sdfunpart1f
    &=\Sdfmult
      \Comp\Sflipl d
      \Comp\Sdfunpart0^{d+1}\Mlin\Sif_{\Sdfun X}
      \Comp\Tuple{\Sdfunpart1p,\Tuple{\Sdfunpart1{q_0},\Sdfunpart1{q_1}}}\\
    &=\Sdfmult
      \Comp\Sflipl d
      \Comp\Sdfunpart0^{d+1}\Mlin\Sif_{\Sdfun X}
      \Comp\Tuple{\Psem{\Ldletv xP}{\Gamma,x:\Tdiff A},
      \Tuple{\Psem{\Ldletv x{Q_0}}{\Gamma,x:\Tdiff A},
      \Psem{\Ldletv x{Q_1}}{\Gamma,x:\Tdiff A}}}\\
    & \text{\Textsep by ind.~hyp.}\\
    &=\Psem{\Lsum{\Lflipl d{\Lif{d+1}{\Ldletv
      xP}{\Ldletv x{Q_0}}{\Ldletv x{Q_1}}}}}{\Gamma,x:\Tdiff A}\\
    &=\Psem{\Ldletv xM}{\Gamma,x:\Tdiff A}
  \end{align*}
  as required.
  The case $M=\Llet dyPQ$ is completely similar.

  \Proofcase %
  Assume that $M=M_0+M_1$. Whatever the last rule of $\delta$ is, we
  know that $(\Tseq{\Gamma,x:A}{M_i}B)_{i=0,1}$ and, by
  Theorem~\ref{th:sem-defined-sum}, that %
  $(g_i=\Psem{M_i}{\Gamma,x:A}\in\Kl\cL(\With ZU,X))_{i=0,1}$ are well
  defined and summable and that %
  $f=g_0+g_1$, where $f=\Psem M{\Gamma,x:A}$.
  By inductive hypothesis we have $\Psem{\Ldletv x{M_i}}\Gamma=g_i$
  for $i=0,1$.
  By left-linearity of composition in $\Kl\cL$, we have
  \begin{align*}
    f
    &=g_0+g_1\\
    &=\Psem{\Ldletv x{M_0}}{\Gamma,x:\Tdiff A}
      +\Psem{\Ldletv x{M_0}}{\Gamma,x:\Tdiff A}\\
    &=\Psem{\Ldletv x{M_0}+\Ldletv x{M_1}}{\Gamma,x:\Tdiff A}
      \text{\quad by Theorem~\ref{th:sem-defined-sum}}\\
    &=\Psem{\Ldletv x{M}}{\Gamma,x:\Tdiff A}
  \end{align*}
  as required.

  \Proofcase %
  Assume that $M=\Lfix P$ so that $\delta$ ends %
  with~\ref{rl:trfix} 
  and that we have %
  $\Tseq{\Gamma,x:A}P{\Timpl BB}$ so that, setting %
  $g=\Psem M{\Gamma,x:A}$ we have %
  $g\in\Kl\cL(\With ZU,\Simpl XX)$ and %
  $f=\Psem M{\Gamma,x:A}=\Sfix^X\Comp g$.
  We have
  \begin{align*}
    \Sdfunpart 1f
    &=\Sdfun{\Sfix^X}
      \Comp\Sdfunpart1g\\
    &=\Sfix^{\Sdfun X}
      \Comp\Cur(\Sdfun\Ev^{X,X})\Comp\Sdfunpart1g
      \text{\quad by Theorem~\ref{th:sdfun-sfix}}\\
    &=\Sfix^{\Sdfun X}
      \Comp\Cur(\Sdfmult
      \Comp\Sdfun\Ev^{X,\Sdfun X}
      \Comp\Sdfstr^1_{\Simpl X{\Sdfun X},X})
      \Comp\Sdfunpart1g
      \text{\quad by Lemma~\ref{lemma:sdfun-Ev-expression}}\\
    &=\Sfix^{\Sdfun X}
      \Comp{\Simplp{\Sdfun X}{\Sdfmult}}
      \Comp\Cur(\Sdfun\Ev^{X,\Sdfun X}
      \Comp\Sdfstr^1_{\Simpl X{\Sdfun X},X})
      \Comp\Sdfunpart1g
      \text{\quad by cartesian closedness}\\
    &=\Sfix^{\Sdfun X}
      \Comp{\Simplp{\Sdfun X}{\Sdfmult}}
      \Comp\Sdfunint^{X,\Sdfun X}
      \Comp\Sdfunpart1g
      \text{\quad by definition of }\Sdfunint^{X,\Sdfun X}\\
    &=\Sfix^{\Sdfun X}
      \Comp\Sdfmult
      \Comp\Sdfunint^{X,\Sdfun X}
      \Comp\Sdfunpart1g
  \end{align*}
  and hence
  \begin{align*}
    \Sdfunpart1f
    &=\Sfix^{\Sdfun X}
      \Comp\Sdfmult
      \Comp\Sdfunint^{X,\Sdfun X}
      \Comp\Sdfunpart1g\\
    &=\Sfix^{\Sdfun X}
      \Comp\Sdfmult
      \Comp\Sdfunint^{X,\Sdfun X}
      \Comp{\Psem{\Ldletv xP}{\Gamma,x:\Tdiff A}}\\
    &=\Psem{\Lfix{(\Lsum{\Ldiff{\Ldletv xP}})}}{\Gamma,x:\Tdiff A}\\
    &=\Psem{\Ldletv xM}{\Gamma,x:\Tdiff A}
  \end{align*}
  as required.
\end{proof}

\subsubsection{Soundness theorem}

\begin{theorem}[Soundness of the semantics] %
  \label{th:sem-term-invariant}
  Assume that $\Tseq\Gamma MA$ and $M\Red M'$ (so that
  $\Tseq\Gamma{M'}A$). Then we have $\Psem M\Gamma=\Psem{M'}\Gamma$.
\end{theorem}
\begin{proof}
  We consider the various cases in the definition of $\Red$. We set
  $Z=\Tsem\Gamma$ and $X=\Tsem A$.

  \Proofcase %
  Assume first that $M\Linred M'$. The fact that %
  $\Psem M\Gamma=\Psem{M'}\Gamma$ results simply from the linearity of
  the semantic constructs corresponding to the linear contexts. As
  an example, consider the situation where %
  $M=\Llet dx{M_0+M_1}N$ and %
  $M'=\Llet dx{M_0}N+\Llet dx{M_1}N$ (so that we have %
  $\Tseq\Gamma{M_i}{\Tdiffm d\Tnat}$ and %
  $\Tseq\Gamma NB$ for a type $B$, and we have $A=\Tdiffm dB$). Then
  we have %
  \begin{align*}
    \Psem M\Gamma
    &=\Sdfunpart0^d\Mlin\Slet
      \Comp\Tuple{\Psem{M_0+M_1}\Gamma,\Psem N\Gamma}\\
    &=\Sdfunpart0^d\Mlin\Slet
      \Comp\Tuple{\Psem{M_0}\Gamma+\Psem{M_1}\Gamma,\Psem N\Gamma}\\
    &=\Sdfunpart0^d\Mlin\Slet
      \Comp\Tuple{\Psem{M_0}\Gamma,\Psem N\Gamma}
      +\Sdfunpart0^d\Mlin\Slet
      \Comp\Tuple{\Psem{M_1}\Gamma,\Psem N\Gamma}
  \end{align*}
  by the bilinearity of $\Slet$.

\Proofcase %
Assume that %
$M=\App{\Abst xAP}{Q}$ and %
$M'=\Subst PQx$, we directly apply Lemma~\ref{lemma:sem-subst}.

\Proofcase %
Assume that \(M=\Lfix N\) and %
\(M'=\App NM\). We directly apply Equation~\Eqref{eq:Sfix-charact}.

\Proofcase %
Assume that %
$M=\Ldiffp{\Abst xBP}$ and $M'=\Abst x{\Tdiff B}{\Ldletv xP}$ so
that %
$A=\Timplp{\Tdiff B}{\Tdiff C}$ and $\Tseq{\Gamma,x:B}PC$, so that
setting $U=\Tsem B$ and $Y=\Tsem C$ we have
$f=\Psem P{\Gamma,x:B}\in\Kl\cL(\With ZU,Y)$. Then we have
\begin{align*}
  \Psem M\Gamma
  &=\Sdfunint\Comp\Cur f\\
  &=\Cur(\Sdfun\Ev^{U,Y}\Comp\Sdfstr^1_{\Simpl UY,U})\Comp\Cur f\\
  &=\Cur(\Sdfun\Ev^{U,Y}\Comp\Sdfstr^1_{\Simpl UY,U}
    \Comp\Withp{\Cur f}{\Sdfun U})\\
  &=\Cur(\Sdfun\Ev^{U,Y}\Comp\Withp{\Sdfun(\Cur f)}{\Sdfun U}
    \Comp\Sdfstr^1_{Z,U})
    \text{\quad by nat.~of }\Sdfstr^1\\
  &=\Cur(\Sdfun(\Ev^{U,Y}\Comp\Withp{\Cur f}{U})\Comp\Sdfstr^1_{Z,U})\\
  &=\Cur(\Sdfun f\Comp\Sdfstr^1_{Z,U})\\
  &=\Cur(\Sdfunpart1f)\\
  &=\Cur(\Psem{\Ldletv xP}{\Gamma,x:\Tdiff B})
    \text{\quad by Lemma~\ref{lemma:sem-diff-subst}}\\
  &=\Psem{\Abst x{\Tdiff B}{\Ldletv xP}}\Gamma
\end{align*}

\Proofcase %
Assume that $M=\Lif0{\Num{\nu+1}}{P_0}{P_1}$ and $M'=P_1$ so that,
setting %
$p_i=\Psem{P_i}\Gamma\in\Kl\cL(Z,X)$ for $i=0,1$, we have %
$\Psem M\Gamma=\Mlin\Sif \Comp\Tuple{\Snum{\nu+1}\Comp
  0,\Tuple{p_0,p_1}}=p_1$ by the second diagram
of~\Eqref{eq:diag-conditional}, where %
$0$ is the unique element of $\Kl\cL(Z,\Top)$. The cases where %
$M=\Lsucc0{\Num \nu}\Red\Num{\nu+1}$, %
$M=\Lpred0{\Num 0}\Red\Num 0$, %
$M=\Lpred0{\Num{\nu+1}}\Red\Num \nu$, and %
$M=\Lif0{\Num 0}{P}{Q}\Red P$ are similar.

\Proofcase %
Assume that %
$M=\Llet 0x{\Num \nu}P$ and %
$M'=\Subst P{\Num \nu}x$. Let %
$f=\Psem P{\Gamma,x:\Tnat}\in\Kl\cL(\With Z\Snat,X)$.
We have $\Snum \nu\in\Kl\cL(Z,\Snat)$ and
$\Coalg\Snat\Compl\Snum \nu=\Prom{\Snum \nu}$ (actually $\Snum \nu$ is a
morphism in the Eilenberg-Moore category of $\Excl\_$) from which it
follows that
\begin{align*}
  \Psem M\Gamma
  &=\Slet^0(\Snum \nu,f)    \text{\quad using the notation of
    Section~\ref{sec:syn-basic-constr}}\\
  &=\Mlin\Slet\Comp\Tuple{\Snum \nu,\Cur f}
\\
  &=\Slet\Compl\Tensp{\Snum \nu}{\Cur f}\Compl\Contr Z
  \\
  &=\Evlin\Compl\Sym\Compl\Tensp{\Coalg\Snat}{\Limplp{\Excl\Snat}{X}}
    \Tensp{\Snum \nu}{\Cur f}\Compl\Contr Z\\
  &=\Evlin\Compl\Tensp f{\Prom{\Snum \nu}}\Compl\Contr Z\\
  &=f\Comp\Tuple{Z,\Snum \nu}\\
  &=\Psem{\Subst P{\Num \nu}x}\Gamma
\end{align*}
by
Lemma~\ref{lemma:sem-subst}.

\Proofcase %
Assume that %
$M=\Lprojd jd{\Abst xBP}$ and %
$M'=\Abst xB{\Lprojd idP}$ so that we must have %
$\Tseq{\Gamma,x:B}P{\Tdiffm{d+1}C}$ with $A=\Tdiffm{d}{\Timplp BC}$, %
$X=\Simplp U{\Sdfun^dY}$ where $U=\Tsem B$ and $Y=\Tsem C$. %
Let $f=\Psem P{\Gamma,x:B}\in\Kl\cL(\With ZU,\Sdfun^{d+1}Y)$ so that %
$\Psem{\Abst xBP}\Gamma=\Cur f\in\Kl\cL(Z,\Simpl U{\Sdfun^{d+1}Y})$
and %
$\Sdfun^d\Sproj i\Comp f\in\Kl\cL(\With ZU,\Sdfun^d Y)$.
Remember that, for any object $V$,
$\Sstr^0_{\Simpl UV,\Excl U} \in\cL(\Sdfun{\Limplp{\Excl
    U}V}\ITens\Excl X, \Sdfun(\Tens{\Limplp{\Excl X}{U}}{\Excl X}))$
and that by \Saxfun{} of~\cite{Ehrhard23a} the morphism %
$\Sstrc =\Cur({(\Sdfun\Evlin) \Compl\Sstr^0_{\Simpl UV,\Excl U}})
\in\cL(\Sdfun(\Simpl UV),\Simpl U{\Sdfun V})$ is an iso. %
Thanks to this iso we identify the objects %
$\Sdfun^{d+1}(\Simpl UY)$ and $\Simpl U{\Sdfun^{d+1} Y}$. %
Under this identification the morphisms
$\Sdfun^d\Sproj j\in\cL(\Sdfun^{d+1}(\Simpl UY),\Sdfun^d(\Simpl UY))$
and %
$\Simpl U{\Sdfun^d\Sproj j}\in\cL(\Simpl U{\Sdfun^{d+1} Y},\Simpl
U{\Sdfun^dY})$ are identified as well. Since %
$(\Simpl U{\Sdfun^d\Sproj j})\Comp\Cur f=\Cur(\Sdfun^d\Sproj j\Comp
f)$ we have
\begin{align*}
  \Psem M\Gamma
  &=\Sdfun^d\Sproj i(\Psem{\Abst xBP}\Gamma)\\
  &=(\Simpl U{\Sdfun^d\Sproj j})\Comp\Cur f\\
  &=\Cur(\Sdfun^d\Sproj j\Comp f)\\
  &=\Cur(\Sdfun^d\Sproj j\Comp \Psem P{\Gamma,x:B})\\
  &=\Cur(\Psem{\Lprojd idP}{\Gamma,x:B})\\
  &=\Psem{\Abst xB{\Lprojd idP}}\Gamma
\end{align*}
as required.

\Proofcase %
Assume that %
$M=\Lprojd id{\App PQ}$ and %
$M'=\App{\Lprojd i{d}P}{Q}$ %
so that we must have %
$\Tseq\Gamma P{\Timpl B{\Tdiffm{d+1}C}}$ and %
$\Tseq\Gamma QB$ with $A=\Tdiffm dC$. %
Setting %
$U=\Tsem B$, $Y=\Tsem C$, %
$p=\Psem P\Gamma$ and $q=\Psem Q\Gamma$ we have %
$\Psem M\Gamma=\Sdfun^d\Sproj i\Comp\App pq\in\Kl\cL(Z,\Tdiffm dY)$.
Then, by naturality of
evaluation, we have
$\Psem M\Gamma=(\Sdfun^d\Sproj
j)\Comp\Ev\Comp\Tuple{p,q}
=\Ev\Comp\Withp{\Simplp{U}{(\Sdfun^{d}\Sproj j)}}{U}
\Comp\Tuple{p,q} =\Ev\Comp\Tuple{(\Sdfun^{d+1}\Sproj j)\Comp p,q}$
under the same identification as in the previous case. That is %
$\Psem M\Gamma=\Psem{M'}\Gamma$.

\Proofcase %
Assume $M=\Lprojd jd{\Lsucc{e}P}$ and %
$M'=\Lsucc{e-1}{\Lprojd jdP}$ so that we must have $d<e$, %
$\Tseq\Gamma{\Lsucc eP}{\Tdiffm{e}\Tnat}$
and $A=\Tdiffm{e-1}\Tnat$.
We have $\Ssuc\in\cL(\Snat,\Snat)$ that we consider as usual as a
morphism in $\Kl\cL(\Snat,\Snat)$ so that
$\Sdfun^e\Ssuc\in\Kl\cL(\Sdfun^e\Snat,\Sdfun^e\Snat)$ and since $d<e$
we can write $\Sdfun^e\Snat=\Sdfun^{d+1}\Sdfun^{e-d-1}\Snat$ so that
$\Sdfun^d\Sproj i\in\Kl\cL(\Sdfun^e\Snat,\Sdfun^{e-1}\Snat)$ and we
have
\begin{align*}
  \Sdfun^d\Sproj i\Comp\Sdfun^e\Ssuc
  &=\Sdfun^d(\Sproj j\Comp\Sdfun^{e-d}\Ssuc)\\
  &=\Sdfun^d(\Sdfun^{e-d-1}\Ssuc\Comp\Sproj j)
    \text{\quad by nat.~of }\Sproj j\\
  &=\Sdfun^{e-1}\Ssuc\Comp\Sdfun^d\Sproj j
\end{align*}
so that $\Psem M\Gamma=\Psem{M'}\Gamma$. %
The cases $M=\Lprojd jd{\Lpred{e}P}$ and %
$M'=\Lpred{e-1}{\Lproj jP}$ (still with $d<e$) are completely similar.

\Proofcase %
Assume that %
$M=\Lprojd id{\Lif{e} N{P_0}{P_1}}$ and %
$M'=\Lif{e-1}{\Lprojd idN}{P_0}{P_1}$ with $d<e$. We must have %
$\Tseq\Gamma N{\Tdiffm e\Tnat}$ and %
$\Tseq\Gamma{P_j}B$ for $j=0,1$, for a type \(B\) uniquely determined
by \(A=\Tdiffm{e-1}B\). %
Let $Y=\Tsem B$, %
we have $X=\Sdfun^{e-1}Y$. Let %
$f=\Psem N\Gamma\in\Kl\cL(Z,\Sdfun^e\Snat)$ and %
$p_j=\Psem{P_j}\Gamma\in\Kl\cL(Z,Y)$ for $j=0,1$.
We have
\begin{align*}
  \Sdfun^d\Sproj i
  &\Comp\Sdfunpart0^e\Mlin\Sif
  \Comp\Tuple{f,\Tuple{p_0,p_1}}\\
  &=\Sdfun^d\Sproj i
    \Comp\Sdfun^e\Mlin\Sif
    \Comp\Sdfstr^0_{\Snat,\With YY}(e)
    \Comp\Tuple{f,\Tuple{p_0,p_1}}
    \text{\quad by Lemma~\ref{lemma:sdfunpart-it-expr}}\\
  &=\Sdfun^d(\Sproj i\Comp\Sdfun^{e-d}\Mlin\Sif)
    \Comp\Sdfstr^0_{\Snat,\With YY}(e)
    \Comp\Tuple{f,\Tuple{p_0,p_1}}\,.
\end{align*}
Now we have to consider the two cases $i=0$ and $i=1$. The first case
is dealt with as follows.
\begin{align*}
  \Sdfun^d\Sproj 0
  &\Comp\Sdfunpart0^e\Mlin\Sif
  \Comp\Tuple{f,\Tuple{p_0,p_1}}\\
  &=\Sdfun^d(\Sdfun^{e-d-1}\Mlin\Sif\Comp\Withp{\Sproj 0}{\Sproj 0}) 
    \Comp\Sdfstr^0_{\Snat,\With YY}(e)
    \Comp\Tuple{f,\Tuple{p_0,p_1}}\\
  &  \text{\Textsep by Theorems~\ref{th:sdiff-multilin}
    and~\ref{th:sproj-sdfun-mlin}}\\
  &=\Sdfun^{e-1}\Mlin\Sif
    \Comp\Withp{\Sdfun^d\Sproj 0}{\Sdfun^d\Sproj 0}
    \Comp\Withp{\Sdfun^e\Snat}{\Sin0(e)}
    \Comp\Tuple{f,\Tuple{p_0,p_1}}\\
    &\text{\Textsep by Lemma~\ref{lemma:sdfstr-it-charact}}\\
  &=\Sdfun^{e-1}\Mlin\Sif
    \Comp\Withp{\Sdfun^d\Sproj0}{\Sin0(e-1)}
    \Comp\Tuple{f,\Tuple{p_0,p_1}}
    \text{\quad by Lemma~\ref{lemma:proj-sin-deep}}\\
  &=\Sdfun^{e-1}\Mlin\Sif
    \Comp\Withp{\Sdfun^{e-1}\Snat}{\Sin0(e-1)}
    \Comp\Tuple{\Sdfun^d\Sproj 0\Comp f,\Tuple{p_0,p_1}}\\
  &=\Sdfunpart0^{e-1}\Mlin\Sif
    \Comp\Tuple{\Sdfun^d\Sproj 0\Comp f,\Tuple{p_0,p_1}}\,.
\end{align*}
Let us deal with the second case.
\begin{align*}
  \Sdfun^d\Sproj 1
  \Comp\Sdfunpart0^e\Mlin\Sif
  \Comp\Tuple{f,\Tuple{p_0,p_1}}
  &=\Sdfun^d(\Sdfun^{e-d-1}\Mlin\Sif\Comp\Withp{\Sproj 0}{\Sproj 1}) 
    \Comp\Sdfstr^0_{\Snat,\With XX}(e)
    \Comp\Tuple{f,\Tuple{p_0,p_1}}\\
  &\ +\Sdfun^d(\Sdfun^{e-d-1}\Mlin\Sif\Comp\Withp{\Sproj 1}{\Sproj 0}) 
    \Comp\Sdfstr^0_{\Snat,\With XX}(e)
    \Comp\Tuple{f,\Tuple{p_0,p_1}}\\
  &\text{\Textsep by Theorems~\ref{th:sdiff-multilin}
    and~\ref{th:sproj-sdfun-mlin}.}
\end{align*}
By Lemmas~\ref{lemma:sdfstr-it-charact} and~\ref{lemma:proj-sin-deep}
we have %
$\Withp{\Sdfun^d{\Sproj 1}}{\Sdfun^d{\Sproj 0}}
\Comp\Sdfstr^{\Snat,\With XX}_0(e) =\Withp{\Sdfun^d{\Sproj
    1}}{\Sin0(e-1)}$ and %
$\Withp{\Sdfun^d{\Sproj 0}}{\Sdfun^d{\Sproj 1}}
\Comp\Sdfstr^{\Snat,\With XX}_0(e) =\Withp{\Sdfun^d{\Sproj 0}}{0}$.
It follows by bilinearity of $\Sdfun^{e-1}\Mlin\Sif$ that
\begin{align*}
  \Sdfun^d\Sproj 1
  \Comp\Sdfunpart0^e\Mlin\Sif
  \Comp\Tuple{f,\Tuple{p_0,p_1}}
  &=\Sdfun^{e-1}\Mlin\Sif
    \Comp\Withp{\Sdfun^d\Sproj1}{\Sin0(e-1)}
    \Comp\Tuple{f,\Tuple{p_0,p_1}}\\
  &=\Sdfunpart0^{e-1}\Mlin\Sif
    \Comp\Tuple{\Sdfun^d\Sproj 1\Comp f,\Tuple{p_0,p_1}}\,.
\end{align*}
To summarize, for $i=0,1$ and if $d<e$ we have
\begin{align*}
  \Sdfun^d\Sproj i
  \Comp\Sdfunpart0^e\Mlin\Sif
  \Comp\Tuple{f,\Tuple{p_0,p_1}}
  =\Sdfunpart0^{e-1}\Mlin\Sif
  \Comp\Tuple{\Sdfun^d\Sproj i\Comp f,\Tuple{p_0,p_1}}\,.
\end{align*}
It follows that $\Psem M\Gamma=\Psem{M'}\Gamma$.

The case $M=\Lprojd id{\Llet{e}xMP}$ and
$M'=\Llet{e-1}x{\Lprojd idM}P$ with $d<e$ is handled similarly.

\Proofcase %
Assume that %
$M=\Lprojd id{\Lif eN{P_0}{P_1}}$ and %
$M'=\Lif{e}N{\Lprojd i{d-e}{P_0}}{\Lprojd i{d-e}{P_1}}$, with %
$e\leq d$. %
We must have $\Tseq\Gamma N{\Tdiffm e\Tnat}$ and there must be a type
$B$ such that %
$\Tseq\Gamma{P_j}B$ for $j=0,1$ so that %
\(\Tseq\Gamma{\Lif eN{P_0}{P_1}}{\Tdiffm eB}\). %
Since $\Tseq\Gamma MA$ the type $B$ must be of shape %
$\Tdiffm{d-e+1}C$ for a (uniquely determined) type $C$, and then %
$A=\Tdiffm{d}C$ and %
$\Tseq\Gamma{\Lif eN{P_0}{P_1}}{\Tdiffm{d+1}C}$. %
We set $U=\Tsem C$ and $Y=\Tsem B=\Sdfun^{d-e+1}U$, %
$f=\Psem N\Gamma\in\Kl\cL(Z,\Sdfun^e\Snat)$ and %
$p_j=\Psem{P_j}\Gamma\in\Kl\cL(Z,Y)$.
We have
\begin{align*}
  \Sdfun^d\Sproj i
  &\Comp\Sdfunpart0^e\Mlin\Sif
  \Comp\Tuple{f,\Tuple{p_0,p_1}}\\
  &=\Sdfun^d\Sproj i
    \Comp\Sdfun^e\Mlin\Sif
    \Comp\Sdfstr^0_{\Snat,\With YY}(e)    
    \Comp\Tuple{f,\Tuple{p_0,p_1}}\\
  & =\Sdfun^e(\Sdfun^{d-e}\Sproj i\Comp\Mlin\Sif)
    \Comp\Sdfstr^0_{\Snat,\Sdfun^{d-e+1}\Withp UU}(e)
    \Comp\Tuple{f,\Tuple{p_0,p_1}}\\
  &
    =\Sdfun^e(\Mlin\Sif    
    \Comp\Withp{\Snat}{\Sdfun^{d-e}\Sproj i})
    \Comp\Sdfstr^0_{\Snat,\Sdfun^{d-e+1}\Withp UU}(e)
    \Comp\Tuple{f,\Tuple{p_0,p_1}}\\
  &\text{\Textsep\Textsep by naturality of }\Mlin\Sif\\
  &
    =\Sdfun^e\Mlin\Sif
    \Comp\Withp{\Sdfun^e\Snat}{\Sdfun^{d}\Sproj i}
    \Comp\Withp{\Sdfun^e\Snat}{\Sin0(e)}
    \Comp\Tuple{f,\Tuple{p_0,p_1}}\\
  &
    =\Sdfun^e\Mlin\Sif
    \Comp\Withp{\Sdfun^e\Snat}{\Sin0(e)}
    \Comp\Withp{\Sdfun^e\Snat}{\Sdfun^{d-e}\Sproj i}
    \Comp\Tuple{f,\Tuple{p_0,p_1}}
    \text{\quad by nat.~of }\Sin0(e)\\
  &
    =\Sdfunpart0^e\Mlin\Sif
    \Comp\Tuple{f,\Tuple{(\Sdfun^{d-e}\Sproj i)\Comp p_0,
    (\Sdfun^{d-e}\Sproj i)\Comp p_1}}
\end{align*}
and it follows that %
$\Psem M\Gamma=\Psem{M'}\Gamma$ since %
$(\Sdfun^{d-e}\Sproj i)\Comp p_j=\Psem{\Lprojd i{d-e}{P_j}}\Gamma$ for
$j=0,1$.

The case where $M=\Lprojd id{\Llet{e}xMP}$ and %
$M'=\Llet exM{\Lprojd i{d-e}P}$ with $e\leq d$ is handled
similarly.

\Proofcase %
Assume that %
$M=\Lprojd 0d{\Lsumd dN}$ and %
$M'=\Lprojd0d{\Lprojd0d N}$ so that we must have %
$\Tseq\Gamma N{\Tdiffm{d+2}B}$ for a (uniquely determined) type $B$
such that $A=\Tdiffm{d}B$. Let $Y=\Tsem B$ so that $X=\Sdfun^{d}Y$ and
$f=\Psem N\Gamma\in\Kl\cL(Z,\Sdfun^{d+2}Y)$.
We have
\begin{align*}
  (\Sdfun^d\Sproj0)
  \Comp(\Sdfun^d\Sdfmult)
  \Comp f
  &=\Sdfun^d(\Sproj0\Comp\Sdfmult)
    \Comp f\\
  &=\Sdfun^d(\Sproj0\Comp\Sproj0)
    \Comp f\\
  &=(\Sdfun^d\Sproj0)\Comp(\Sdfun^d\Sproj0)\Comp f\\
  &=\Psem{M'}\Gamma\,.
\end{align*}
Notice that we are using the fact that we are composing \emph{linear}
morphisms in $\Kl\cL$ so that the equation
$\Sproj0\Comp\Sdfmult=\Sproj0\Comp\Sproj0$ holds in $\Kl\cL$ because %
$\Sproj0\Compl\Sfunadd=\Sproj0\Compl\Sproj0$ holds in $\cL$.

The case where $M=\Lprojd 1d{\Lsumd dN}$ and %
$M'=\Lprojd1d{\Lprojd0d N}+\Lprojd0d{\Lprojd1d N}$ is similar, using
the equation %
$\Sproj1\Compl\Sfunadd=\Sproj1\Compl\Sproj0+\Sproj0\Compl\Sproj1$ in $\cL$.

\Proofcase %
Assume that %
$M=\Lprojd id{\Lsumd eN}$ and %
$M'=\Lsumd{e-1}{\Lprojd idN}$ with $d<e$
so that we must have %
$\Tseq\Gamma N{\Tdiffm{e+2}B}$ for a (uniquely determined) type %
$B$ such that $A=\Tdiffm{e}B$. Let %
$Y=\Tsem B$ so that $X=\Sdfun^eY$. Let %
$f=\Psem N\Gamma\in\Kl\cL(Z,\Sdfun^{e+2}Y)$  so that %
$(\Sdfun^e\Sdfmult)\Comp f\in\Kl\cL(Z,\Sdfun^{e+1}X)$ and hence %
$(\Sdfun^d\Sproj i)\Comp(\Sdfun^e\Sdfmult)\Comp
f\in\Kl\cL(Z,\Sdfun^{e}X)$. We have
\begin{align*}
  (\Sdfun^d\Sproj i)
  \Comp(\Sdfun^e\Sdfmult)
  \Comp f
  &=\Sdfun^d(\Sproj i\Comp(\Sdfun^{e-d}\Sdfmult))
    \Comp f\\
  &=\Sdfun^d(\Sdfun^{e-d-1}\Sdfmult\Comp\Sproj i)
    \Comp f
    \text{\quad by naturality of }\Sproj i\\
  &=(\Sdfun^{e-1}\Sdfmult)
    \Comp(\Sdfun^d\Sproj i)
    \Comp f\\
  &=\Psem{M'}\Gamma\,.
\end{align*}

\Proofcase %
Assume that %
$M=\Lprojd id{\Lsumd eN}$ and %
$M'=\Lsumd{e}{\Lprojd i{d+1}N}$ with $e<d$ so that we must have %
$\Tseq\Gamma N{\Tdiffm{e+2}B}$ for some type $B$ and we have %
$\Tseq\Gamma{\Lsumd eN}{\Tdiffm{e+1}B}$.
There must be a type $C$ such
that %
$\Tdiffm{e+1}B=\Tdiffm{d+1}C$, and then
$A=\Tdiffm dC=\Tdiffm eB$. In other words %
$B=\Tdiffm{d-e}C$ %
and of course $C$ is uniquely determined by
$A$. Let %
$Y=\Tsem C$ so that we have %
$f=\Psem N\Gamma\in\Kl\cL(Z,\Tdiffm{d+2}Y)$ and hence %
$(\Sdfun^d\Sproj i)\Comp(\Sdfun^e\Sdfmult)\Comp
f\in\Kl\cL(Z,\Sdfun^{d}Y)$.
We have
\begin{align*}
  (\Sdfun^d\Sproj i)\Comp(\Sdfun^e\Sdfmult)\Comp f
  &=\Sdfun^e(\Sdfun^{d-e}\Sproj i\Comp\Sdfmult)
    \Comp f\\
  &=\Sdfun^e(\Sdfmult\Comp\Sdfun^{d-e+1}\Sproj i)
    \Comp f
    \text{\quad by nat.~of }\Sdfmult\text{, observing that }d-e\geq 1\\
  &=\Sdfun^e\Sdfmult
    \Comp\Sdfun^{d+1}\Sproj i
    \Comp f\\
  &=\Psem{M'}\Gamma\,.
\end{align*}

\Proofcase %
Assume that %
$M=\Lprojd id{\Lflipdl elN}$ and %
$M'=\Lflipdl{e-1}l{\Lprojd idN}$ with $d<e$. %
Then there must be a type $B$ such that %
$\Tseq\Gamma{N}{\Tdiffm{e+l+2}{B}}$ and hence %
$\Tseq\Gamma{\Lflipdl elN}{\Tdiffm{e+l+2}{B}}$, and %
$\Tseq\Gamma M{\Tdiffm{e+l+1}{B}}$ since $d<e+l+2$. So we have %
$A=\Tdiffm{e+l+1}{B}$. Let $Y=\Tsem B$ so that %
$f=\Psem N\Gamma\in\Kl\cL(Z,\Sdfun^{e+l+2}Y)$ and we have
\begin{align*}
  \Psem M\Gamma
  &=(\Sdfun^d\Sproj i)
    \Comp(\Sdfun^e\Sflipl l)
    \Comp f\\
  &=\Sdfun^d(\Sproj i\Comp\Sdfun^{e-d}\Sflipl l)
    \Comp f\\
  &=\Sdfun^d(\Sdfun^{e-d-1}\Sflipl l\Comp\Sproj i)
    \Comp f
    \text{\quad by naturality of }\Sproj i\\
  &=(\Sdfun^{e-1}\Sflipl l)
    \Comp(\Sdfun^d\Sproj i)
    \Comp f\\
  &=\Psem{M'}\Gamma\,.
\end{align*}

\Proofcase %
Assume that %
$M=\Lprojd id{\Lflipdl elN}$ and %
$M'=\Lflipdl{e}l{\Lprojd idM}$ with $e+l+2\leq d$. Then there must be
a type $B$ such that %
$\Tseq\Gamma N{\Tdiffm{e+l+2}B}$ and hence %
$\Tseq\Gamma{\Lflipdl elN}{\Tdiffm{e+l+2}B}$ and therefore there must
be a type $C$ such that %
$\Tdiffm{d+1}C=\Tdiffm{e+l+2}B$, that is %
$B=\Tdiffm{d-e-l-1}C$. Setting $Y=\Tsem C$ we have %
$f=\Psem N\Gamma\in\Kl\cL(Z,\Sdfun^{d+1}Y)$ and
\begin{align*}
  \Psem M\Gamma
  &=(\Sdfun^d\Sproj i)
    \Comp(\Sdfun^e\Sflipl l)
    \Comp f\\
  &=\Sdfun^e((\Sdfun^{d-e}\Sproj i)\Comp\Sflipl l)
    \Comp f\\
  &=\Sdfun^e((\Sdfun^{l+2}\Sdfun^{d-e-l-2}\Sproj i)\Comp\Sflipl l)
    \Comp f\\
  &=\Sdfun^e(\Sflipl l\Comp(\Sdfun^{d-e}\Sproj i))
    \Comp f
    \text{\quad by naturality of }\Sflipl l\\
  &=(\Sdfun^e\Sflipl l)
    \Comp(\Sdfun^d\Sproj i)
    \Comp f\\
  &=\Psem{M'}\Gamma\,.
\end{align*}

\Proofcase %
Assume that %
$M=\Lprojd{i_{l+1}}d{\cdots\Lprojd{i_0}d{\Lflipdl dlN}}$ and %
$M'=\Lprojd{i_0}d{\Lprojd{i_{l+1}}d{\cdots\Lprojd{i_1}d{N}}}$ so that
there must be a type $B$ such that %
$\Tseq\Gamma N{\Tdiffm{d+l+2}B}$ and hence
$\Tseq\Gamma M{\Tdiffm dB=A}$. Let $Y=\Tsem B$ and %
$f=\Psem N\Gamma\in\Kl\cL(Z,\Sdfun^{d+l+2}Y)$. We have
\(
(\Sdfun^d\Sproj{i_{l+1}})\Comp\cdots\Comp(\Sdfun^d\Sproj{i_{0}})
\in\Kl\cL(\Sdfun^{d+l+2}X,\Sdfun^{d}X)
\)
and then
\begin{align*}
  \Psem M\Gamma
  &=(\Sdfun^d\Sproj{i_{l+1}})
    \Comp\cdots
    \Comp(\Sdfun^d\Sproj{i_{0}})
    \Comp(\Sdfun^d\Sflipl l)
    \Comp f\\
  &=(\Sdfun^d\Sproj{i_{0}})
    \Comp
    (\Sdfun^d\Sproj{i_{l+1}})
    \Comp\cdots
    \Comp(\Sdfun^d\Sproj{i_{1}})
    \Comp f\\
  &=\Psem{M'}\Gamma
\end{align*}
by Lemma~\ref{lemma:sflipl-circular-permutation}.

\Proofcase %
Assume that %
$M=\Lprojd id{\Linjd jdN}$, $M'=N$ if $i=j$ and $M'=\Lzero$ if
$i\not=j$. %
We must have $\Tseq\Gamma N{\Tdiffm dB}$ for a type $B$ such that %
$A=\Tdiffm dB$. Setting $Y=\Tsem B$ we have %
$f=\Psem N\Gamma\in\Kl\cL(Z,\Sdfun^dY)$
and 
\begin{align*}
  (\Sdfun^d\Sproj i)
  \Comp(\Sdfun^d\Sin j)
  \Comp f
  &=\Sdfun^d(\Sproj i\Comp\Sin j)\Comp f\\
  &=\Kronecker ij f
\end{align*}
since $\Sproj i\Comp\Sin j=\Kronecker ij\Id$.

\Proofcase %
Assume that %
$M=\Lprojd id{\Linjd jeN}$ and %
$M'=\Linjd j{e-1}{\Lprojd idN}$ with $d<e$ so that we must have %
$\Tseq\Gamma N{\Tdiffm eB}$ and %
$A=\Tdiffm dC$ with $\Tdiffm{d+1}C=\Tdiffm{e+1}B$ so that %
$C=\Tdiffm{e-d}B$. Let $Y=\Tsem B$ so that we have %
$f=\Psem N\Gamma\in\Kl\cL(Z,\Sdfun^eY)$ and
\begin{align*}
  \Psem M\Gamma
  &=(\Sdfun^d\Sproj i)
    \Comp(\Sdfun^e\Sin j)
    \Comp f\\
  &=\Sdfun^d(\Sproj i\Comp\Sdfun^{e-d}\Sin j)
    \Comp f\\
  &=\Sdfun^d((\Sdfun^{e-d-1}\Sin j)\Comp\Sproj i)
    \Comp f
    \text{\quad by naturality of }\Sproj i\\
  &=(\Sdfun^{e-1}\Sin j)
    \Comp(\Sdfun^d\Sproj i)
    \Comp f\\
  &=\Psem{M'}\Gamma\,.
\end{align*}

\Proofcase %
Assume that %
$M=\Lprojd id{\Linjd jeN}$ and %
$M'=\Linjd j{e}{\Lprojd i{d-1}N}$ with $e<d$ so that we must have %
$\Tseq\Gamma N{\Tdiffm eB}$ and %
$A=\Tdiffm dC$ with $\Tdiffm{d+1}C=\Tdiffm{e+1}B$ so that %
$B=\Tdiffm{d-e}C$. Let $Y=\Tsem C$ so that we have %
$f=\Psem N\Gamma\in\Kl\cL(Z,\Sdfun^dY)$ and
\begin{align*}
  \Psem M\Gamma
  &=(\Sdfun^d\Sproj i)
    \Comp(\Sdfun^e\Sin j)
    \Comp f\\
  &=\Sdfun^e((\Sdfun^{d-e}\Sproj i)\Comp\Sin j)
    \Comp f\\
  &=\Sdfun^e(\Sin j\Comp(\Sdfun^{d-e-1}\Sproj i))
    \Comp f
    \text{\quad by naturality of }\Sin j\\
  &=(\Sdfun^e\Sin j)\Comp(\Sdfun^{d-1}\Sproj i)
    \Comp f\\
  &=\Psem{M'}\Gamma\,.
\end{align*}

\Proofcase %
Assume that %
$M=\Lprojd id{\Lprojd jeN}$ and %
$M'=\Lprojd j{e-1}{\Lprojd idN}$ with $d<e$. We must have %
$\Tseq\Gamma N{\Tdiffm{e+1}B}$ for a type $B$ such that %
$A=\Tdiffm{e-1}B$. Let $Y=\Tsem B$ so that we have %
$f=\Psem N\Gamma\in\Kl\cL(Z,\Sdfun^{e+1}Y)$ and
\begin{align*}
  \Psem M\Gamma
  &=(\Sdfun^d\Sproj i)
    \Comp(\Sdfun^e\Sproj j)
    \Comp f\\
  &=\Sdfun^d(\Sproj i\Comp(\Sdfun^{e-d}\Sproj j))
    \Comp f\\
  &=\Sdfun^d((\Sdfun^{e-d-1}\Sproj j)\Comp\Sproj i)
    \Comp f
    \text{\quad by naturality of }\Sproj i\\
  &=(\Sdfun^{e-1}\Sproj j)
    \Comp(\Sdfun^d\Sproj i)
    \Comp f\\
  &=\Psem{M'}\Gamma\,.
\end{align*}

\Proofcase %
Assume that %
$M=\Lprojd id{\Lprojd jeN}$ and %
$M'=\Lprojd j{e}{\Lprojd i{d+1}N}$ with $e\leq d$. We must have %
$\Tseq\Gamma N{\Tdiffm{e+1}B}$ for a type $B$ and then %
\(\Tseq\Gamma{\Lprojd jeN}{\Tdiffm eB}\) so that
\(B=\Tdiffm{d-e+1}C\) for a type \(C\) such that %
$A=\Tdiffm{d}C$. Let $Y=\Tsem C$ so that we have %
$f=\Psem N\Gamma\in\Kl\cL(Z,\Sdfun^{d+2}Y)$ and
\begin{align*}
  \Psem M\Gamma
  &=(\Sdfun^d\Sproj i)
    \Comp(\Sdfun^e\Sproj j)
    \Comp f\\
  &=\Sdfun^e((\Sdfun^{d-e}\Sproj i)\Comp\Sproj j)
    \Comp f\\
  &=\Sdfun^e(\Sproj j\Comp(\Sdfun^{d-e+1}\Sproj i))
    \Comp f
    \text{\quad by naturality of }\Sproj j\\
  &=(\Sdfun^{e}\Sproj j)
    \Comp(\Sdfun^{d+1}\Sproj i)
    \Comp f\\
  &=\Psem{M'}\Gamma\,.%
    \qedhere
\end{align*}
\end{proof}

\begin{definition} %
  \label{def:term-multiset-summability}
  Let \(S=\Mset{\List M1k}\in\Rsca{\Mfin\Lang}\) and assume that %
  \(\Tseq\Gamma SA\), that is \((\Tseq\Gamma{M_i}A)_{i=1}^k\).
  We say that \(S\) is \emph{\(\cL\)-summable} if the family %
  \((\Psem{M_i}\Gamma)_{i=1}^k\) is summable in
  \(\Kl\cL(\Tsem\Gamma,\Tsem A)\).
  When this property holds we set %
  \( \Psem S\Gamma
  =\sum_{i=1}^k\Psem{M_i}\Gamma\in\Kl\cL(\Tsem\Gamma,\Tsem A) \). %
\end{definition}

Notice that this is not a purely syntactic notion, it depends on the
notion of summability of \(\cL\).

\begin{remark}
  We can introduce an absolute notion of summability by quantifying
  universally on \(\cL\) but this will not be really useful here.
\end{remark}

\begin{theorem}[multiset soundness] %
  \label{th:term-multiset-soundness}
  If \(S\in\Rsca{\Mfin\Lang}\) is such that %
  \(\Tseq\Gamma SA\) and \(S\) is \(\cL\)-summable and if
  \(S\Topred S'\) then %
  \(S'\) is \(\cL\)-summable and \(\Psem{S'}\Gamma=\Psem{S}\Gamma\).
\end{theorem}
Remember that \(\Tseq\Gamma{S'}A\) by
Theorem~\ref{th:subj-reduction-ms}.

\begin{proof}
  The following cases are possible, for some %
  \(S_0=\Mset{\List N1k}\).

  \Proofcase %
  \(S=S_0+\Mset M\), \(M\Red 0\) and %
  \(S'=S_0\). Since %
  \((\Psem{N_1}\Gamma,\dots,\Psem{N_k}\Gamma,\Psem M\Gamma)\) is a
  summable family, we know that the family
  \((\Psem{N_1}\Gamma,\dots,\Psem{N_k}\Gamma)\) is summable by %
  Theorem~5
  of~\cite{Ehrhard23a}. %
  Moreover \(\Psem{S'}\Gamma=\Psem S\Gamma\) because
  \(\Psem M\Gamma=\Psem 0\Gamma=0\) by
  Theorem~\ref{th:sem-term-invariant}.

  \Proofcase %
  \(S=S_0+\Mset M\), \(M\Red M'\) and %
  \(S'=S_0+\Mset{M'}\). Then we have %
  \(\Psem{M'}\Gamma=\Psem{M}\Gamma\) by
  Theorem~\ref{th:sem-term-invariant} and hence \(S'\) is
  \(\cL\)-summable and \(\Psem{S'}\Gamma=\Psem S\Gamma\).

  \Proofcase %
  \(S=S_0+\Mset M\), \(M\Red M_0+M_1\) and %
  \(S'=S_0+\Mset{M_0,M_1}\). Then by
  Theorem~\ref{th:sem-term-invariant} we know that
  \(\Psem{M_0}\Gamma,\Psem{M_1}\Gamma\) are summable and
  \(\Psem M\Gamma=\Psem{M_0}\Gamma+\Psem{M_1}\Gamma\). %
  By
  Theorem~5 %
  of~\cite{Ehrhard23a}
  we know that %
  \(
  (\Psem{N_1}\Gamma,\dots,\Psem{N_k}\Gamma, %
  \Psem{M_0}\Gamma,\Psem{M_1}\Gamma)
  \) %
  is a summable family with
  \(\Psem S\Gamma=\Psem{N_1}\Gamma+\cdots+\Psem{N_k}\Gamma %
  +\Psem{M_0}\Gamma+\Psem{M_1}\Gamma\), %
  that is \(S'\) is \(\cL\)-summable and %
  \(\Psem{S'}\Gamma=\Psem S\Gamma\).
\end{proof}

\begin{remark}
  This semantic notion of multiset summability seems counter to the
  general philosophy of this work which is to take the notion of summability
  into account at a purely syntactic level.
  We hope to be able to avoid it by means of a ``parallel reduction''
  as explained in Section~\ref{sec:reducing-sums}; this does not seem
  trivial and is postponed to further work since it is not strictly
  necessary for our current purpose.
  We will see indeed in Section~\ref{sec:determinism} that these sums can
  be controlled (see Theorem~\ref{cor:sem-adequacy-det} which uses
  Theorem~\ref{th:term-multiset-soundness}), and even disposed of in
  Section~\ref{sec:det-machine}.
\end{remark}

\section{Completeness of the reduction rules} %
\label{sec:completeness}

\subsection{A differential abstract machine}\label{sec:stacks} %
\label{sec:diff-machine}
Our goal now is to define a specific reduction strategy within the
rewriting system \(\Lang\). We do this by means of an ``abstract
machine'' which has no environment (just as in the work of Krivine on
classical realizability). In further papers we will certainly develop
an environment machine more suitable for implementations.

So a state of our machine will consist of a term and a stack. It will
also contain an \emph{access word} which will be a finite sequence of
\(0\) and \(1\) describing which component of the output is sought:
the machine will allow terms of type \(\Tdiffm d\Tnat\) to be
evaluated and the word, of length \(d\), specifies which leaf of the
associated tree of height \(d\) has to be computed by the machine.

Let us say that a type $A$ is \emph{\Sharp{}} if it is not of shape %
$\Tdiff B$; in our type language, this is equivalent to %
$A=(\Timpl{A_1}{\cdots\Timpl{A_k}{\Tnat}})$ for some uniquely
determined types $\List A1k$. %
Any type $A$ can be written uniquely %
$A=\Tdiffm dE$ where $E$ is sharp and \(d\in\Nat\).
We use the letters $E,F,\dots$ for \Sharp{} types.

We say that a term is \emph{\Simplicit{}} if it does not contain the
following constructs: %
$\Lzero$ and  $M_0+M_1$.

A \emph{\Path{}} is an element $\alpha$ of $\Eset{0,1}^{<\omega}$, we
refer to Section~\ref{sec:not-multisets} for notations concerning
words.
Then we define the stacks as follows:
\begin{align*}
  s,t,\dots\Bnfeq
  \Stempty
  \Bnfor \Starg Ms
  \Bnfor \Stsucc s
  \Bnfor \Stpred s
  \Bnfor \Stif\alpha{M_0}{M_1}s
  \Bnfor \Stlet\alpha xMs
  \Bnfor \Stdiff is
\end{align*}
where $\alpha$ is a word, $i\in\Eset{0,1}$, and $M$, $M_0$ and $M_1$ are
\Simplicit. Stacks are typed by judgments of shape $\Stseq sE$ where
$E$ is a \Sharp{} type.
The typing rules for stacks are given in Figure~\ref{fig:stack-type}.

\begin{figure}
  \begin{center}
    \begin{prooftree}
      \infer0{\Stseq\Stempty\Tnat}
    \end{prooftree}
    \Treesep
    \begin{prooftree}
      \hypo{\Stseq s\Tnat}
      \infer1{\Stseq{\Stsucc s}\Tnat}
    \end{prooftree}
    \Treesep
    \begin{prooftree}
      \hypo{\Stseq s\Tnat}
      \infer1{\Stseq{\Stpred s}\Tnat}
    \end{prooftree}
  \end{center}
  \begin{center}
    \begin{prooftree}
      \hypo{\Stseq sE}
      \hypo{\Tseq{}{M_0}{\Tdiffm dE}}
      \hypo{\Tseq{}{M_1}{\Tdiffm dE}}
      \hypo{\delta\in\Intoset^d}
      \infer4{\Stseq{\Stif\delta{M_0}{M_1}s}{\Tnat}}
    \end{prooftree}
  \end{center}
  \begin{center}
    \begin{prooftree}
      \hypo{\Stseq sE}
      \hypo{\Tseq{x:\Tnat}{M}{\Tdiffm dE}}
      \hypo{\delta\in\Intoset^d}
      \infer3{\Stseq{\Stlet\delta{x}{M}s}{\Tnat}}      
    \end{prooftree}
    \Treesep
    \begin{prooftree}
      \hypo{\Tseq{}MA}
      \hypo{\Stseq sE}
      \infer2{\Stseq{\Starg Ms}{\Timpl AE}}
    \end{prooftree}
  \end{center}
  \begin{center}
    \begin{prooftree}
      \hypo{\Stseq{s}{\Timpl{\Tdiff A}E}}
      \hypo{i\in\Intoset}
      \infer2{\Stseq{\Stdiff is}{\Timpl AE}}
    \end{prooftree}
  \end{center}
  \caption{Typing rules for stacks}\label{fig:stack-type}
\end{figure}

\subsubsection{The machine} %
\label{sec:machine}
We define the states of our machine as follows:
\begin{align*}
  c,c_0,c_1,\dots
  \Bnfeq \State\delta Ms
  \Bnfor 0
  \Bnfor c_0+c_1
\end{align*}
where $\delta$ is a word, $M$ is a \Simplicit{} term and $s$ is a
stack. The state \(\State\delta Ms\) is well typed if
\begin{align*}
  \Stseq sE
  \text{\quad and\quad}
  \Tseq{}{M}{\Tdiffm{\Len\delta}E}
\end{align*}
for some sharp type \(E\) (uniquely determined by \(M\)). The state
\(0\) is always well typed and \(c_0+c_1\) is well typed if \(c_0\)
and \(c_1\) are. We define a rewriting system
\(\States\) such that \(\Rsca\States\) is the set of all states. A state is %
\Simplicit{} if it is not of shape \(0\) or \(c_1+c_2\).

\renewcommand\Stred{\Rel{\Rsred{\States}}} %
The associated reduction relation \(\Stred\) is defined in
Figure~\ref{fig:state-reduction}. It is a deterministic reduction
relation on states: determinism results from the fact that the rule to
be applied on $\State\alpha Ms$ is completely determined by the shape
of $M$ (actually, by the outermost construct of $M$). Notice that
states which are not \Simplicit{} (that is, which are of shape \(0\)
or \(c_0+c_1\)) cannot be reduced by the system \(\States\). The
associated reduction system \(\Msrs\States\), see
Section~\ref{sec:ms-rewriting}, is precisely designed to reduce such
sums.

\begin{figure}
{\footnotesize
  \centering
  \begin{align*}
    \State\delta{\App MN}{s}
    &\Stred\State{\delta}{M}{\Starg Ns}
    &\State{\delta}{\Abst xAM}{\Starg Ns}
    &\Stred\State{\delta}{\Subst MNx}{s}
    \\
    \State{\delta i}{\Ldiff M}{s}
    &\Stred\State\delta M{\Stdiff is}
    &\State{\delta}{\Abst xAM}{\Stdiff is}
    &\Stred\State{\delta i}{\Abst x{\Tdiff A}{\Ldletv xM}}{s}
    \\
    \State\delta{\Lfix M}s
    &\Stred\State\delta M{\Starg{\Lfix M}s}
    &&\\
    \State{\delta}{\Lsucc dM}{s}
    &\Stred\State{\delta}{M}{\Stsucc s}
    &\State{\delta}{\Lpred dM}{s}
    &\Stred\State{\delta}{M}{\Stpred s}
    \\
    \State{\Pempty}{\Num n}{\Stsucc s}
    &\Stred\State{\Pempty}{\Num{n+1}}{s}
    &\State{\Pempty}{\Num 0}{\Stpred s}
    &\Stred\State{\Pempty}{\Num{0}}{s}
    \\
    \State{\Pempty}{\Num{n+1}}{\Stpred s}
    &\Stred\State{\Pempty}{\Num{n}}{s}
    &&\\
    \State{\epsilon\delta}{\Lif dM{P_0}{P_1}}{s}
    &\Stred\State\delta M{\Stif\epsilon{P_0}{P_1}s}
    &\State{\Pempty}{\Num 0}{\Stif{\epsilon}{P_0}{P_1}s}
    &\Stred\State{\epsilon}{P_0}{s}\\
    \State{\Pempty}{\Num{n+1}}{\Stif{\epsilon}{P_0}{P_1}s}
    &\Stred\State{\epsilon}{P_1}{s}
    &&\\
    \State{\epsilon\delta}{\Llet dxNM}{s}
    &\Stred\State{\delta}{N}{\Stlet\epsilon xNs}
    &\State{\Pempty}{\Num n}{\Stlet\epsilon xNs}
    &\Stred\State\epsilon{\Subst N{\Num n}x}s
    \\
    \State{\epsilon i\delta}{\Linjd idM}{s}
    &\Stred\State{\epsilon\delta}{M}{s}
    &\State{\epsilon i\delta}{\Linjd{1-i}dM}{s}
    &\Stred0
    \\
    \State{\epsilon\alpha\delta}{\Lflipdl dlM}{s}
    &\Stred\State{\epsilon\Rcycle\alpha\delta}{M}{s}
    &\State{\epsilon\delta}{\Lprojd idM}{s}
    &\Stred\State{\epsilon i\delta}{M}{s}
    &\\  
    \State{\epsilon 0\delta}{\Lsumd dM}{s}
    &\Stred\State{\epsilon 00\delta}{M}{s}
    &\State{\epsilon 1\delta}{\Lsumd dM}{s}
    &\Stred\State{\epsilon 10\delta}{M}{s}\\
    &&&\hspace{6em}+\State{\epsilon 01\delta}{M}{s}
  \end{align*}}
  \caption{Reduction rules for states. Convention: $d=\Len\delta$,
    $e=\Len\epsilon$.}
  \label{fig:state-reduction}
\end{figure}

\begin{proposition}\label{prop:machine-subj-red}
  If \(c\Stred c'\) and $c$ is a well typed state, then $c'$ is a well
  typed state.
\end{proposition}
\begin{proof}
  Since \(c\Stred c'\) we must have $c=\State\alpha Ms$ with $\Stseq sE$ and
  $\Tseq{}M{\Tdiffm{\Len\alpha}E}$. We have to consider each rewrite
  rule of Figure~\ref{fig:state-reduction} so we reason by cases on
  the shape of $M$, we focus on the most interesting cases, the other
  ones are similar and easier.

  \Proofcase %
  Assume that $M=\Lif dN{P_0}{P_1}$ so that we must have %
  \(\Tseq{}N{\Tdiffm d\Tnat}\), %
  \(\Tseq{}{P_i}{\Tdiffm eE}\) for $i=0,1$ and hence %
  \(\Tseq{}{M}{\Tdiffm{d+e}E}\),
  \(\Stseq{s}{E}\) and
  $\Len\alpha=e+d$ so that we can write $\alpha=\epsilon\delta$ with %
  $\Len\delta=d$ and $\Len\epsilon=e$. Then we have %
  \(\Stseq{\Stif\epsilon{P_0}{P_1}s}{\Tnat}\) %
  and hence 
  \(c'=\State{\delta}{N}{\Stif\epsilon{P_0}{P_1}s}\) %
  is well typed.

  \Proofcase %
  Assume $M=\Ldiff N$ so that we have $\Len\alpha>0$ and %
  $\Tseq{}{N}{\Tdiffm{\Len\alpha-1}E=\Timplp{A}{\Tdiffm{\Len\alpha-1}F}}$
  for some type $A$, and $E=\Timplp AF$ where $F$ is \Sharp. %
  We have %
  \[
    \Tseq{}M{\Timplp{\Tdiff A}{\Tdiffm{\Len\alpha}F}
      =\Tdiffm{\Len\alpha}{\Timplp AF}}
  \]
  and we can write %
  $\alpha=\delta i$ and with this notation %
  $c'=\State{\delta}{N}{\Stdiff is}$ is well typed since %
  $\Stseq s{\Timpl{\Tdiff A}{F}}$ and hence %
  $\Stseq{\Stdiff is}{\Timpl A{F}}$.

  \Proofcase %
  Assume that $M=\Lsumd dN$ then we have %
  $\Tseq{}{N}{\Tdiffm{d+2}{B}}$ for some type $B=\Tdiffm eE$ where %
  $E$ is \Sharp{} and $\Stseq{s}{E}$. And %
  $\Tseq{}{M}{\Tdiffm{d+1}B=\Tdiffm{d+e+1}E}$ and hence we can write %
  $\alpha=\epsilon i\delta$ with %
  $i\in\Eset{0,1}$, $\Len\epsilon=e$ and $\Len\delta=d$. %
  If $i=0$ we have %
  $c'=\State{\epsilon 00\delta}{N}{s}$ which is well typed since %
  $\Tdiffm{d+2}{B}=\Tdiffm{d+e+2}E=\Len{\epsilon 00\delta}$. %
  If $i=1$ we have %
  $c'=c_0+c_1$ with $c_j=\State{\epsilon(1-j)j\delta}{N}{s}$ for
  $j=0,1$, and $c_0$ and $c_1$ are well typed for the same reason.

  \Proofcase %
  Assume that $M=\Linjd idN$ with $i\in\Eset{0,1}$. Then we have %
  $\Tseq{}{N}{\Tdiffm dB}$ for some type $B=\Tdiffm eE$ where $E$ is
  \Sharp, $\Stseq sE$ and $\Tseq{}{M}{\Tdiffm{e+d+1}{E}}$. So we can
  write %
  $\alpha=\epsilon j\delta$ where $j\in\Eset{0,1}$. %
  If $j\not=i$ we have $c'=0$ and hence \(c'\) is the sum of the empty
  family of well typed states. %
  If $j=i$ then %
  $c'=\State{\epsilon\delta}{N}{s}$ which is well typed since %
  $\Tdiffm dB=\Tdiffm{e+d}E$.
\end{proof}

\begin{proposition}
  If $c$ is a \Simplicit{} well typed state which is $\Stred$-normal
  then there is $\nu\in\Nat$ such that
  $c=\State{\Pempty}{\Num\nu}{\Stempty}$.
\end{proposition}
\begin{proof}
  We have $c=\State\alpha Ms$ with $\Tseq{}M{\Tdiffm{\Len\alpha}E}$
  and $\Stseq sE$, we reason by cases on the last typing rule of $M$
  which is \Simplicit, that is, on the structure of $M$.

  \Proofcase %
  $M$ cannot be a variable because it is closed.

  \Proofcase %
  If $M=\Num\nu$ then $E=\Tnat$ and $\alpha=\Pempty$ and hence
  $\Stseq s\Tnat$. According to the typing rule for stacks, $s$ must
  be of one of the following shapes: $\Stempty$, $\Stsucc t$,
  $\Stpred t$, $\Stif\epsilon{P_0}{P_1}t$, $\Stlet\epsilon xPt$ for
  some stack $t$ and the first case only is possible since $c$ is
  normal.

  \Proofcase %
  $M$ cannot be $\App NP$ or $M=\Lfix N$ since $c$ is normal.

  \Proofcase %
  Assume that $M=\Abst xAN$ so that %
  $E=\Timplp AF$ and we must have $\Stseq s{\Timpl AF}$. According to
  the typing rule for stacks we must have $s=\Starg Pt$ or
  $s=\Stdiff it$. In both cases a reduction rule applies, contradicting
  the assumption that $c$ is normal. This case is impossible.

  \Proofcase %
  Assume that $M=\Lif dN{P_0}{P_1}$ so that we must have %
  $\Tseq{}N\Tnat$ and %
  $\Tseq{}{P_i}{B=\Tdiffm eE}$ for $i=0,1$ where $E$ is a \Sharp{}
  type, and $\Stseq sE$, and $\Tseq{}M{\Tdiffm{e+d}E}$. %
  Since $c$ is well typed we must have $\Len\alpha=e+d$ %
  with $e=\Len\epsilon$, $d=\Len\delta$ and
  $\alpha=\epsilon\delta$. It follows that $c$ is not normal,
  reducing to $\State\delta N{\Stif\epsilon{P_0}{P_1}{s}}$.

  \Proofcase %
  The case $M=\Llet dxPM$ is completely similar to the previous one.

  \Proofcase %
  The remaining cases are dealt with as in the proof of
  Proposition~\ref{prop:machine-subj-red}: in each case it appears
  that, because $c$ is well typed, it cannot be normal.
\end{proof}

\subsubsection{Context associated with a stack, term associated with a state}
Given %
$\delta\in\Eset{0,1}^d$, %
\(e\in\Nat\) and a term $M$, we define a term %
$\Lprojd\delta eM$ by %
$\Lprojd{\Seqempty}e{M}=M$ and %
$\Lprojd{\delta i}eM=\Lprojd\delta e{\Lprojd ieM}$. %

Given a stack \(s\) such that %
\(\Stseq sE\) we define a context \(\Stctx s\) in
Figure~\ref{fig:stack-ctx}. Notice that this context is closed (it has
no free occurrences of variables). Remember that the notion of linear
context is defined in Equation~\Eqref{eq:linear-context}.

\begin{figure}
  \centering
  \begin{align*}
    \Stctx\Stempty &= \Echole
    & \Stctx{\Stsucc s} &= \Stctx s[\Lsucc0{\Echole}] \\
    \Stctx{\Stpred s} &= \Stctx s[\Lpred0{\Echole}]
    & \Stctx{\Stif\delta{M_0}{M_1}s}
                        &= \Stctx s[\Lproj\delta{\Lif 0\Echole{M_0}
                          {M_1}}]\\
    \Stctx{\Stlet\delta xMs} &= \Stctx s[\Lproj\delta{\Llet 0x\Echole{M}}]
    & \Stctx{\Starg Ms} &= \Stctx s[\App\Echole M]\\
    \Stctx{\Stdiff is} &= \Stctx s[\Lproj i{\Ldiff\Echole}]
  \end{align*}
  \caption{Context associated with a stack}
  \label{fig:stack-ctx}
\end{figure}

\begin{lemma}
  Let \(s\) be a well typed state such that \(\Stseq sE\). Then the
  context \(\Stctx s\) is linear and satisfies
  \(\Tseq{x:E}{\Stctx s[x]}{\Tnat}\).
\end{lemma}
\begin{proof}
  Straightforward induction on the typing derivation of \(s\).
\end{proof}

\begin{lemma}
  If $\Tseq\Gamma M{\Tdiffm{d+e}A}$ then %
  $\Tseq\Gamma{\Lprojd\delta eM}{\Tdiffm eA}$ if $d=\Len\delta$.
\end{lemma}
The proof is straightforward. We set
\(\Lproj\delta M=\Lprojd\delta 0M\). 
Given a state \(c\) of shape \(c=\State\delta Ms\) with %
\(\Stseq sE\) we set %
\(\Tofst c=\Stctx s[\Lproj\delta{M}]\). We extend this definition to all
states by
\begin{align*}
  \Tofst 0&=0
            &\Tofst{c_0+c_1}&=\Tofst{c_0}+\Tofst{c_1}\,.
\end{align*}

\begin{lemma}
  If \(c\) is a well typed state which is not of shape \(c_0+c_1\)
  then \(\Tseq{}{\Tofst c}{\Tnat}\).
\end{lemma}

\begin{theorem} %
  \label{th:red-state-term}
  If \(c\) is a well typed state and \(c\Stred c'\) then %
  \(\Tofst c\Trcl\Red\Tofst{c'}\).
\end{theorem}
\begin{proof}
  We must have \(c=\State\delta Qs\) with %
  \(\Stseq sE\) and %
  \(\Tseq{}{Q}{\Tdiffm dE}\) where \(\Len\delta=d\) so that %
  \(\Tofst c=\Stctx s[\Lproj\delta Q]\) and we reason by considering
  the various cases in the transition $c\Stred c'$ as listed in
  Figure~\ref{fig:state-reduction}.

  \Proofcase %
  \(Q=\App MN\) so that %
  \(\Tofst c=\Stctx s[\Lproj\delta{\App MN}]\), and %
  \(c'=\State{\delta}{M}{\Starg Ns}\) and hence %
  \(
  \Tofst{c'}
  =\Stctx{\Starg Ns}[\Lproj\delta M]
  =\Stctx s[\App{\Lproj\delta M}N]
  \).
  We have %
  \(\Lproj\delta{\App MN}\Trcl\Red\App{\Lproj\delta M}{N}\) %
  (in \(\Len\delta\) steps)
  and hence \(\Tofst c\Trcl\Red\Tofst{c'}\) %
  since \(\Stctx s\) is an evaluation context.

  \Proofcase %
  \(Q=\Abst xAM\) and \(s=\Starg Nt\) so that %
  \(
  \Tofst c
  =\Tofst{\Starg Nt}[\Lproj\delta{\Abst xAM}]
  =\Tofst t[\App{\Lproj\delta{\Abst xAM}}N]
  \). %
  And \(c'=\State{\delta}{\Subst MNx}{t}\) so that %
  \(\Tofst{c'}=\Tofst t[\Lproj\delta{\Subst MNx}]\). %
  We have %
  \begin{align*}
  \App{\Lproj\delta{\Abst xAM}}N
  &\Trcl\Red\App{\Abst xA{\Lproj\delta M}}{N}\\
  &\Red\Subst{\Lproj\delta M}{N}x
  =\Lproj\delta{\Subst MNx}
  \end{align*} %
  and hence \(\Tofst c\Trcl\Red\Tofst{c'}\). %

  \Proofcase %
  \(Q=\Ldiff M\) and \(\delta=\epsilon i\) so that %
  \(
  \Tofst c
  =\Tofst{s}[\Lproj{\epsilon i}{\Ldiff M}]
  \). %
  And %
  \(
  c'=\State{\epsilon}{M}{\Stdiff is}
  \) %
  so that %
  \(
  \Tofst{c'}
  =\Stctx{\Stdiff is}[\Lproj\epsilon M]
  =\Stctx s[\Lproj i{\Ldiff{\Lproj\epsilon M}}]
  \). %
  We have %
  \begin{align*}
  \Lproj{\epsilon i}{\Ldiff M}
  &\Trcl\Red
  \Lproj i{\Lprojd\epsilon 1{\Ldiff M}}\\
  &\Red
  \Lproj i{\Ldiff{\Lproj\epsilon M}}
  \end{align*} %
  and hence \(\Tofst c\Trcl\Red\Tofst{c'}\).

  \Proofcase %
  \(Q=\Abst xAM\) and \(s=\Stdiff it\) with %
  \(i\in\Eset{0,1}\) so that we have %
  \(
  \Tofst c
  =\Stctx{\Stdiff it}[\Lproj\delta{\Abst xAM}]
  =\Stctx t[\Lproj i{\Ldiffp{\Lproj\delta{\Abst xAM}}}]
  \). %
  And
  \(
  c'=\State{\delta i}{\Abst x{\Tdiff A}{\Ldletv xM}}{t}
  \), so %
  \(
  \Tofst{c'}
  =\Stctx t[\Lproj{\delta i}{\Abst x{\Tdiff A}{\Ldletv xM}}]
  \). %
  We have %
  \begin{align*}
  \Lproj i{\Ldiffp{\Lproj\delta{\Abst xAM}}}
  &\Trcl\Red
  \Lproj i{\Lprojd\delta 1{\Ldiffp{\Abst xAM}}}\\
  &\Trcl\Red
  \Lproj\delta{\Lproj i{\Ldiffp{\Abst xAM}}}\\
  &\Red
  \Lproj\delta{\Lproj i{\Abst x{\Tdiff A}{\Ldletv xM}}}
  =\Lproj{\delta i}{\Abst x{\Tdiff A}{\Ldletv xM}}
  \end{align*} %
  and hence \(\Tofst c\Trcl\Red\Tofst{c'}\).

  \Proofcase %
  \(Q=\Lfix M\) so that %
  \(
  \Tofst c
  =\Stctx s[\Lproj\delta{\Lfix M}]
  \) %
  and
  \(
  \Tofst{c'}
  =\Stctx{\Starg{\Lfix M}s}[\Lproj\delta{M}]
  =\Stctx s[\App{\Lproj\delta{M}}{\Lfix M}]
  \). %
  We have %
  \begin{align*}
  \Lproj\delta{\Lfix M}
  &\Red\Lproj\delta{\App M{\Lfix M}}\\
  &\Trcl\Red\App{\Lproj\delta{M}}{\Lfix M}
  \end{align*} %
  and hence \(\Tofst c\Trcl\Red\Tofst{c'}\).

  \Proofcase %
  \(Q=\Lsucc dM\) with \(d=\Len\delta\), so %
  \(
  \Tofst c
  =\Stctx s[\Lproj\delta{\Lsucc dM}]
  \). And %
  \(
  \Tofst{c'}
  =\Stctx{\Stsucc s}[\Lproj\delta M]
  =\Stctx s[\Lsucc0{\Lproj\delta M}]
  \). %
  We have
  \(
  \Lproj\delta{\Lsucc dM}
  \Trcl\Red
  \Lsucc0{\Lproj\delta M}
  \) %
  in $d$ steps
  and hence \(\Tofst c\Trcl\Red\Tofst{c'}\).

  \Proofcase %
  The case $Q=\Lpred dM$ is similar and the cases $Q=\Num n$ are
  straightforward (and then $\delta=\Pempty$ and there are several
  cases to consider as to the shape of $s$, they are all easy).

  \Proofcase %
  \(Q=\Lif e{M}{P_0}{P_1}\) and \(\delta=\eta\epsilon\) with %
  \(e=\Len\epsilon\). 
  We have %
  \(
  \Tofst c
  =\Stctx s[\Lproj{\eta\epsilon}{\Lif e{M}{P_0}{P_1}}]
  \) %
  and %
  \(
  \Tofst{c'}
  =\Tofst{\Stif{\eta}{P_0}{P_1}s}[\Lproj\epsilon M]
  =\Tofst s[\Lproj\eta{\Lif 0{\Lproj\epsilon M}{P_0}{P_1}}]
  \). %
  We have %
  \[
  \Lproj{\eta\epsilon}{\Lif e{M}{P_0}{P_1}}
  =\Lproj\eta{\Lproj\epsilon{\Lif e{M}{P_0}{P_1}}}
  \Trcl\Red
  \Lproj\eta{\Lif 0{\Lproj\epsilon M}{P_0}{P_1}}
  \] %
  and hence \(\Tofst c\Trcl\Red\Tofst{c'}\).

  \Proofcase %
  The case $Q=\Llet dxNM$ is similar.

  \Proofcase %
  \(Q=\Linjd ieM\) and \(\delta=\alpha i\epsilon\) so that %
  \(
  \Tofst c
  =\Stctx s[{\Lproj{\alpha i\epsilon}{\Linjd ieM}}]
  \) and %
  \(
  \Tofst{c'}
  =\Stctx s[{\Lproj{\alpha\epsilon}{M}}]
  \). We have %
  \begin{align*}
  \Lproj{\alpha i\epsilon}{\Linjd ieM}
  &=\Lproj\alpha{\Lproj i{\Lproj\epsilon{\Linjd ieM}}}\\
  &\Trcl\Red
  \Lproj\alpha{\Lproj i{\Linjd i0{\Lproj\epsilon M}}}\\
  &\Red
  \Lproj\alpha{\Lproj\epsilon M}
  =\Lproj{\alpha\epsilon}M
  \end{align*} %
  and hence \(\Tofst c\Trcl\Red\Tofst{c'}\).

  \Proofcase %
  \(Q=\Linjd ieM\) and \(\delta=\alpha{(1-i)}\epsilon\) so that %
  \(
  \Tofst c
  =\Stctx s[{\Lproj{\alpha{(1-i)}\epsilon}{\Linjd ieM}}]
  \) and %
  \(
  \Tofst{c'}
  =\Tofst0
  =0
  \). We have %
  \begin{align*}
  \Lproj{\alpha(1-i)\epsilon}{\Linjd ieM}
  &=\Lproj\alpha{\Lproj{1-i}{\Lproj\epsilon{\Linjd ieM}}}\\
  &\Trcl\Red
  \Lproj\alpha{\Lproj{1-i}{\Linjd i0{\Lproj\epsilon M}}}
  \Red
  0
  \end{align*} %
  and hence \(\Tofst c\Trcl\Red\Tofst{c'}\).

  \Proofcase %
  \(Q=\Lflipdl elM\) and %
  \(\delta=\alpha\Listter i{l+1}0\epsilon\) with %
  \(\Len\epsilon=e\).
  Then %
  \(
  \Tofst c
  =\Stctx s[{\Lproj{\alpha\Listter i{l+1}0\epsilon}{\Lflipdl elM}}]
  \) and %
  \(
  \Tofst{c'}
  =\Stctx s[{\Lproj{\alpha i_0\Listter i{l+1}{1}\epsilon}{M}}]
  \).
  Next observe that %
  \begin{align*}
  \Lproj{\alpha\Listter i{l+1}0\epsilon}{\Lflipdl elM}
  &\Trcl\Red
  \Lproj{\alpha\Listter i{l+1}0}{\Lflipdl 0l{\Lproj\epsilon M}}\\
  &\Red
  \Lproj{\alpha i_0\Listter i{l+1}{1}}{\Lproj\epsilon M}
  =\Lproj{\alpha i_0\Listter i{l+1}{1}\epsilon}{M}
  \end{align*} %
  and hence \(\Tofst c\Trcl\Red\Tofst{c'}\).

  \Proofcase %
  \(Q=\Lsumd eM\) and \(\delta=\alpha 0\epsilon\) with %
  \(\Len\epsilon=e\).
  Then we have %
  \(
  \Tofst c
  =\Stctx s[{\Lproj{\alpha 0\epsilon}{\Lsumd eM}}]
  \) and %
  \(
  \Tofst{c'}
  =\Stctx s[{\Lproj{\alpha 00\epsilon}{M}}]
  \). %
  And %
  \begin{align*}
  \Lproj{\alpha 0\epsilon}{\Lsumd eM}
  &\Trcl\Red
  \Lproj\alpha{\Lproj0{\Lsumd0{\Lproj\epsilon M}}}\\
  &\Red
  \Lproj{\alpha 00\epsilon}{M}
  \end{align*} %
  and hence \(\Tofst c\Trcl\Red\Tofst{c'}\).
  
  \Proofcase %
  \(Q=\Lsumd eM\) and \(\delta=\alpha 1\epsilon\) with %
  \(\Len\epsilon=e\).
  Then we have %
  \(
  \Tofst c
  =\Stctx s[{\Lproj{\alpha 1\epsilon}{\Lsumd eM}}]
  \) and %
  \(
  \Tofst{c'}
  =\Stctx s[{\Lproj{\alpha 10\epsilon}{M}}]
  +\Stctx s[{\Lproj{\alpha 01\epsilon}{M}}]
  \).
  And %
  \begin{align*}
  \Lproj{\alpha 1\epsilon}{\Lsumd eM}
  &\Trcl\Red
  \Lproj\alpha{\Lproj1{\Lsumd0{\Lproj\epsilon M}}}\\
  &\Red
  \Lproj{\alpha 10\epsilon}{M}
  +\Lproj{\alpha 01\epsilon}{M}
  \end{align*} %
  and hence \(\Tofst c\Trcl\Red\Tofst{c'}\).
\end{proof}

For \(C=\Mset{\List c1k}\in\Mfin{\Rsca\States}\) we define %
\(\Tofst
C=\Mset{\Tofst{c_1},\dots,\Tofst{c_k}}\in\Mfin{\Rsca\Lang}\). We say
that \(C\) is well typed if \(\List c1k\) are well typed.

\begin{theorem}
  \label{th:}
  If \(C\in\Mfin{\Rsca\States}\) is well typed and %
  \(C\Rel{\Rsred{\Msrs\States}}C'\) then %
  \(\Tofst C\Rel{\Rsred{\Msrs\Lang}}\Tofst{C'}\).
\end{theorem}
This is an immediate corollary of Theorem~\ref{th:red-state-term} and
of the definition of \(\Msrs\States\) and \(\Msrs\Lang\) follow the
same pattern. Let us say that
\(C=\Mset{\List c1k}\in\Rsca{\Msrs\States}\) is \(\cL\)-summable if
\(C\) is well typed and the family of terms
\((\Tofst{c_1},\dots,\Tofst{c_k})\) is \(\cL\)-summable (see
Definition~\ref{def:term-multiset-summability}).

\begin{theorem}%
  \label{cor:K-mset-summability-pres}
  If \(C\in\Rsca{\Msrs\States}\) is \(\cL\)-summable and
  \(C\Rel{\Rsred{\Msrs\States}}C'\) then \(C'\) is \(\cL\)-summable.
\end{theorem}
\begin{proof}
  Use Theorem~\ref{th:term-multiset-soundness}.
\end{proof}

\subsection{Soundness}\label{sec:K-machine-sound}
Let \(s\) be a stack such that \(\Stseq sE\) and a variable \(x\) we
know that \(\Stctx s[x]\) is a term such that %
\(\Tseq{x:E}{\Stctx s[x]}{\Tnat}\).
So we can define the semantics of \(s\) by
\(\Psem s{}=\Psem{\Stctx s[x]}{x:E}\in\Kl\cL(\Tsem E,\Snat)\).

\begin{lemma}
  If \(\Stseq sE\) then \(\Psem s{}\) is linear.
\end{lemma}
This is due to the fact that the context \(\Stctx s\) is always
linear.
The proof is straightforward.

If \(c\) is a well typed state, we set
\(\Psem c{}=\Psem{\Tofst c}{}\in\cL(\Sone,\Snat)\) and if %
\(C=\Mset{\List c1k}\in\Rsca{\Msrs\States}\) is \(\cL\)-summable we
set %
\(\Psem C{}=\sum_{i=1}^k\in\Psem{c_i}{}\) which is a well defined
element of \(\cL(\Sone,\Snat)\) by definition of \(\cL\)-summability.

\begin{theorem}
  \label{th:machine-sound}
  If \(C\in\Rsca{\Msrs\States}\) is \(\cL\)-summable and %
  \(C\Rel{\Rsred{\Msrs\States}}C'\) then %
  \(\Psem{C'}{}=\Psem C{}\).
\end{theorem}
This makes sense by Theorem~\ref{cor:K-mset-summability-pres} which
entails that \(C'\) is \(\cL\)-summable. The proof is a
straightforward application of the definition of
\(\Rsred{\Msrs\States}\).

\subsection{Intersection typing system for terms}
\label{sec:inter-types-terms}
We present the interpretation of terms in \(\REL\) (see
Section~\ref{sec:rel-model}) as a deduction system so as to make the
statement and the proofs of our normalization result and of the lemmas
it uses (see Section~\ref{sec:int-seq-normalization}) more readable and
natural.
Indeed these proofs are most naturally presented as proofs by
induction on derivations in this typing system.
The main idea is to consider the elements of the sets interpreting
types as intersection types, which are non-idempotent since we use
multisets and not sets in the definition of \(\Excl X\).

\subsubsection{Complements about \(\REL\)}
\label{sec:complements-REL}
Before describing the relational interpretation of types and terms, we
need to complement Section~\ref{sec:rel-model} with a few words about
the generalized flip of Lemma~\ref{lemma:sflipl-circular-permutation},
the additive strength and the basic operations on integers in
\(\REL\).

For \(l\in\Nat\), the generalized flip of length \(l+2\) is %
\( \Sflipl l= \{((\alpha,a),(\Rcycle\alpha,a)) \St
\alpha\in\Into^{l+2}\text{ and }a\in X\}
\in\REL(\Scfun^{l+2}X,\Scfun^{l+2}X) \). %
We prove this property this by induction on \(l\). %
For \(l=0\) this is due to the fact that \(\Sflipl 0=\Sflip\). %
Assume that this property holds for \(l\) %
and remember that by definition %
\(\Sflipl{l+1}=\Sflip\Compl(\Scfun\Sflipl l)\). We have %
\(
\Scfun\Sflipl l
=\{(r\alpha,a),(r\Rcycle\alpha,a))\St r\in\Into,\ a\in X\text{ and }
\alpha\in\Into^{l+2}\}
\) %
and hence
\(
\Sflip\Compl(\Scfun\Sflipl l)
=\{((r\alpha,a),(\Rcycle{r\alpha},a)
\St r\in\Into,\alpha\in\Into^{l+1} \text{ and }a\in X\}
=\Sflipl{l+1}
\).

The additive strength (see Section~\ref{sec:additive-strength}) of the
monad \(\Sdfun\) on \(\Kl\REL\) is %
\( \Sdfstr^i_{X_0,\dots,X_n} \in\Kl\REL(X_0\IWith\cdots\IWith\Sdfun
X_i\IWith\cdots\IWith X_n, \Sdfun(X_0\IWith\cdots\IWith X_n) \) %
given by
\begin{align*}
\Sdfstr^i_{X_0,\dots,X_n}
&=\Eset{(\Mset{(j,a)},(0,j,a))
  \St j\in\Eset{0,\dots,n}\setminus\Eset i
  \text{ and }a\in X_j}\\
&\Textsep\cup\Eset{(\Mset{(i,0,a)},(0,i,a))\St a\in X_i}\\
&\Textsep\cup\Eset{(\Mset{(i,1,a)},(1,i,a))\St a\in X_i}\,.
\end{align*}

Using the notations of Section~\ref{sec:basic-categ-constr} we set
$\Snat=\Nat$ and then %
\(\Szero=\Eset{(\Sonelem,0)}\in\REL(\Sone,\Snat)\), %
\(\Ssuc=\Eset{(\nu,\nu+1)\St n\in\Nat}\in\REL(\Snat,\Snat)\) and%
\footnote{To avoid confusions with integers used as indices or
  multiplicities, we use Greek letters \(\kappa,\nu\) to denote
  numerals.}
\(\Spred=\Eset{(0,0)}\cup\Eset{(\nu+1,\nu)\St \nu\in\Nat}\).
The %
\(\oc\)-coalgebra structure of $\Snat$ is given by %
\( \Coalg\Snat =\Eset{(\nu,k\Mset\nu)\St
  k,\nu\in\Nat}\in\REL(\Snat,\Excl\Snat) \).
Given a set $X$, the morphism %
\(\Slet\in\REL(\Tens\Snat{(\Limpl{\Excl\Snat}{X})},X)\) is given by %
\(\Slet=\Eset{(\nu,(k\Mset \nu,a),a)\St \nu,k\in\Nat\text{ and }a\in X}\) %
and the morphism %
\(\Sif\in\cL(\Snat\ITens(\With XX),X)\) is given by %
\(
\Sif=\Eset{(0,(0,a),a)\St a\in X}
\cup\Eset{(\nu+1,(1,a),a)\St a\in X}
\).

\subsubsection{The interpretation in \(\REL\) as a deduction system}
\label{sec:REL-deduction}
The interpretation \(\Tsemrel A\) of a type \(A\) in \(\REL\) is given by %
\(\Tsemrel{\Tdnat d}=\Into^d\times\Nat\) and %
\(\Tsemrel{\Timpl AB}=\Mfin{\Tsemrel A}\times\Tsemrel B\).
We consider the elements of $\Into^d$ as %
sequences of length $d$ of elements of
$\Into=\With\Sone\Sone=\Eset{0,1}$, that is, as sequences of bits of
length $d$.
Given \(\delta\in\Into^d\) and \(a\in\Tsemrel A\) one defines %
\(\Tseqact\delta a\in\Tsemrel{\Tdiffm dA}\) by induction on \(A\): %
if %
\(A=\Tdnat e\) then \(a=(\epsilon,\nu)\) where %
\(\nu\in\Nat\) and \(\epsilon\in\Into^e\) and we set %
\(\Tseqact\delta a=(\delta\epsilon,\nu)\in\Tdnat{d+e}\), %
and if \(A=\Timplp BC\) then \(a=(p,c)\) where %
\(p\in\Mfin{\Tsemrel B}\) and \(c\in\Tsemrel C\) and we set %
\( \Tseqact\delta a =(p,\Tseqact\delta c) \in\Tsemrel{\Timpl B{\Tdiffm
    dC}}=\Tsemrel{\Tdiffm d{\Timplp BC}} \).
Any type \(A\) can be written uniquely \(A=\Tdiffm dF\) where %
\(F\) is sharp, and then any element \(a\in\Tsemrel A\) can be
written %
uniquely \(a=\Tseqact\delta f\) where %
\(\delta\in\Into^d\) and \(f\in\Tsemrel F\).

An intersection typing context is a sequence %
\(\Phi=(\Ttermi{x_1}{m_1}{A_1},\dots,\Ttermi{x_n}{m_n}{A_n})\) where the %
\(x_i\)'s are pairwise distinct variables, %
and \(m_i\in\Mfin{\Tsemrel{A_i}}\) for each \(i\). %
Then we use \(\Contca\Phi\) for the underlying typing context defined
as \(\Contca\Phi=(x_1:A_1,\dots,x_n:A_n)\). If %
\(\Contca\Phi=\Contca{\Phi'}\), so that
\(\Phi=(\Ttermi{x_1}{m_1}{A_1},\dots,\Ttermi{x_n}{m_n}{A_n})\) and %
\(\Phi'=(\Ttermi{x_1}{m'_1}{A_1},\dots,\Ttermi{x_n}{m'_n}{A_n})\), %
then we set %
\(
\Phi+\Phi'=(\Ttermi{x_1}{m_1+m'_1}{A_1},\dots,\Ttermi{x_n}{m_n+m'_n}{A_n})
\) %
and when we use this notation we always implicitly assume that all
the intersection typing contexts involved have the same underlying
typing context. Given a typing context %
\(\Gamma=(x_1:A_1,\dots,x_n:A_n)\) we set %
\(
\Contz\Gamma
=(\Ttermi{x_1}{\Msetempty}{A_1},\dots,\Ttermi{x_n}{\Msetempty}{A_n})
\). %

The intersection typing rules are given in
Figure~\ref{fig:int-typing-rules}.
\begin{figure}
  \begin{center}
    \begin{prooftree}
      \infer0[\Itrvar]{\Tseqi{\Contz\Gamma,\Ttermi x{\Mset a}A}{x}{a}{A}}
    \end{prooftree}
    \labeltext{$\Itrvar$}{rl:itrvar}
    \Treesep
    \begin{prooftree}
      \hypo{\Tseqi{\Phi,\Ttermi xmA}{M}b{B}}
      \infer1[\Itrabs]{\Tseqi{\Phi}{\Abst xAM}{(m,b)}{\Timpl AB}}
    \end{prooftree}
    \labeltext{$\Itrabs$}{rl:itrabs}
  \end{center}
  \begin{center}
    \begin{prooftree}
      \hypo{\Tseqi{\Phi^0}M{(\Mset{\List a1n},b)}{\Timpl AB}}
      \hypo{(\Tseqi{\Phi^j}N{a_j}A)_{j=1}^n}
      \infer2[\Itrapp]{\Tseqi{\sum_{j=0}^n\Phi^j}{\App MN}bB}
    \end{prooftree}
    \labeltext{$\Itrapp$}{rl:itrapp}
  \end{center}
  \begin{center}
    \begin{prooftree}
      \hypo{\Tseqi{\Phi^0}M{(\Mset{\List a1n},a)}{\Timpl AA}}
      \hypo{(\Tseqi{\Phi^j}{\Lfix M}{a_j}A)_{j=1}^n}
      \infer2[\Itrfix]{\Tseqi{\sum_{j=0}^n\Phi^j}{\Lfix M}aA}
    \end{prooftree}
    \labeltext{$\Itrfix$}{rl:itrfix}
  \end{center}
  \begin{center}
    \begin{prooftree}
      \hypo{\nu\in\Nat}
      \infer1[\Itrnum]{\Tseqi{\Contz\Gamma}{\Num\nu}\nu\Tnat}
    \end{prooftree}
    \labeltext{$\Itrnum$}{rl:itrnum}
    \Treesep
    \begin{prooftree}
      \hypo{\Tseqi\Phi M{\Tseqact\delta\nu}{\Tdiffm d\Tnat}}
      \hypo{\Len\delta=d}
      \infer2[\Itrsuc]
      {\Tseqi\Phi{\Lsucc dM}{\Tseqact\delta{(\nu+1)}}{\Tdiffm d\Tnat}}
    \end{prooftree}
    \labeltext{$\Itrsuc$}{rl:itrsuc}
  \end{center}
  \begin{center}
    \begin{prooftree}
      \hypo{\Tseqi\Phi M{\Tseqact\delta 0}{\Tdiffm d\Tnat}}
      \hypo{\Len\delta=d}
      \infer2[\Itrpredz]
      {\Tseqi\Phi{\Lpred dM}{\Tseqact\delta 0}{\Tdiffm d\Tnat}}
    \end{prooftree}
    \labeltext{$\Itrpredz$}{rl:itrpredz}
  \end{center}
  \begin{center}
    \begin{prooftree}
      \hypo{\Tseqi\Phi M{\Tseqact\delta{(\nu+1)}}{\Tdiffm d\Tnat}}
      \hypo{\Len\delta=d} \infer2[\Itrpredp] {\Tseqi\Phi{\Lpred
          dM}{\Tseqact\delta\nu}{\Tdiffm d\Tnat}}
    \end{prooftree}
    \labeltext{$\Itrpredp$}{rl:itrpredp}
  \end{center}
  \begin{center}
    \begin{prooftree}
      \hypo{\Tseqi{\Phi^0}M{\Tseqact\delta 0}{\Tdiffm d\Tnat}}
      \hypo{\Tseqi{\Phi^1}PaA} \hypo{\Tseq{\Contca{\Phi^0}} QA}
      \hypo{\Len\delta=d} \infer4[\Itrifz]
      {\Tseqi{\Phi^0+\Phi^1}{\Lift AdMPQ}{\Tseqact\delta a}{\Tdiffm d
          A}}
    \end{prooftree}
    \labeltext{$\Itrifz$}{rl:itrifz}
  \end{center}
  \begin{center}
    \begin{prooftree}
      \hypo{\Tseqi{\Phi^0}M{\Tseqact\delta{(\nu+1)}}{\Tdiffm d\Tnat}}
      \hypo{\Tseq{\Contca{\Phi^0}} PA} \hypo{\Tseqi{\Phi^1}QaA}
      \hypo{\Len\delta=d} \infer4[\Itrifp]
      {\Tseqi{\Phi^0+\Phi^1}{\Lift AdMPQ}{\Tseqact\delta a}{\Tdiffm d
          A}}
    \end{prooftree}
    \labeltext{$\Itrifp$}{rl:itrifp}
  \end{center}
  \begin{center}
    \begin{prooftree}
      \hypo{\Tseqi\Phi{M}{\Tseqact{\delta r}a}{\Tdiffm{d+1}A}}
      \hypo{\Len\delta=d}
      \infer2[\Itrproj]
      {\Tseqi\Phi{\Lprojd rdM}{\Tseqact\delta a}{\Tdiffm dA}}
    \end{prooftree}
    \labeltext{$\Itrproj$}{rl:itrproj}
  \end{center}
  \begin{center}
    \begin{prooftree}
      \hypo{\Tseqi\Phi{M_i}a{A}}
      \hypo{\Tseq{\Contca\Phi}{M_0+M_1}{A}}
      \hypo{i\in\Into}
      \infer3[\Itradd]{\Tseqi\Phi{M_0+M_1}aA}
    \end{prooftree}
    \labeltext{$\Itradd$}{rl:itradd}
  \end{center}
  \begin{center}
    \begin{prooftree}
      \hypo{\Tseqi\Phi M{\Tseqact\delta a}{\Tdiffm dA}}
      \hypo{\Len\delta=d}
      \hypo{r\in\Into}
      \infer3[\Itrinj]
      {\Tseqi\Phi{\Linjd rdM}{\Tseqact{\delta r}a}{\Tdiffm{d+1}A}}
    \end{prooftree}
    \labeltext{$\Itrinj$}{rl:itrinj}
  \end{center}
  \begin{center}
    \begin{prooftree}
      \hypo{\Tseqi\Phi M{\Tseqact{\delta r_0r_1}a}{\Tdiffm{d+2}A}}
      \hypo{r,r_0,r_1\in\Into\text{ and }r=r_0+r_1}
      \hypo{\Len\delta=d}
      \infer3[\Itrsum]
      {\Tseqi\Phi{\Lsumd dM}{\Tseqact{\delta r}a}{\Tdiffm{d+1}A}}
    \end{prooftree}
    \labeltext{$\Itrsum$}{rl:itrsum}
  \end{center}
  \begin{center}
    \begin{prooftree}
      \hypo{\Tseqi\Phi M{(m,b)}{\Timpl AB}}
      \hypo{r\in\Into\text{ and }(m',(r,m))\in\Sdiff_{\Tsem A}}
      \infer2[\Itrdiff]
      {\Tseqi\Phi{\Ldiff M}{(m',\Tseqact rb)}
        {\Timpl{\Tdiff A}{\Tdiff B}}}
    \end{prooftree}
    \labeltext{$\Itrdiff$}{rl:itrdiff}
  \end{center}
  \begin{center}
    \begin{prooftree}
      \hypo{\Tseqi\Phi M{\Tseqact{\delta\alpha}{a}}
        {\Tdiffm{d+l+2}A}}
      \hypo{\Len\delta=d,\ \Len\alpha=l+2}
      \infer2[\Itrcirc]
      {\Tseqi\Phi{\Lflipdl dlM}{\Tseqact{\delta\Rcycle\alpha}{a}}
        {\Tdiffm{d+l+2}A}}
    \end{prooftree}
    \labeltext{$\Itrcirc$}{rl:itrcirc}
  \end{center}
  \begin{center}
    \begin{prooftree}
      \hypo{\Tseqi{\Phi^0}M{\Tseqact\delta\nu}{\Tdiffm d\Tnat}}
      \hypo{\Tseqi{\Phi^1,\Ttermi x{k\Mset\nu}\Tnat}N{b}B}
      \hypo{k,\nu\in\Nat\text{ and }\Len\delta=d}
      \infer3[\Itrlet]
      {\Tseqi{\Phi^0+\Phi^1}{\Llett BdxMN}{\Tseqact\delta b}{\Tdiffm dB}}
    \end{prooftree}
    \labeltext{$\Itrlet$}{rl:itrlet}
  \end{center}
  \caption{Intersection typing rules for terms}
  \label{fig:int-typing-rules}
\end{figure}
\begin{figure}
  \begin{center}
    \begin{prooftree}
      \infer0{\Stseqi\Stempty\nu\Tnat\nu}
    \end{prooftree}
    \Treesep
    \begin{prooftree}
      \hypo{\Stseqi s{\kappa+1}\Tnat\nu} %
      \infer1{\Stseqi{\Stsucc s}\kappa\Tnat\nu} %
    \end{prooftree}
  \end{center}
  \begin{center}
    \begin{prooftree}
      \hypo{\Stseqi s0\Tnat\nu} %
      \infer1{\Stseqi{\Stpred s}0\Tnat\nu}
    \end{prooftree}
    \Treesep
    \begin{prooftree}
      \hypo{\Stseqi s\kappa\Tnat\nu} %
      \infer1{\Stseqi{\Stpred s}{\kappa+1}\Tnat\nu}
    \end{prooftree}
  \end{center}
  \begin{center}
    \begin{prooftree}
      \hypo{\Stseqi sfF\nu} %
      \hypo{\Tseqi{}{M_0}{\Tseqact\delta f}{\Tdiffm dF}} %
      \hypo{\Tseq{}{M_1}{\Tdiffm dF}} %
      \hypo{\Len\delta=d} %
      \infer4{\Stseqi{\Stif\delta{M_0}{M_1}s}0{\Tnat}\nu} %
    \end{prooftree}
  \end{center}
  \begin{center}
    \begin{prooftree}
      \hypo{\Stseqi sfF\nu} %
      \hypo{\Tseq{}{M_0}{\Tdiffm dF}} %
      \hypo{\Tseqi{}{M_1}{\Tseqact\delta f}{\Tdiffm dF}} %
      \hypo{\Len\delta=d} %
      \infer4{\Stseqi{\Stif\delta{M_0}{M_1}s}{\kappa+1}{\Tnat}\nu} %
    \end{prooftree}
  \end{center}
  \begin{center}
    \begin{prooftree}
      \hypo{\Stseqi sfF\nu} %
      \hypo{\Tseqi{\Ttermi x{k\Mset\kappa}\Tnat}{M}
        {\Tseqact\delta f}{\Tdiffm dF}} %
      \hypo{\Len\delta=d} %
      \infer3{\Stseqi{\Stlet\delta{x}{M}s}\kappa{\Tnat}\nu} %
    \end{prooftree}    
  \end{center}
  \begin{center}
    \begin{prooftree}
      \hypo{(\Tseqi{}M{a_i}A)_{i=1}^k}
      \hypo{\Stseqi sfF\nu}
      \infer2{\Stseqi{\Starg Ms}{(\Mset{\List a1k},f)}{\Timpl AF}\nu}
    \end{prooftree}
  \end{center}
  \begin{center}
    \begin{prooftree}
      \hypo{\Stseqi{s}{(m',f)}{\Timpl{\Tdiff A}F}\nu} %
      \hypo{r\in\Into\text{ and }(m',(r,m))\in\Sdiff_{\Tsem A}} %
      \infer2{\Stseqi{\Stdiff rs}{(m,f)}{\Timpl AF}\nu} %
    \end{prooftree}
  \end{center}
  \begin{center}
    \begin{prooftree}
      \hypo{\Tseqi{}M{\Tseqact\delta f}{\Tdiffm dF}}
      \hypo{\Stseqi sfF\nu}
      \infer2{\Etseqi{\State\delta Ms}\nu}
    \end{prooftree}
  \end{center}
  \caption{Intersection typing rules for stacks and states}
  \label{fig:int-typing-rules-stacks}
\end{figure}
We also define an intersection typing system for stacks in
Figure~\ref{fig:int-typing-rules-stacks}. An intersection stack typing
judgment is an expression %
\(\Stseqi sfF\nu\) where \(s\) is a stack, \(F\) is a sharp type,
\(f\in\Tsemrel F\) and \(\nu\in\Nat\).
The intuition is that $s$ produces $\nu$ if it receives $f$ as a
linear argument.

The main feature of this system is the following result which is
obtained by a simple induction structured exactly as the proof of
Theorem~\ref{th:sem-defined-sum}, and expresses that the rules of
Figure~\ref{fig:int-typing-rules-stacks} are a simple reformulation of
the denotational interpretation of Figure~\ref{fig:term-interp} in the
case of the model \(\REL\).
\begin{theorem}\label{th:sem-intersection-equiv}
  For any term \(M\) such that %
  \(\Tseq{x_1:A_1,\dots,x_n:A_n}MB\) and any %
  \((m_i\in\Mfin{\Tsemrel{A_i}})_{i=1}^n\) and \(b\in\Tsemrel B\), the two
  following properties are equivalent.
  \begin{itemize}
  \item \((\List m1n,b)\in\Psemrel M\Gamma\)
  \item the judgment
    \(
    \Tseqi{\Ttermi{x_1}{m_1}{A_1},\dots,\Ttermi{x_n}{m_n}{A_n}}{M}{b}{B}
    \) %
    is provable in the system of Figure~\ref{fig:int-typing-rules}.
  \end{itemize}
  For any stack \(s\) such that \(\Stseq sF\) and any %
  \(f\in\Tsemrel F\) one has \((f,n)\in\Psemrel s{}\) iff %
  \(\Stseqi sfFn\). For any well-typed state \(c\) and \(n\in\Nat\),
  one has \(\Etseqi cn\) iff \(n\in\Psemrel c{}\).
\end{theorem}
Of course \(\Psemrel M\Gamma\) is the interpretation of \(M\) in
\(\REL\), and similarly for stacks and states.

Given a typing context %
\(\Gamma=(x_1:A_1,\dots,x_n:A_n)\) and \(j\in\Eset{1,\dots,n}\) %
we define \(\Sdiffctx j\Gamma\) as the set of all triples %
\((\Phi',r,\Phi)\) where \(r\in\Into\), %
\(\Contca\Phi=\Contca{\Phi'}=\Gamma\) and, if we set %
\(\Phi=(x_i:m_i:A_i)_{i=1}^n\) and \(\Phi'=(x_i:m'_i:A'_i)_{i=1}^n\)
then we have
\begin{align*}
  \begin{cases}
  A'_i=\Tdiff A_i\text{ and }
  (m'_j,(r,m_j))\in\Sdiff_{\Tsem{A_j}} &\\
  A'_i=A_i\text{ and }m'_i=m_i&\text{ if }i\not=j\,.
  \end{cases}
\end{align*}

By Lemma~\ref{lemma:sem-diff-subst} we have %
\[
  \Psemrel{\Ldletv xM}{\Gamma,x:\Tdiff A} =\Sdfunpart1{\Psemrel
    M{\Gamma,x:A}} \in\Kl\cL(\With{\Tsemrel\Gamma}{\Sdfun{\Tsemrel
      A}},\Sdfun{\Tsemrel B})\,.
\]
If we rephrase this property in the model \(\REL\), we get the
following.

\begin{theorem}
  Assume that \(\Tseqi\Phi MbB\) and
  \((\Phi',r,\Phi)\in\Sdiffctx i\Gamma\). Then
  \(\Tseqi{\Phi'}{\Ldletv{x_i}M}{\Tseqact rb}{\Tdiff B}\).
\end{theorem}
\begin{proof}
  By Theorem~\ref{th:sem-intersection-equiv}.
\end{proof}

We set \(\Ldletvm 0xM=M\) and
\(\Ldletvm{d+1}xM=\Ldletv x{\Ldletvm dxM}\).

\subsubsection{Normalization} %
\label{sec:int-seq-normalization}
Given \(\nu\in\Nat\) we define %
\(\Itintc\nu\) as the set of all well-typed states \(c\) such that %
\(
\Mset c
\Rel{\Trcl{\Rsred{\Msrs\States}}}
C+\Mset{\State{\Pempty}{\Num\nu}{\Stempty}}
\) %
for some \(C\in\Mfin{\Rsca\States}\).

With any type \(A\) and \(a\in\Tsemrel A\), we associate a set %
\(\Itintt aA\) of terms \(M\) such that \(\Tseq{}MA\).
If \(F\) is a sharp type, \(f\in\Tsemrel F\) and \(\nu\in\Nat\) we also
define a set %
\(\Itints fF\nu\) of stacks \(s\) such that \(\Stseq sF\).
The definition is by mutual induction on types, and more precisely on
the number of ``\(\Timpl{}{}\)'' in types. Remember that \(\Rev\delta\) is
the reversed word of \(\delta\).

\begin{itemize}
\item \(
  \Itintt{\Tseqact\delta f}{\Tdiffm dF}
  =\Eset{M\St \Tseq{}M{\Tdiffm dF}\text{ and }
    \forall\nu\in\Nat\,\forall s\in\Itints fF\nu\
    \State{\Rev\delta}Ms\in\Itintc\nu}
  \)
\item \(
  \Itints{\kappa}\Tnat\nu
  =\Eset{s\St \Stseq s\Tnat\text{ and }
    \State{\Pempty}{\Num\kappa}{s}\in\Itintc\nu}
  \)
\item let \(m\in\Mfin{\Tsemrel A}\) and \(f\in\Tsemrel F\), then %
  \begin{align*}
    \Itints{(m,f)}{\Timpl AF}\nu
    &=\{\Stdiff{r_1}{\cdots\Stdiff{r_d}{\Starg Ms}}\St
      d\in\Nat,\,\List r1d\in\Into,\\ 
    &\quad\quad \exists(m_i\in\Mfin{\Tsemrel{\Tdiffm iA}})_{i=0}^d\ m_0=m\\
    &\quad\quad \text{and }
      ((m_i,(r_i,m_{i-1}))\in\Sdiff_{\Tsemrel{\Tdiffm{i-1}{A}}})_{i=1}^d,\\
    &\quad\quad M\in\Itintt{m_d}{\Tdiffm dA}
      \text{ and } s\in\Itints fF\nu\}.
  \end{align*}
\end{itemize}
In this definition, we use the following notation: if
\(m=\Mset{\List a1k}\in\Mfin{\Tsemrel A}\) with %
\(\List a1n\in\Tsemrel A\), then %
\( \Itintt mA=\bigcap_{j=1}^k\Itintt{a_j}A \).

Given an intersection typing context \(\Phi=(x_i:m_i:A_i)_{i=1}^k\) we
define %
\(\Itintctx\Phi\) as the set of all substitutions %
\(\Vect P=(P_1/x_1,\dots,P_k/x_k)\) where %
\((P_i\in\Itintt{m_i}{A_i})_{i=1}^k\).  Given also \(b\in\Tsemrel B\) %
we define a set \(\Itopenn\Phi bBn\) of terms \(M\) such that %
\(\Tseq{\Contca\Phi}{M}{B}\), the definition is by induction on \(n\):
\begin{align*}
  \Itopenn\Phi bB0
  &=\{M\St\Tseq{\Contca\Phi}{M}{B}
  \text{ and }
    \forall\Vect P\in\Itintctx\Phi\ \Substbis M{\Vect P}\in\Itintt bB\}\\
  \Itopenn\Phi bB{n+1}
  &=\{M\St\Tseq{\Contca\Phi}{M}{B}\text{ and }
    \forall r\in\Into\,
    \forall i\in\Eset{1,\dots,k},
    \\
  &\Textsep\forall\Phi'\ (\Phi',(r,\Phi))\in\Sdiffctx i{\Contca\Phi}
  \Implies\Ldletv{x_i}{M}\in\Itopenn{\Phi'}{\Tseqact rb}{\Tdiff B}n\}
\end{align*}
and last we set
\begin{align*}
  \Itopen\Phi bB=\bigcap_{n\in\Nat}\Itopenn\Phi bBn\,.
\end{align*}

\begin{remark}
  Given \(M\), \(\Gamma\) and \(B\) such that \(\Tseq\Gamma MB\),
  saying that \(M\in\Itopenn\Phi bB{n}\) for all \(\Phi\) such that
  \(\Contca\Phi=\Gamma\) and all \(b\in\Tsemrel B\) intuitively means
  that \(M\) is \(n\) times differentiable with respect to its free
  variables, and \(M\in\Itopen\Phi bB\) for all \(\Phi\) and \(b\)
  means that \(M\) is infinitely differentiable. This definition bears
  some similarity with the definition of \(C^\infty\) maps in
  Analysis.
\end{remark}

\begin{lemma}%
  \label{lemma:sdiff-split}
  Let \(X\) be a set, let \(r\in\Into\) and let %
  \((m',(r,m))\in\Sdiff_X\). %
  Let \(\List m1k\in\Mfin X\) be such that \(m=m_1+\cdots+m_k\). %
  There are \(\List r1k\in\Into\) and \(\List{m'}{1}{k}\in\Mfin X\)
  such that %
  \(r=r_1+\cdots+r_k\), \(m'=m'_1+\cdots+m'_k\) and %
  \(((m'_i,(r_i,m_i))\in\Sdiff_X)_{i=1}^k\).
\end{lemma}
\begin{proof}
  Categorically this is a simple consequence of the fact that the
  morphism \(\Sdiffca\in\REL(\Into,\Excl\Into)\) of
  Section~\ref{sec:rel-model} is a \(\oc\)-coalgebra structure, and
  can also be checked directly as we do now. Setting
  \(m=\Mset{a_j\St j\in J}\), the assumption %
  \((m',(r,m))\in\Sdiff_X\) means that %
  \(m'=\Mset{(r'_j,a_j)\St j\in J}\) with \((r'_j\in\Into)_{j=1}^k\)
  and %
  \(r=\sum_{j\in J}r'_j\). Since \(m=m_1+\cdots+m_k\) we can find
  pairwise disjoint \((J_i)_{i=1}^k\) with \(J=\bigcup_{i=1}^kJ_i\)
  and \(m_i=\Mset{a_j\St j\in J_i}\). Then we take
  \(m'_i=\Mset{(r'_j,a_j)\St j\in J_i}\) and
  \(r_i=\sum_{j\in J_i}r'_i\) for \(i=1,\dots,k\). We have
  \(r_i\in\Into\) because \(\sum_{j\in J}r'_j=r\in\Into\).
\end{proof}

Now we prove a series of lemmas %
(from~\ref{lemma:itopenn-flip} to~\ref{lemma:itopenn-abst})
which express that the sets
\(\Itopenn\Phi bBn\) have some stability properties with respect to
the syntactic constructs of the language. They will make the proof of
Theorem~\ref{th:cdpcf-adequacy} essentially trivial.

\begin{remark}
  These lemmas are proven by induction on \(n\) (the superscript in
  \(\Itopenn\Phi bBn\)) and it is essential to notice that the
  statement we prove by induction on \(n\) is \emph{universally
    quantified} on \(\Phi\), \(B\) and \(b\) because, when proving the
  implication for \(M\), we have to apply the inductive hypothesis to
  \(\Ldletv xM\) for all the variables \(x\) in the context.
\end{remark}


\begin{lemma}%
  \label{lemma:itopenn-flip}
  If
  \(
  M\in\Itopenn{\Phi}{\Tseqact{\delta\alpha}{b}}{\Tdiffm{d+h+2}{B}}n
  \) with %
  \(d=\Len\delta\) and \(h+2=\Len\alpha\) then we have %
  \(
  \Lflipdl dhM
  \in\Itopenn{\Phi}{\Tseqact{\delta\Rcycle\alpha}{b}}{\Tdiffm{d+h}{B}}n
  \).
\end{lemma}
\begin{proof}
  By induction on \(n\). For \(n=0\): let \(\Vect P\in\Itintctx\Phi\)
  and let us set \(M'=\Substbis M{\Vect P}\). We can write %
  \(B=\Tdiffm eF\) where \(F\) is sharp and %
  \(b=\Tseqact\epsilon f\) where \(\Len\epsilon=e\) and
  \(f\in\Tsemrel F\). %
  Let \(\nu\in\Nat\). %
  Let \(s\in\Itints fF\nu\), we have %
  \begin{align*}
  \State{\Rev{\delta\Rcycle\alpha\epsilon}}{\Lflipdl dh{M'}}{s}
  &=\State{\Rev{\epsilon}\Lcycle{(\Rev\alpha)}\Rev\delta}
  {\Lflipdl dh{M'}}{s}\\
  &\Rel\Stred
  \State{\Rev\epsilon\Rcycle{\Lcycle{(\Rev\alpha)}}\Rev\delta}
  {M'}{s}\\
  &=\State{\Rev{\delta\alpha\epsilon}}{M'}{s}\in\Itintc\nu
  \end{align*}
  since %
  \(M\in\Itopenn{\Phi}{\Tseqact{\delta\alpha\epsilon}f}
  {\Tdiffm{d+h+e+2}F}{0}\).

  For the inductive step, we assume that the implication holds for
  \(n\) and we prove it for \(n+1\) so assume that %
  \(
  M\in\Itopenn{\Phi}{\Tseqact{\delta\alpha}{b}}{\Tdiffm{d+h+2}{B}}{n+1}
  \) and let %
  \(r\in\Into\), \(l\in\Eset{1,\dots,\Len\Phi}\) and %
  \(\Phi'\) be such that %
  \((\Phi',r,\Phi)\in\Sdiffctx{l}{\Contca\Phi}\), we have
  \( \Ldletv{x_l}M
  \in\Itopenn{\Phi'}{\Tseqact{r\delta\alpha}b}{\Tdiffm{d+h+3}b}{n} \)
  and hence, by inductive hypothesis, %
  \( \Ldletv{x_l}{\Lflipdl dhM} =\Lflipdl{d+1}h{\Ldletv{x_l}M}
  \in\Itopenn{\Phi'}{\Tseqact{r\delta\Rcycle\alpha}b}{\Tdiffm{d+h+3}b}{n}
  \). %
  Since we have this property for all \(l\) and \(r,\Phi'\) such that
  \((\Phi',r,\Phi)\in\Sdiffctx{l}{\Contca\Phi}\), we have proven
  that %
  \( \Lflipdl dhM
  \in\Itopenn{\Phi}{\Tseqact{\delta\Rcycle\alpha}b}{\Tdiffm{d+h+2}B}{n+1}
  \) as required.
\end{proof}

\begin{lemma}%
  \label{lemma:itopenn-inj}
  If \(M\in\Itopenn\Phi{\Tseqact\delta b}{\Tdiffm dB}{n}\) and
  \(r\in\Into\) then %
  \(\Linjd rdM\in\Itopenn{\Phi}{\Tseqact{\delta r}{b}}{\Tdiffm{d+1}B}{n}\) %
  where \(d=\Len\delta\).
\end{lemma}
\begin{proof}
  By induction on \(n\). For \(n=0\): let \(\Vect P\in\Itintctx\Phi\)
  and let us set \(M'=\Substbis M{\Vect P}\). We can write %
  \(B=\Tdiffm eF\) where \(F\) is sharp and %
  \(b=\Tseqact\epsilon f\) where \(\Len\epsilon=e\) and
  \(f\in\Tsemrel F\). %
  Let \(\nu\in\Nat\). %
  Let \(s\in\Itints fF\nu\), we have %
  \(
  \State{\Rev{\delta r\epsilon}}{\Linjd rd{M'}}{s}
  =\State{\Rev\epsilon r\Rev\delta}{\Linjd rd{M'}}{s}
  \Rel\Stred
  \State{\Rev\epsilon\Rev\delta}{M'}{s}
  =\State{\Rev{\delta\epsilon}}{M'}{s}\in\Itintc\nu
  \) %
  by our assumption about \(M\), hence %
  \(\Linjd rd{M'}\in\Itintt{\Tseqact{\delta r}b}{\Tdiffm{d+1}B}\)
  as required.

  For the inductive step we assume that %
  \(M\in\Itopenn\Phi{\Tseqact\delta b}{\Tdiffm{d}B}{n+1}\). %
  Let \(r'\in\Into\), \(l\in\Eset{1,\dots,\Len\Phi}\), and \(\Phi'\) be
  such that \((\Phi',r',\Phi)\in\Sdiffctx l{\Contca\Phi}\). %
  We have %
  \(
  \Ldletv{x_l}M\in\Itopenn{\Phi'}{\Tseqact{r'\delta}b}
  {\Tdiffm{d+1}B}{n}
  \) and hence, by inductive hypothesis, %
  \(
  \Ldletv{x_l}{\Linjd rdM}
  =\Linjd r{d+1}{\Ldletv{x_l}{M}}
  \in\Itopenn{\Phi'}{\Tseqact{r'\delta r}{b}}{\Tdiffm{d+2}B}{n}
  \) since \(\Len{r'\delta}=d+1\).
\end{proof}

\begin{lemma}%
  \label{lemma:itopenn-lsum}
  If %
  \(
  M\in\Itopenn\Phi{\Tseqact{\delta r_0r_1}{b}}{\Tdiffm{d+2}B}{n}
  \) %
  with \(d=\Len\delta\) then %
  \(
  \Lsumd dM\in\Itopenn\Phi{\Tseqact{\delta r}{b}}{\Tdiffm{d+2}B}{n}
  \) %
  if \(r=r_0+r_1\in\Into\).
\end{lemma}
\begin{proof}
  By induction on \(n\). For \(n=0\): let \(\Vect P\in\Itintctx\Phi\)
  and let us set \(M'=\Substbis M{\Vect P}\). We can write %
  \(B=\Tdiffm eF\) where \(F\) is sharp and %
  \(b=\Tseqact\epsilon f\) where \(\Len\epsilon=e\) and
  \(f\in\Tsemrel F\). %
  Let \(\nu\in\Nat\). %
  Let \(s\in\Itints fF\nu\). Assume first \(r=0\) so that \(r_0=r_1=0\),
  we have %
  \(
  \State{\Rev{\delta 0\epsilon}}{\Lsumd d{M'}}{s}
  \Rel\Stred
  \State{\Rev\delta00\epsilon}{M'}{s}\in\Itintc\nu
  \) %
  by our assumption about \(M\), hence %
  \(\Lsumd d{M'}\in\Itintt{\Tseqact{\delta 0}b}{\Tdiffm{d+1}B}\)
  as required. Assume now \(r=1\), we have
  \(
  \State{\Rev{\delta 1\epsilon}}{\Lsumd d{M'}}{s}
  \Rel\Stred
  \State{\Rev{\delta01\epsilon}}{M'}{s}
  +\State{\Rev{\delta10\epsilon}}{M'}{s}
  \) %
  and since one of these summands belongs to \(\Itintc\nu\) %
  by our assumption about \(M\), we get
  \(\State{\Rev{\delta 1\epsilon}}{\Lsumd d{M'}}{s}\in\Itintc\nu\).

  For the inductive step we assume that %
  \(M\in\Itopenn\Phi{\Tseqact{\delta r_0r_1}b}{\Tdiffm{d+2}B}{n+1}\). %
  Let \(r'\in\Into\), \(l\in\Eset{1,\dots,\Len\Phi}\), and \(\Phi'\) be
  such that \((\Phi',r',\Phi)\in\Sdiffctx l{\Contca\Phi}\). %
  We have %
  \(
  \Ldletv{x_l}M\in\Itopenn{\Phi'}{\Tseqact{r'\delta r_0r_1}b}
  {\Tdiffm{d+3}B}{n}
  \) and hence, by inductive hypothesis, %
  \(
  \Ldletv{x_l}{\Lsumd dM}
  =\Lsumd{d+1}{\Ldletv{x_l}{M}}
  \in\Itopenn{\Phi'}{\Tseqact{r'\delta r_0r_1}{b}}{\Tdiffm{d+3}B}{n}
  \) since \(\Len{r'\delta}=d+1\).  
\end{proof}

\begin{lemma}%
  \label{lemma:itopenn-proj}
  For all \(n\in\Nat\), %
  if \(r\in\Into\) and %
  \(M\in\Itopenn{\Phi}{\Tseqact{\delta r}b}{\Tdiffm{d+1}B}n\)
  then \(\Lprojd rdM\in\Itopenn{\Phi}{\Tseqact\delta b}{\Tdiffm dB}n\)
\end{lemma}
\begin{proof}
  By induction on \(n\). For \(n=0\): let \(\Vect P\in\Itintctx\Phi\)
  and let us set \(M'=\Substbis M{\Vect P}\). We can write %
  \(B=\Tdiffm eF\) where \(F\) is sharp and %
  \(b=\Tseqact\epsilon f\) where \(\Len\epsilon=e\) and
  \(f\in\Tsemrel F\). %
  Let \(\nu\in\Nat\). %
  Let \(s\in\Itints fF\nu\), we have %
  \(
  \State{\Rev{\delta\epsilon}}{\Lprojd rd{M'}}{s}
  \Rel\Stred
  \State{\Rev{\delta r\epsilon}}{M'}{s}\in\Itintc\nu
  \) %
  by our assumption about \(M\), hence %
  \(\Lprojd rd{M'}\in\Itintt{\Tseqact{\delta}b}{\Tdiffm{d+1}B}\)
  as required.

  For the inductive step we assume that %
  \(M\in\Itopenn\Phi{\Tseqact{\delta r}b}{\Tdiffm{d+1}B}{n+1}\). %
  Let \(r'\in\Into\), \(l\in\Eset{1,\dots,\Len\Phi}\), and \(\Phi'\) be
  such that \((\Phi',r',\Phi)\in\Sdiffctx l{\Contca\Phi}\). %
  We have %
  \(
  \Ldletv{x_l}M\in\Itopenn{\Phi'}{\Tseqact{r'\delta r}b}
  {\Tdiffm{d+2}B}{n}
  \) and hence, by inductive hypothesis, %
  \(
  \Ldletv{x_l}{\Lprojd rdM}
  =\Lprojd r{d+1}{\Ldletv{x_l}{M}}
  \in\Itopenn{\Phi'}{\Tseqact{r'\delta}{b}}{\Tdiffm{d+1}B}{n}
  \) since \(\Len{r'\delta}=d+1\).
\end{proof}

\begin{lemma}%
  \label{lemma:itopenn-ldiff}
  If \(M\in\Itopenn\Phi{(m,b)}{\Timpl AB}{n}\) and %
  \((m',(r,m))\in\Sdiff_{\Tsemrel A}\) then %
  \(\Ldiff M\in\Itopenn\Phi{(m',\Tseqact rb)}
  {\Timpl{\Tdiff A}{\Tdiff B}}{n}\).
\end{lemma}
\begin{proof}
  By induction on \(n\). For \(n=0\): let \(\Vect P\in\Itintctx\Phi\)
  and let us set \(M'=\Substbis M{\Vect P}\). We can write %
  \(B=\Tdiffm eF\) where \(F\) is sharp and %
  \(b=\Tseqact\epsilon f\) where \(\Len\epsilon=e\) and
  \(f\in\Tsemrel F\). %
  Let \(\nu\in\Nat\). %
  Let \(s\in\Itints{(m',f)}{\Timpl{\Tdiff A}{F}}\nu\), we have %
  \(\Stdiff rs\in\Itints{(m,f)}{\Timpl AF}{\nu}\) since %
  \((m',(r,m))\in\Sdiff_{\Tsemrel A}\), and hence %
  \(\State{\Rev\epsilon}{M'}{\Stdiff rs}\in\Itintc\nu\) and therefore %
  \(\State{\Rev{r\epsilon}}{\Ldiff{M'}}{s}\in\Itintc\nu\) since %
  \(
  \State{\Rev{r\epsilon}}{\Ldiff{M'}}{s}
  =\State{\Rev{\epsilon}r}{\Ldiff{M'}}{s}
  \Rel\Stred\State{\Rev\epsilon}{M'}{\Stdiff rs}
  \). %
  It follows that %
  \(\Ldiff{M'}\in\Itintt{(m',\Tseqact{r\epsilon}F)}{\Tdiffm{e+1}F}\)
  as required.

  For the inductive step, assume that %
  \(M\in\Itopenn\Phi{(m,b)}{\Timpl AB}{n+1}\) and let %
  \(r'\in\Into\), \(l\in\Eset{1,\dots,\Len\Phi}\) and %
  \(\Phi'\) be such that %
  \((\Phi',r',\Phi)\in\Sdiffctx{l}{\Contca\Phi}\), we have
  \(
  \Ldletv{x_l}M
  \in\Itopenn{\Phi'}{(m,\Tseqact{r'}b)}{\Timpl{A}{\Tdiff B}}{n}
  \) %
  and hence, by inductive hypothesis, %
  \(
  \Ldiff{\Ldletv{x_l}M}
  \in\Itopenn{\Phi'}{(m',\Tseqact{rr'}b)}{\Timpl{A}{\Tdiffm2B}}{n}
  \) since %
  \((m',(r,m))\in\Sdiff_{\Tsemrel A}\), and hence %
  \(
  \Ldletv{x_l}{\Ldiff M}
  =\Lflip{\Ldiff{\Ldletv{x_l}M}}
  \in\Itopenn{\Phi'}{(m',\Tseqact{r'r}b)}{\Timpl{A}{\Tdiffm2B}}{n}
  \) by Lemma~\ref{lemma:itopenn-flip}. Since we have this property for all
  \(l\) and \(r',\Phi'\) such that 
  \((\Phi',r',\Phi)\in\Sdiffctx{l}{\Contca\Phi}\), we have proven that %
  \(
  \Ldiff M\in\Itopenn{\Phi}{(m',\Tseqact rb))}
  {\Timpl{\Tdiff A}{\Tdiff B}}{n+1}
  \) as required.
\end{proof}

\begin{lemma}%
  \label{lemma:itopenn-var}
  Given \(\Phi=(x_1:m_1:A_1,\dots,x_k:m_k:A_k)\) %
  and \(l\in\{1,\dots,k\}\) such that %
  \(m_j=\Emptymset\) if \(j\not=l\) and \(m_l=\Mset a\), one has %
  \(x_l\in\Itopenn{\Phi}{a}{A}{n}\).
\end{lemma}
\begin{proof}
  By induction on \(n\). For \(n=0\), let \(\Vect P\in\Itintctx\Phi\)
  so that \(P_l\in\Itintt a{A_l}\) so that %
  \(x_l\in\Itopenn{\Phi}{a}{A_l}{0}\). We prove the property for
  \(n+1\) assuming that it holds for \(n\). So with the notations of
  the lemma we must prove that \(x_l\in\Itopenn\Phi a{A_l}{n+1}\). Let
  \(j\in\Eset{1,\dots,k}\), \(r\in\Into\) and \(\Phi'\) be such that %
  \((\Phi',r,\Phi)\in\Sdiffctx j{\Contca\Phi}\). We set
  \(\Phi'=(x_p:m'_p:A'_p)_{p=1}^k\), we must prove that %
  \(\Ldletv{x_j}{x_l}\in\Itopenn{\Phi'}{\Tseqact ra}{\Tdiff{A_l}}n\). %
  There are two cases.
  \begin{itemize}
  \item If \(j\not=l\) then \(m_j=\Emptymset\). We have %
    \(m'_j=\Emptymset\) and \(A'_j=\Tdiff{A_j}\), and %
    \(m'_p=m_p\) and \(A'_p=A_p\) for \(p\not=j\). %
    It follows that \(r=0\).
    We have \(x_l\in\Itopenn{\Phi'}{a_l}{A_l}{n}\) %
    by inductive hypothesis and hence  %
    \(
    \Ldletv{x_j}{x_l}
    =\Linj 0{x_l}
    \in\Itopenn{\Phi'}{\Tseqact 0a}{\Tdiff{A_l}}{n}
    \) by Lemma~\ref{lemma:itopenn-inj}, that
    is %
    \(\Ldletv{x_j}{x_l}\in\Itopenn{\Phi'}{\Tseqact
      ra}{\Tdiff{A_i}}{n}\) %
    as required.
  \item Assume now that \(j=l\). We have \(m'_p=\Emptymset\) and
    \(A'_p=A_p\) if \(p\not=l\), and \(m'_l=\Mset{\Tseqact r{a_l}}\)
    and \(A'_l=\Tdiff A_l\). Then by inductive hypothesis we have %
    \(x_l\in\Itopenn{\Phi'}{\Tseqact r{a_l}}{\Tdiff{A_l}}{n}\) which
    is exactly what we need since \(\Ldletv{x_j}{x_l}=x_l\) in that
    case.\qedhere
  \end{itemize}
\end{proof}

\begin{lemma}%
  \label{lemma:itopenn-num}
  Given \(\Phi=(x_1:\Emptymset:A_1,\dots,x_k:\Emptymset:A_k)\) %
  and \(\kappa\in\Nat\) one has %
  \(\Num\kappa\in\Itopenn{\Phi}{\kappa}{\Tnat}{n}\).
\end{lemma}
\begin{proof}
  By induction on \(n\). For \(n=0\) it suffices to observe that
  \(\Num\kappa\in\Itintt\kappa\Tnat\) which results obviously from the
  definition of \(\Itints\kappa\Tnat\nu\).  We prove the property for
  \(n+1\) assuming that it holds for \(n\). So with the notations of
  the lemma we must prove that
  \(\Num\kappa\in\Itopenn\Phi a{A_l}{n+1}\). Let
  \(l\in\Eset{1,\dots,k}\), \(r\in\Into\) and \(\Phi'\) be such that %
  \((\Phi',r,\Phi)\in\Sdiffctx j{\Contca\Phi}\). Setting
  \(\Phi'=(x_j:m'_j:A'_j)_{j=1}^k\), this means that
  \(m'_j=\Emptymset\) for each \(j\), that %
  \(A'_j=A_j\) for \(j\not=l\) and that %
  \(A'_l=\Tdiff{A_l}\). This also implies that \(r=0\). %
  By inductive hypothesis we have %
  \(\Num\kappa\in\Itopenn{\Phi'}{\kappa}{\Tnat}n\) and hence %
  \(
  \Ldletv{x_l}{\Num\kappa}
  =\Linj 0{\Num\kappa}
  \in\Itopenn{\Phi'}{(0,\kappa)}{\Tdiff\Tnat}{n}
  \) by Lemma~\ref{lemma:itopenn-inj}.
\end{proof}

\begin{lemma}%
  \label{lemma:itopenn-succ}
  For all \(n\in\Nat\), if %
  \(M\in\Itopenn{\Phi}{\Tseqact\delta\kappa}{\Tdiffm d\Tnat}{n}\) then %
  \(
  \Lsucc dM
  \in\Itopenn{\Phi}{\Tseqact\delta{(\kappa+1)}}{\Tdiffm{d}\Tnat}{n}
  \)
  where \(d=\Len\delta\).
\end{lemma}
\begin{proof}
  By induction on \(n\). Assume that \(n=0\), let %
  \(\Vect P\in\Itintctx\Phi\), \(M'=\Substbis M{\Vect P}\). %
  Let \(\nu\in\Nat\) and \(s\in\Itints{\kappa+1}\Tnat\nu\). %
  We have %
  \(
  \State{\Pempty}{\Num\kappa}{\Stsucc{s}}
  \Rel\Stred
  \State{\Pempty}{\Num{\kappa+1}}{s}
  \in\Itintc\nu
  \) and hence %
  \(\Stsucc s\in\Itints{\kappa}{\Tnat}\nu\).
  Since \(\Vect P\in\Itintctx{\Phi}\) %
  we have %
  \(M'\in\Itintt{\Tseqact\delta\kappa}{\Tdiffm d\Tnat}\) and hence %
  \(
  \State{\Rev{\delta}}{\Lsucc d{M'}}{s}
  \Rel\Stred\State{\Rev\delta}{M'}{\Stsucc s}
  \in\Itintc\nu
  \). %
  This shows that %
  \(
  \Lsucc d{M'}
  \in\Itintt{\Tseqact{\delta}{(\kappa+1)}}{\Tdiffm{d}\Tnat}
  \) as required.

  For the inductive step, we assume that the implication holds for
  \(n\) and we prove it for \(n+1\), so we assume that
  \(M\in\Itopenn{\Phi}{\Tseqact\delta\kappa}{\Tdiffm d\Tnat}{n+1}\). %
  Let \(r\in\Into\), \(l\in\{1,\dots,\Len\Phi\}\) and \(\Phi'\) such
  that \((\Phi',r,\Phi)\in\Sdiffctx l{\Contca\Phi}\). %
  We have %
  \(
  \Ldletv{x_l}M
  \in\Itopenn{\Phi}{\Tseqact{r\delta}{\kappa}}{\Tdiffm{d+1}\Tnat}{n}
  \) %
  and hence, by inductive hypothesis,
  \(
  \Ldletv{x_l}{\Lsucc dM}
  =\Lsucc{d+1}{\Ldletv{x_l}M}
  \in\Itopenn{\Phi'}{\Tseqact{r\delta}{(\kappa+1)}}{\Tdiffm{d+1}\Tnat}{n}
  \)
  and so we have shown that %
  \(
  \Lsucc dM\in\Itopenn{\Phi}{\Tseqact\delta{(\kappa+1)}}{\Tdiffm d\Tnat}{n+1}
  \). %
\end{proof}

\begin{lemma}%
  \label{lemma:itopenn-predz}
  For all \(n\in\Nat\), if %
  \(M\in\Itopenn{\Phi}{\Tseqact\delta0}{\Tdiffm d\Tnat}{n}\) then %
  \(
  \Lpred dM
  \in\Itopenn{\Phi}{\Tseqact\delta{0}}{\Tdiffm{d+1}\Tnat}{n}
  \)
  where \(d=\Len\delta\).
\end{lemma}
\begin{proof}
  Similar to that of Lemma~\ref{lemma:itopenn-succ}.
\end{proof}

\begin{lemma}%
  \label{lemma:itopenn-predp}
  For all \(n\in\Nat\), if %
  \(M\in\Itopenn{\Phi}{\Tseqact\delta{(\kappa+1)}}{\Tdiffm d\Tnat}{n}\) then %
  \(
  \Lpred dM
  \in\Itopenn{\Phi}{\Tseqact\delta{\kappa}}{\Tdiffm{d+1}\Tnat}{n}
  \)
  where \(d=\Len\delta\).
\end{lemma}
\begin{proof}
  Similar to that of Lemma~\ref{lemma:itopenn-succ}.
\end{proof}

\begin{lemma}%
  \label{lemma:itopenn-ifz}
  For all \(n\in\Nat\), if %
  \(M\in\Itopenn{\Phi_0}{\Tseqact\delta 0}{\Tdiffm d\Tnat}{n}\), %
  \(Q_0\in\Itopenn{\Phi_1}{b}{B}{n}\) and %
  \(\Tseq{\Contca\Phi}{Q_1}{B}\) then %
  \(\Lif dM{Q_0}{Q_1}\in\Itopenn{\Phi}{\Tseqact\delta b}{\Tdiffm
    dB}{n}\), %
  where \(d=\Len\delta\) and \(\Phi=\Phi_0+\Phi_1\).
\end{lemma}
\begin{proof}
  By induction on \(n\). Assume that \(n=0\), let %
  \(\Vect P\in\Itintctx\Phi\), \(M'=\Substbis M{\Vect P}\) %
  and \(Q_i'=\Substbis{Q_i}{\Vect P}\) for \(i=0,1\). We can write %
  \(B=\Tdiffm eF\) where \(F\) is sharp and \(b=\Tseqact\epsilon f\)
  where \(f\in\Tsemrel F\) and \(e=\Len\epsilon\). %
  Let \(\nu\in\Nat\) and \(s\in\Itints fF\nu\). %
  We have %
  \(
  \State{\Pempty}{\Num 0}{\Stif{\Rev\epsilon}{Q'_0}{Q'_1}{s}}
  \Rel\Stred
  \State{\Rev\epsilon}{Q'_0}{s}
  \). %
  Since \(\Vect P\in\Itintctx{\Phi}\subseteq\Itintctx{\Phi_1}\) %
  we have \(Q'_0\in\Itintt{\Tseqact\epsilon f}{\Tdiffm eF}\) %
  and hence \(\State{\Rev\epsilon}{Q'_0}{s}\in\Itintc\nu\) %
  and therefore %
  \(\Stif{\Rev\epsilon}{Q'_0}{Q'_1}{s}\in\Itints{0}{\Tnat}{\nu}\). %
  Since \(\Vect P\in\Itintctx{\Phi}\subseteq\Itintctx{\Phi_0}\) %
  we have \(M'\in\Itintt{\Tseqact\delta 0}{\Tdiffm d\Tnat}\) and hence %
  \(
  \State{\Rev{\delta\epsilon}}{\Lif d{M'}{Q'_0}{Q'_1}}{s}
  =\State{\Rev\epsilon\Rev\delta}{\Lif d{M'}{Q'_0}{Q'_1}}{s}
  \Rel\Stred
  \State{\Rev\delta}{M'}{\Stif{\Rev\epsilon}{Q'_0}{Q'_1}{s}}
  \in\Itintc\nu
  \). %
  This shows that %
  \(
  \Lif d{M'}{Q'_0}{Q'_1}
  \in\Itintt{\Tseqact{\delta\epsilon}{f}}{\Tdiffm{d+e}F}
  \) as required.

  For the inductive step, we assume that the implication holds for
  \(n\) and we prove it for \(n+1\), so we assume that
  \(M\in\Itopenn{\Phi_0}{\Tseqact\delta 0}{\Tdiffm d\Tnat}{n+1}\), %
  \(Q_0\in\Itopenn{\Phi_1}{b}{B}{n+1}\) and %
  \(\Tseq{\Contca\Phi}{Q_1}{B}\). %
  Let \(r\in\Into\), \(l\in\{1,\dots,\Len\Phi\}\) and \(\Phi'\) such
  that \((\Phi',r,\Phi)\in\Sdiffctx l{\Contca\Phi}\). %
  By Lemma~\ref{lemma:sdiff-split}, since \(\Phi=\Phi_0+\Phi_1\), %
  we can find \(r_0,r_1\in\Into\) such that \(r=r_0+r_1\) as well as %
  \(\Phi'_0,\Phi'_1\) such that %
  \((\Phi'_i,r_i,\Phi_i)\in\Sdiffctx l{\Contca\Phi}\) for \(i=0,1\), %
  and \(\Phi'=\Phi'_0+\Phi'_1\). If follows that %
  \(
  \Ldletv{x_l}{M}
  \in\Itopenn{\Phi'_0}{\Tseqact{r_0\delta}{0}}{\Tdiffm{d+1}\Tnat}{n}
  \) and %
  \(
  \Ldletv{x_l}{Q_0}
  \in\Itopenn{\Phi'_1}{\Tseqact{r_1}{b}}{\Tdiff B}{n}
  \). %
  Since
  \(\Tseq{\Contca{\Phi'}}{\Ldletv{x_l}{Q_1}}{\Tdiff B}\), %
  we have %
  \(
  \Lif{d+1}{\Ldletv{x_l}{M}}{\Ldletv{x_l}{Q_0}}{\Ldletv{x_l}{Q_1}}
  \in\Itopenn{\Phi'}{\Tseqact{r_0\delta r_1}{b}}{\Tdiffm{d+2}B}{n}
  \) by inductive hypothesis. %
  By Lemma~\ref{lemma:itopenn-flip} we have %
  \(
  \Lflipl d{\Lif{d+1}{\Ldletv{x_l}{M}}{\Ldletv{x_l}{Q_0}}{\Ldletv{x_l}{Q_1}}}
  \in\Itopenn{\Phi'}{\Tseqact{\Rcycle{r_0\delta r_1}}{b}}{\Tdiffm{d+2}B}{n}
  \), that is
  \(
  \Lflipl d{\Lif{d+1}{\Ldletv{x_l}{M}}{\Ldletv{x_l}{Q_0}}{\Ldletv{x_l}{Q_1}}}
  \in\Itopenn{\Phi'}{\Tseqact{r_1r_0\delta}{b}}{\Tdiffm{d+2}B}{n}
  \). Therefore %
  \begin{align*}
    \Ldletv{x_l}{\Lif dM{Q_0}{Q_1}}
    =\Lsum{\Lflipl d{\Lif{d+1}{\Ldletv{x_l}{M}}
    {\Ldletv{x_l}{Q_0}}{\Ldletv{x_l}{Q_1}}}}
    \in\Itopenn{\Phi'}{\Tseqact{r\delta}{b}}{\Tdiffm{d+2}B}{n}
  \end{align*}
  by Lemma~\ref{lemma:itopenn-lsum}, since \(r=r_0+r_1\).
\end{proof}
\begin{remark}
  In some sense this proof motivates syntactically the introduction of
  the cyclic permutation combinator \(\Lflipl{d}\_\) in the definition
  of \(\Ldletv x{\Lif dM{Q_0}{Q_1}}\) in Figure~\ref{fig:diff-subst}:
  it allows the \(\Lsum\_\) combinator to act at the right level. We have
  already motivated it denotationally by
  Theorem~\ref{th:sem-term-invariant}.
\end{remark}

\begin{lemma}%
  \label{lemma:itopenn-ifp}
  For all \(n\in\Nat\), if %
  \(M\in\Itopenn{\Phi_0}{\Tseqact\delta{(\kappa+1)}}{\Tdiffm d\Tnat}{n}\), %
  \(Q_1\in\Itopenn{\Phi_1}{b}{B}{n}\) for some \(\kappa\in\Nat\) and %
  \(\Tseq{\Contca\Phi}{Q_0}{B}\) then %
  \(\Lif dM{Q_0}{Q_1}\in\Itopenn{\Phi}{\Tseqact\delta b}{\Tdiffm
    dB}{n}\), %
  where \(d=\Len\delta\) and \(\Phi=\Phi_0+\Phi_1\).
\end{lemma}
\begin{proof}
  Similar to that of Lemma~\ref{lemma:itopenn-ifz}.
\end{proof}

\begin{lemma}
  \label{lemma:itopenn-let}
  For all \(n\in\Nat\), if %
  \(M\in\Itopenn{\Phi_0}{\Tseqact\delta\kappa}{\Tdiffm d\Tnat}{n}\) and %
  \(N\in\Itopenn{\Phi_1,x:k\Mset\kappa}{b}{B}{n}\)
  where \(k,\kappa\in\Nat\), then %
  \(
  \Llet dxMN\in\Itopenn{\Phi}{\Tseqact\delta b}{\Tdiffm dB}{n}
  \) where \(d=\Len\delta\) and \(\Phi=\Phi_0+\Phi_1\).
\end{lemma}
\begin{proof}
  By induction on \(n\).  Assume that \(n=0\), let %
  \(\Vect P\in\Itintctx\Phi\), \(M'=\Substbis M{\Vect P}\) %
  and \(N'=\Substbis{N}{\Vect P}\). We can write %
  \(B=\Tdiffm eF\) where \(F\) is sharp and \(b=\Tseqact\epsilon f\)
  where \(f\in\Tsemrel F\) and \(e=\Len\epsilon\). %
  Let \(\nu\in\Nat\) and \(s\in\Itints fF\nu\). %
  We have %
  \(
  \State{\Pempty}{\Num\kappa}{\Stlet{\Rev\epsilon}{x}{N'}{s}}
  \Rel\Stred
  \State{\Rev\epsilon}{\Subst{N'}{\Num\kappa}{x}}{s}
  \). %
  Since \(\Vect P\in\Itintctx{\Phi}\subseteq\Itintctx{\Phi_0}\) %
  and \(\Num\kappa\in\Itintt{k\Mset\kappa}{\Tnat}\)
  we have %
  \(\Subst{N'}{\Num\kappa}{x}\in\Itintt{\Tseqact\epsilon f}{\Tdiffm eF}\) %
  and hence %
  \(\State{\Rev\epsilon}{\Subst{N'}{\Num\kappa}{x}}{s}\in\Itintc\nu\) %
  and therefore %
  \(\Stlet{\Rev\epsilon}{x}{N'}{s}\in\Itints{\kappa}{\Tnat}{\nu}\). %
  Since \(\Vect P\in\Itintctx{\Phi}\subseteq\Itintctx{\Phi_1}\) %
  we have \(M'\in\Itintt{\Tseqact\delta\kappa}{\Tdiffm d\Tnat}\) and hence %
  \(
  \State{\Rev{\delta\epsilon}}{\Llet dx{M'}{N'}}{s}
  =\State{\Rev\epsilon\Rev\delta}{\Llet dx{M'}{N'}}{s}
  \Rel\Stred
  \State{\Rev\delta}{M'}{\Stlet{\Rev\epsilon}{x}{N'}{s}}
  \in\Itintc\nu
  \). %
  This shows that %
  \(
  \Llet dx{M'}{N'}
  \in\Itintt{\Tseqact{\delta\epsilon}{f}}{\Tdiffm{d+e}F}
  \) as required.

  For the inductive step, we assume that the implication holds for
  \(n\) and we prove it for \(n+1\), so we assume that
  \(M\in\Itopenn{\Phi_0}{\Tseqact\delta\kappa}{\Tdiffm d\Tnat}{n+1}\), %
  and \(N\in\Itopenn{\Phi_1,x:k\Mset\kappa}{b}{B}{n+1}\). %
  Let \(r\in\Into\), \(l\in\{1,\dots,\Len\Phi\}\) and \(\Phi'\) %
  be such that \((\Phi',r,\Phi)\in\Sdiffctx l{\Contca\Phi}\). %
  By Lemma~\ref{lemma:sdiff-split}, since %
  \(\Phi=\Phi_0+\Phi_1\), we can find \(r_0,r_1\in\Into\)
  such that \(r=r_0+r_1\) as well as %
  \(\Phi'_0,\Phi'_1\) such that %
  \((\Phi'_0,r_0,\Phi_0)\in\Sdiffctx l{\Contca\Phi}\) so that %
  \(((\Phi'_0,x:k\Mset\kappa:\Tnat),r_0,(\Phi_0,x:k\Mset\kappa:\Tnat))
  \in\Sdiffctx l{(\Contca\Phi,x:\Tnat)}\), and %
  \((\Phi'_1,r_1,\Phi_1)\in\Sdiffctx l{\Contca\Phi}\) and
  \(\Phi'=\Phi'_0+\Phi'_1\). If follows that %
  \( \Ldletv{x_l}{M}
  \in\Itopenn{\Phi'_0}{\Tseqact{r_0\delta}{\kappa}}{\Tdiffm{d+1}\Tnat}{n}
  \) and %
  \( \Ldletv{x_l}{N}
  \in\Itopenn{\Phi'_1,x:k\Mset\kappa:\Tnat}{\Tseqact{r_1}{b}}{\Tdiff
    B}{n} \). %
  We have %
  \( \Llet{d+1}x{\Ldletv{x_l}{M}}{\Ldletv{x_l}{N}}
  \in\Itopenn{\Phi'}{\Tseqact{r_0\delta r_1}{b}}{\Tdiffm{d+2}B}{n} \)
  by inductive hypothesis. %
  By Lemma~\ref{lemma:itopenn-flip} we have %
  \( \Lflipl d{\Llet{d+1}x{\Ldletv{x_l}{M}}{\Ldletv{x_l}{N}}}
  \in\Itopenn{\Phi'}{\Tseqact{\Rcycle{r_0\delta
        r_1}}{b}}{\Tdiffm{d+2}B}{n} \), that is
  \( \Lflipl d{\Llet{d+1}x{\Ldletv{x_l}{M}}{\Ldletv{x_l}{N}}}
  \in\Itopenn{\Phi'}{\Tseqact{r_1r_0\delta}{b}}{\Tdiffm{d+2}B}{n}
  \). Therefore %
  \begin{align*}
    \Ldletv{x_l}{\Llet dx{M}{N}}
    =\Lsum{\Lflipl d{\Llet{d+1}x{\Ldletv{x_l}{M}}{\Ldletv{x_l}{N}}}}
    \in\Itopenn{\Phi'}{\Tseqact{r\delta}{b}}{\Tdiffm{d+2}B}{n}
  \end{align*}
  by Lemma~\ref{lemma:itopenn-lsum}, since \(r=r_0+r_1\).
\end{proof}

\begin{lemma}%
  \label{lemma:itopenn-app}
  If \(M\in\Itopenn{\Phi_0}{(m,b)}{\Timpl AB}n\) with %
  \(m=\Mset{\List a1k}\) and %
  \((N\in\Itopenn{\Phi_j}{a_j}{A}n)_{j=1}^k\) then %
  \(\App MN\in\Itopenn\Phi bBn\) where %
  \(\Phi=\sum_{j=0}^k\Phi_j\) (so that
  \(\forall j\ \Contca{\Phi_j}=\Contca\Phi\)).
\end{lemma}
\begin{proof}
  By induction on \(n\). Assume that \(n=0\). So let
  \(\Vect P\in\Itintctx\Phi\). Since for each \(j=0,\dots,k\) we have
  \(\Itintctx\Phi\subseteq\Itintctx{\Phi_j}\) by the assumption that %
  \(\Phi=\sum_{j=0}^k\Phi_j\), we have %
  \(M'\in\Itintt{(m,b)}{\Timpl AB}\) (setting
  \(M'=\Substbis M{\Vect P}\), remember that we use this convention
  systematically) and \(N'\in\bigcap_{j=1}^k\Itintt{a_j}A\). %
  We can write uniquely \(B=\Tdiffm eF\) with \(F\) sharp and %
  \(b=\Tseqact\epsilon f\) with \(\epsilon\in\Into^e\) and
  \(f\in\Tsemrel F\). Let \(\nu\in\Nat\) and \(s\in\Itints fF\nu\), we have %
  \(
  \State{\Rev\epsilon}{\App{M'}{N'}}{s}
  \Rel\Stred
  \State{\Rev\epsilon}{M'}{\Starg{N'}{s}}\in\Itintc\nu
  \) %
  since \(\Starg{N'}s\in\Itints fF\nu\).

  Now assume that the implication holds for \(n\) and let us prove it
  for \(n+1\) so assume that %
  \(M\in\Itopenn{\Phi_0}{(m,b)}{\Timpl AB}{n+1}\) and %
  \((N\in\Itopenn{\Phi_j}{a_j}{A}{n+1})_{j=1}^k\). %
  Let \(r\in\Into\) and let \(\Phi'\) be such that %
  \((\Phi',r,\Phi)\in\Sdiffctx l{\Contca\Phi}\) for some %
  \(1\leq l\leq\Len\Phi\). By Lemma~\ref{lemma:sdiff-split} we can
  find %
  \((r_i\in\Into)_{i=0}^k\) such that \(r=\sum_{i=0}^kr_i\) as well
  as %
  \((\Phi'_i)_{i=0}^k\) such that \(\Phi'=\sum_{i=0}\Phi'_i\) and %
  \(((\Phi'_i,r_i,\Phi_i)\in\Sdiffctx{l}{\Contca\Phi} )_{i=0}^k\). So
  by our assumptions we have %
  \(\Ldletv{x_l}{M}\in\Itopenn{\Phi'_0}{(m,\Tseqact{r_0}b)}
  {\Timpl A{\Tdiff B}}{n}\) and %
  \(\Ldletv{x_l}N\in\bigcap_{i=1}^k\Itopenn{\Phi'_i}
  {\Tseqact{r_j}{a_j}}{\Tdiff A}{n}\). %
  Let \(r'=\sum_{i=1}^kr_i\in\Into\) and %
  \(m'=\Mset{\Tseqact{r_1}{a_1},\dots,\Tseqact{r_k}{a_k}}\).
  By Lemma~\ref{lemma:itopenn-ldiff} we have %
  \(
  \Ldiff{\Ldletv{x_l}{M}}
  \in\Itopenn{\Phi'_0}{(m',\Tseqact{r'r_0}b)}
  {\Timpl{\Tdiff A}{\Tdiffm2B}}{n}
  \) and hence by Lemma~\ref{lemma:itopenn-lsum} we have %
  \(
  \Lsum{\Ldiff{\Ldletv{x_l}{M}}}
  \in\Itopenn{\Phi'_0}{(m',\Tseqact{(r'+r_0)}b)}
  {\Timpl{\Tdiff A}{\Tdiff B}}{n}
  \). %
  By inductive hypothesis we get %
  \(
  \Ldletv{x_l}{\App MN}
  =\App{\Lsum{\Ldiff{\Ldletv{x_l}{M}}}}{\Ldletv{x_l}N}
  \in\Itopenn{\Phi'}{\Tseqact{r}{b}}{\Tdiff B}{n}
  \) since \(r=r'+r_0\).
  Since we have proven this for all choices of \(l\), \(r\) and
  \(\Phi'\), our contention follows.
\end{proof}

\begin{lemma}%
  \label{lemma:itopenn-lfix}
  For any \(n\in\Nat\), if %
  \(\App M{\Lfix M}\in\Itopenn\Phi bBn\)
  then \(\Lfix M\in\Itopenn\Phi bBn\).
\end{lemma}
\begin{proof}
  By induction on \(n\). For \(n=0\), let \(\Vect P\in\Itintctx\Phi\)
  and let \(M'=\Substbis M{\Vect P}\). We write \(B=\Tdiffm eF\) where
  \(F\) is sharp and \(b=\Tseqact\epsilon f\) where \(f\in\Tsemrel
  F\). Let \(\nu\in\Nat\) and \(s\in\Itints fF\nu\), we have %
  \(
  \State{\Rev\epsilon}{\Lfix{M'}}{s}
  \Rel\Stred
  \State{\Rev\epsilon}{M'}{\Starg{\Lfix{M'}}s}
  \in\Itintc\nu
  \) %
  since we have %
  \(
  \State{\Rev\epsilon}{\App{M'}{\Lfix{M'}}}{s}
  \Rel\Stred
  \State{\Rev\epsilon}{M'}{\Starg{\Lfix{M'}}s}  
  \)
  and we have \(\App{M'}{\Lfix{M'}}\in\Itintt bB\) by our assumption. %
  So we have shown that \(\Lfix M\in\Itopenn\Phi bB0\).

  Assume that the implication holds for \(n\) and let us prove it for
  \(n+1\). So we assume that %
  \(\App{M}{\Lfix M}\in\Itopenn{\Phi}{b}{B}{n+1}\). %
  Let \(l\in\Eset{1,\dots,\Len\Phi}\), \(r\in\Into\) and \(\Phi'\) be
  such that \((\Phi',r,\Phi)\in\Sdiffctx l{\Contca\Phi}\). We know that %
  \(
  \Ldletv{x_l}{\App{M}{\Lfix M}}\in\Itopenn{\Phi'}{\Tseqact rb}{\Tdiff B}{n}
  \). On the other hand %
  \(
  \Ldletv{x_l}{\App{M}{\Lfix M}}
  =\App{\Lsum{\Ldiff{\Ldletv{x_l}M}}}{\Ldletv{x_l}{\Lfix M}}
  \), see Figure~\ref{fig:diff-subst}. %
  On the same figure we see that %
  \(\Ldletv{x_l}{\Lfix M}=\Lfix{(\Lsum{\Ldiff{\Ldletv{x_l}M}})}\). %
  So we have %
  \(
  \App{\Lsum{\Ldiff{\Ldletv{x_l}M}}}{\Lfix{(\Lsum{\Ldiff{\Ldletv{x_l}M}})}}
  \in\Itopenn{\Phi'}{\Tseqact rb}{\Tdiff B}{n}
  \) and hence %
  \(
  \Lfix{(\Lsum{\Ldiff{\Ldletv{x_l}M}})}
  \in\Itopenn{\Phi'}{\Tseqact rb}{\Tdiff B}{n}
  \) by inductive hypothesis.
\end{proof}

\begin{remark}
  The next lemma has a different structure: the hypothesis in the
  implication we prove by induction is stronger. This feature is used
  only in the base case where it is absolutely crucial.
\end{remark}

\begin{lemma}%
  \label{lemma:itopenn-abst0}
  For any \(n\in\Nat\), if
  \(M\in\bigcap_{h\geq n}\Itopenn{\Phi,x:m:A}{b}{B}{h}\) then
  \(\Abst xAM\in\Itopenn{\Phi}{(m,b)}{\Timpl AB}{n}\).
\end{lemma}
\begin{proof}
  We prove the statement by induction on \(n\). For \(n=0\) our
  assumption is
  \(\forall h\in\Nat\ M\in\Itopenn{\Phi,x:m:A}{b}{B}{h}\) and we prove
  \(\Abst xAM \in\Itopenn{\Phi}{(m,b)}{\Timpl AB}{0}\).
  So let %
  \(\Vect P\in\Itintctx\Phi\) and let \(M'=\Substbis M{\Vect P}\). %
  We can write uniquely \(B=\Tdiffm eF\) with \(F\) sharp and %
  \(b=\Tseqact\epsilon f\) with \(f\in\Tsemrel F\). Let \(\nu\in\Nat\)
  and \(s\in\Itints{(m,f)}{\Timpl AF}{\nu}\). Then, by the typing
  rules for stacks, we have %
  \(s=\Stdiff{r_1}{\cdots\Stdiff{r_d}{\Starg Pt}}\) where %
  \(d\in\Nat\), %
  \(\List r1d\in\Into\) and there are %
  \(m_1\in\Mfin{\Tsemrel{\Tdiff A}},\dots,m_d\in\Mfin{\Tsemrel{\Tdiffm
      dA}}\) such that %
  \(((m_i,(r_i,m_{i-1}))\in\Sdiff_{\Tsemrel{\Tdiffm{i-1}{A}}})_{i=1}^d\),
  where we set \(m_0=m\). And last %
  \(P\in\Itintt{m_d}{\Tdiffm dA}\) and \(t\in\Itints fF\nu\). %
  By our assumption%
  \footnote{It is here that the special form of the hypothesis is
    crucial.} %
  about \(M\) we have in particular %
  \(M\in\Itopenn{\Phi,x:m:A}{b}{B}{d}\) and hence %
  \( \Ldletvm d{x}{M} \in\Itopenn{\Phi,x:m_d:\Tdiffm
    dA}{\Tseqact{r_d\cdots r_1}b}{\Tdiffm dB}{0} \) so that %
  \( \Subst{\Ldletvm dx{M'}}Px =\Substbis{\Ldletvm dxM}{\Vect P,P/x}
  \in\Itintt{\Tseqact{r_d\cdots r_1}b}{\Tdiffm dB} \).
  So we have
  \begin{align*}
    \State{\Rev\epsilon}{\Abst xA{M'}}{s}
    &\Rel\Stred
      \State{\Rev\epsilon r_1}{\Abst x{\Tdiff A}{\Ldletv x{M'}}}
      {\Stdiff{r_2}{\cdots\Stdiff{r_d}{\Starg Pt}}}\\
    &\Rel\Stred\cdots\Rel\Stred    
      \State{\Rev\epsilon r_1\cdots r_d}{\Abst x{\Tdiffm dA}
      {\Ldletvm dx{M'}}}
      {\Starg Pt}\\
    &\Rel\Stred
      \State{\Rev{r_d\cdots r_1\epsilon}}{\Subst{\Ldletvm dx{M'}}Px}{t}
      \in\Itintc\nu
  \end{align*}
  which entails that %
  \(\State{\Rev\epsilon}{\Abst xA{M'}}{s}\in\Itintc\nu\) and hence %
  \(\Abst xA{M'}\in\Itintt{(m,b)}{\Timpl AB}\) as required.

  For the inductive step we assume that the implication holds for
  \(n\) and prove it for \(n+1\). %
  Remember that in this implication, \(M\), \(\Phi\), 
  \(m\), \(A\), \(b\) and \(B\) are universally quantified.  So we
  assume that %
  \(\forall h\geq n+1\ M\in\Itopenn{\Phi,x:m:A}{b}{B}{h}\) and prove
  that %
  \(\Abst xAM\in\Itopenn{\Phi}{(m,b)}{\Timpl AB}{n+1}\). %
  We set \(\Phi=(x_1:m_1:A_1,\dots,x_k:m_k:A_k)\).
  Let \(r\in\Into\), \(l\in\Eset{1,\dots,k}\) and \(\Phi'\) be such
  that %
  \((\Phi',r,\Phi)\in\Sdiffctx{l}{\Contca\Phi}\). We have %
  \(((\Phi',x:m:A),r,(\Phi,x:m:A))\in\Sdiffctx{l}{(\Contca\Phi,x:A)}\) %
  and hence %
  \(
  \forall h\geq n\ \Ldletv{x_l}M
  \in\Itopenn{\Phi',x:A}{(m,\Tseqact rb)}{\Timpl A{\Tdiff B}}{h}
  \) %
  by our assumption about \(M\). It follows by inductive hypothesis that %
  \(
  \Ldletv{x_l}{\Abst xAM}=\Abst xA{\Ldletv{x_l}M}
  \in\Itopenn{\Phi'}{(m,\Tseqact rb)}{\Timpl A{\Tdiff B}}{n}\). %
  Therefore \(\Abst xAM\in\Itopenn{\Phi}{(m,b)}{\Timpl AB}{n+1}\) as required.
\end{proof}

\begin{lemma}%
  \label{lemma:itopenn-abst}
  If
  \(M\in\Itopen{\Phi,x:m:A}{b}{B}\) then
  \(\Abst xAM\in\Itopen{\Phi}{(m,b)}{\Timpl AB}\).
\end{lemma}
\begin{proof}
  Apply Lemma~\ref{lemma:itopenn-abst0}.
\end{proof}

\begin{theorem}%
  \label{th:cdpcf-adequacy}
  If \(\Tseqi{\Phi}MbB\) then \(M\in\Itopen\Phi bB\).
\end{theorem}
\begin{proof}
  By induction on the derivation of %
  \(\Tseqi\Phi MbB\).
  We write \(\Phi=(x_1:m_1:A_1,\dots,x_k:m_k:A_k)\).

  \Proofcase %
  Assume that \(M=x_l\) for some %
  \(l\in\Eset{1,\dots,k}\) so that the derivation consists of a %
  rule~\ref{rl:itrvar} 
  \(B=A_l\) and \(m_l=\Mset b\) and we have
  \(m_j=\Emptymset\) for \(j\not=l\).  So we can apply
  Lemma~\ref{lemma:itopenn-var} which gives us %
  \(x_l\in\Itopenn\Phi{b}{B}{n}\) for all \(n\in\Nat\).

  \Proofcase %
  Assume that \(M=\Num\kappa\) so that the derivation consists of a
  rule~\ref{rl:itrnum} 
  and hence \(B=\Tnat\) and \(b=\kappa\), and we have
  \(m_j=\Emptymset\) for each \(j\). By Lemma~\ref{lemma:itopenn-num}
  we get \(M\in\Itopenn\Phi{\Num\kappa}\Tnat n\) for each
  \(n\in\Nat\).
  
  \Proofcase %
  Assume that \(M=\Abst xAN\) so that \(B=\Timplp AC\) and %
  \(b=(m,c)\) and we have \(\Tseqi{\Phi,x:m:A}{M}{c}{C}\) and hence %
  \(M\in\Itopen{\Phi,x:m:A}{c}{C}\) by inductive hypothesis from
  which we get %
  \(\Abst xAM\in\Itopen{\Phi}{(m,c)}{\Timpl AC}\) by
  Lemma~\ref{lemma:itopenn-abst}.

  \Proofcase %
  Assume that \(M=\App NQ\) with %
  \(\Tseqi{\Phi_0}{N}{(m,b)}{\Timpl AB}\) with %
  \(m=\Mset{\List a1n}\) and \((\Tseqi{\Phi_j}{Q}{a_j}{A})_{j=1}^n\), %
  and \(\Phi=\sum_{j=0}^n\Phi_j\). By inductive hypothesis we have %
  \(N\in\Itopen{\Phi_0}{(m,b)}{\Timpl AB}\) and %
  \((Q\in\Itopen{\Phi_j}{a_j}{A})_{j=1}^n\) so we get %
  \(\App NQ\in\Itopen{\Phi}{b}{B}\) by Lemma~\ref{lemma:itopenn-app}.

  \Proofcase %
  Assume that \(M=\Ldiff N\) with %
  \(B=\Timplp{\Tdiff A}{\Tdiff C}\),
  \(b=(m',\Tseqact rc)\),
  \((m',(r,m))\in\Sdiff_{\Tsemrel A}\),
  \(\Tseqi\Phi{N}{(m,c)}{\Timpl AC}\).
  By inductive hypothesis we have %
  \(N\in\Itopen{\Phi}{(m,c)}{\Timpl AC}\) and hence %
  \(M\in\Itopen{\Phi}{(m',\Tseqact rc)}{\Timpl{\Tdiff A}{\Tdiff C}}\) %
  by Lemma~\ref{lemma:itopenn-ldiff}.
  
  \Proofcase %
  Assume that \(M=\Lfix N\) with %
  \(\Tseqi{\Phi_0}N{(m,b)}{\Timpl BB}\) with \(m=\Mset{\List b1n}\), %
  \((\Tseqi{\Phi_j}{\Lfix N}{b_j}{B})_{j=1}^n\) and %
  \(\Phi=\sum_{j=0}^n\Phi_j\). By inductive hypothesis we get %
  \(N\in\Itopen{\Phi_0}{(m,b)}{\Timpl BB}\) and %
  \((\Lfix N\in\Itopen{\Phi_j}{b_j}{B})_{j=1}^n\) and hence
  \(\App N{\Lfix N}\in\Itopen\Phi bB\) by
  Lemma~\ref{lemma:itopenn-app}. It follows that
  \(M=\Lfix N\in\Itopen\Phi bB\) by Lemma~\ref{lemma:itopenn-lfix}.

  \Proofcase %
  Assume that \(M=\Lsucc dN\) with
  \(\Tseqi{\Phi}{N}{\Tseqact\delta\kappa}{\Tdiffm d\Tnat}\) where
  \(\kappa\in\Nat\) and \(d=\Len\delta\). By inductive hypothesis we have %
  \(N\in\Itopen{\Phi}{\Tseqact\delta\kappa}{\Tdiffm d\Tnat}\) and hence %
  \(M\in\Itopen{\Phi}{\Tseqact\delta{\kappa+1}}{\Tdiffm d\Tnat}\) %
  by Lemma~\ref{lemma:itopenn-succ}.

  \Proofcase %
  The cases where \(M=\Lpred dN\) are similar, using
  Lemmas~\ref{lemma:itopenn-predz} and~\ref{lemma:itopenn-predp}.
  
  \Proofcase %
  Assume that \(M=\Lif d{N}{Q_0}{Q_1}\) and that
  \(\Tseqi{\Phi_0}{N}{\Tseqact\delta 0}{\Tdiffm d\Tnat}\) (with
  \(d=\Len\delta\)), \(B=\Tdiffm dC\), \(b=\Tseqact\delta c\) with
  \(\Tseqi{\Phi_1}{Q_0}{c}{C}\), \(\Tseq{\Contca\Phi}{Q_1}{C}\) and %
  \(\Phi=\Phi_0+\Phi_1\). By inductive hypothesis we have %
  \(N\in\Itopen{\Phi_0}{\Tseqact\delta 0}{\Tdiffm d\Tnat}\) and %
  \(Q_0\in\Itopen{\Phi_1}{c}{C}\) and hence %
  \(M\in\Itopen{\Phi}{\Tseqact\delta c}{\Tdiffm dC}\) %
  by Lemma~\ref{lemma:itopenn-ifz}.

  \Proofcase %
  The case \(M=\Lif d{N}{Q_0}{Q_1}\),
  \(\Tseqi{\Phi_0}{N}{\Tseqact\delta{(\kappa+1)}}{\Tdiffm d\Tnat}\) %
  (with \(\kappa\in\Nat\) and \(d=\Len\delta\)), \(B=\Tdiffm dC\),
  \(b=\Tseqact\delta c\) with \(\Tseqi{\Phi_1}{Q_1}{c}{C}\),
  \(\Tseq{\Contca\Phi}{Q_0}{C}\) and %
  \(\Phi=\Phi_0+\Phi_1\) is similar to the previous one, using
  Lemma~\ref{lemma:itopenn-ifp}.

  \Proofcase %
  Assume that \(M=\Llet dxNQ\) with %
  \(\Tseqi{\Phi_0}{N}{\Tseqact\delta\kappa}{\Tdiffm d\Tnat}\) and %
  \(\Tseqi{\Phi_1,x:k\Mset\kappa:\Tnat}{Q}{b}{B}\), %
  \(d=\Len\delta\) and \(\Phi=\Phi_0+\Phi_1\). By inductive hypothesis
  we have %
  \(N\in\Itopen{\Phi_0}{\Tseqact\delta\kappa}{\Tdiffm d\Tnat}\) and %
  \(Q\in\Itopen{\Phi_1,x:k\Mset\kappa:\Tnat}{b}{B}\) and hence by
  Lemma~\ref{lemma:itopenn-let} we have %
  \(\Llet dxNQ\in\Itopen\Phi{\Tseqact\delta b}{\Tdiffm dB}\).

  \Proofcase %
  Assume that \(M=\Linjd rdN\) with %
  \(\Tseqi\Phi N{\Tseqact\delta c}{\Tdiffm dC}\) so that %
  \(B=\Tdiffm{d+1}C\) and \(b=\Tseqact{\delta r}c\) (of course
  \(d=\Len\delta\)). By inductive hypothesis we have %
  \(N\in\Itopen{\Phi}{\Tseqact\delta c}{\Tdiffm dC}\) and hence %
  \(M\in\Itopen\Phi{\Tseqact{\delta r}c}{\Tdiffm{d+1}C}\)
  by Lemma~\ref{lemma:itopenn-inj}.
  
  \Proofcase %
  Assume that \(M=\Lsumd dN\) with %
  \(\Tseqi\Phi N{\Tseqact{\delta r_0r_1} c}{\Tdiffm{d+2}C}\) so that %
  \(B=\Tdiffm{d+1}C\) and \(b=\Tseqact{\delta r}c\) with
  \(r=r_0+r_1\in\Into\) (of course \(d=\Len\delta\)). By inductive
  hypothesis we have %
  \(N\in\Itopen{\Phi}{\Tseqact{\delta r_0r_1}c}{\Tdiffm{d+2}C}\) and hence %
  \(M\in\Itopen\Phi{\Tseqact{\delta r}c}{\Tdiffm{d+1}C}\) by
  Lemma~\ref{lemma:itopenn-lsum}.
  
  \Proofcase %
  Assume that \(M=\Lprojd rdN\) with %
  \(\Tseqi{\Phi}{N}{\Tseqact{\delta r}c}{\Tdiffm{d+1}C}\) so that %
  \(B=\Tdiffm dC\) and \(b=\Tseqact\delta c\) (of course
  \(d=\Len\delta\)). By inductive hypothesis we have %
  \(N\in\Itopen{\Phi}{\Tseqact{\delta r}c}{\Tdiffm{d+1}C}\) and hence
  \(M\in\Itopen{\Phi}{\Tseqact{\delta}c}{\Tdiffm{d}C}\) by %
  Lemma~\ref{lemma:itopenn-proj}.

  \Proofcase %
  Assume that \(M=\Lflipdl dlN\) with %
  \(\Tseqi{\Phi}N{\Tseqact{\delta\alpha}c}{\Tdiffm{d+l+2}C}\) so
  that %
  \(B=\Tdiffm{d+l+2}C\) and \(b=\Tseqact{\delta\Rcycle\alpha}c\) (of
  course \(\Len\delta=d\) and \(\Len\alpha=l+2\)). %
  By inductive hypothesis we have %
  \(N\in\Itopen{\Phi}{\Tseqact{\delta\alpha}{c}}{\Tdiffm{d+l+2}C}\) %
  and hence
  \(M\in\Itopen{\Phi}{\Tseqact{\delta\Rcycle\alpha}{c}}{\Tdiffm{d+l+2}C}\) %
  by Lemma~\ref{lemma:itopenn-flip}.
\end{proof}

\begin{theorem}%
  \label{cor:sem-adequacy}
  Let \(M\) be a closed term and let \(\nu\in\Nat\) be such that
  \(\Tseqi{}M\nu\Tnat\). %
  Then %
  \(
  \Mset{\State\Pempty M\Stempty}
  \Rel{\Trcl{\Rsred{\Msrs\States}}}
  C=C_0+\Mset{\State\Pempty{\Num\nu}\Stempty}
  \) %
  for some multiset of well typed states \(C_0\) such that \(C\) is
  \(\cL\)-summable in any model \(\cL\).
\end{theorem}

\section{Determinism and probabilities} %
\label{sec:determinism}
Our goal is to refine Theorem~\ref{cor:sem-adequacy} by showing that
none of the elements of \(C_0\) reduces to a value.
To this end we use the model of probabilistic coherence spaces that we
have presented in Section~\ref{sec:PCS-definition}.

\subsection{Integers and fixpoints in \(\PCOH\)} %
\label{sec:pcoh-fix}
The category \(\PCOH\) is \Cpolike{} in the sense of
Section~\ref{sec:cat-recursion}. The order relation \(\leq\) on
morphisms in \(\PCOH(X,Y)\) is given by \(s\leq t\) iff
\(\forall(a,b)\in\Web{\Limpl XY}\ s_{a,b}\leq t_{a,b}\). It is a
standard fact that for any PCS \(X\) the poset \(\Pcoh X\) (equipped
with the pointwise order of \(\Web X\)-indexed families of
non-negative real numbers) is a cpo and that all the operations
(composition of morphisms, tensor product, \(\Excl\_\) functor) preserve
the lubs of directed families of morphisms, that is \(\PCOH\) is
\Cpolike{}, see for instance~\cite{DanosEhrhard08}.

As a consequence for any PCS \(X\) we have a fixpoint operator
\(\Sfix\in\PCOH(\Simpl XX,X)\) which is characterized by %
\(\Fun\Sfix(t)=\sup_{n\in\Nat}\Fun t^n(0)\). Concerning derivatives,
this operator satisfies Theorem~\ref{th:sdfun-sfix}.

For the categorical axiomatization of integers we refer to
Section~\ref{sec:basic-categ-constr}.  We define \(\Snat\) by %
\(\Web\Snat=\Nat\) and \(u\in\Pcoh\Snat\) if \(u\in\Realpto\Nat\)
satisfies \(\sum_{\nu\in\Nat}u_\nu\leq 1\). This structure \(\Snat\)
is a PCS such that \(\Pcoh{\Orth\Snat}=\Intcc01^\Nat\) as easily
checked. Notice that \(\Norm u_\Snat=\sum_{\nu\in\Nat}u_\nu\). %
This is an \(\ell^1\) norm whereas in \(\Orth\Snat\) the norm is
\(\Norm{u'}_{\Orth\Snat}=\sup_{\nu\in\Nat}u'_\nu\) which is an
\(\ell^\infty\) norm. The matrix %
\(\chi\in\Realpto{\Web{\Limpl{\Plus\Sone\Snat}{\Snat}}}\) defined by %
\(\chi_{(0,\Sonelem),n}=\Kronecker 0n\) and %
\(\chi_{(1,n'),n}=\Kronecker{n'+1}n\) satisfies %
\(\chi\in\PCOH(\Plus\Sone\Snat,\Snat)\) and is an isomorphism between
these two PCSs. %
Then, given \(t\in\PCOH(\Plus\Sone X,X)\), let
\((s_n\in\PCOH(\Snat,X))_{n\in\Nat}\) be the sequence of morphisms
defined by %
\begin{align*}
  s_0=0
  \text{\quad and\quad}
  s_{n+1}=t\Compl(\Plus\Sone{s_n})\Compl\Funinv\chi\,.
\end{align*}
An easy induction shows that \(\forall n\in\Nat\ s_n\leq s_{n+1}\) and
so %
\(s=\sup_{n\in\Nat}s_n\in\PCOH(\Snat,X)\) satisfies %
\(s=t\Compl(\Plus\Sone s)\Compl\Funinv\chi\). This means that
\((\Snat,\chi)\) is an initial algebra for the functor
\(\Plus\Sone\_:\PCOH\to\PCOH\) and so \(\PCOH\) %
satisfies~\ref{ax:laxint}. 

The associated morphisms \(\Ssuc,\Spred\in\PCOH(\Snat,\Snat)\) are
characterized by %
\(\Matappa\Ssuc u=\sum_{\nu\in\Nat}u_\nu\Base{\nu+1}\) and %
\(\Matappa\Spred u=u_0\Base 0+\sum_{\nu\in\Nat}u_{\nu+1}\Base\nu\). %
The morphism \(\Sif\in\PCOH(\Tens\Snat{\Withp XX},X)\) is
characterized by %
\[
  \Matappa{\Sif}{\Tensp{u}{\Tuple{x^0,x^1}}}
  =u_0x^0+\big(\sum_{\nu=1}^\infty u_\nu\big)x^1\,.
\]
Last the morphism %
\(\Slet\in\PCOH(\Tens\Snat{(\Limpl{\Excl\Snat}{X})},X)\) is
characterized by %
\[
  \Matappa\Slet{\Tensp ut}=\sum_{\nu\in\Nat}u_\nu\Fun t(\Base\nu)\,.
\]

So we have an interpretation of terms and states in \(\PCOH\) which is
sound, that is, invariant by reduction.
More precisely, following the general pattern of
Section~\ref{sec:types-terms-interp}, we associate with each type
\(A\) an object \(\Tsempcoh A\) of \(\PCOH\) in such a way that %
\(\Tsempcoh{\Tdiffm d\Tnat}=\Sdfun^d\Snat\), %
\(\Tsempcoh{\Timpl AB}=(\Simpl{\Tsempcoh A}{\Tsempcoh B})\).

And with any term \(M\) such that %
\(\Tseq{x_1:A_1,\dots,x_k:A_k}MB\) we can associate the morphism %
\(\Psempcoh M\Gamma\in\Kl\PCOH(\Bwith_{i=1}^k\Tsempcoh{A_i},\Tsempcoh
B)\) and this interpretation satisfies that if %
\(M\Rel\Red M'\) then %
\(\Psempcoh M\Gamma=\Psempcoh{M'}\Gamma\).
Remember that the reduction relation \(\Red\) can be extended into the
reduction relation %
\(\Rsred{\Msrs\Lang}\) on \(\PCOH\)-summable multisets of terms, that
is to multisets \(S=\Mset{\List M1n}\) such that
\((\Tseq\Gamma{M_j}B)_{j=1}^n\) and %
\( \Psem S\Gamma=\sum_{j=1}^n\Psempcoh{M_j}\Gamma
\in\Kl\PCOH(\Bwith_{i=1}^k\Tsempcoh{A_i},\Tsempcoh B) \) and this
extended relation satisfies %
\(S\Rel{\Rsred{\Msrs\Lang}}S'\Implies\Psempcoh
S\Gamma=\Psempcoh{S'}\Gamma\).

\subsection{A forgetful functor} %
\label{sec:pcoh-rel-forget}
Given \(s\in\PCOH(X,Y)\), we set
\(
\Qforg s=\Eset{(a,b)\in\Web X\times\Web Y\St s_{a,b}\not=0}
\in\REL(\Web X,\Web Y)
\).

\begin{theorem}
  The operation \(\Qforg\) extended to objects by %
  \(\Qforg X=\Web X\) is a functor \(\PCOH\to\REL\) which preserves
  all the structure of model of \LL{}.
\end{theorem}
\begin{proof}
  This is essentially trivial. Let us prove functoriality: let
  \(s\in\PCOH(X,Y)\) and \(t\in\PCOH(Y,Z)\). If
  \((a,c)\in\Qforg(\Matapp ts)\) then
  \(\sum_{b\in\Web Y}s_{a,b}t_{b,c}\not=0\) and so there must be
  \(b\in\Web Y\) such that \(s_{a,b}\not=0\) and \(t_{b,c}\not=0\) and
  hence \((a,c)\in\Matapp{(\Qforg t)}{(\Qforg s)}\). %
  Conversely assume that \((a,c)\in\Matapp{(\Qforg t)}{(\Qforg s)}\)
  and let \(b\in\Web Y\) be such that \(s_{a,b}\not=0\) and
  \(t_{b,c}\not=0\). Since all coefficients are non-negative we have
  \(\sum_{b'\in\Web Y}s_{a,b'}t_{b',c}\geq s_{a,b}t_{b,c}>0\) and
  hence \((a,c)\in\Qforg(\Matapp ts)\).

  As another example, let us prove that if \(s\in\PCOH(X,Y)\) then %
  \(\Qforg(\Excl s)=\Exclp{\Qforg s}\). Let
  \((m,p)\in\Qforg(\Excl s)\), so let \(r\in\Mstrans mp\) be such that
  \(s^r=\prod_{(a,b)\in\Web X\times\Web
    Y}s_{a,b}^{r(a,b)}\not=0\). This implies that
  \(\Supp r\subseteq\Qforg s\) and hence we have
  \(r\in\Exclp{\Qforg s}\). %
  Conversely if \((m,p)\in\Exclp{\Qforg s}\) we can write %
  \(m=\Mset{\List a1n}\) and \(p=\Mset{\List b1n}\) in such a way that %
  \(((a_i,b_i)\in\Qforg s)_{i=1}^k\) which means that %
  \(\prod_{i=1}^ks_{a_i,b_i}\not=0\). Now setting %
  \(r=\Mset{(a_1,b_1),\dots,(a_k,b_k)}\) we have %
  \(r\in\Mstrans mp\) and \(s^r\not=0\) and hence %
  \((m,p)\in\Qforg(\Excl s)\).

  To prove that \(\Qforg\) applied to the \(\Sfix\) operator of
  \(\PCOH\) yields the \(\Sfix\) operator of \(\REL\), it suffices to
  observe that the functor \(\Qforg\) is locally continuous.
\end{proof}

\begin{theorem}%
  \label{cor:psem-qforg}
  For any type \(A\) we have \(\Qforg\Tsempcoh A=\Tsemrel A\) and for
  any term \(M\) such that %
  \(\Tseq\Gamma MB\) we have
  \(\Qforg\Psempcoh M\Gamma=\Psemrel M\Gamma\), and similarly for
  stacks and states.
\end{theorem}

\begin{theorem}
  For any term \(M\) such that \(\Tseq{(x_i:A_i)_{1=1}^k} MB\) all
  coefficients of the matrix
  \(\Psempcoh M\Gamma\in\PCOH(\Bwith_{i=1}^k\Tsempcoh{A_i},\Tsempcoh B)\)
  belong to \(\Nat\). The same holds for stacks and states.
\end{theorem}
\begin{proof}
  It suffices to observe that if \(t\in\PCOH(X,Y)\) belongs to
  \(\Nat^{\Web X\times\Web Y}\), then
  \(\Excl t\in\Nat^{\Web{\Excl X}\times\Web{\Excl Y}}\) and that
  \(\Sdiff_X\in\Nat^{\Web{\Excl{\Sfun X}}\times\Web{\Sfun{\Excl X}}}\).
\end{proof}

\begin{remark}
  Of course the above property is lost if we extend the \(\Lang\) with
  probabilistic choice, which is perfectly possible and compatible
  with the \(\PCOH\) semantics.
\end{remark}

\begin{theorem}%
  \label{cor:psem-pcoh-int}
  If \(\Mset{\List c1n}\) is a \(\PCOH\)-summable multiset of
  states, we have only the two following possibilities:
  \begin{itemize}
  \item either \(\Psempcoh{c_i}{}=0\) for all \(i=1,\dots,n\)
  \item or there is exactly one \(i\in\Eset{1,\dots,n}\) and one
    \(\nu\in\Nat\) such that \(\Psempcoh{c_i}{}=\Base\nu\), and
    \(\Psempcoh{c_j}{}=0\) for \(j\not=i\).
  \end{itemize}
\end{theorem}
\begin{proof}
  Observe that if \(u\in\Pcoh\Snat\) belongs to \(\Nat^\Nat\) then we
  have either \(u=0\) of \(u=\Base\nu\) for a uniquely determined
  \(\nu\in\Nat\).
\end{proof}

\begin{theorem}%
  \label{cor:sem-adequacy-det}
  Let \(M\) be a closed term and let \(\nu\in\Nat\) be such that
  \(\Tseqi{}M\nu\Tnat\). %
  Then %
  \(
  \Mset{\State\Pempty M\Stempty}
  \Rel{\Trcl{\Rsred{\Msrs\States}}}
  C=C_0+\Mset{\State\Pempty{\Num\nu}\Stempty}
  \) %
  for some multiset of well typed states \(C_0\) such that \(C\) is
  \(\cL\)-summable in any model \(\cL\), and all the elements \(c\) of
  \(C_0\) satisfy \(\Psemrel c{}=\emptyset\).
\end{theorem}
\begin{proof}
  Using Theorem~\ref{cor:sem-adequacy} we have %
  \( \Mset{\State\Pempty M\Stempty} \Rel{\Trcl{\Rsred{\Msrs\States}}}
  C=C_0+\Mset{\State\Pempty{\Num\nu}\Stempty} \) %
  and by Theorem~\ref{cor:K-mset-summability-pres} we know that %
  \(C\) is \(\PCOH\)-summable and hence by
  Theorem~\ref{cor:psem-pcoh-int} we know that all the elements
  \(c\) of \(C_0\) satisfy \(\Psempcoh c{}=0\) and hence
  \(\Psemrel c{}=\emptyset\) by Theorem~\ref{cor:psem-qforg}.
\end{proof}

\subsection{A deterministic machine} %
\label{sec:det-machine}
We slightly modify the Krivine machine presented in
Section~\ref{sec:diff-machine} so as to make it fully deterministic.
Guillaume Geoffroy must be credited for the key idea of this
determinization which consists in making the access word component of
the Krivine machine writable. To this end we introduce the set %
\(\Intow=\Into\uplus\Eset{\Wcell n\St n\in\Nat}\) where for each
\(n\in\Nat\) the symbol \(\Wcell n\) represents a %
``writable cell of name \(n\)''.

A state of the deterministic machine is \(0\) or a tuple
\begin{align*}
  \gamma=\Statew\zeta kM\sigma
\end{align*}
where \(\zeta\in\Intow^d\) for some \(d\in\Nat\), \(k\in\Nat\), \(M\)
is a \(\Lang\)-term and \(\sigma\) is a stack (defined exactly as in
Section~\ref{sec:diff-machine} apart that words are now taken in
\({\Intow}^{<\omega}\)).  We use \(\Swcell\gamma\) for the set of all
\(n\in\Nat\) such that \(\Wcell n\) occurs in \(\zeta\) or \(\sigma\).
When \(k\) is \(>\) than all the elements of \(\Swcell\gamma\), we say
that \(\gamma\) is \emph{well-formed}.

We use letters \(u,v\) to denote general
elements of \(\Intow\), letters \(i,j\) for elements of
\(\Into\subset\Intow\) and \(\Statesw\) for the set of states of this
new machine. The main novelty is that that the \(+\) operation on
states is no more required. The second component of a state of the
machine is used as a \(\mathtt{gensym}\) for creating cell names
on request, which are fresh by the well-formedness condition.

The typing rules for stacks and states are exactly the same as in
Figure~\ref{fig:stack-type}, replacing \(\Into\) with \(\Intow\). Most
transition rules are the same as in
Figure~\ref{fig:state-reduction}. The modified transition rules are
given in Figure~\ref{fig:state-det-reduction}. Notice that the rule
producing a \(0\) is still present
(\(\State{\eta i\zeta}{\Linjd{1-i}dM}{\sigma}\Stred0\) for
\(i\in\Into\) and \(d=\Len\zeta\)) but that the rule producing a sum
has been replaced by the deterministic rule~\Eqref{eq:stredw-lsumd-1}
so that the machine \(\Statesw\) equipped with the transitions
\(\Stredw\) is fully deterministic. In rule~\Eqref{eq:stredw-linjd1},
namely
\(\Statew{\eta\Wcell n\zeta}{k}{\Linjd 1dM}{\sigma}
\Subst{\Stredw\Statew{\eta\zeta}{k}{M}{\sigma}}{0}{\Wcell n}\), notice
that the cell \(\Wcell n\) can occur in \(\eta\zeta\) and in
\(\sigma\). The intuition behind this rule is that we give the value
\(1\) to the occurrence of \(\Wcell n\) singled out in its left
member, so we know that all the other occurrences must take value
\(0\), see Remark~\ref{rem:dwords-intuition} below.
Rule~\Eqref{eq:stredw-linjd0} means intuitively that the occurrence of
\(\Wcell n\) singled out in its left member receives value \(0\) and
we do not know for the others; what we know however is that exactly
one will receive value \(1\), see Remark~\ref{rem:dwords-intuition}
below, so at least one occurrence of \(\Wcell n\) must be present in
\(\gamma\).
This explains why the proviso in the application of this rule makes
sense; it will be essential in the proof of
Lemma~\ref{lemma:Strew-co-simulation}.

Notice that if \(\gamma\Stredw\gamma'\) and \(\gamma\) is well-formed
then \(\gamma'\) is also well-formed thanks to
rule~\Eqref{eq:stredw-lsumd-1} which is the only one which introduces
a new writable cell: it advances the gensym counter by \(1\).

\begin{figure}
  \centering
  \begin{align}
    \Statew{\zeta u}k{\Ldiff M}{\sigma}
    &\Stredw\Statew\zeta kM{\Stdiff u\sigma}\\
    \Statew{\zeta}{k}{\Abst xAM}{\Stdiff u\sigma}
    &\Stredw\Statew{\zeta u}{k}{\Abst x{\Tdiff A}{\Ldletv xM}}{\sigma}\\
    \Statew{\eta\Wcell n\zeta}{k}{\Linjd 0dM}{s}
    &\Stredw\gamma=\Statew{\eta\zeta}{k}{M}{\sigma}
    \text{\quad if }n\in\Swcell\gamma\label{eq:stredw-linjd0}\\
    \Statew{\eta\Wcell n\zeta}{k}{\Linjd 1dM}{\sigma}
    &\Subst{\Stredw\Statew{\eta\zeta}{k}{M}{\sigma}}{0}{\Wcell n}
    \label{eq:stredw-linjd1}\\
    \Statew{\eta\Wcell n\zeta}k{\Lsumd dM}{\sigma}
    &\Stredw\Statew{\eta\Wcell n\Wcell n\zeta}{k}{M}{\sigma}\\
    \Statew{\eta1\zeta}k{\Lsumd dM}{\sigma}
    &\Stredw\Statew{\eta\Wcell k\Wcell k\zeta}{k+1}{M}{\sigma}
      \label{eq:stredw-lsumd-1}
  \end{align}
  \caption{Deterministic reduction rules, with the convention that
    \(d=\Len\zeta\)}
  \label{fig:state-det-reduction}
\end{figure}

Given \(\zeta\in\Intow^{<\omega}\) we use %
\(\Dwordl\zeta\) for the set of all \(\delta\in\Into^{\Len\zeta}\)
such that, for all \(l\in\Eset{1,\dots,\Len\zeta}\),
\(\zeta_l\in\Into\Implies\delta_l=\zeta_l\).
For a stack \(\sigma\) of \(\Statesw\) we define a set of stacks
\(\Dwordl\sigma\) of \(\Rsca\States\) in
Figure~\ref{fig:stack-shape}. Last given a \(\Statesw\)-state
\(\gamma=\Statew\zeta kM\sigma\) we define a set of
\(\Rsca\States\)-states \(\Dwordl\gamma\) by %
\[
  \Dwordl\gamma=\Eset{\State\delta Ms\St\delta\in\Dwordl\zeta\text{
      and }s\in\Dwordl\sigma}\,.
\]
\begin{figure}
  \centering
  \begin{align*}
    \Dwordl{\Starg M\sigma}
    &=\{\Starg Ms\St s\in\Dwordl\sigma\}\\
    \Dwordl{\Stsucc\sigma}
    &=\{\Stsucc s\St s\in\Dwordl\sigma\}\\
    \Dwordl{\Stpred\sigma}
    &=\{\Stpred s\St s\in\Dwordl\sigma\}\\
    \Dwordl{\Stif\zeta{M_1}{M_2}\sigma}
    &=\Eset{\Stif\delta{M_1}{M_2}s
      \St \delta\in\Dwordl\zeta\text{ and }s\in\Dwordl\sigma}\\
    \Dwordl{\Stlet\zeta{x}{M}\sigma}
    &=\Eset{\Stlet\delta{x}{M}s
      \St \delta\in\Dwordl\zeta\text{ and }s\in\Dwordl\sigma}\\
    \Dwordl{\Stdiff{i}\sigma}
    &=\Eset{\Stdiff is\St s\in\Dwordl\sigma}\\
    \Dwordl{\Stdiff{\Wcell n}\sigma}
    &=\Eset{\Stdiff is\St i\in\Into\text{ and }s\in\Dwordl\sigma}
  \end{align*}
  \caption{The set of \(\Rsca\States\)-stacks associated with a
    \(\Statesw\)-stack }
  \label{fig:stack-shape}
\end{figure}

Next given \(n\in\Nat\) and \(\delta\in\Dwordl\zeta\) we define %
\(\Nwcell n\delta\zeta\in\Nat\), which is the sum of all values taken
by the cell \(\Wcell n\) in \(\delta\), by
\begin{align*}
  \Nwcell n{\Wempty}{\Wempty}
  &=0\\
  \Nwcell n{\delta i}{\zeta i}
  &=\Nwcell n\delta\zeta\\
  \Nwcell n{\delta i}{\zeta\Wcell k}
  &=\Kronecker nki+\Nwcell n\delta\zeta
    \text{\quad where }\Kronecker nk\text{ is the Kronecker symbol.}
\end{align*}
Given \(s\in\Dwordl\sigma\) we define \(\Nwcell ns\sigma\in\Nat\) as
follows.
\begin{align*}
  \Nwcell n{\Starg Ms}{\Starg M\sigma}
  &=\Nwcell ns\sigma\\
  \Nwcell n{\Stsucc s}{\Stsucc \sigma}
  &=\Nwcell ns\sigma\\
  \Nwcell n{\Stpred s}{\Stpred \sigma}
  &=\Nwcell ns\sigma\\
  \Nwcell n{\Stif\delta{M_1}{M_2}s}{\Stif\zeta{M_1}{M_2}\sigma}
  &=\Nwcell n\delta\zeta+\Nwcell ns\sigma\\
  \Nwcell n{\Stlet\delta{x}{M}s}{\Stlet\zeta{x}{M}\sigma}
  &=\Nwcell n\delta\zeta+\Nwcell ns\sigma\\
  \Nwcell n{\Stdiff is}{\Stdiff i\sigma}
  &=\Nwcell ns\sigma\\
  \Nwcell n{\Stdiff is}{\Stdiff{\Wcell k}\sigma}
  &=\Kronecker kni+\Nwcell ns\sigma\,.
\end{align*}
Last, for states \(\gamma=\Statew\zeta kM\sigma\) and
\(c=\State\delta Ms\) such that \(c\in\Dwordl\gamma\) we set %
\begin{align*}
  \Nwcell nc\gamma=\Nwcell n\delta\zeta+\Nwcell ns\sigma\,.
\end{align*}
Finally we define
\begin{align*}
  \Dwords\gamma
  =\Eset{c\in\Dwordl\gamma
  \St\forall n\in\Swcell\gamma\quad\Nwcell nc\gamma=1}\,.
\end{align*}

\begin{remark} %
  \label{rem:dwords-intuition}
  Intuitively \(\Dwords\gamma\) is the set of all states of the
  machine \(\Rsca\States\) obtained from \(\gamma\in\Statesw\) as
  follows:
  for each \(n\) such that \(\Wcell n\) occurs at some place in
  \(\gamma\), we choose one occurrence of \(\Wcell n\) and replace it
  with \(1\), and we replace all the other occurrences of \(\Wcell n\)
  with \(0\)'s.
\end{remark}

Let \(\Strednd\) be the transition relation on \(\Statessi\), the set
of all \Simplicit{} elements of \(\Rsca\States\), defined as follows:
\begin{align*}
  c\Strednd c'\text{\quad if\quad}
  \begin{cases}
    & c\Stred c'\\
    \text{or }&c\Stred c'+c''\\
    \text{or }&c\Stred c''+c'
  \end{cases}
\end{align*}
so that now
\begin{align*}
  \State{\epsilon1\delta}{\Lsumd dM}{s}
  &\Strednd\State{\epsilon10\delta}{M}{s}\\
  \State{\epsilon1\delta}{\Lsumd dM}{s}
  &\Strednd\State{\epsilon01\delta}{M}{s}
\end{align*}
and \(\Strednd\) is defined exactly as \(\Stred\) in the other cases.

\begin{lemma} %
  \label{lemma:statew-zero-subst}
  Assume that \(c\in\Dwordl\gamma\) and that %
  \(\Nwcell nc\gamma=0\) for some \(n\in\Nat\). Then %
  \(c\in\Dwordl{\Subst\gamma 0{\Wcell n}}\).
\end{lemma}
\begin{proof}[Sketch of the proof]
  By the assumption that \(\Nwcell nc\gamma=0\) we know that at all
  the places in \(c\) corresponding to occurrences of \(\Wcell n\) in
  \(\gamma\) we have the value \(0\). The conclusion follows readily.
\end{proof}

\begin{lemma} %
  \label{lemma:Strew-simulation}
  Let \(\gamma\in\Statesw\) and \(c\in\Dwords\gamma\).  If
  \(c\Strednd c'\not=0\) then %
  \(\gamma\Stredw\gamma'\) for \(\gamma'\) such that
  \(c'\in\Dwords{\gamma'}\).
\end{lemma}
\begin{proof} %
  We use our convention that \(d=\Len\delta=\Len\zeta\).  Most cases
  are straightforward, we deal first with one of these to illustrate its
  triviality.

  \Proofcase %
  \(c=\State{\epsilon\delta}{\Lif dM{N_1}{N_2}}{s}\) and %
  \(c'=\State{\delta}{M}{\Stif\epsilon{N_1}{N_2}s}\). %
  Since \(c\in\Dwords\gamma\) we can write %
  \(\gamma=\Statew{\eta\zeta}{k}{\Lif dM{N_1}{N_2}}{\sigma}\) and we
  have %
  \(\gamma\Stredw\gamma'=\Statew\zeta
  k{M}{\Stif\eta{N_1}{N_2}{\sigma}}\)
  so that \(c'\in\Dwords{\gamma'}\) because for each %
  \(n\in\Swcell\gamma=\Swcell{\gamma'}\) we have %
  \(\Nwcell n{c'}{\gamma'}=\Nwcell{n'}{c}{\gamma}=1\).

  Next we consider the more interesting cases.
    
  \Proofcase %
  \(c=\State{\delta i}{\Ldiff M}{s}\) and
  \(c'=\State{\delta}{M}{\Stdiff is}\). %
  Then \(\gamma=\Statew{\zeta u}k{\Ldiff M}{\sigma}\) and we have
  \(\gamma\Stredw\gamma'=\Statew\zeta Mk{\Stdiff u\sigma}\) and hence
  \(c'\in\Dwords{\gamma'}\). To be more explicit we should consider
  the two possible cases for \(u\).
  \begin{itemize}
  \item If \(u\in\Into\) we have \(u=i\) because \(c\in\Dwordl\gamma\)
    and hence \(\gamma'=\Statew\zeta kM{\Stdiff i\sigma}\). %
    For each \(n\in\Swcell\gamma=\Swcell{\gamma'}\) we have %
    \(\Nwcell n{c'}{\gamma'}=\Nwcell
    n\delta\zeta+\Nwcell ns\sigma=\Nwcell n{c}{\gamma}=1\) so that %
    \(c'\in\Dwords{\gamma'}\).
  \item If \(u=\Wcell n\) for some \(n\in\Nat\) we know that %
    \(\Nwcell{n'}c\delta
    =\Kronecker{n'}ni+\Nwcell{n'}\delta\zeta+\Nwcell {n'}s\sigma=1\)
    for each \(n'\in\Swcell\gamma=\Swcell{\gamma'}\) and since we also
    have %
    \(\Nwcell{n'}{c'}{\gamma'}=\Kronecker{n'}ni+\Nwcell
    {n'}\delta\zeta+\Nwcell{n'}s\sigma\) it follows that %
    \(c'\in\Dwords{\gamma'}\).
  \end{itemize}

  \Proofcase %
  \(c=\State\delta{\Abst xAM}{\Stdiff is}\) and %
  \(c'=\State{\delta i}{\Abst x{\Tdiff A}{\Ldletv xM}}{s}\). %
  Then \(\gamma=\Statew{\zeta}{k}{\Abst xAM}{\Stdiff u\sigma}\) %
  and we have
  \(\gamma\Stredw\gamma' =\Statew{\zeta u}k{\Abst x{\Tdiff A}{\Ldletv
      xM}}\sigma\) so that %
  \(c'\in\Dwordl{\gamma'}\). We check easily as above that %
  \(\Nwcell n{c'}{\gamma'}=\Nwcell n{c}{\gamma}\) for all
  \(n\in\Swcell\gamma=\Swcell{\gamma'}\) and hence
  \(c'\in\Dwords{\gamma'}\).

  \Proofcase %
  \(c=\State{\epsilon i\delta}{\Linjd idM}s\) and %
  \(c'=\State{\epsilon\delta}{M}{s}\) with \(i=0\).  Assume that
  \(\gamma=\Statew{\eta\Wcell n\zeta}{k}{\Linjd 0dM}{\sigma}\) %
  and let \(\gamma'=\Statew{\eta\zeta}{k}{M}{\sigma}\).  Clearly
  \(c'\in\Dwordl{\gamma'}\). Moreover
  \(\Nwcell{n'}{c'}{\gamma'}=\Nwcell{n'}{c}{\gamma}=1\) for all
  \(n'\in\Swcell{\gamma'}\subseteq\Swcell{\gamma}\), including when
  \(n'=n\) in which case we use the fact that \(i=0\) ---~assuming that
  \(n\in\Swcell{\gamma'}\) which is not only possible but necessary as
  we see now~---. Since \(\Nwcell nc\gamma=1\) and \(i=0\), \(\Wcell n\)
  must occur in \(\eta\), \(\zeta\) or \(\sigma\) so that
  \(n\in\Swcell{\gamma'}\) and hence \(\gamma\Stredw\gamma'\) (see in
  Figure~\ref{fig:state-det-reduction} the restrictive condition for
  this reduction~\Eqref{eq:stredw-linjd0}).

  The case where \(\gamma=\Statew{\eta0\zeta}{k}{\Linjd 0dM}{\sigma}\)
  is dealt with straightforwardly.

  \Proofcase %
  \(c=\State{\epsilon i\delta}{\Linjd idM}s\) and %
  \(c'=\State{\epsilon\delta}{M}{s}\) with \(i=1\).
  Assume that
  \(\gamma=\Statew{\eta\Wcell n\zeta}{k}{\Linjd 1dM}{\sigma}\) %
  so that we have 
  \(\gamma\Stredw\gamma'
  =\Subst{\Statew{\eta\zeta}{k}{M}{\sigma}}{0}{\Wcell n}\) and we
  clearly have \(c'\in\Dwordl{\gamma'}\). %
  Since \(i=1\), for \(n'\in\Swcell\gamma\) we have %
  \(\Nwcell{n'}{c}{\gamma}
  =\Kronecker{n}{n'}+\Nwcell{n'}{\epsilon\delta}{\eta\zeta}
  +\Nwcell{n'}{s}{\sigma}\)
  and we know that \(\Nwcell{n'}{c}{\gamma}=1\).
  It follows that %
  \(\Nwcell{n'}{c'}{\gamma'}=\Nwcell{n'}{\epsilon\delta}{\eta\zeta}
  +\Nwcell{n'}{s}{\sigma}=1\) for all
  \(n'\in\Swcell{\gamma'}=\Swcell\gamma\setminus\Eset n\) and moreover
  that
  \(\Nwcell{n}{c'}{\gamma'} =\Nwcell{n}{\epsilon\delta}{\eta\zeta}
  +\Nwcell{n}{s}{\sigma}=0\). %
  By Lemma~\ref{lemma:statew-zero-subst} it follows that %
  \(c'\in\Dwordl{\gamma'}\), and hence that \(c'\in\Dwords{\gamma'}\)
  as expected.

  The case where \(\gamma=\Statew{\eta1\zeta}{k}{\Linjd 0dM}{\sigma}\)
  is dealt with straightforwardly.

  \Proofcase %
  \(c=\State{\epsilon1\delta}{\Lsumd dM}{s}\) and %
  \(c'=\State{\epsilon i_1i_2\delta}{M}{s}\) for some
  \(i_1,i_2\in\Into\) such that \(i_1+i_2=1\).

  Assume first that
  \(\gamma=\Statew{\eta\Wcell n\zeta}{k}{\Lsumd dM}{\sigma}\). %
  Then 
  \(\gamma'=\Statew{\eta\Wcell n\Wcell n\delta}{k}{M}{\sigma}\).
  We clearly have \(c'\in\Dwordl{\gamma'}\), %
  and if \(n'\in\Swcell\gamma=\Swcell{\gamma'}\) we have
  \begin{align*}
    \Nwcell{n'}{c'}{\gamma'}
    &=\Nwcell{n'}{\epsilon\delta}{\eta\zeta}+\Kronecker{n'}{n}(i_1+i_2)
      +\Nwcell{n'}{s}{\sigma}\\
    &=\Nwcell{n'}{\epsilon\delta}{\eta\zeta}+\Kronecker{n'}{n}
      +\Nwcell{n'}{s}{\sigma}\\
    &=\Nwcell{n'}{c}{\gamma}=1.
  \end{align*}

  Assume next that
  \(\gamma=\Statew{\eta1\zeta}{k}{\Lsumd dM}{\sigma}\). %
  Then
  \(\gamma\Stredw\gamma' =\Statew{\eta\Wcell k\Wcell
    k\zeta}{k+1}{M}{\sigma}\) so that we clearly have %
  \(c'\in\Dwordl{\gamma'}\). %
  For \(n'\in\Swcell{\gamma'}=\Swcell\gamma\uplus\Eset k\) %
  (since \(\gamma\) is well-formed) we have
  \begin{align*}
    \Nwcell{n'}{c'}{\gamma'}
    =\Nwcell{n'}{\epsilon\delta}{\eta\zeta}+\Nwcell{n'}{s}{\sigma}
    =\Nwcell{n'}{c}{\gamma}=1\text{ if }n'\not=k
  \end{align*}
  and
  \begin{align*}
    \Nwcell{n'}{c'}{\gamma'}
    =\Nwcell{n'}{\epsilon\delta}{\eta\zeta}+i_1+i_2+\Nwcell{n'}{s}{\sigma}
    =0+1+0\text{ if }n'=k
  \end{align*}
  since \(i_1+i_2=1\) and by our assumption that \(\gamma\) is
  well-formed which implies that \(\Nwcell kc\gamma=0\).

  The case where \(c=\State{\epsilon0\delta}{\Lsumd dM}{s}\), so that %
  \(c'=\State{\epsilon00\delta}{M}{s}\), is easy.
\end{proof}

\begin{lemma}
  \label{lemma:Strew-co-simulation}
  Let \(\gamma,\gamma'\in\Statesw\) with \(\gamma\Stredw\gamma'\).  If
  \(c'\in\Dwords{\gamma'}\) then there is \(c\in\Dwords\gamma\) such
  that \(c\Strednd c'\).
\end{lemma}
\begin{proof}
  We consider only a few cases, the other ones being straightforward.

  \Proofcase %
  \(\gamma=\Statew{\zeta u}k{\Ldiff M}{\sigma}\) and
  \(\gamma'=\Statew\zeta Mk{\Stdiff u\sigma}\). %
  Then we have \(c'=\State{\delta}{M}{\Stdiff is}\) for some
  \(i\in\Into\) (with \(i=u\) is \(u\in\Into\)) and then
  \(c=\State{\delta i}{\Ldiff M}{s}\) satisfies the required
  conditions.

  \Proofcase %
  \(\gamma=\Statew{\zeta}{k}{\Abst xAM}{\Stdiff u\sigma}\) and
  \(\gamma' =\Statew{\zeta u}k{\Abst x{\Tdiff A}{\Ldletv
      xM}}\sigma\). %
  Then we have %
  \(c'=\State{\delta i}{\Abst x{\Tdiff A}{\Ldletv xM}}{s}\) for some
  \(i\in\Into\) (with \(i=u\) if \(u\in\Into\)) and then
  \(c=\State\delta{\Abst xAM}{\Stdiff is}\) satisfies the required
  conditions.

  \Proofcase %
  \(\gamma=\Statew{\eta\Wcell n\zeta}{k}{\Linjd 0dM}{\sigma}\) %
  and \(\gamma'=\Statew{\eta\zeta}{k}{M}{\sigma}\) applying
  rule~\Eqref{eq:stredw-linjd0}, so that we know that
  \(n\in\Swcell{\gamma'}\).
  Then we have \(c'=\State{\epsilon\delta}{M}{s}\).
  Taking \(c=\State{\epsilon 0\delta}{\Linjd 0dM}s\) we have
  \(c\Strednd c'\) and \(c\in\Dwords\gamma\) since for all
  \(n'\in\Swcell\gamma=\Swcell{\gamma'}\) we have
  \(\Nwcell{n'}{c}{\gamma}=\Nwcell{n'}{c'}{\gamma'}=1\).
  It is here that the restrictive condition in the
  rule~\Eqref{eq:stredw-linjd0} of
  Figure~\ref{fig:state-det-reduction} is essential.

  \Proofcase %
  \(\gamma=\Statew{\eta\Wcell n\zeta}{k}{\Linjd 1dM}{\sigma}\) %
  and
  \(\gamma'=\Subst{\Statew{\eta\zeta}{k}{M}{\sigma}}{0}{\Wcell n}\). %
  Then we have \(c'=\State{\epsilon\delta}{M}{s}\) and we take
  \(c=\State{\epsilon 1\delta}{\Linjd 1dM}s\) so that %
  \(c\Strednd c'\). For
  \(n'\in\Swcell{\gamma}=\Swcell{\gamma'}\uplus\Eset n\) we have %
  \(\Nwcell{n'}c\gamma=\Nwcell{n'}{c'}{\gamma'}=1\) if
  \(n'\not=n\). Since %
  \(c'\in\Dwordl{\Subst{\Statew{\eta\zeta}{k}{M}{\sigma}}{0}{\Wcell
      n}}\) we have
  \(\Nwcell n{c'}{\Statew{\eta\zeta}{k}{M}{\sigma}}=0\) (that is, all
  the occurrences of \(\Wcell n\) in
  \(\Statew{\eta\zeta}{k}{M}{\sigma}\) are filled with \(0\)'s in
  \(c'\)) and hence \(\Nwcell nc\gamma=1\).

  \Proofcase %
  \(\gamma=\Statew{\eta\Wcell n\zeta}{k}{\Lsumd dM}{\sigma}\) and
  \(\gamma'=\Statew{\eta\Wcell n\Wcell n\delta}{k}{M}{\sigma}\).  Then
  \(c'=\State{\epsilon i_1i_2\delta}{M}{s}\) for some
  \(i_1,i_2\in\Into\) such that \(i=i_1+i_2\in\Into\). Let %
  \(c=\State{\epsilon i\delta}{\Lsumd dM}{s}\) so that %
  \(c\Strednd c'\) and \(c\in\Dwordl\gamma\). %
  Then for \(n'\in\Swcell\gamma=\Swcell{\gamma'}\) we have %
  \begin{align*}
    \Nwcell{n'}{c}{\gamma}
    &=\Nwcell{n'}{\epsilon\delta}{\eta\zeta}
      +\Kronecker{n'}{n}i+\Nwcell{n'}{s}{\sigma}\\
    &=\Nwcell{n'}{\epsilon\delta}{\eta\zeta}
      +\Kronecker{n'}{n}i_1+\Kronecker{n'}{n}i_2+\Nwcell{n'}{s}{\sigma}\\
    &=\Nwcell{n'}{c'}{\gamma'}=1
  \end{align*}
  and hence \(c\in\Dwords\gamma\).

  \Proofcase %
  \(\gamma=\Statew{\eta1\zeta}{k}{\Lsumd dM}{\sigma}\) and
  \(\gamma'=\Statew{\eta\Wcell k\Wcell k\delta}{k+1}{M}{\sigma}\).
  Then \(c'=\State{\epsilon i_1i_2\delta}{M}{s}\) for some
  \(i_1,i_2\in\Into\) such that \(i_1+i_2=1\): this is due to the fact
  that we must have \(\Nwcell k{c'}{\gamma'}=1\) and we know that
  \(\Wcell k\) does not occur in \(\eta\), \(\zeta\) and \(\sigma\)
  since \(\gamma\) is well-formed. Let %
  \(c=\State{\epsilon 1\delta}{\Lsumd dM}{s}\) so that %
  \(c\Strednd c'\) and \(c\in\Dwordl\gamma\). %
  Then for \(n'\in\Swcell\gamma=\Swcell{\gamma'}\setminus\Eset k\) we
  have \(\Nwcell{n'}{c}{\gamma}=\Nwcell{n'}{c'}{\gamma'}=1\)
  and hence \(c\in\Dwords\gamma\).
\end{proof}

\begin{theorem}%
  \label{th:det-machine-complete}
  Let \(M\) be a term such that \(\Tseq{}M\Tnat\) and let
  \(\nu\in\Nat\). Then
  \begin{align*}
    [\exists k\in\Nat\ \Statew{\Pempty}{0}{M}{\Stempty}
    \Trcl\Stredw\Statew\Pempty{k}{\Num\nu}\Stempty]
    \Equiv
    \State\Pempty M\Stempty\Trcl\Strednd\State\Pempty{\Num\nu}\Stempty
  \end{align*}
  Moreover when one of these two reduction converges, the other one
  does, with the same number of steps.
\end{theorem}
\begin{proof}
  \Proofcase %
  (\(\Leftarrow\)) By an obvious induction on the length of the
  \(\Strednd\)-reduction using Lemma~\ref{lemma:Strew-simulation} one
  proves the following statement: %
  if \(\State{\Pempty}{M}{\Stempty}\Trcl\Strednd c\not=0\) then %
  \(\Statew{\Pempty}{0}{M}{\Stempty}\Trcl\Stredw\gamma\) with %
  \(c\in\Dwords\gamma\). We apply this statement to the case where %
  \(c=\State\Pempty{\Num\nu}\Stempty\) and obtain that %
  \(\Statew{\Pempty}{0}{M}{\Stempty}\Trcl\Stredw\gamma\) with %
  \(\State\Pempty{\Num\nu}\Stempty\in\Dwords\gamma\), which means that %
  \(\gamma=\Statew\Pempty k{\Num\nu}\Stempty\) for some \(k\in\Nat\).

  \Proofcase %
  (\(\Implies\))  By an obvious induction on the length of the
  \(\Stredw\)-reduction using Lemma~\ref{lemma:Strew-co-simulation}
  one proves the following statement: %
  if \(\gamma\Trcl\Stredw\Statew\Pempty k{\Num\nu}\Stempty\) for some
  \(k\in\Nat\) then there is \(c\in\Dwords\gamma\) such that %
  \(c\Trcl\Strednd\State\Pempty{\Num\nu}\Stempty\). We apply this
  statement to the case where \(\gamma=\Statew\Pempty0M\Stempty\) and
  obtain \(c\) such that
  \(c\Trcl\Strednd\State\Pempty{\Num\nu}\Stempty\) and
  \(c\in\Dwords\gamma\). This latter property means that
  \(c=\State\Pempty M\Stempty\).
\end{proof}

\begin{remark}
  The fact that the lengths of the deterministic \(\Stredw\)-reduction
  and of the non-deterministic \(\Strednd\)-reduction in
  Theorem~\ref{th:det-machine-complete} are equal is of course
  essential since it is always possible to simulate a
  non-deterministic reduction by a deterministic one using
  interleaving techniques. One can use the simulation and
  co-simulation Lemmas~\ref{lemma:Strew-simulation}
  and~\ref{lemma:Strew-co-simulation} for proving various
  generalizations of that theorem, using the fact that if
  \(c\in\Dwords\gamma\) and \(\gamma\) contains no \(\Wcell n\)'s then
  \(c\) and \(\gamma\) are essentially the same thing (the only
  difference is the counter contained in \(\gamma\)).
\end{remark}

\begin{theorem}
  Let \(M\) be a term such that \(\Tseq{}M\Tnat\) and let
  \(\nu\in\Nat\). Then we have \(\Tseqi{}M\nu\Tnat\) iff %
  \(\Statew{\Pempty}{0}{M}{\Stempty}
  \Trcl\Stredw\Statew{\Pempty}{k}{\Num\nu}{\Stempty}\)
  for some \(k\in\Nat\).
\end{theorem}
\begin{proof}
  By Theorems~\ref{th:det-machine-complete}, \ref{th:machine-sound}
  and~\ref{cor:sem-adequacy-det}, observing that %
  \(c\Trcl{\Strednd}c'\) is equivalent to the existence of \(C\) such that %
  \(\Mset c\Trcl{\Rsred{\Msrs\States}}\Mset{c'}+C\).
\end{proof}

\section*{Conclusion}
Building on the categorical axiomatization of coherent differentiation
introduced in~\cite{Ehrhard23a} we have defined a differential
extension \(\Lang\) of the standard Turing complete functional
programming language \PCF{}.

The rewriting system of \(\Lang\) has many reduction rules and
therefore it would be probably rather difficult to prove this
determinism property syntactically (as a Church-Rosser property) so
our use of denotational semantics for this purpose seems really
crucial.
We also use semantics for proving that our rewriting system is
complete in the sense that it allows any closed term of type \(\Tnat\)
whose interpretation in the relational model contains \(\nu\) to be
reduced to \(\Num\nu\).
This proof is based on the use of a Krivine machine which is a way of
extracting from the general \(\Lang\) rewriting system a fairly simple
though sufficiently expressive subsystem.
The completeness proof is based on a reducibility method that we have
been obliged to modify drastically in order to adapt it to this
differential setting.
To make the proof more readable, we present the relational semantics
of \(\Lang\) as a non-idempotent intersection typing system and
completeness can be understood as a normalization property for this
typing system.

One major novelty\footnote{And improvement in some sense.} of coherent
differentiation wrt.~\(\DILL\) is the fact that it is
deterministic.
This was already clear in~\cite{Ehrhard23a} where we showed that
coherent differentiation admits deterministic models, that is models
where the superposition of the values \(\True\) and \(\False\) in the
type \(\Bool\) is rejected.
The present article provides a syntactic evidence of this determinism
via the fully deterministic version of our machine which is shown to
be sound and complete wrt.~the execution of closed \(\Lang\)-terms of
type \(\Tnat\).

\subsection*{Future work}
Our Krivine machine has no environment and uses actual substitutions
in terms for implementing \(\beta\)-reduction, as well as a syntactic
differential operation \(\Ldletv xM\) defined by induction on \(M\) to
implement the differential reduction of \(\Lang\).
From the viewpoint of efficiency this is of course not satisfactory
and we will present in a forthcoming paper a machine using a stack as
well as a (differential) environment not invoking any external
operation defined by induction on terms for executing the expressions
of our language.

The most puzzling questions however remain of a theoretical nature and
concern the exact operational meaning of our language \(\Lang\)%
\footnote{Or of its variants, we can of course expect to design more
  syntactically elegant versions of \(\Lang\) in the next few months;
  the version presented in this paper has been chosen for its
  relatively straightforward denotational semantics.}, %
which has now fully satisfactory deterministic operational and
denotational semantics.
From a programming point of view, what is
exactly the meaning of the type construction \(\Tdiff A\) and in what
kind of programming situation could it be useful, as well as the
syntactic term construction \(\Ldiff M\)?
The few examples provided in Section~\ref{sec:term-examples} do
clearly not answer this question.
One way to address it could be to consider a probabilistic
extension of \(\Lang\), for which differentiation has a clear
mathematical meaning easily expressed in \(\PCOH\) as we have seen in
Section~\ref{sec:PCS-definition}.

Another interesting direction, which might require the extension of
\(\Lang\) with richer types%
\footnote{Most models of \(\LL\) support inductive and coinductive
  definitions so such extensions should not be problematic.}, %
would be to understand if it has connection with incremental
programming where syntactic constructs of a differential nature are
also used. Such a connection remains however highly conjectural.
More specifically, in~\cite{EhrhardLaurent07a,EhrhardLaurent07b}, we
have suggested possible connections between \(\DILL\) and various
process calculi, it might be worthwhile to understand if such
connections could be related to incremental computing and benefit from
coherent differentiation.

Last, as suggested in Section~\ref{sec:term-examples}, the fact that
\(\Lang\) has at the same time a general fixpoint construct and a
differential construct means that it is possible to define programs by
some kind of ``differential equations'' (recursive definitions of
functions whose body contains the possibly higher derivatives of the
functions being defined) and that such programs can be executed in our
Krivine machine(s); this is a very exciting feature of our setting
which justifies investigations \emph{per se}.

\section*{Acknowledgments}
This work was partly supported by the ANR project %
\emph{Probabilistic Programming Semantics (PPS)} ANR-19-CE48-0014. I
am very grateful to Christine Tasson, Michele Pagani, Paul-André
Melliès, Adrien Guatto, Guillaume Geoffroy and Aymeric Walch for many
discussions on this work. In particular I owe to Guillaume Geoffroy
the key idea of the final step in the determinization of the Krivine
machine for \(\Lang\): make the access words writable!

\bibliographystyle{alphaurl}
\bibliography{newbiblio.bib}

\end{document}

%% file: mainnot.tex
\newenvironment{theorem}{\begin{thm}}{\end{thm}}
\newenvironment{lemma}{\begin{lem}}{\end{lem}}

\newenvironment{proposition}{\begin{prop}}{\end{prop}}

\newenvironment{definition}{\begin{defi}}{\end{defi}}

\newenvironment{remark}{\begin{rem}}{\end{rem}}
\newenvironment{example}{\begin{exa}}{\end{exa}}

\newenvironment{Example}{%
   \bigbreak\noindent%
   \refstepcounter{thm}
   \textbf{$\blacktriangleright$\ Example \thethm.\ }%
   }{\ \hfill$\blacktriangleleft$\par\bigbreak}
\numberwithin{examplectr}{section}

\newcommand\Ie{\textit{ie.}}
\newcommand\Etc{\textit{etc.}}

\makeatletter
\newcommand*{\inlineeq}[2][]{%
  \begingroup
    \refstepcounter{equation}%
    \ifx\\#1\\%
    \else
      \label{#1}%
    \fi
    \relpenalty=10000 %
    \binoppenalty=10000 %
    \ensuremath{%
      #2%
    }%
    ~\@eqnnum
  \endgroup
}
\makeatother

\makeatletter
\newcommand{\labeltext}[2]{%
  \@bsphack
  \csname phantomsection\endcsname 
  \def\@currentlabel{#1}{\label{#2}}%
  \@esphack
}
\makeatother

\newenvironment{Axicond}[2]
{\smallbreak\noindent{#1}\labeltext{#1}{#2}\,}
{\smallbreak}

\newcommand\Proofcase{\smallbreak\noindent$\blacktriangleright$\ }

\newcommand{\Endproof}{
  \ifmmode 
  \else \leavevmode\unskip\penalty9999 \hbox{}\nobreak\hfill
  \fi
  \quad\hbox{$\Box$}
  \par\medskip}

\newcommand\Eqref[1]{(\ref{#1})}

\renewcommand{\phi}{\varphi}
\renewcommand\epsilon{\varepsilon}

\newcommand{\Implies}{\Rightarrow}

\newcommand\Equiv{\Leftrightarrow}
\newcommand{\St}{\mid}

\newcommand{\Bott}{{\mathord{\perp}}}
\newcommand{\Top}{\top}

\newcommand\Seqempty{\Tuple{}}

\newcommand\cC{\mathcal{C}}

\newcommand\cL{\mathcal{L}}

\newcommand\cT{\mathcal{T}}

\newcommand\Fini{{\mathrm{fin}}}

\newcommand\Part[1]{{\mathcal P}\left({#1}\right)}

\newcommand\Union{\bigcup}

\newcommand{\Linarrow}{\multimap}

\newcommand\Myleft{}
\newcommand\Myright{}

\newcommand\Web[1]{\Myleft|{#1}\Myright|}

\newcommand\Supp[1]{\operatorname{\mathsf{supp}}({#1})}

\newcommand\Emptymset{[\,]}
\newcommand\Mset[1]{[{#1}]}

\newcommand\ITens{\otimes}
\newcommand\Tens[2]{{#1}\ITens{#2}}
\newcommand\Comma[2]{{#1},{#2}}
\newcommand\Tensp[2]{({#1}\ITens{#2})}
\newcommand\IWith{\mathrel{\&}}
\newcommand\With[2]{{#1}\IWith{#2}}

\newcommand\Withep[2]{\left({#1}\IWith^\Klsymb{#2}\right)}
\newcommand\Withp[2]{\left({#1}\IWith{#2}\right)}
\newcommand\IPlus{\oplus}
\newcommand\Plus[2]{{#1}\IPlus{#2}}
\newcommand\Orth[2][]{#2^{\mathord\perp_{#1}}}

\newcommand\Klsymb{\mathsf K}

\newcommand\Bwith{\mathop{\&}}
\newcommand\Bwithe{\mathop{\&^{\Klsymb}}}
\newcommand\Bplus{\mathop\oplus}
\newcommand\Bunion{\mathop\cup}

\newcommand\Inj[1]{\overline\pi_{#1}}

\newcommand\Biorth[1]{#1^{\bot\bot}}

\newcommand\One{1}

\newcommand\LL{\hbox{\textsf{LL}}}

\newcommand\Card[1]{\#{#1}}

\newcommand\Locun[1]{1^J}

\newcommand\Isom\simeq

\newcommand\NUCS{\mathbf{nCoh}}

\newcommand\Comp{\mathrel\circ}

\newcommand\Funinv[1]{{#1}^{-1}}

\newcommand\Limpl[2]{{#1}\Linarrow{#2}}
\newcommand\Limplp[2]{\left({#1}\Linarrow{#2}\right)}

\newcommand\Nat{{\mathbb{N}}}

\newcommand\Natnz{{\Nat^+}}

\newcommand\Snat{\mathsf N}

\newcommand\Bool{\mathbf{Bool}}
\newcommand\True{\mathbf t}
\newcommand\False{\mathbf f}

\newcommand\Diffsymb{\mathsf D}

\newcommand\Zero{0}
\newcommand\App[2]{({#1})\Appsep{#2}}

\newcommand\Abst[3]{\lambda#1^{#2}\,{#3}}

\newcommand\List[3]{#1_{#2},\dots,#1_{#3}}

\newcommand\Listter[3]{#1_{#2}\cdots #1_{#3}}

\newcommand\Kronecker[2]{\boldsymbol\delta_{{#1},{#2}}}

\newcommand\Subst[3]{{#1}\left[{#2}/{#3}\right]}

\newcommand\Substbis[2]{{#1}[{#2}]}

\newcommand\Factor[1]{{#1}!}

\newcommand\Multinomb[2]{\genfrac{[}{]}{0pt}{}{#1}{#2}}

\newcommand\Real{\mathbb{R}}
\newcommand\Realp{\mathbb{R}_{\geq 0}}
\newcommand\Realpto[1]{(\Realp)^{#1}}

\newcommand\Realpc{\overline{\Realp}}
\newcommand\Realpcto[1]{\Realpc^{#1}}

\newcommand\Linapp[2]{\left\langle{#1}\right\rangle{#2}}

\newcommand\Mfin[1]{\mathcal M_\Fini({#1})}

\newcommand\Ev{\operatorname{\mathsf{Ev}}}
\newcommand\Evlin{\operatorname{\mathsf{ev}}}
\newcommand\REL{\operatorname{\mathbf{Rel}}}

\newcommand\Norm[1]{\|{#1}\|}

\newcommand\Rel[1]{\mathrel{#1}}

\newcommand\Red{\beta_\Delta}
\newcommand\Redst[1]{\mathop{\mathsf{Red}}}

\newcommand\Tofst[1]{\langle#1\rangle}

\newcommand\Tuple[1]{\langle{#1}\rangle}
\newcommand\Cotuple[1]{\left[{#1}\right]}

\newcommand\Msetofsubst[1]{\bar F}

\newcommand\Inv[1]{{#1}^{-1}}
\newcommand\Invp[1]{({#1})^{-1}}

\newcommand\Pcoh[1]{\mathsf P{#1}}
\newcommand\Pcohp[1]{\Pcoh{(#1)}}

\newcommand\Base[1]{\mathsf e_{#1}}
\newcommand\Matapp[2]{{#1}\Compl{#2}}

\newcommand\PCOH{\mathbf{Pcoh}}

\newcommand\Leftu{\lambda}
\newcommand\Rightu{\rho}
\newcommand\Assoc{\alpha}
\newcommand\Sym{\gamma}
\newcommand\Msetempty{\Mset{\,}}

\newcommand\Retri\zeta
\newcommand\Retrp\rho

\newcommand\Impl[2]{{#1}\Rightarrow{#2}}
\newcommand\Implp[2]{({#1}\Rightarrow{#2})}

\newcommand\Tsem[1]{\llbracket{#1}\rrbracket}
\newcommand\Tsemrel[1]{\llbracket{#1}\rrbracket^{\REL}}
\newcommand\Tsempcoh[1]{\llbracket{#1}\rrbracket^{\PCOH}}

\newcommand\Psem[2]{\llbracket{#1}\rrbracket_{#2}}
\newcommand\Psemrel[2]{\llbracket{#1}\rrbracket^{\REL}_{#2}}

\newcommand\Psempcoh[2]{\llbracket{#1}\rrbracket^{\PCOH}_{#2}}

\newcommand\Tnat\iota

\newcommand\Succ[1]{\operatorname{\mathsf{succ}}(#1)}
\newcommand\Num[1]{\underline{#1}}
\newcommand\Loop\Omega

\newcommand\Tseq[3]{{#1}\vdash{#2}:{#3}}

\newcommand\Timpl\Impl
\newcommand\Timplp\Implp
\newcommand\Simpl\Impl
\newcommand\Simplp[2]{(\Impl{#1}{#2})}

\newcommand\PCFD{\Lang}
\newcommand\PCF{\textsf{PCF}}

\newcommand\Weak[1]{\operatorname{\mathsf{weak}}_{#1}}
\newcommand\Contr[1]{\operatorname{\mathsf{contr}}_{#1}}

\newcommand\Der[1]{\operatorname{\mathsf{der}}_{#1}}

\newcommand\Digg[1]{\operatorname{\mathsf{dig}}_{#1}}

\newcommand\Fun[1]{\widehat{#1}}

\newcommand\Id{\operatorname{\mathsf{Id}}}

\newcommand\Proj[1]{\mathsf{pr}_{#1}}
\newcommand\Proje[1]{\mathsf{pr}^{\Klsymb}_{#1}}

\newcommand\Excl[1]{\oc{#1}}
\newcommand\Excll[1]{\oc\oc{#1}}
\newcommand\Exclp[1]{\oc({#1})}

\newcommand\Prom[1]{#1^!}
\newcommand\Promp[1]{(#1)^!}
\newcommand\Promm[1]{{#1}^{!!}}
\newcommand\Prommm[1]{{#1}^{!!!}}

\newcommand\Relincl\eta
\newcommand\Relrestr\rho

\newcommand\Seely{\mathsf m}
\newcommand\Seelyz{\Seely^0}
\newcommand\Seelyt{\Seely^2}

\newcommand\Compl{\,}
\newcommand\Curlin{\operatorname{\mathsf{cur}}}
\newcommand\Curlinp[1]{\Curlin(#1)}
\newcommand\Cur{\operatorname{\mathsf{Cur}}}

\newcommand\Kl[1]{{#1}_\oc}
\newcommand\Em[1]{{#1}^\oc}

\newcommand\Coalg[1]{h_{#1}}
\newcommand\Coalgw[1]{\mathsf w_{#1}}
\newcommand\Coalgc[1]{\mathsf c_{#1}}

\newcommand\Eval[2]{\langle#1\mid#2\rangle}

\newcommand\Snum[1]{\overline{#1}}
\newcommand\Sif{\overline{\mathsf{if}}}
\newcommand\Ssuc{\overline{\mathsf{suc}}}
\newcommand\Spred{\overline{\mathsf{pred}}}
\newcommand\Szero{\overline{\mathsf{zero}}}

\newcommand\Sfix{\mathcal{Y}}

\newcommand\Vect[1]{\vec{#1}}

\newcommand\Argsep{\,}

\newcommand\Bnfeq{\mathrel{\mathord:\mathord=}}
\newcommand\Bnfor{\,\,\mathord|\,\,}

\newcommand\Len[1]{\mathsf{len}(#1)}

\newcommand\Eset[1]{\{#1\}}

\newcommand\Sfun{\operatorname{\mathsf S}}
\newcommand\Scfun{\Sfun_{\Into}}
\newcommand\Sproj[1]{\operatorname\pi_{#1}}
\newcommand\Sin[1]{\operatorname{\iota}_{#1}}
\newcommand\Ssum{\sigma}
\newcommand\Sflip{\operatorname{\mathsf c}}
\newcommand\Sflipl[1]{\Sflip(#1)}

\newcommand\Stuple[1]{\Tuple{#1}_{\Sfun}}

\newcommand\Smont{\operatorname{\mathsf L}}
\newcommand\Smontz{\Smont^0}
\newcommand\Scmont{\overline{\Smont}}
\newcommand\Sstr{\operatorname\phi}

\newcommand\Sstrc{\Sstr^{\multimap}}
\newcommand\Sdiff{\operatorname\partial}
\newcommand\Sdiffcaold{\delta}
\newcommand\Sdiffca{\tilde\partial}
\newcommand\Sdiffctx[2]{\Sdiff(#1,#2)}

\newcommand\Sdfun{\operatorname{\widetilde{\mathsf D}}}
\newcommand\Sdfunc{\Sdfun_{\mathsf{cur}}} 
\newcommand\Sdfunint{\Sdfun_{\mathsf{int}}} 

\newcommand\Matappa[2]{{#1}\cdot{#2}}

\newcommand\Saxcom{(\textbf{S-com})}
\newcommand\Saxzero{(\textbf{S-zero})}

\newcommand\Saxwit{(\textbf{S-witness})}
\newcommand\Saxdist{(\textbf{S$\ITens$-dist})}
\newcommand\Saxfun{(\textbf{S$\ITens$-fun})}

\newcommand\Sone{\One}
\newcommand\Sonelem{\ast}

\newcommand\Win[1]{\Inj{#1}^{\with}}
\newcommand\Wdiag{\Delta^\with}
\newcommand\Into{\mathbb{D}}
\newcommand\Intoset{\{0,1\}}

\newcommand\Obj[1]{\mathsf{Obj(#1)}}

\newcommand\Intcc[2]{[#1,#2]}

\newcommand\Sfunadd{\tau}

\newcommand\Daxchain{($\partial$\textbf{-chain})}
\newcommand\Daxlocal{($\partial$\textbf{-local})}

\newcommand\Daxlin{($\partial$\textbf{-lin})}
\newcommand\Daxwith{($\partial$\textbf{-}$\mathord{\IWith}$)}

\newcommand\Daxschwarz{($\partial$\textbf{-Schwarz})}

\newcommand\Treesep{\quad\quad}
\newcommand\Textsep{\hspace{6em}}

\newcommand\Sdfunit{\zeta}
\newcommand\Sdfmult{\theta}
\newcommand\Sdfstr{\psi}

\newcommand\Kllin{\mathsf{Lin}_{\mathord\oc}}

\newcommand\Tdiffsymb{\widetilde{\mathsf D}}
\newcommand\Tdiff[1]{\Tdiffsymb{#1}}
\newcommand\Tdiffm[2]{\Tdiffsymb^{#1}{#2}}
\newcommand\Tdnat[1]{\Tdiffsymb^{#1}\Tnat}
\newcommand\Lproj[2]{\Sproj{#1}(#2)}
\newcommand\Lprojd[3]{\Sproj{#1}^{#2}(#3)}

\newcommand\Linj[2]{\Sin{#1}(#2)}
\newcommand\Linjd[3]{\Sin{#1}^{#2}(#3)}
\newcommand\Lin[2]{\Sin{#1}(#2)}
\newcommand\Lsum[1]{\Sdfmult(#1)}
\newcommand\Lsumd[2]{\Sdfmult^{#1}(#2)}
\newcommand\Lflip[1]{\Sflip(#1)}
\newcommand\Lflipd[2]{\Sflip^{#1}(#2)}
\newcommand\Lflipdl[3]{\Sflip^{#1}_{#2}(#3)}
\newcommand\Lflipl[2]{\Sflip_{#1}(#2)}

\newcommand\Ldletv[2]{\partial({#1},{#2})}
\newcommand\Ldletvm[3]{\partial^{#1}({#2},{#3})}
\newcommand\Lzero{0}
\newcommand\Lzerot[1]{0^{#1}}
\newcommand\Lplus[2]{#1+#2}
\newcommand\Ldiff[1]{\Tdiff{#1}}

\newcommand\Ldiffp[1]{\Tdiff({#1})}
\newcommand\Lfix[1]{\mathsf Y#1}
\newcommand\Llet[4]{\mathsf{let}^{#1}(#2,#3,#4)}
\newcommand\Llett[5]{\mathsf{let}_{#1}^{#2}(#3,#4,#5)}
\newcommand\Lif[4]{\mathsf{if}^{#1}(#2,#3,#4)}
\newcommand\Lift[5]{\mathsf{if}_{#1}^{#2}(#3,#4,#5)}
\newcommand\Lpred[2]{\operatorname{\mathsf{pred}}^{#1}(#2)}
\newcommand\Lsucc[2]{\operatorname{\mathsf{succ}}^{#1}(#2)}

\newcommand\Sdfunpart[2]{\Sdfun_{#1}#2}

\renewcommand\Red{\to}

\newcommand\Cpolike{Scott}

\newcommand\Mlin[1]{\mathsf{m}(#1)}

\newcommand\Laxint{(\textbf{Int})}

\newcommand\Snatit[1]{\overline{\mathsf{it}}(#1)}

\newcommand\Slet{\overline{\mathsf{let}}}

\newcommand\Sfunpart[1]{\Sfun_{#1}}

\newcommand\Linred{\to_{\mathsf{lin}}}
\newcommand\Topred{\to_{\mathsf{top}}}

\newcommand\Echole{[\ ]}
\newcommand\Ecfilled[2]{#1[#2]}

\newcommand\Trvar{(\textbf{var})}
\newcommand\Trabs{(\textbf{abs})}
\newcommand\Trapp{(\textbf{app})}
\newcommand\Trfix{(\textbf{fix})}
\newcommand\Trnum{(\textbf{num})}
\newcommand\Trsuc{(\textbf{suc})}
\newcommand\Trpred{(\textbf{prd})}
\newcommand\Trif{(\textbf{if})}
\newcommand\Trzero{(\textbf{0})}

\newcommand\Trprojo{(\textbf{proj1})}
\newcommand\Trprojt{(\textbf{proj2})}
\newcommand\Trprojd{(\textbf{projd})}
\newcommand\Trinj{(\textbf{inj})}
\newcommand\Trsum{(\textbf{sum})}
\newcommand\Trcirc{(\textbf{circ})}
\newcommand\Trdiff{(\textbf{diff})}
\newcommand\Trlet{(\textbf{let})}
\newcommand\Trlin{(\textbf{lin})}

\newcommand\Itrvar{(\textbf{i-var})}
\newcommand\Itrabs{(\textbf{i-abs})}
\newcommand\Itrapp{(\textbf{i-app})}
\newcommand\Itrfix{(\textbf{i-fix})}
\newcommand\Itrnum{(\textbf{i-num})}
\newcommand\Itrsuc{(\textbf{i-suc})}
\newcommand\Itrpredz{(\textbf{i-prd}$0$)}
\newcommand\Itrpredp{(\textbf{i-prd}$+$)}
\newcommand\Itrifz{(\textbf{i-if}$0$)}
\newcommand\Itrifp{(\textbf{i-if}$+$)}

\newcommand\Itrproj{(\textbf{i-proj})}
\newcommand\Itradd{(\textbf{i-add})}

\newcommand\Itrinj{(\textbf{i-inj})}

\newcommand\Itrsum{(\textbf{i-sum})}
\newcommand\Itrcirc{(\textbf{i-circ})}

\newcommand\Itrdiff{(\textbf{i-diff})}
\newcommand\Itrlet{(\textbf{i-let})}

\newcommand\Trcl[1]{#1^{\ast}}

\newcommand\Tdersize[1]{\mathsf{sz}(#1)}
\newcommand\Lcht[1]{\mathsf{lh}(#1)}

\newcommand\Simplicit{sum-implicit}
\newcommand\Sharp{sharp}
\newcommand\Path{word}

\newcommand\Stempty{()}
\newcommand\Starg[2]{\mathsf{arg}(#1)\cdot#2}
\newcommand\Stdiff[2]{\mathsf{D}(#1)\cdot#2}
\newcommand\Stsucc[1]{\mathsf{succ}\cdot#1}
\newcommand\Stpred[1]{\mathsf{pred}\cdot#1}
\newcommand\Stif[4]{\mathsf{if}(#1,#2,#3)\cdot#4}

\newcommand\Stlet[4]{\mathsf{let}(#1,#2,#3)\cdot#4}

\newcommand\Stseq[2]{#1:#2\vdash\Tnat}
\newcommand\Stseqi[4]{#1:#2:#3\vdash#4:\Tnat}

\newcommand\Etseqi[2]{\vdash #1:#2:\Tnat}

\newcommand\Stred{\to}
\newcommand\Stredw{\to_{\mathsf{w}}}

\newcommand\State[3]{(#1,#2,#3)}
\newcommand\Statew[4]{(#1,#2,#3,#4)}

\newcommand\Pempty{\langle\,\rangle}

\newcommand\Stauto{$\ast$-autonomous}

\newcommand\DILL{\mathsf{DiLL}}

\newcommand\Tonat{\Rightarrow}

\newcommand\Monz{\mu^0}
\newcommand\Mont{\mu^2}
\newcommand\Mong[1]{\mu^{#1}}

\newcommand\Tseqact[2]{#1\cdot#2}

\newcommand\Ttermi[3]{#1:#2:#3}
\newcommand\Tseqi[4]{#1\vdash\Ttermi{#2}{#3}{#4}}
\newcommand\Contca[1]{\underline{#1}}
\newcommand\Contz[1]{0_{#1}}

\newcommand\Itintt[2]{|#1|_{#2}}
\newcommand\Itintc[1]{\Bot(#1)}
\newcommand\Itints[3]{\|#1\vdash #3\|_{#2}}
\newcommand\Itintctx[1]{|#1|}

\newcommand\Msrs[1]{\Mfin{#1}}

\newcommand\Rsca[1]{\underline{#1}}
\newcommand\Rsred[1]{\to_{#1}}

\newcommand\Cdsymb{\mathsf{cd}}
\newcommand\Lang{\Lambda_\Cdsymb}
\newcommand\States{\Theta_\Cdsymb}
\newcommand\Statessi{\Theta_\Cdsymb^0}
\newcommand\Statesw{\check\Theta_\Cdsymb}
\newcommand\Strednd{\to_{\mathsf{nd}}}

\newcommand\Stctx[1]{\langle{#1}\rangle}

\newcommand\Rev[1]{\overline{#1}}

\newcommand\Rcycle[1]{\underrightarrow{#1}}
\newcommand\Lcycle[1]{\underleftarrow{#1}}

\newcommand\Itopenn[4]{|#1\vdash #2:#3|^{(#4)}}
\newcommand\Itopen[3]{|#1\vdash #2:#3|}

\newcommand\Mstrans[2]{\mathsf L(#1,#2)}
\newcommand\Mscard[1]{\##1}

\newcommand\Mspmap[2]{{#1}\ast{#2}}

\newcommand\Qforg{\mathsf Q}

\newcommand\Wcell[1]{\check{#1}}
\newcommand\Intow{\check\Into}

\newcommand\Dwords[1]{\mathsf{w}(#1)}
\newcommand\Dwordl[1]{\mathsf w_0(#1)}
\newcommand\Nwcell[3]{\Sigma_{#1}(#2:#3)}

\newcommand\Wempty{\Pempty}
\newcommand\Swcell[1]{\mathsf{ws}(#1)}

\newcommand\Ldiffd{{\mathsf D}}

\newcommand\Appsep{\,}